\documentclass[journal, onecolumn]{IEEEtran}
\def\I{\boldsymbol{I}}
\newcommand{\GN}{\mathcal N}
\newcommand{\D}{\mathcal D}
\def\y{\boldsymbol{y}}

\def\S{\boldsymbol{S}}
\def\X{\boldsymbol{X}}
\def\x{\boldsymbol{x}}

\def\Y{\boldsymbol{Y}}

\def\N{\boldsymbol{N}}
\def\I{\boldsymbol{I}}
\def\A{\boldsymbol{A}}
\def\M{\boldsymbol{M}}
\def\U{\boldsymbol{U}}
\def\V{\boldsymbol{V}}

\newcommand{\BE}{\begin{equation}}
	\newcommand{\EE}{\end{equation}}
\newcommand{\BS}{\begin{subequations}}
	\newcommand{\ES}{\end{subequations}}
\usepackage{graphicx, subfig}

\usepackage{amssymb,amsmath}
\usepackage{cite}
\usepackage{caption}
\usepackage{ifpdf}
\usepackage{enumerate}
\usepackage{amsmath,amsthm}
\theoremstyle{plain}
\newtheorem{theorem}{Theorem}
\newtheorem{lemma}{Lemma}
\newtheorem{definition}{Definition} 
\newtheorem{assumption}{Assumption}
\newtheorem{corollary}{Corollary} 

\newtheoremstyle{mystyle}
{}
{}
{\upshape}
{}
{\bfseries}
{.}
{ }
{}

\theoremstyle{mystyle}

\newtheorem{remark}{Remark}

\usepackage{algorithm}
\usepackage{algorithmic}
\usepackage{setspace}
\usepackage{array}
\usepackage{xcolor}

\usepackage{booktabs}
\usepackage{multirow}
\usepackage{makecell}
\usepackage{fixltx2e}

\usepackage{color}

\usepackage{stfloats}

\usepackage{algorithm} 

\usepackage{algorithmic} 
\usepackage{multirow} 
\usepackage{xcolor}

	\newcounter{TempEqCnt} 
\begin{document}
		
		\title{Improved Turbo Message Passing for Compressive Robust Principal Component Analysis:\\ Algorithm Design and Asymptotic Analysis}

		\author{Zhuohang~He, Junjie~Ma,  and~Xiaojun~Yuan,~\IEEEmembership{Senior~Member,~IEEE}
			\thanks{Zhuohang He and Xiaojun Yuan are with the National Key Laboratory of Wireless Communications, University of Electronic Science and Technology of China, Chengdu 611731, China (e-mail: 201911220520@std.uestc.edu.cn; xjyuan@uestc.edu.cn).}
			\thanks{Junjie Ma is with the Institute of Computational Mathematics and Scientific/Engineering Computing, Academy of Mathematics and Systems Science, Chinese Academy of Sciences, Beijing 100864, China (e-mail:majunjie@lsec.cc.ac.cn). 
				
				This paper was presented in part at 2024 IEEE International Symposium on Information Theory, Athens, Greece, Jul 2024\cite{he2024}.}	 
	}

		\maketitle

		\begin{abstract}
			Compressive Robust Principal Component Analysis (CRPCA) naturally arises in various applications as a means to recover a low-rank matrix  low-rank matrix $\boldsymbol{L}$ and a sparse matrix $\boldsymbol{S}$ from compressive measurements. In this paper, we approach the problem from a Bayesian inference perspective. We establish a probabilistic model for the problem and develop an improved turbo message passing (ITMP) algorithm based on the sum-product rule and the appropriate approximations. Additionally, we establish a state evolution framework to characterize the asymptotic behavior of the ITMP algorithm in the large-system limit. By analyzing the established state evolution, we further propose sufficient conditions for the global convergence of our algorithm. Our numerical results validate the theoretical results, demonstrating that the proposed asymptotic framework accurately characterize the dynamical behavior of the ITMP algorithm, and the phase transition curve specified by the sufficient condition agrees well with numerical simulations.					
		\end{abstract}


		\IEEEpeerreviewmaketitle

		\section{Introduction}
		\IEEEPARstart{M}{odern}  high-dimensional data usually possess low intrinsic dimensionality due to the natural dynamics or human influences.
		To identify the low-dimensional  structure of high dimensional data is a universal problem in statistical signal processing. Principal component analysis (PCA) is a widely used statistical procedure that produces a low-dimensional representation of high-dimensional data. However, such a representation    provided by PCA is sensitive to the outliers contained in   the data. To address this issue, robust
		principal component analysis (RPCA) \cite{candes2011robust} has been proposed as a promising solution to suppress sparse outliers. The RPCA problem can be formulated as to recover a low-rank matrix $\boldsymbol{L}$  from grossly corrupted measurements
		\BE
		\begin{aligned}
			\X=\boldsymbol{L}+\S+\boldsymbol{N}
		\end{aligned}
		\EE
		where $\S$ accounts for sparse outliers, $\boldsymbol{N}$ accounts for the additive noise with the elements having relatively small magnitudes, and $n_1$ (or $n_2$) is the number of rows (columns) of $\X$. 
		This model and its variants have been extensively applied in addressing various real-world problems, including but not limited to:		
			\begin{itemize}
				\item \textbf{Community Detection:} RPCA plays a pivotal role in identifying clusters or communities within networks, where nodes exhibit higher connectivity within the group than with the broader network. These methods are especially useful in detecting evolving or temporal communities efficiently \cite{8861143}.				
				\item \textbf{Network Traffic Anomalies:} In network security and performance monitoring, RPCA is applied to detect abnormal patterns or behaviors in traffic data. This anomaly detection is crucial for identifying potential security threats and ensuring stable network performance \cite{7472443,6497613}.
				\item \textbf{Acoustic Separation:} RPCA is widely used in the field of audio processing, particularly for separating overlapping sound sources. This is vital for improving sound quality in applications such as speech enhancement and audiovisual separation \cite{8574947,swann2024aeroacoustic}.
				\item \textbf{Computer Vision:} In video and image analysis, RPCA is employed to decompose visual data into a low-rank component representing the background and a sparse component representing the foreground, such as moving objects. This has numerous applications in object tracking, video surveillance, and scene analysis \cite{8425659,8450718}.	
				\item \textbf{Recommendation Systems:} RPCA has been applied in recommendation systems to predict user preferences based on sparse and incomplete data, a problem commonly known as collaborative filtering. The low-rank matrix models typical user preferences, while the sparse matrix accounts for anomalies or errors in the data. This helps improve the accuracy of recommendations in online platforms like movie or product recommendation systems \cite{hofmann2004latent}.
				\item \textbf{Face Recognition:} RPCA is also used in face recognition, where face images can be represented by a low-rank matrix due to their inherent low-dimensional nature. Imperfections such as shadows or lighting variations are modeled as sparse errors. Successfully separating these components can significantly enhance the performance of face recognition systems, even when images are occluded or distorted \cite{zhang2013toward}.
		\end{itemize}
In a mathematical form, the problem can be formulated as:
		\BE
		\begin{aligned}
			&\operatorname{minimize} &&\operatorname{rank}(\boldsymbol{L})+\lambda\|\S\|_0 \\
			&\text { subject to } &&\boldsymbol{L}+\S=\X \notag,
		\end{aligned}
		\EE 
		where $\operatorname{rank}(\cdot)$ denotes the rank  and $\|\cdot\|_0$ denotes $\l_0$.
		The RPCA problem is known as NP-hard in general \cite{mahajan2010complexity}. Approximate approaches have been proposed in the literature to solve the RPCA problem \cite{candes2011robust,chandrasekaran2011rank,chen2015fast,cherapanamjeri2017nearly,zhang2018robust}. 
		The following convex relaxation to RPCA, which is referred to as Principal Component Pursuit (PCP), has been extensively studied \cite{zhou2010stable,lin2010augmented,yuan2009sparse} 
		\BE\label{pcpprob}
		\begin{aligned}
			\operatorname{minimize} \|\boldsymbol{L}\|_*+\lambda\|\S\|_1 \quad
			\text { subject to } \boldsymbol{L}+\S=\X.
		\end{aligned}
		\EE
		Here, $\|\cdot\|_*$ denotes the nuclear norm, and $\|\cdot\|_1$ denotes the $\ell_1$ norm. The PCP problem in (\ref{pcpprob}) can be solved by various iterative methods, such as the alternating direction method \cite{yuan2009sparse}, augmented Lagrange multiplier \cite{lin2010augmented}, and iterative thresholding \cite{daubechies2004iterative}, among others. These algorithms typically involve a singular value decomposition (SVD) step in each iteration, leading to high computational complexity for large-scale problems. 
		
		Non-convex approaches have been explored to alleviate the computational burden of SVD and improve recovery performance. For example, the authors in \cite{gu2016low} propose to reparametrize the low-rank matrix as the product of a tall matrix and a wide matrix (i.e., $\boldsymbol{L}=\U\V^T$) and then perform alternating minimization on the reparametrized model. This reparametrization technique is also considered in \cite{yi2016fast}, where the projected gradient method is applied to the factorized spaces spanned by $\U$ and $\V$. In \cite{zhang2018unified}, a generic projected gradient descent algorithm for general loss functions is proposed with the factorization $\boldsymbol{L}=\U\V^T$. In \cite{chen2021bridging}, the authors improve the statistical performance guarantee compared to the methodology developed in prior work \cite{gross2011recovering}. 
		In \cite{li2015identifying}, the authors consider a matrix identification problem, where $\S$ is a matrix of outliers that is nonzero in only a fraction of its columns. The proposed algorithms proceed with a two-step adaptive sampling and inference procedure to identify the locations of the nonzero columns of $\S$.
		
		To exploit the prior information of $\boldsymbol{L}$ and $\S$, another line of research aims to solve the RPCA problem based on probabilistic models from a Bayesian perspective. For example, a sparse Bayesian learning algorithm was proposed in \cite{SDerinBabacan2012SparseBM} for RPCA, where marginal distributions are approximated based on mean-field variational Bayes \cite{MatthewJBeal2003VariationalAF}. Early work by Chen et al. \cite{chen2013variational} introduces a variational Bayesian approach for sparse component estimation and low-rank matrix recovery. In this work, the low-rank matrix is modeled as the product of two random Gaussian matrices, which may not align well with real-world scenarios. A similar variational Bayesian framework is also adopted in \cite{nakajima2013variational} for sparse additive matrix factorization. However, the assumption of Gaussian distribution for the additive dense noise may not be entirely appropriate for many practical applications. In \cite{han2017bayesian}, the authors further leverage realistic sparse distributions that go beyond simple sparsity by exploiting structural dependencies between the values and locations of the sparse coefficients in real applications, such as face modeling experiments and background subtraction.
			In \cite{jia2018bayesian}, the authors present an adaptive regularizer learning strategy that improves low-rank and sparse matrix decomposition by addressing the limitations of the traditional $\ell_1$-norm, which tends to over-penalize. By leveraging a Bayesian framework and an Alternating Direction Method of Multipliers (ADMM) optimization algorithm, the proposed method achieves superior results in image denoising and foreground-background extraction. In \cite{liu2019sparse}, the author proposes a robust sparse linear regression objective function to address the RPCA problem, equivalent to the fundamental ``rank+sparsity''
		 minimization objective, with MAP-EM applied to solve the Bayesian inference problem.	BiG-AMP (bilinear generalized approximate message passing) \cite{parker2014bilinear} recasts the RPCA problem as a probabilistic inference problem, constructs the corresponding factor graph, and derives an approximate message passing algorithm based on the sum-product rule. In this algorithm, the low-rank matrix $\boldsymbol{L}$ is treated as the product of two random i.i.d. matrices, which may be unrealistic for certain applications. Recently, a Bayesian algorithm \cite{yuan2022unitary} developed for general bi-linear models, based on the framework of variational inference and the UAMP algorithm \cite{guo2015approximate}, can be applied to the RPCA problem.

		As inspired by  the success of  compressed sensing\cite{1614066,candes2008introduction,5419072,recht2010guaranteed,li2013compressed},
		compressive robust principal component analysis (CRPCA) aims to recover low-rank matrix $\boldsymbol{L}$ and sparse matrix $\S$ from noisy compressive measurements of the mixture $\boldsymbol{L} + \S$. Mathematically, the noisy compressive measurements can be modeled as 
		\BE\label{crpcamodel}
		\begin{aligned}
			\boldsymbol{y}=\mathcal{A}(\boldsymbol{X})+\boldsymbol{n}=\mathcal{A}(\boldsymbol{L}+\boldsymbol{S})+\boldsymbol{n}
		\end{aligned}
		\EE
		where $\mathcal{A}:\mathbb{R}^{n}\to\mathbb{R}^{m}$ is a linear operator with $m\ll n_1n_2$, $\boldsymbol{L}\in\mathbb{R}^{n_1\times n_2}$ is a rank-$r$ matrix with $r \ll n_1,n_2$, $\boldsymbol{S}\in\mathbb{R}^{n_1\times n_2}$ is a sparse matrix with sparsity level $\rho$ (i.e., a fraction $\rho$ of the elements of $\S$ are nonzero),
		and $\boldsymbol{n}\in\mathbb{R}^{m\times 1}$ is an additive white Gaussian noise with mean zero and variance $\sigma_{\boldsymbol{n}}^2{\boldsymbol{I}}$. In (\ref{crpcamodel}), “compressive” means that the number of measurements (i.e., $m$) is much less than the number of unknown variables (i.e., $n_1n_2$). The goal of CRPCA is to recover $\boldsymbol{L}$ and $
		\S$ from the noisy compressive measurements in $\y$.
		
		Similarly to PCP, a convex program,  termed  Compressive Principal Component Pursuit (CPCP)\cite{2013Compressive}, was proposed to solve the compressive RPCA problem:
		\BE
		\begin{aligned}
			&\operatorname{minimize} &&\|\boldsymbol{L}\|_*+\lambda\|\S\|_1 \\
			&\text { subject to } &&\mathcal{A}(\boldsymbol{L}+\S)=\y. \notag
		\end{aligned}
		\EE
		The  optimization methods introduced in RPCA models for PCP are also applicable to CPCP with some necessary modifications \cite{aravkin2014variational,zhang2018unified,chen2021bridging}. 	A greedy algorithm \cite{waters2011sparcs}, named SpaRCS, combined the advantages of CoSaMP \cite{2009CoSaMP} for sparse signal recovery and  ADMiRA\cite{2010ADMiRA} for low-rank matrix recovery. 
		 Building upon SpaRCS, R-SpaRCS\cite{ramesh2015r} further accounted for the connectedness of the support of the sparse foreground component in videos, resulting in improved performance.  In \cite{8309163}, the authors address the CRPCA problem in an online setting. Unlike the offline scenario, previous frames provide additional prior information that aids in recovering the $t$-th frame. The proposed COCRPCA algorithm leverages this prior information for both the sparse and low-rank components. The proposed $n-\ell_1$ minimization problem with adaptive weights is solved by the proximal gradient methods. Following \cite{8309163}, a cluster-based formulation is added to the objective function of the sparse component in \cite{8416578}. With this constraint, the proposed CODA algorithm outperforms COCRPCA, as demonstrated in their experiments. Theoretical guarantees regarding the number of measurements required for convergence are provided in \cite{8450742}.
 
To alleviate the high computational complexity involved  in the above convex methods, a multilevel optimization algorithm \cite{hovhannisyan2019fast} was proposed based on the Frank-Wolfe thresholding method developed in \cite{mu2016scalable}. 

	      To facilitate comparison with the previous algorithms, we have created a table to help readers easily reference the differences.
		 	Table~\ref{tab:optimization} in Appendix. \ref{comp_alg} includes the  details and provides a clear comparison between our proposed method and existing approaches.
		
		In this paper, we tackle the CRPCA problem from the perspective of Bayesian inference. We establish a probability model for the problem, and then develop an improved turbo message passing (ITMP) algorithm based on the sum-product rule \cite{2001Factor} and appropriate approximations \cite{8502093}. The proposed ITMP algorithm consists of three denoisers, namely, a linear denoiser that handles the noisy linear observation $\y$, a sparse denoiser that handles the sparsity of $\S$, and a low-rank denoiser that handles the low-rank constraint of $\boldsymbol{L}$. The three denoisers are iteratively executed until convergence is achieved. Compared to previous work \cite{8502093,9345494} on TMP algorithms for CRPCA, we provide a new interpretation for deriving the CRPCA algorithm, along with algorithmic analysis. Our analysis demonstrates that, under proper assumptions, the parameters in the dynamic algorithm converge asymptotically to an analytical form. This convergence ensures the validity and accuracy of certain empirical approximations within the algorithm. Furthermore, based on our asymptotic analysis, we establish a guarantee of algorithmic convergence. Such rigorous analysis, along with the derived asymptotic parameter forms, has not been available in prior work.

		The main contributions of our work are summarized as follows:
		\begin{enumerate}[$\bullet$]
			\item \textbf{Contribution 1.}  \emph{(A new interpretation of TMP for the CRPCA problem)}: To facilitate the calculation of messages, we introduce two basic assumptions to model the input errors of the sparse denoiser and the low-rank denoiser as independent and identically distributed (i.i.d.) sequences in each iteration. These two assumptions allows us to readily calculate the statistical coefficients (such as means and variances) required in turbo message passing. Interestingly, our analysis shows that the operation of the ITMP algorithm in each iteration preserves the uncorrelation between the signal to be recovered and the corresponding error at the input of the sparse/low-rank denoiser. Since independence implies uncorrelation, the two assumptions are to some extent self-consistent from iteration to iteration.
			\item \textbf{Contribution 2.}  \emph{(State evolution (SE) analysis)}: We further establish the state evolution to characterize the asymptotic behaviors of the proposed ITMP algorithm in the large-system limit. Concretely speaking, under certain regularity conditions, we analyze the asymptotic behaviors of the  denoisers, and derive the corresponding asymptotic mean-square error (MSE) transfer functions. To the best of our knowledge, this is the first time to establish an analytical expression for the asymptotic MSE transfer function of the spectral denoiser (a family of popularly used low-rank denoisers) under arbitrary i.i.d. input noise with zero mean and bounded variance. 		
			\item \textbf{Contribution 3.}  \emph{(Convergence analysis)}: Based on the developed state evolution, we analyze the global convergence of the proposed ITMP algorithm. We provide a sufficient condition for the global convergence of the proposed algorithm based on the monotonicity of the MSE transfer functions. The phase transition boundary predicted by the sufficient condition agrees well with the empirical phase transition curve obtained from numerical experiments.
		\end{enumerate}
		\subsection{Prior Work Related to Contribution 1 }
		Message passing and its variants have been widely used to recover structured signals with near-optimal performance \cite{2001Factor,yedidia2005constructing }. For example, approximate message passing (AMP)\cite{donoho2009message}, initially developed for the recovery of sparse signals in underdetermined linear systems, is a popular compressed sensing algorithm due to its low complexity. The AMP algorithm consists of two recursively executed steps: One step is to produce a coarse estimate of the signal to be recovered based on the observed linear measurements; the other step is to refine the estimate based on the sparsity of the signed to be recovered.     The AMP algorithm has been extended to solve various statistical estimation problems, including sparse signal reconstruction with generalized linear models  \cite{6033942,mondelli2021approximate}, low-rank matrix recovery \cite{matsushita2013low,kabashima2016phase,romanov2018near}  and 
		signal reconstruction with more sophisticated signal structures \cite{metzler2016denoising}. It is known that the AMP algorithm has guaranteed performance only with a Gaussian or sub-Gaussian measurement matrix \cite{5695122,berthier2020state}, where the elements of the measurement matrix are independently drawn from a Gaussian or sub-Gaussian distribution. 
		
		Compressed sensing applications, however, usually employ non-i.i.d. measurement matrix based on fast orthogonal transforms, such as the discrete Fourier transform (DFT) and the discrete cosine transform (DCT).  
	 Similarly to AMP, the TMP algorithm has been extended to solve various structured signal reconstruction problems \cite{Ma2014Turbo,ma2017orthogonal,2017Denoising,xue2019tarm,8502093}.  The state evolution of TMP has been rigorously established when the measurement is a right-orthogonally invariant linear (ROIL) operator\cite{rangan2019vector}.		
		The above message passing algorithms are all established based on a critical assumption, i.e., the input errors of the non-linear denoisers are assumed to be Gaussian \cite{Ma2014Turbo,ma2017orthogonal,2017Denoising,xue2019tarm}. It has been shown that this approximation is asymptotically true in the large system limit of AMP and turbo message passing for sparse signal recovery and low-rank matrix recovery problems\cite{5695122,berthier2020state}.
		However, we find that in message passing for the CRPCA problem, the input errors of the low-rank denoiser do not behave as Gaussian distributed\footnote{The message passing algorithm for the low-rank matrix recovery problem  involves only two modules, namely, a linear denoiser for handling the linear measurements, and a low-rank denoiser for handling the low-rank property of $\boldsymbol{L}$. Due to the linear mixing effect of the linear module, the input errors of the low-rank denoiser (i.e., the output errors of the linear module) follows a Gaussian distribution in the large-system limit. However, as aforementioned, the ITMP algorithm for CRPCA involves the iteration of three modules, where the input of the low-rank denoiser is a mixture of the outputs of the linear denoiser and the sparse denoiser. The output errors of the sparse denoiser are non-Gaussian, and so are the input errors of the low-rank denoiser.}. As such, the approaches in \cite{Ma2014Turbo,ma2017orthogonal,2017Denoising,xue2019tarm}, if directly applied to the CRPCA problem, suffer from performance degradation due to the Gaussian error approximation.
		
		In this paper, we formulate the CRPCA problem as a probabilistic inference problem and derive an approximate message passing algorithm based on the sum-product rule.  The proposed ITMP algorithm  consists of three modules, where a linear denoiser is employed to linearly estimate $\X$ based on $\y$ in (\ref{crpcamodel}) and the messages from the other two modules, and two nonlinear denoisers that refine the coarse input estimates of $\boldsymbol{L}$ and $\S$ by exploiting the low-rankness of $\boldsymbol{L}$ and the sparsity of $\S$. We assume that the input errors of the low-rank denoiser 
		(and the sparse denoiser) are i.i.d., but not necessarily Gaussian distributed. 
		We systematically derive the operations of the three denoisers by following the principle of extrinsic message exchange \cite{8502093}. 
		With the new formulation, it is helpful to  analyze the state evolution.
		Moreover, based on the proposed algorithm design and starting from the i.i.d. assumptions for the current iteration input noise of non-linear denoisers, we prove that after each iteration, the input noise  of each non-linear denoiser  is uncorrelated with signal to be recovered. This uncorrelation is consistent with the two i.i.d. assumptions.  This means that the i.i.d. assumptions are to some extent self-consistent under the proposed formulation during the iterative process.
		\subsection{Prior Work Related to Contribution 2 }
		Under certain randomness assumptions on the measurement matrix, the dynamic processes of message-passing based compressed sensing algorithms \cite{berthier2020state,bayati2011dynamics} admit scalar functions to characterize the evolution of the mean-square errors (MSEs) (a.k.a. the state evolution). This work aims to establish the state evolution of ITMP for the CRPCA problem. We note that the CRPCA problem reduces to the sparse signal recovery problem by setting $\boldsymbol{L}$ to zero, and to the low-rank matrix recovery problem by setting $\S$ to zero. These two reduced problems can be solved by message passing algorithms; see, e.g.,\cite{donoho2009message,romanov2018near,Ma2014Turbo,5695122,rangan2019vector}. However, the existing state evolution analysis cannot be straightforwardly extended to the TMP algorithm for the CRPCA problem. The main obstacle is that the empirical distribution of the input errors of the low-rank denoiser (as one of the three modules of the TMP algorithm) is generally non-Gaussian. As such, the main challenge of the state evolution is to establish the MSE  transfer function of the low-rank denoiser with non-Gaussian input errors.
		
		We now  introduce the existing literature on the analysis of the low-rank denoiser. In \cite{candes2013unbiased,hansen2018stein}, the Stein’s unbiased risk estimation (SURE) can be employed to estimate the output MSE of the low-rank denoiser. These analyses cannot be used in our work since they are limited to case of Gaussian input noise. Another approach is to formulate low-rank denoising as a minimax risk problem that optimizes the denoiser function to obtain the minimal risk guarantee under a fixed input noise power. In the minimax risk formulation, the optimal hard threshold and the optimal shrinkage of singular values for fixed rank $r$ are obtained in the large-system limit \cite{gavish2017optimal}. In \cite{donoho2014minimax}, the authors analyzed the optimal singular-value soft threshold in the large-system limit with a proportionally increasing rank. Yet, the asymptotic MSEs obtained in \cite{gavish2017optimal}, \cite{donoho2014minimax} are obtained by fixing the noise level and letting the signal power go to infinity, which only represents a point in our desired MSE transfer function.
		
		In this paper, we derive an asymptotic MSE transfer function of the low-rank denoiser with i.i.d. non-Gaussian input noise under certain regularity conditions.  Our results are applicable to commonly used denoisers, including the best-rank-r denoiser, the singular-value soft-threshold denoiser, and the singular-value hard-threshold denoiser.
		
		\subsection{Prior Work Related to Contribution 3 }  
		A phase transition is a rapid and significant transformation in the nature of a computational problem as its parameters undergo variations. This phenomenon emerges in a variety of compressed sensing problems; for example, see \cite{donoho2009message,donoho2013accurate,2013Compressive,amelunxen2014living,romanov2018near}. For the CRPCA problem, we consider a parameter space with axes quantifying sparsity level $k/n=\rho$, rank fraction $ r/n_1=\gamma$ and  undersampling ratio $m/n=\alpha$. In the limit of large dimensions $n$, the parameter space splits in two phases: one presenting that the proposed algorithm is successful in accurately reconstructing $\S$ and $\boldsymbol{L}$, and the other representing the unsuccessful case. For the message passing based algorithms, state evolution, which tracks the MSE between the real signals and the estimations, can be used to analyze the phase transition of the algorithms. When the MSE approaches zero, it means that the signal has been accurately recovered. Phase transitions have been extensively studied in the sparse signal recovery problem \cite{donoho2009message,donoho2013accurate} and low-rank matrix recovery problem \cite{romanov2018near}, and tight bounds have been obtained by analysing the convergence of SE for the AMP based algorithms. Ref. \cite{amelunxen2014living} further provides an explanation of phase transitions by using high-dimensional geometric tools.
		
		For the CRPCA problem, a preliminary analysis of the phase transition bound has been conducted in \cite{2013Compressive} under the framework of finite dimensional settings. This result, however, becomes trivial in the large system limit, since the global convergence bound of $\gamma$ goes to zero as $n\to\infty$. To the best of our knowledge, the phase transition of the CRPCA problem in the large system limit remains a challenging open problem.
		
		In this paper, based on the developed SE, we establish a sufficient condition for the global convergence of ITMP in the large system limit. The tightness of the sufficient condition is supported by the observation that the phase transition curve specified by the sufficient condition closely matches that of ITMP in various numerical simulations. We further establish a necessary condition for the global convergence, corresponding to a lower bound of the phase transition curve. These upper and lower bounds contribute to a more comprehensive understanding of the performance of ITMP in the considered asymptotic regime.
		
		\subsection{Organizations and Notations}
		
		The rest of the paper is organized as follows.
		In Section \ref{Sec3}, we present the CRPCA problem. We derive the overall message passing algorithm  based on sum-product rule and approximate the message as a denoising problem. In Section \ref{Sec4}, we characterize the asymptotic behavior for the three generic denoisers and develop the state evolution of the algorithm. In Section \ref{Sec5}, we present the general convergence analysis of the algorithm by tracking the  dynamic system given by SE. In Section \ref{Sec6}, we present the SE and the global convergence for certain settings of the algorithm and provide numerical experiments. 
		
		Throughout this paper, we use bold lower-case and upper case letters to denote vectors and matrices respectively. The nomenclature  illustrates the notations used in this paper.
		

	\begin{table*}[htbp]
		\caption*{NOMENCLATURE}
		\centering
		\renewcommand{\arraystretch}{1.2} 
		\setlength{\tabcolsep}{5pt} 
		\begin{tabular}{p{60pt}p{180pt}p{60pt}p{180pt}}
			\hline
			$\mathbb{R}$ & Set of real numbers &  $\mathcal{A}$ & A linear operator $\mathcal{A}:\mathbb{R}^n\to\mathbb{R}^m$ defined by $\mathcal{A}(\X)=\A\operatorname{vec}(\X)$ \\    
			$\mathbb{N}_{>0}$ & Set of positive integers &
			$\mathcal{A}^T$ & Adjoint linear operator of $\mathcal{A}$\\
			$\mathbb{R}^{N}$ & $N$-dimensional Euclidean space & $v_{a\rightarrow b}^{(t)}$ & Variance of message $m_{a\rightarrow b}^{(t)}(x)$\\
			$\X,\S,\boldsymbol{L}$ & Matrices $\X,\S,\boldsymbol{L}$ & $\langle\boldsymbol{X},\boldsymbol{Y}\rangle$ & Inner product of $\X$ and $\Y$: $\sum_{i,j} X_{i,j}Y_{i,j}$\\    
			$n_1$ & Number of rows of $\X,\S,\boldsymbol{L}$ & $\X\circ\Y$ & Hadamard product of matrices $\X$ and $\Y$ \\ 
			$n_2$ & Number of columns of $\X,\S,\boldsymbol{L}$ & $m_{a\rightarrow b}^{(t)}(x)$ & Message from node $a$ to node $b$ with respect to $x$ at the $t$-th iteration \\ 
			$n$ & Number of entries of $\X$ (e.g., $n=n_1n_2$) & $x_{a\rightarrow b}^{(t)}$ & Mean of message $m_{a\rightarrow b}^{(t)}(x)$ \\ 
			$X_{i,j}$ & $(i,j)$-th element of matrix $\X$ & $v_{a\rightarrow b}^{(t)}$ & Variance of message $m_{a\rightarrow b}^{(t)}(x)$ \\ 
			$\A$ & Measurement matrix $\A\in\mathbb{R}^{m\times n}$ & $m_{x}^{(t)}(x)$ & Message of variable node $x$ at the $t$-th iteration \\ 
			$\boldsymbol{x}_{i}$ & $i$-th column of matrix $\X$ & $\delta(\cdot)$ & Dirac delta function \\ 
			$\alpha$ & Undersampling ratio (e.g., $\alpha=\frac{m}{n_1n_2}$) & $(x)_{+}$ & $x$ if $x>0$, and zero otherwise \\ 
			$\X_{/(i,j)}$ & Set of the elements of $\X$ except the $(i,j)$-th element & $\mathbb{I}$ & Indicator function \\ 
			$\beta$ & Aspect ratio (e.g., $\beta=\frac{n_1}{n_2}$) & $\mathbb{E}\left[\cdot \right]$ & Expectation operation over all variables involved \\ 
			$r$ & Rank of matrix $\boldsymbol{L}$ & $\mathbb{E}_{x}\left[\cdot \right]$ & Expectation operation over $x$ \\ 
			$\gamma$ & Rank fraction of $\boldsymbol{L}$ (e.g., $\frac{r}{n_1}$) & $\mathbb{E}\left[x|y\right]$ & Expectation of $x$ conditioned on $y$ \\ 
			$k$ & Non-zero numbers in matrix $\S$ & $x\overset{p}{=}y$ & $x$ converges to $y$ in probability \\ 
			$\rho$ & Sparsity level of $\S$ (e.g., $\frac{k}{n_1n_2}$) & $x\overset{w}{=}y$ & $x$ converges weakly to $y$ \\ 
			$\boldsymbol{X}^{T}$ & Transpose of matrix $\X$ & $\mathcal{D}_{\X}$ & Denoiser of matrix $\X$ \\ 
			$\operatorname{vec}(\X)$ & $\left[\boldsymbol{x}_{1}^{T},\boldsymbol{x}_{2}^{T}, \cdots, \boldsymbol{x}_{n_{2}}^{T}\right]^{T}$ & $\mathcal{D}_{\boldsymbol{L}}$ & Denoiser of matrix $\boldsymbol{L}$ \\ 
			$\|\X\|_{F}$ & Frobenius norm of matrix $\X$ & $\mathcal{D}_{\S}$ & Denoiser of matrix $\S$ \\ 
		\end{tabular}
	\end{table*}

		\section{Improved Turbo Message Passing for CRPCA}\label{Sec3}
		\subsection{System Model}
		In this paper, we focus on special families of linear operators defined below.
		It is well known that the generation model of the linear operator $\mathcal{A}$ affects the reconstruction of the structured signals \cite{bayati2011dynamics,8502093,ma2017orthogonal,berthier2020state,6033942}. The linear operator $\mathcal{A}$ is commonly assumed to sample the signal ``uniformly'' enough. For example, in \cite{2013Compressive}, the rows of $\A$ (i.e., the matrix form of $\mathcal{A}$) are assumed to follow the Haar measure, and in \cite{waters2011sparcs}, $\A$ is considered to satisfy the restricted isometry property (RIP) and rank-restricted isometry property (RRIP). 
		In this paper, we focus on special families of linear operators defined below.
		\begin{definition}[Rotationally invariant matrix]
			A random matrix $\M$ is rotationally invariant if the distribution of $\M$ is equal to that of $\M\U$ and $\U\M$ for any orthogonal matrix $\U$, independent of $\M$.	
		\end{definition}
		\begin{definition}[ROIL]\label{defROIL}
			A linear operator $\mathcal{A}$ is a rotationally invariant linear (ROIL) operator if the matrix form $\A$ of $\mathcal{A}$ is rotationally invariant, where $\A$ is given by  $\mathcal{A}(\X)=\A\operatorname{vec}(\X)$.
		\end{definition}
	\begin{definition}[Haar (distributed) matrix]
				The set of $n \times n$ orthogonal matrices over $\mathbb{R}^{n \times n}$ is denoted by $\mathcal{O}{n}$. We say an $n \times n$ random matrix $\mathbf{V} \in \mathcal{O}{n}$ is Haar distributed on $\mathcal{O}{n}$ if for any fixed $\mathbf{U} \in \mathcal{O}{n}$,
				\BE
				\begin{aligned}
					\V\overset{p}{=}\U\V\overset{p}{=}\V\U.
				\end{aligned} 
				\EE
				We also refer to such a matrix $\mathbf{V}$ as the Haar matrix in this paper.
		\end{definition}

		Consider the singular value decomposition ($\operatorname{SVD}$) of  $\M$ 
		\BE
		\begin{aligned}
			\M=\U_{\M}\Sigma_{\M}\V_{\M}^{T}
		\end{aligned}
		\EE
		where $\Sigma_{\M}\in\mathbb{R}^{m\times n}$ is a diagonal matrix, $\U_{\M}\in\mathcal{O}{m}$ and $\V_{\M}\in\mathcal{O}{n}$. $\M$ is a rotationally invariant matrix, if and only if $\V_{\M}$ and $\U_{\M}$  are  Haar matrices independent of each other and  independent of  $\Sigma_{\M}$.

		We say that a linear operator $\mathcal{A}$ is a partial  Haar operator if the matrix form $\A$ is given by $	\A=\U_{s}\V$, where  $\U_{s}$ consists of uniformly and randomly selected rows of the identity  matrix and $\V$ is a Haar matrix. (When constructing the matrix $\U_{s}$, rows are  chosen without repetition. The partial Haar operator requires selecting rows from the identity matrix while ensuring that each row vector are orthogonal.) We say that a linear operator $\mathcal{A}$ is a Gaussian operator if the entries of the matrix form $\A$ are $i.i.d.$ Gaussian with zero mean and $\frac{1}{n}$ variance. Clearly, both  partial Haar operators and Gaussian operators  belong to the family of ROIL operators. In the following, we always assume that $\mathcal{A}$ is a ROIL operator, unless otherwise specified. 
		
		\subsection{Improved Turbo Message Passing Algorithm}
		We take a Bayesian approach to tackle the CRPCA problem. We first introduce a probability model for the variables in (\ref{crpcamodel}). Assume that 
		$\boldsymbol{L}$, $\S$, and $\boldsymbol{n}$ are  independent of each other. Denote by $p(\boldsymbol{L})$ and $p(\S)$ the prior distributions of $\boldsymbol{L}$ and $\S$, respectively. The joint distribution of $\y$, $\X$, $\boldsymbol{L}$ and $\S$ can be expressed as
		\BE\label{promodel}
		\begin{aligned}
			p(\boldsymbol{y}, \boldsymbol{X}, \boldsymbol{S}, \boldsymbol{L})=p(\y|\X)p(\X|\boldsymbol{L},\S) p(\boldsymbol{S}) p(\boldsymbol{L})
		\end{aligned}
		\EE 
		where $p(\boldsymbol{y}|\boldsymbol{X})=\mathcal{N}\left(\boldsymbol{y}-\mathcal{A}(\boldsymbol{X}); \boldsymbol{0},\sigma_{\boldsymbol{n}}^2\I\right)$, and $p(\boldsymbol{X}|\boldsymbol{L},\boldsymbol{S})=\delta(\boldsymbol{X}-\boldsymbol{S}-\boldsymbol{L})$ with $\delta(\cdot)$ being the Dirac delta function. Here,  $\boldsymbol{L}\in\mathbb{R}^{n_1\times n_2}$ is a rank-$r$ matrix with $r \ll n_1,n_2$, $\boldsymbol{S}\in\mathbb{R}^{n_1\times n_2}$ is a sparse matrix with sparsity level $\rho$ (i.e., a fraction $\rho$ of the elements of $\S$ are nonzero),
			and $\boldsymbol{n}\in\mathbb{R}^{m\times 1}$ is an additive white Gaussian noise with mean zero and variance $\sigma_{\boldsymbol{n}}^2{\boldsymbol{I}}$.
		
		Computing the maximum a posteriori probability (MAP) estimate or minimum mean square error (MMSE)  estimate based on (\ref{promodel}) is usually intractable for high dimensional problems. In this paper, we obtain a message passing algorithm following the sum-product rule together with necessary approximations.
		
		We start with a factor graph representation of the considered probability model. Recall that $p(\y|\X)$ is a Gaussian distribution with mean $\mathcal{A} (\boldsymbol{X})$ and covariance $\sigma_{\boldsymbol{n}}^2\I$, and 
		\BE
		\begin{aligned}
			p(\X|\S,\boldsymbol{L})=\displaystyle\prod_{i, j}\delta(X_{i,j}-S_{i,j}-L_{i,j}),
		\end{aligned}
		\EE
		where $X_{i,j}$, $S_{i,j}$ and $L_{i,j}$ are respectively the $(i,j)$-th elements of $\X$, $\S$, and $\boldsymbol{L}$. Then, (2) can be rewritten as 
		\BE\label{e4}
		\begin{aligned}
			p(\boldsymbol{y}, \boldsymbol{X}, \boldsymbol{S}, \boldsymbol{L}) =\ &\GN\left(\y;\mathcal{A} (\boldsymbol{X}), \sigma_{\boldsymbol{n}}^2\I\right)p(\boldsymbol{S}) p(\boldsymbol{L})
\displaystyle\prod_{i, j}\delta(X_{i,j}-S_{i,j}-L_{i,j}).
		\end{aligned}
		\EE
		\begin{figure}[htbp]\label{factorgr1}
			\begin{minipage}{\textwidth}
				\centering\includegraphics[width=0.5\textwidth]{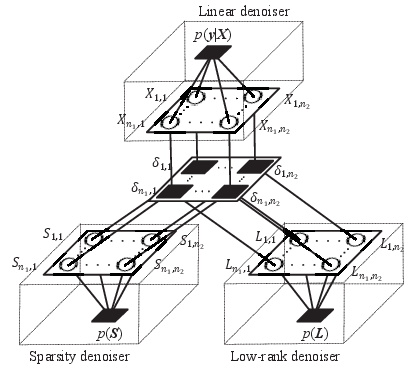}
				\caption{The factor graph representation of (\ref{e4})}
				\label{tcrpca}
			\end{minipage}
		\end{figure}
		
		The factor graph representation of (\ref{e4}) is shown in Fig. \ref{tcrpca}, where each white circle represents a variable node, and each black square represents a factor node. For notational convenience, we abbreviate $\delta(X_{i,j}-L_{i,j}-S_{i,j})$ to $\delta_{i,j}$ and $p(\y|\X)$ to $\y$ in situations without causing ambiguity. In addition, we denote by $\delta$ the collection of all ${\delta_{i,j}}$. Denote by $m_{a\to b}^{(t)}(x)$ the message from node $a$ to node $b$ with respect to variable $x$ in the $t$-th iteration, where $x_{a\to b}^{(t)}$ and $v_{a\to b}^{(t)}$ are the corresponding mean and variance respectively, and by $m_{a}^{(t)}(x)$ the message of node $a$ with respect to variable $x$ in the  $t$-th iteration. We now describe the proposed algorithm  based on the sum-product rule.

		We observe the symmetry of the nodes $p(\S)$ and $p(\boldsymbol{L})$ in the factor graph. Below, we provide derivation of the messages  associated with $\S$. By replacing $\S$ in the derived messages related to $\S$, we obtain the corresponding messages associated with $p(\boldsymbol{L})$ node.
		\subsubsection{Messages associated with  $\S$ and $\boldsymbol{L}$}
		The message from variable node $S_{i,j}$ to factor node $p(\S)$ and the message from variable node $S_{i,j}$ to factor node $\delta_{i,j}$ are  respectively given by
		\BS \label{message_S_1}
		\begin{align} 
			m^{(t)}_{S_{i,j}\to p(\S)}&=m^{(t)}_{\delta_{i,j}\to S_{i,j}}\\
			m^{(t)}_{S_{i,j}\to\delta_{i,j}}&=m^{(t)}_{p(\S)\to S_{i,j}}.
		\end{align} 
		\ES
		The messages $m^{(t)}_{\S\to p(\S)}$ and $m^{(t)}_{p(\S)\to \S}$ are respectively defined as
		\BS\label{Sdelta}
		\begin{align}
			m^{(t)}_{\S\to p(\S)} &= \displaystyle\prod_{i,j}m^{(t)}_{S_{i,j}\to p(\S)}\\
			m^{(t)}_{p(\S)\to \S} &= \displaystyle\prod_{i,j}m^{(t)}_{p(\S)\to S_{i,j}},
		\end{align}  
		\ES
		where
		\BE\label{e9}
		\begin{aligned} 
			m^{(t)}_{p(\S)\to S_{i,j}}=\frac{\displaystyle\int m^{(t)}_{\S\to p(\S)} p(\S)\,d\S_{/(i,j)}}{ m_{S_{i,j} \to p(\S)}^{(t)}}
 		\end{aligned} 
		\EE
		with 
		\BE\label{message_S_2}
		\begin{aligned} 
			m^{(t)}_{p(\S)}=\displaystyle\int m^{(t)}_{\S\to p(\S)} p(\S)\,d\S_{/(i,j)}.
		\end{aligned} 
		\EE
		Clearly, $m^{(t)}_{p(\S)\to \S}$  cannot be evaluated when $p(\S)$ is unavailable (which is usually true in practice). Even if $p(\S)$ is available, the evaluation of $m_{p(\S)\to\S}^{(t)}$ is still difficult since it involves high-dimensional integrals. Considering the  symmetry of the nodes $p(\S)$ and $p(\boldsymbol{L})$ in the factor graph, we can obtain the messages associated with $\boldsymbol{L}$ by just replacing $\S$ with $\boldsymbol{L}$ in (\ref{message_S_1})-(\ref{message_S_2}). 
			Again, $m_{p(\boldsymbol{L})\to\boldsymbol{L}}^{(t)}$ is difficult to evaluate since it involves $p(\boldsymbol{L})$ (which is unavailable in practice) and the high-dimensional integral.  
		
		To avoid the above difficulties, we use the mean and variance for message passing, following the idea of approximate message passing \cite{donoho2009message}, expectation propagation \cite{minka2013expectation} and mean field \cite{barabasi1999mean}. Inspires by the construction of the \textbf{extrinsic message}, which have been developed in \cite{Ma2014Turbo,ma2017orthogonal,2017Denoising,xue2019tarm,8502093}, we approximately evaluate the mean and variance of the message.

			The idea of exchanging extrinsic messages between components was  first introduced in decoding turbo-codes \cite{539767}. 
			The  output message by one component becomes the prior message for the other, and this process iterates multiple times. The output of each component decoder should be the so-called 'extrinsic message', which is obtained by excluding the prior message of a component decoder from the corresponding posterior message. 
		
			Inspired by the exchange of extrinsic messages in turbo decoding, the idea of exchanging extrinsic messages has been applied to message-passing algorithms \cite{Ma2014Turbo, ma2015performance} that solve sparse linear inverse problems.			
	For example,	denoising-based Turbo-CS \cite{2017Denoising} was proposed to replace the MMSE denoiser in \cite{Ma2014Turbo} with a generic denoiser, even without prior knowledge of the signal distribution. Inspired by Gaussian message passing, the authors in \cite{Ma2014Turbo} extend the concept of independence to uncorrelatedness to accommodate a broader range of posterior denoisers. Additionally, the mean of the extrinsic message is constructed through a linear combination of the  \textit{a priori} mean and the \textit{a posteriori} mean.
			This formulation has also been applied to construct message passing algorithms in various signal reconstruction problems \cite{xue2019tarm,8502093}. 
		
			 Let $\S_{\S\to p(\S)}^{(t)}$ and $v_{\S\to p(\S)}^{(t)}$ ( $\boldsymbol{L}_{\boldsymbol{L}\to p(\boldsymbol{L})}^{(t)}$ and $v_{\S\to p(\S)}^{(t)}$ ) be the mean and variance of message $m_{\S\to p(\S)}^{(t)}$ (and $m_{\boldsymbol{L}\to p(\boldsymbol{L})}^{(t)}$), respectively. We next describe how to approximately evaluate the mean and variance of $m_{p(\S)\to S_{i,j}}^{(t)}$ (denoted respectively by $\S^{(t)}_{p(\S)\to S_{i,j}}$ and $v^{(t)}_{p(\S)\to S_{i,j}}$). 
		To start with, we make the following assumptions on $\S_{\S\to p(\S)}^{(t)}$ and $\boldsymbol{L}_{\boldsymbol{L}\to p(\boldsymbol{L})}^{(t)}$ .
		
		\begin{assumption}\label{as1}
			For each iteration $t$, the elements of the estimation error matrix  $\N_{\S}^{(t)}=\S^{(t)}_{\S\to p(\S)}-\S$ are independently drawn from a common distribution  $\mathcal{Q}_{N_{\S}}^{(t)}(N_{\S})$ with mean zero and variance $v^{(t)}_{\S\to p(\S)}$.
		\end{assumption}
		\begin{assumption}\label{as2}
		For each iteration $t$, the elements of the estimation error matrix  $\N_{\boldsymbol{L}}^{(t)}=\boldsymbol{L}^{(t)}_{\boldsymbol{L}\to p(\boldsymbol{L})}-\boldsymbol{L}$ are independently drawn from a common distribution  $\mathcal{Q}^{(t)}_{N_{\boldsymbol{L}}}(N_{\boldsymbol{L}})$ with zero mean and variance $v^{(t)}_{\boldsymbol{L}\to p(\boldsymbol{L})}$.
	\end{assumption}
		The above assumption is supported by some empirical evidences presented in the next section. Similar assumptions have been previously used in the development of efficient iterative algorithms; see, e.g.,\cite{Ma2014Turbo} and \cite{xue2019tarm}. 
		 This assumption is useful in  approximate calculation of $\S^{(t)}_{p(\S)\to\S}$ and  $\boldsymbol{L}^{(t)}_{p(\boldsymbol{L})\to\boldsymbol{L}}$. 
		In the following, we will only describe how to approximate the messages related to $\S$; a similar method can be applied to $\boldsymbol{L}$.  Without necessarily knowing the exact expression of $p(\S)$, we generally approximate the mean of $m_{\S}^{(t)}$ as
		\BE 
		\begin{aligned} \label{denoiser_formS}    
			\S^{(t)}=\D_{\S}\left(\S^{(t)}_{\S\to p(\S)};v_{\S\to p(\S)}^{(t)}\right)
		\end{aligned} 
		\EE 
		where $\D_{\S}$ is a denoising function that suppresses the noise in the input by exploiting certain prior knowledge of $\S$ (such as the sparsity of $\S$, but not necessarily the complete knowledge of   $p(\S)$).

		Then,  the mean of $m^{(t)}_{p(\S)\to\S}$ is constructed by 
		\BE\label{formS}
		\begin{aligned}
			\S^{(t)}_{p(\S)\to\S}= c_{\S}^{(t)}\left(\S^{(t)}-a_{\S}^{(t)}\S^{(t)}_{\S\to p(\S)} \right)
		\end{aligned}
		\EE
		where $a_{\S}^{(t)}$ and $c_{\S}^{(t)}$  are scalar parameters. Note that $\S_{p(\S)\to\S}^{(t)}$ in (\ref{formS}) is also referred to as  ``extrinsic mean'' since the linear combination of $\S^{(t)}$ and $\S_{\S\to p(\S)}^{(t)}$ in (\ref{formS}) can be interpreted as excluding the contribution of the ``prior mean'' $\S_{\S\to p(\S)}^{(t)}$ from the ``posterior mean'' $\S^{(t)}$ \cite{xue2019tarm, 8502093,2017Denoising}. The parameters $a_{\S}^{(t)}$ and $c_{\S}^{(t)}$  are determined by solving the following optimization problem \cite{8502093}:
		\BS \label{turboprinciple}
		\begin{align} 
			&\min_{a^{(t)}_{\S},c^{(t)}_{\S}}\mathbb{E}\lVert  \S^{(t)}_{p(\S)\to\S}-\S \rVert^2_{F} \\
			&\begin{array}{r@{\quad}r@{}l@{\quad}l}
				{\rm s.t.}&\mathbb{E}\left \langle \S^{(t)}_{p(\S)\to\S}-\S,\S^{(t)}_{\S\to p(\S)}-\S\right \rangle=0 
			\end{array}.
		\end{align} 
		\ES
		where the expectations are taken over the joint distribution of $\S$ and $\S^{(t)}_{\S\to p(\S)}$, i.e., $p(\S)\prod_{i,j}\mathcal{Q}_{N_{\S}}^{(t)}\left( (\S^{(t)}_{\S\to p(\S)})_{i,j}-S_{i,j}\right)$ from Assumption \ref{as1}. In the above, (\ref{turboprinciple}b) ensures that the output estimation error is statistically orthogonal to the input estimation error, and (\ref{turboprinciple}a) ensures that the output mean square error (MSE) is minimized. By solving (\ref{turboprinciple}), we obtain $a_{\S}^{(t)}$ and $c_{\S}^{(t)}$ as  
		\BS\label{Sac}
		\begin{align} 
			a^{(t)}_{\S} &=\frac{\mathbb{E}\left \langle \S^{(t)}_{\S\to p(\S)} -\S, \S^{(t)}\right \rangle}{ \mathbb{E}\left\langle \S^{(t)}_{\S\to p(\S)} -\S, \S^{(t)}_{\S\to p(\S)}\right \rangle}\\ c^{(t)}_{\S}&=\frac{\mathbb{E}\left\langle\left(\S^{(t)}-a_{\S}^{(t)} \S^{(t)}_{\S\to p(\S)}\right) ,\S_{\S\to p(\S)}^{(t)}\right\rangle}{\mathbb{E}\left\lVert \S^{(t)}-a_{\S}^{(t)} \S^{(t)}_{\S\to p(\S)}\right\rVert^2_{F}} .
		\end{align} 
		\ES 
		Correspondingly, the variance of message $m^{(t)}_{p(\S)\to\S}$ is given by
		\BE \label{Svar}
		\begin{aligned} v_{p(\S)\to\S}^{(t)} =\frac{1}{n}\mathbb{E}\left\| \S_{p(\S)\to\S}^{(t)}-\S \right\|_{F}^2
			=\frac{1}{n}\mathbb{E}\left\|\S \right\| _{F}^2-c_{\S}^{(t)}\frac{1}{n}\mathbb{E} \left\langle \S^{(t)},\S^{(t)}_{\S\to p(\S)} \right\rangle
		 +c_{\S}^{(t)} a_{\S}^{(t)}\left(\frac{1}{n}\mathbb{E}\left\|\S \right\| _{F}^2+v_{\S\to p(\S)}^{(t)} \right).
		\end{aligned}
		\EE
		The detailed derivation of (\ref{Svar}) can be found in Appendix \ref{proof_eq16}.
		When $p(\S)$ and $\mathcal{Q}_{N_{\S}}^{(t)}(N_{\S}) $ are unavailable, the above expectations are computationally intractable. In Section \ref{Sec:sparsity_denoiser}, we use statistical signal processing techniques to approximately calculate  (\ref{Sac}) and (\ref{Svar}) without requiring the full knowledge of $p(\S)$ and $\mathcal{Q}_{N_{\S}}^{(t)}(N_{\S})$. 	 Note that $v_{p(\boldsymbol{L})\to\boldsymbol{L}}^{(t)}$ can be obtained from (\ref{Svar}) by replacing $\S$ by $\boldsymbol{L}$. But for
	 the convenience of practical calculation, we express the variance $v^{(t)}_{p(\boldsymbol{L})\to\boldsymbol{L}}$ equivalently as 
		\BE\label{varL} 
		\begin{aligned}
			v_{p(\boldsymbol{L})\to\boldsymbol{L}}^{(t)}
			=\frac{1}{n}\mathbb{E}\left\| \boldsymbol{L}_{p(\boldsymbol{L})\to\boldsymbol{L}}^{(t)}-\boldsymbol{L} \right\|_{F}^2
			= \frac{\mathbb{E}\!\left[ \left\|\!\boldsymbol{L}^{(t)}_{\boldsymbol{L}\to p(\boldsymbol{L})}\!\right\|_{F}^{2}\!\left\|\!\boldsymbol{L}^{(t)}\!\right\|_{F}^{2}\!\!-\!\left\langle\! \boldsymbol{L}^{(t)}_{\boldsymbol{L}\to p(\boldsymbol{L})},\boldsymbol{L}^{(t)}\!\right\rangle^2 \!\right] }{n\mathbb{E}\left[\left\| \boldsymbol{L}^{(t)}- a_{\boldsymbol{L}}^{(t)} \boldsymbol{L}^{(t)}_{\boldsymbol{L}\to p(\boldsymbol{L})} \right\|_{F}^2\right] } -\! v_{\boldsymbol{L}\to p(\boldsymbol{L})}^{(t)}.
		\end{aligned}
		\EE
		The detailed derivation of (\ref{varL}) is given in Appendix \ref{proof_var_exp}. Similarly to (\ref{Sac}) and (\ref{Svar}), the expectations in  (\ref{varL}) depend on $p (\boldsymbol{L})$ and $\mathcal{Q}^{(t)}_{N_{\boldsymbol{L}}}(N_{\boldsymbol{L}})$  which are usually unknown   in practice. We describe how to approximately evaluate $a_{\boldsymbol{L}}^{(t)}$, $c_{\boldsymbol{L}}^{(t)}$ and $v_{p(\boldsymbol{L})\to\boldsymbol{L}}^{(t)}$ without requiring exact knowledge of $p(\boldsymbol{L})$ and $\mathcal{Q}^{(t)}_{N_{\boldsymbol{L}}}(N_{\boldsymbol{L}})$  in Section \ref{Sec:low_rank_denoiser}. 
		\subsubsection{Messages associated with check node  $\delta$}
		Based on the sum-product rule, the message from  $\delta$ to $\X$ is given by 
			\BE\label{message_delta}
			\begin{aligned}
				m^{(t)}_{\delta\to \X}=\int\displaystyle\prod_{i, j}\delta(X_{i,j}-S_{i,j}-L_{i,j})m^{(t)}_{\boldsymbol{L}\to \delta}m^{(t)}_{\S\to \delta}\,d\boldsymbol{L}\,d\S.
			\end{aligned} 
			\EE
			Then, the mean and variance of message $m^{(t)}_{\delta\to \X}$  are respectively given by
			\BE\label{mx}
			\begin{aligned}
				\X^{(t)}_{\delta\to \X}&=\boldsymbol{L}^{(t)}_{\boldsymbol{L}\to \delta}+\S^{(t)}_{\S\to \delta},\quad
				v^{(t)}_{\delta\to \X}&=v^{(t)}_{\boldsymbol{L}\to \delta}+v^{(t)}_{\S\to \delta}.
			\end{aligned} 
			\EE 			
			Based on the sum-product rule, similar to (\ref{message_delta}) and (\ref{mx}), the mean and variance of message $m^{(t+1)}_{\delta\to \boldsymbol{L}}$ are respectively  given by
			\BE\label{ml}
			\begin{aligned}
				\boldsymbol{L}^{(t+1)}_{\delta\to \boldsymbol{L}}&=\X^{(t+1)}_{\X\to \delta}-\S^{(t)}_{\S\to \delta}, \quad v^{(t+1)}_{\delta\to \boldsymbol{L}}&=v^{(t+1)}_{\X\to \delta}+v^{(t)}_{\S\to \delta}.
			\end{aligned} 
			\EE 
			and
			the mean and variance of message $m^{(t+1)}_{\S\to \delta}$ are respectively given by
			\BE\label{ms}
			\begin{aligned}
				\S^{(t+1)}_{\delta\to \S}&=\X^{(t+1)}_{\X\to \delta}-\boldsymbol{L}^{(t)}_{\boldsymbol{L}\to \delta},\quad
				v^{(t+1)}_{\delta\to \S}&=v^{(t+1)}_{\X\to \delta}+v^{(t)}_{\boldsymbol{L}\to \delta}.
			\end{aligned} 
			\EE 	
		\subsubsection{Messages associated with  $\X$}
		The message from variable node $X_{i,j}$ to factor node $p(\y|\X)$  and the message from factor node $p(\y|\X)$ to variable node $X_{i,j}$  are  respectively given by
		\BE\label{x0}
		\begin{aligned} 
			m^{(t)}_{\delta\to X_{i,j}}&=m^{(t)}_{X_{i,j}\to \y},\quad
			m^{(t+1)}_{\y\to X_{i,j}}&=m^{(t+1)}_{X_{i,j}\to\delta}.
		\end{aligned} 
		\EE
		Recall $p(\y|\X)=\GN\left(\y;\mathcal{A} (\boldsymbol{X}),\sigma_{\boldsymbol{n}}^2\I\right)$. Then,  message $m^{(t+1)}_{\y\to X_{i,j}}$ is given by
		\BE\label{xx}
		\begin{aligned}
			m^{(t+1)}_{\y\to X_{i,j}}=\frac{\displaystyle\int \displaystyle\prod_{k,l} m_{X_{k,l} \rightarrow \y}^{(t)} \GN\left(\y;\mathcal{A} (\boldsymbol{X}),\sigma_{\boldsymbol{n}}^2\I\right)\,d\X_{/(i,j)}}{ m_{X_{i,j}^{(t)}\to \y }}
		\end{aligned}
		\EE
		with 
		\BE
		\label{x1}
		\begin{aligned}
			m^{(t)}_{X_{i,j}}=\int \displaystyle\prod_{k,l} m_{X_{k,l} \rightarrow \y}^{(t)} \GN\left(\y;\mathcal{A} (\boldsymbol{X}),\sigma_{\boldsymbol{n}}^2\I\right)\, d\X_{/(i,j)}.
		\end{aligned}
		\EE
		For notational convenience, define the compound messages between factor node $p(\y|\X)$ and $\X$  as
		\BE
		\label{x2}
		\begin{aligned}
			m^{(t)}_{\X\to \y}&=\displaystyle\prod_{i,j}m^{(t)}_{X_{i,j}\to\y},\quad
			m^{(t+1)}_{\y\to \X}&=\displaystyle\prod_{i,j}m^{(t+1)}_{\y\to X_{i,j}}.
		\end{aligned}
		\EE
		The compound message $m_{\X}^{(t)}$ is defined as
		\BE
		\label{x3}
		\begin{aligned}
			m_{\X}^{(t)}=\displaystyle\prod_{i,j}m^{(t)}_{X_{i,j}}.
		\end{aligned}
		\EE
		From (\ref{x1}), (\ref{x2}) and (\ref{x3}), we obtain 
		\BE
		\begin{aligned}
			m^{(t+1)}_{\y\to \X}=\frac{m^{(t)}_{\X}}{m^{(t)}_{\X\to \y}}.
		\end{aligned}
		\EE
		The high-dimensional integrals in (\ref{xx}) can be solved explicitly by assuming $m_{\delta\to\X}^{(t)}=\GN\left( vec(\X); vec(\X^{(t)}_{\delta\to\X}), v_{\delta\to\X}^{(t)}\I \right)$. This is referred to as the LMMSE approach in \cite{8502093} and the resultant message-update expressions can be  found therein. Here we introduce a denoising-based approach similar to the treatment for approximating the messages associated with $\boldsymbol{L}$ and $\S$. The denoising-based approach presented here is a generalization of the LMMSE approach. More importantly, we show that, with the denoising conditions satisfied,  Assumptions \ref{as1} and \ref{as2} are self-consistent in the iterative process. 
		
		Denote by $\D_{\X} $ the denoiser of matrix $\X$.  Then, the mean of message $m^{(t)}_{\X}$ is approximated by  
		\BE
		\begin{aligned}      
			\X^{(t)}=\D_{\X}\left(\X^{(t)}_{\X\to \y};v_{\X\to \y}^{(t)}\right)
		\end{aligned} 
		\EE
		where $\X^{(t)}_{\X\to \y}$ and $v_{\X\to \y}^{(t)}$ are the mean and variance of message $m^{(t)}_{\X\to \y}$. Then, the mean of $m^{(t+1)}_{\y\to \X}$  is given by 
		\BE\label{formX}
		\begin{aligned}
			\X^{(t+1)}_{\y\to \X}= c_{\X}^{(t)}\left(\X^{(t)}-a_{\X}^{(t)}\X_{\X\rightarrow \y}^{(t)} \right)
		\end{aligned}
		\EE
		
		where $a_{\X}^{(t)}$ and $c_{\X}^{(t)}$  are  scalar parameters satisfying
		\BS
		\label{xtp}
		\begin{align} 
			&\mathbb{E}\left \langle \X^{(t+1)}_{\y\to\X}-\X,\X^{(t)}_{\X\to \y}-\X\right \rangle=0\\ 
			&\mathbb{E}\left \langle \X^{(t+1)}_{\y\to\X}-\X,\X\right \rangle=0.
		\end{align} 
		\ES 
		The constraint in (\ref{xtp}a) is similar to the conditions for sparse denoiser (see (\ref{turboprinciple}b)) and low-rank denoiser (see  (\ref{Sac}b)) to ensure the uncorrelation between the prior message and the extrinsic message. The constraint (\ref{xtp}b), which is introduced in \cite{ma2017orthogonal}, ensures the uncorrelation between the extrinsic message and the signal to be recovered.
		Under Assumptions \ref{as1} and \ref{as2}, the expectations in (\ref{xtp}) are taken over the joint distribution $p(\boldsymbol{L})p(\S)p(\boldsymbol{L}_{\boldsymbol{L}\to\delta}^{(t)}|\boldsymbol{L})p(\S_{\S\to\delta}^{(t)}|\S)\delta(\X-\S-\boldsymbol{L})p(\y|\X)$. We see that (\ref{xtp}a), similar to (\ref{turboprinciple}b), requires that the output error is statistically orthogonal to the input error. More interestingly, (\ref{xtp}b) requires that the output error is statistically orthogonal to $\X$. The conditions in (\ref{xtp}) are important to ensure that Assumptions \ref{as1} and \ref{as2} are self-consistent in the iterative process. With (\ref{xtp}) and (\ref{Xformac}), we obtain 
		\BS \label{Xformac}
		\begin{align} 
			&a_{\X}^{(t)}\!=\!\frac{\!\mathbb{E}\!\left\| \X\right\|_{F}^2\!\mathbb{E}\left \langle\! \X^{(t)}\!,\!\X^{(t)}_{\X\to \y}\!\right \rangle\!-\!\mathbb{E}\left \langle\! \X^{(t)}\!,\!\X\!\right \rangle\mathbb{E}\left \langle\! \X^{(t)}_{\X\to \y}\!,\!\X\!\right \rangle}{nv^{(t)}_{\X\rightarrow \y}\mathbb{E}\left\| \X\right\|_{F}^2\!-\!\left( \mathbb{E}\left\| \X\right\|_{F}^2-\mathbb{E}\left \langle \X^{(t)}_{\X\to \y},\X\right\rangle\right)^2 }\\
			&c_{\X}^{(t)}\!=\!\left( \frac{ \mathbb{E}\left\langle  \X^{(t)}, \X\right \rangle-a_{\X}^{(t)}\mathbb{E}\left \langle \X^{(t)}_{\X\to \y},\X\right \rangle}{\mathbb{E}\left\| \X\right\|_{F}^2}\right) ^{-1}.
		\end{align} 
		\ES
		Correspondingly, the variance of  message $m^{(t+1)}_{\y\to\X}$ is given by  
		\BE\label{varX} 
		\begin{aligned}
			v_{\y\to\X}^{(t+1)}&=\frac{1}{n}\mathbb{E}\left\| \X^{(t+1)}_{\y\to\X}-\X\right\|_{F}^2
			= \frac{1}{n}\mathbb{E}\left\|c_{\X}^{(t)}\!\left( \X^{(t)}\!-\! a_{\X}^{(t)}\X^{(t)}_{\X\to \y}\right) \right\|_{F}^2\!-\!\frac{1}{n}\mathbb{E}\left\|\X \right\|_{F}^2.
		\end{aligned}
		\EE
		
		\subsubsection{Message from $\delta$ to $\boldsymbol{L}$ and $\S$}

		We now present the main results of this subsection.
		\begin{lemma}\label{lemma1}
			Suppose that Assumptions 1 and 2 hold, and that $\S_{p(\S)\to\S}^{(t)}$ and   $\boldsymbol{L}_{p(\boldsymbol{L})\to\boldsymbol{L}}^{(t)}$ are the solutions of  (\ref{turboprinciple}) Then 
			\BE\label{lem11}
			\begin{aligned}
				\mathbb{E}\left \langle \X^{(t)}_{\X\to \y} -\X, \X\right \rangle	= - v^{(t)}_{\delta\to \X},
			\end{aligned}
			\EE
			where $v^{(t)}_{\delta\to \X}$ is defined by (\ref{mx}b). Further assume that $\X^{(t+1)}_{\y\to \X}$   satisfies  (\ref{xtp}) and that  $\D_{\X}\left(\X^{(t)}_{\X\to \y};v_{\X\to \y}^{(t)}\right)$ is a linear estimator of $\X$ given $\y$ as 
			\BE\label{LE1}
			\begin{aligned}
			\operatorname{vec}\left( \D_{\X}\left(\X^{(t)}_{\X\to \y};v_{\X\to \y}^{(t)}\right)\right)
				=\operatorname{vec}\left( \X^{(t)}_{\X\to \y}\right) +\M^{(t)}\left( \y-\A\operatorname{vec}\left( \X^{(t)}_{\X\to \y}\right)\right) 
			\end{aligned}
			\EE
			where  $\M^{(t)}\in\mathbb{R}^{n\times m}$ is a deterministic combinating matrix usually depending on $v_{\X\to\y}$. Then 
			\BE\label{lem12}
			\begin{aligned}
				\mathbb{E}\left \langle \S^{(t+1)}_{\S\to p(\S)} -\S, \S\right \rangle	= 0,\quad
				\mathbb{E}\left \langle \boldsymbol{L}^{(t+1)}_{\boldsymbol{L}\to p(\boldsymbol{L})} -\boldsymbol{L}, \boldsymbol{L}\right \rangle	= 0.
			\end{aligned}
			\EE
			The expectations in (\ref{lem11}) and (\ref{lem12}) are taken over 
			\BE
			\begin{aligned}
				p\left( \boldsymbol{L},\boldsymbol{L}^{(t)}_{\boldsymbol{L}\to p(\boldsymbol{L})},\S,\S^{(t)}_{\S\to p(\S)}\right)
				= p(\boldsymbol{L})p(\S)\prod_{i,j}\mathcal{Q}^{(t)}_{N_{\S}}\left( (\S^{(t)}_{\S\to p(\S)})_{i,j}-S_{i,j}\right)
			\prod_{i,j}\mathcal{Q}^{(t)}_{N_{\boldsymbol{L}}}\left(  (\boldsymbol{L}^{(t)}_{\boldsymbol{L}\to p(\boldsymbol{L})})_{i,j}-L_{i,j}\right).
			\end{aligned}
			\EE
		\end{lemma}
		\begin{proof}
			See Appendix \ref{prl1}
		\end{proof}
		\begin{remark}
			The linear estimator in the form of (\ref{LE1}) is widely used in  various iterative linear inversion algorithms (e.g. AMP\cite{bayati2011dynamics}, TMP\cite{Ma2014Turbo,2017Denoising,8502093}, OAMP\cite{ma2017orthogonal} and VAMP\cite{rangan2019vector}). In the next subsection, we discuss how to calculate $a_{\X}^{(t)}$, $c_{\X}^{(t)}$, and $v_{\y\to\X}^{(t)}$ in (\ref{Xformac}) and (\ref{varX}) for different choices of the combining matrix $\M^{(t)}$.
		\end{remark}

		\subsection{Details for Sparse Denoiser $\D_{\S}$ }\label{Sec:sparsity_denoiser}
		In this section, we discuss the choice of the denoiser for matrix $\S$.  
		With knowing the prior distribution of the $i.i.d.$ signal $\S$, we can apply MMSE denoiser that is defined as 
		\BE\label{mmseS}
		\begin{aligned}
			\mathcal{D}_{\S}\left(\left( S^{(t)}_{\S \rightarrow p(\S)} \right) _{i,j}\right)&= \mathbb{E}\left(S_{i,j}\Big|\left( S^{(t)}_{\S \rightarrow p(\S)}\right) _{i,j}\right) &\forall i,j
		\end{aligned}
		\EE
		where the expectation is taken over the joint distribution $p\left( S_{i,j},\left( S^{(t)}_{\S \rightarrow p(\S)}\right) _{i,j}\right)$.  We have the following result.
		\begin{lemma}[MMSE denoiser for $\S$]\label{lemss}
			Assume that message $m_{\S\rightarrow p(\S)}^{(t)}$ satisfies Assumption \ref{as1}, and that  $\S_{p(\S)\rightarrow\S}^{(t)}$  in (\ref{turboprinciple}) is the solution of (\ref{formS}). Then,
			\BS\label{lemsse}
			\begin{align}
				a_{\S}^{(t)}&=\frac{v_{\S}^{(t)}}{v^{(t)}_{\S \rightarrow p(\S)}}\\
				c_{\S}^{(t)}&=\frac{v^{(t)}_{\S \rightarrow p(\S)}}{v^{(t)}_{\S \rightarrow p(\S)}-v_{\S}^{(t)}}
			\end{align}
			\ES	
			where $v^{(t)}_{\S \rightarrow p(\S)}$ is the variance of $m^{(t)}_{\S \rightarrow p(\S)}$, and $v_{\S}^{(t)}$ is the variance of $m^{(t)}_{\S}$. Then, the variance of $m^{(t)}_{p(\S) \rightarrow \S}$ is given by 
			\BE\label{lemss2}
			\begin{aligned}
				v^{(t)}_{p(\S) \rightarrow \S}&=\left(\frac{1}{v_{\S}^{(t)}}-\frac{1}{v^{(t)}_{\S \rightarrow p(\S)}} \right)^{-1}.
			\end{aligned} 
			\EE
		\end{lemma}
		
		\begin{proof}
			See Appendix \ref{prl2}.
		\end{proof}
		
		\begin{remark}
			From Lemma \ref{lemss} and (\ref{formS}), the mean of message $m^{(t)}_{p(\S) \rightarrow \S}$ is given by 
			\BE\label{Smmsevar}
			\begin{aligned}
				\S^{(t)}_{p(\S) \rightarrow \S}=	v^{(t)}_{p(\S) \rightarrow \S}\left( \frac{\S^{(t)}}{v_{\S}^{(t)}}-\frac{\S^{(t)}_{\S\rightarrow p(\S)}}{v^{(t)}_{\S\rightarrow p(\S)}}\right).
			\end{aligned}
			\EE
			The  mean in (\ref{Smmsevar}) and variance in (\ref{lemss2}) are referred as the extrinsic message of the sparsity denoiser in \cite{Ma2014Turbo} (that is derived in a different way). This imply that the design of the sparsity denoiser based on (\ref{formS}) and (\ref{turboprinciple}) is indeed a generalization of the previous approach in \cite{Ma2014Turbo}.  
		\end{remark}
		We note that the soft-thresholding denoiser could also be applied to recover the sparse signal without requiring the exact prior distribution. The parameters $ a_{\S}^{(t)}$ and $c_{\S}^{(t)}$ in (\ref{Sac}) can be obtained by invoking Stein's lemma as in \cite{8502093}.
		
		\subsection{Details for Low-rank Denoiser $\D_{\boldsymbol{L}}$}\label{Sec:low_rank_denoiser}
		In this section, we discuss the choice of the denoiser for matrix $\boldsymbol{L}$. 
		Singular value soft-thresholding (SVST), singular value hard-thresholding (SVHT)  and best Rank-$r$  are commonly used spectral denoisers to recover low-rank matrix from noisy observation.

		Let the SVD of $\boldsymbol{L}^{(t)}_{\boldsymbol{L} \rightarrow p(\boldsymbol{L})}$ be $\boldsymbol{L}^{(t)}_{\boldsymbol{L} \rightarrow p(\boldsymbol{L})} = \boldsymbol{U}^{(t)} \boldsymbol{\Sigma}^{(t)} \boldsymbol{\left( V^{(t)}\right) }^{T}$, where $\boldsymbol{U}^{(t)} \in \mathbb{R}^{n_{1} \times n_{1}}$ and $\boldsymbol{V}^{(t)} \in \mathbb{R}^{n_{2} \times n_{2}}$ are orthogonal matrices, and $\boldsymbol{\Sigma}^{(t)}=\operatorname{diag}\left\{\sigma_{1}^{(t)}, \sigma_{2}^{(t)}, \cdots, \sigma^{(t)}_{\min \left(n_{1}, n_{2}\right)}\right\} \in \mathbb{R}^{n_{1} \times n_{2}}$ is a non-negative diagonal matrix with the diagonal elements arranged in the descending order. 
		
		Let $\tilde{\sigma}_{i}^{(t)}=\frac{\sigma_{i}^{(t)}}{\sqrt{n_2v^{(t)}_{\boldsymbol{L} \rightarrow p(\boldsymbol{L})}}}$. The SVST denoiser  has a closed-form expression given by 
		\BE
		\begin{aligned}
			&\quad\mathcal{D}_{\boldsymbol{L}}\left(\boldsymbol{L}^{(t)}_{\boldsymbol{L} \rightarrow p(\boldsymbol{L})} ; \omega\right)
			=\sqrt{n_2v^{(t)}_{\boldsymbol{L} \rightarrow p(\boldsymbol{L})}}\sum_{i=1}^{\min \left(n_{1}, n_{2}\right)}\!\left(\tilde{\sigma}_{i}^{(t)} -\!\omega\right)_{+}\! \boldsymbol{u}_{i}^{(t)} \!\left(\! \boldsymbol{v}_{i}^{(t)}\!\right) ^{T}
		\end{aligned}
		\EE
		where  $\omega$ is a predetermined threshold, $\boldsymbol{u}^{(t)}_{i}$ is the $i$-th column of $\boldsymbol{U}^{(t)}$, and $\boldsymbol{v}^{(t)}_{i}$ is the $i$-th column of $\boldsymbol{V}^{(t)}$. The SVST denoiser applies soft-thresholding denoising to the singular values of the input matrix.

		The best rank-$r$ denoiser is defined by 
		\BE
		\begin{aligned}
			\mathcal{D}_{L}\left(\!\boldsymbol{L}^{(t)}_{\boldsymbol{L} \rightarrow p(\boldsymbol{L})};r \!\right)=\boldsymbol{U}^{(t)} \boldsymbol{\Sigma}_{r}^{(t)} \left( \boldsymbol{V}^{(t)}\right)^{T}
		\end{aligned}
		\EE
		where $\boldsymbol{\Sigma}^{(t)}_{r} \in \mathbb{R}^{n_{1} \times n_{2}}$  is a diagonal matrix with the first $r$ diagonal elements being $\sigma^{(t)}_{1}, \sigma^{(t)}_{2}, \cdots, \sigma^{(t)}_{r}$ and the others being 0.  The parameter $r$ of the best rank-$r$ denoiser is assumed to be known by the algorithm. Otherwise, the rank-contraction method can be used to estimate the value of $r$.

		We need to further compute the extrinsic message of module $\boldsymbol{L}$. To this end, we need to calculate the scalar $a_{\boldsymbol{L}}^{(t)}$ as defined in \eqref{Sac}. 
		\BS\label{aLapro}
		\begin{align}
			a_{\boldsymbol{L}}^{(t)}=\frac{\mathbb{E}\left \langle \boldsymbol{L}^{(t)}_{\boldsymbol{L}\to p(\boldsymbol{L})} -\boldsymbol{L}, \boldsymbol{L}^{(t)}\right \rangle}{ \mathbb{E}\left \langle \boldsymbol{L}^{(t)}_{\boldsymbol{L}\to p(\boldsymbol{L})} -\boldsymbol{L}, \boldsymbol{L}^{(t)}_{\boldsymbol{L}\to p(\boldsymbol{L})}\right \rangle}
			&=\frac{\mathbb{E}\left \langle \boldsymbol{L}^{(t)}_{\boldsymbol{L}\to p(\boldsymbol{L})} -\boldsymbol{L}, \boldsymbol{L}^{(t)}\right \rangle}{  nv^{(t)}_{\boldsymbol{L}\to p(\boldsymbol{L})}}\\
			&\approx\frac{1}{n}\operatorname{div}\left(\mathcal{D}_{\boldsymbol{L}}\left( \boldsymbol{L}^{(t)}_{\boldsymbol{L}\to p(\boldsymbol{L})}\right) \right) 
		\end{align} 
		\ES
		where (\ref{aLapro}a) follows from  Assumption \ref{as2} and (\ref{aLapro}b) follows from the Stein's lemma\cite{Stein1981Estimation,8502093}  under Gaussian noise. We will show that the approximation in (\ref{aLapro}b) is asymptotically accurate for the SVST and best rank-$r$ denoisers even under i.i.d. noise in subsection \ref{LASYM} and Lemma \ref{smooth_bestrankr}. Furthermore, $	c_{\boldsymbol{L}}^{(t)}$ in (\ref{Sac}b) and $v_{p(\boldsymbol{L})\to\boldsymbol{L}}^{(t)}$ in (\ref{varL}) are approximated by dropping the expectation and substituting the approximated $a_{\boldsymbol{L}}^{(t)}$ in (\ref{aLapro}b).
		The divergence of singular value soft-thresholding (SVST) and best Rank-$r$  are respectively given by \cite{candes2013unbiased}:
		\BE\label{softL}
			\begin{aligned}
				\frac{1}{n}\operatorname{div}\left(\mathcal{D}_{\boldsymbol{L}}\!\!\left(\!\boldsymbol{L}^{(t)}_{\boldsymbol{L} \rightarrow p(\boldsymbol{L})} ; \omega\!\right) \right)=\frac{\sqrt{\!n_2v^{(t)}_{\boldsymbol{L} \rightarrow p(\boldsymbol{L})}\!}}{n}\left(\! \left|n_{1}\!-\!n_{2}\right|\!\!\! \sum_{i=1}^{\min \left(n_{1} n_{2}\right)}\!\!\left(1-\frac{\omega}{\tilde{\sigma}^{(t)}_{i}}\right)_{+}\!\!
				+\!\!\sum_{i=1}^{\min \left(n_{1} n_{2}\right)}\!\!\!\! \mathbb{I}\left(\!\tilde{\sigma}^{(t)}_{i}\!>\!\omega\!\right)\!+\!2\!\!\!\!\sum_{i, j=1, i \neq j}\!\! \frac{\tilde{\sigma}^{(t)}_{i}\!\left(\!\tilde{\sigma}^{(t)}_{i}\!-\!\omega\right)_{+}}{\left(\! \tilde{\sigma}_{i}^{(t)}\!\right) ^{2}\!-\!\left(\! \tilde{\sigma}_j^{(t)}\!\right) ^{2}}\right) 
			\end{aligned}
			\EE
			and 
			\BE\label{bestL}
			\begin{aligned}
				\frac{1}{n}\operatorname{div}\left(\mathcal{D}_{L}\left(\boldsymbol{L}^{(t)}_{\boldsymbol{L} \rightarrow p(\boldsymbol{L})};r\right)\right)=\frac{\sqrt{n_2v^{(t)}_{\boldsymbol{L} \rightarrow p(\boldsymbol{L})}}}{n}\left( 
				\left|n_{1}\!-\!n_{2}\right| r\!+r^{2}\!+\!2\sum_{i=1}^{r}\!\sum_{j=r+1}^{min(n_1,n_2)}\!\frac{\left( \tilde{\sigma}^{(t)}_{i}\right) ^{2}}{\left( \tilde{\sigma}^{(t)}_{i}\right) ^{2}-\left( \tilde{\sigma}^{(t)}_{j}\right) ^{2}}\right) ,
			\end{aligned}
			\EE
		provided that  $\boldsymbol{L}^{(t)}_{\boldsymbol{L}\to p(\boldsymbol{L})}$ has $min(n_1,n_2)$ distinct, positive singular values with probability $1$.

		\subsection{Details for Linear Denoiser $\D_{\X}$}
		
		In this section, we describe use linear denoiser to estimate matrix $\X$ given the noisy observation $\y$ in (\ref{crpcamodel}). We first rewrite $\y$ as 
		\BE
		\begin{aligned} \label{mdx}
			\y=\A \operatorname{vec}(\X)+\boldsymbol{n}
		\end{aligned}
		\EE
		where $\operatorname{vec}(\boldsymbol{X})=\left[\boldsymbol{x}_{1}^{T}, \boldsymbol{x}_{2}^{T}, \cdots, \boldsymbol{x}_{n_{2}}^{T}\right]^{T}$ with $\boldsymbol{x}_{i}$ being the $i$-th column  of $\X$, and $\A$ is the matrix form of linear operator $\mathcal{A}$.
		
		Combining the message $m^{(t)}_{\X\to \y}$ and (\ref{mdx}), we obtain the LMMSE estimator of $\X$ given by 
		\BE\label{LMMSE}
		\begin{aligned}
			\operatorname{vec}\left(\!\boldsymbol{X}^{(t)}\!\right)\!=& \operatorname{vec}\left(\!\boldsymbol{X}^{(t)}_{\boldsymbol{X} \rightarrow \boldsymbol{y}}\!\right)\!+\! v^{(t)}_{\boldsymbol{X} \rightarrow \boldsymbol{y}} \boldsymbol{A}^{T}\left(\! v^{(t)}_{\boldsymbol{X} \rightarrow \boldsymbol{y}} \boldsymbol{A} \boldsymbol{A}^{T}\!+\!\sigma^{2}_{\boldsymbol{n}} \boldsymbol{I}\!\right)^{-1}  \left(\boldsymbol{y}-\boldsymbol{A} \operatorname{vec}\left(\boldsymbol{X}^{(t)}_{\boldsymbol{X} \rightarrow \boldsymbol{y}}\right)\right).
		\end{aligned}
		\EE

		We have the following results.
		\begin{lemma}[Parameters for LMMSE denoiser]\label{lemmaX}
			Assume that $\mathcal{A}$ is a ROIL operator.  $a_{\X}^{(t)}$ and $c_{\X}^{(t)}$  for the LMMSE denoiser are respectively given by 
			\BS
			\begin{align}
				a_{\X}^{(t)}&=\frac{v_{\X}^{(t)}}{v_{\X\rightarrow \y}^{(t)}}\\
				c_{\X}^{(t)}&=\frac{v_{\X\rightarrow \y}^{(t)}}{v_{\X\rightarrow \y}^{(t)}-v_{\X}^{(t)}},
			\end{align}
			\ES
			where $	v_{\X}^{(t)}$ is given by 
			\BE\label{var_LMMSE_lemma3}
			\begin{aligned}
				v_{\X}^{(t)}
				= v^{(t)}_{\X \rightarrow \y}-\frac{\left( v_{\X\to\y}^{(t)}\right) ^2}{n}\sum_{i=1}^{m} \frac{\sigma^2_{i,\A}}{v^{(t)}_{\boldsymbol{X} \to \boldsymbol{y}}\sigma^2_{i,\A}+\sigma_{\boldsymbol{n}}^2}.
			\end{aligned}
			\EE
			Here, $\sigma_{i,\A}$ is the $i$-th singular value of $\A$. 
			And
			\BE\label{varx}
			\begin{aligned}
				\X^{(t+1)}_{\boldsymbol{y} \rightarrow \boldsymbol{X}}&=	v^{(t+1)}_{\boldsymbol{y} \rightarrow \boldsymbol{X}}\left( \frac{\X^{(t)}}{v_{\X}^{(t)}}-\frac{\X_{\X\rightarrow\y}^{(t)}}{v_{\X\rightarrow\y}^{(t)}}\right)
			\end{aligned} 
			\EE
			with 
			\BE
			\begin{aligned}
				v_{\y\to\X}^{(t+1)}
				&=\frac{n}{\sum_{i=1}^{m}\frac{\sigma^2_{i,\A}}{v^{(t)}_{\boldsymbol{X} \rightarrow
							\boldsymbol{y}}\sigma^2_{i,\A}+\sigma_{\boldsymbol{n}}^2}}-v_{\X\to\y}^{(t)}
			\end{aligned} 
			\EE
			where $\sigma_{i,\A}$ is the $i$-th singular value of $\A$. 
		\end{lemma}
		\begin{proof}
			See Appendix \ref{prl3}
		\end{proof}
		
		\begin{corollary}\label{coro_le1}
			Assume that $\mathcal{A}$ is a partial  Haar  operator. Then 
			\BS
			\begin{align}
				&v_{\y\to\X}^{(t+1)}=\frac{n-m}{m}v^{(t)}_{\boldsymbol{X} \rightarrow \boldsymbol{y}}+\frac{n}{m}\sigma_{\boldsymbol{n}}^2\\
				&\X^{(t+1)}_{\boldsymbol{y} \rightarrow \boldsymbol{X}}=\operatorname{vec}\left( \X^{(t)}_{\X\to\y}\right) +\frac{n}{m}\A^T\left( \y-\A\operatorname{vec}\left( \X^{(t)}_{\X\to\y}\right)\right) .
			\end{align}
			\ES
		\end{corollary}
		\begin{proof}
			See Appendix \ref{proof_coro_le1}
		\end{proof}

		\subsection{Overall Algorithm}
		\begin{algorithm}
			\caption{Improved Turbo message passing (ITMP) algorithm}\label{algorithm1}
			\begin{algorithmic}[1]
				\setstretch{1.35}
				\REQUIRE $\mathcal{A}$, $\y$, $\sigma_{\boldsymbol{n}}^2$, $t=0$,  		$m_{\boldsymbol{L}\rightarrow\delta}^{(0)}\!=\!\GN({\boldsymbol{L};\mathbf{0},\sigma^2_{\boldsymbol{L}}\boldsymbol{I}})$, $m_{\boldsymbol{S}\rightarrow\delta}^{(0)}\!=\!\GN({\boldsymbol{S};\mathbf{0},\sigma^2_{\S}\boldsymbol{I}})$\\
				\WHILE{the stopping criterion is not met}
				\STATE $\X_{\X\rightarrow \y}^{(t)}\!= \boldsymbol{L}_{\boldsymbol{L}\rightarrow\delta}^{(t)}+\S_{\S \rightarrow\delta}^{(t)}$, \quad
$v_{\X\rightarrow \y}^{(t)}\!=\! v_{\boldsymbol{L}\rightarrow\delta}^{(t)}+v_{\S \rightarrow\delta}^{(t)}$	
				\STATE $\boldsymbol{X}^{(t)}=\D_{\X}\left(\X^{(t)}_{\X\to \y},v_{\X\to \y}^{(t)}\right) \% \emph{linear denoiser}$
				\STATE Update $a_{\X}^{(t)}$ and $c_{\X}^{(t)}$ by (\ref{Xformac})
				\STATE $\X_{\X\rightarrow\delta}^{(t+1)}=c_{\X}^{(t)}\left(\X^{(t)}-a_{\X}^{(t)}\X_{\X\rightarrow \y}^{(t)} \right)$
				\STATE Update $v_{\X\rightarrow\delta}^{(t+1)}$ by (\ref{varX})
				\STATE $\S_{\S\rightarrow p(\S)}^{(t+1)}=\X_{\X\rightarrow\delta}^{(t+1)}-\boldsymbol{L}_{\boldsymbol{L}\rightarrow\delta}^{(t)}$,\quad
 $v_{\S\rightarrow p(\S)}^{(t+1)}=v_{\X\rightarrow\delta}^{(t+1)}+v_{\boldsymbol{L}\rightarrow\delta}^{(t)}$\\
				\STATE $\boldsymbol{L}_{\boldsymbol{L}\rightarrow p(\boldsymbol{L})}^{(t+1)}= \X_{\X\rightarrow\delta}^{(t+1)}-\S_{\S\rightarrow\delta}^{(t)}$,\quad $v_{\boldsymbol{L}\rightarrow p(\boldsymbol{L})}^{(t+1)}= v_{\X\rightarrow\delta}^{(t+1)}+v_{\S\rightarrow\delta}^{(t)}$
				\STATE $t = t+1$
				\STATE $\S^{(t)}\!=\!\D_{\S}\left(\S^{(t)}_{\S\to p(\S)},v_{\S\to p(\S)}^{(t)}\right)\%\emph{sparse denoiser}$
				\STATE Update   $a_{\S}^{(t)}$ and $c_{\S}^{(t)}$ by (\ref{Sac})
				\STATE $\S_{\S\rightarrow\delta}^{(t)}=c_{\S}^{(t)}\left(\S^{(t)}-a_{\S}^{(t)}\S_{\S\rightarrow p(\S)}^{(t)} \right)$
				\STATE  Update	$v_{\S\rightarrow\delta}^{(t)}$ by (\ref{Svar}) 
				\STATE $\boldsymbol{L}^{(t)}\!=\!\D_{\boldsymbol{L}}\left(\boldsymbol{L}^{(t)}_{\boldsymbol{L}\to p(\boldsymbol{L})},v_{\boldsymbol{L}\to p(\boldsymbol{L})}^{(t)}\right)\% \emph{low-rank denoiser}$ 
				\STATE  Update $a_{\boldsymbol{L}}^{(t)}$ and $c_{\boldsymbol{L}}^{(t)}$ by (\ref{Sac})
				\STATE $\boldsymbol{L}_{\boldsymbol{L}\rightarrow\delta}^{(t)}=\!c_{\boldsymbol{L}}^{(t)}\left(\boldsymbol{L}^{(t)}-a_{\boldsymbol{L}}^{(t)}\boldsymbol{L}_{\boldsymbol{L}\rightarrow p(\boldsymbol{L})}^{(t)} \right)$ \\
				\STATE  Update $v_{\boldsymbol{L}\rightarrow\delta}^{(t)}=v_{p(\boldsymbol{L})\rightarrow\boldsymbol{L}}^{(t)}$ by (\ref{varL})
				\ENDWHILE
				\ENSURE  $\boldsymbol{L}^{(t)}$, $\S^{(t)}$
			\end{algorithmic}
		\end{algorithm}
		We summarize the overall improved turbo message passing algorithm in Algorithm~\ref{algorithm1}. The algorithm mainly consists of three parts, corresponding to the linear denoiser, the sparsity denoiser and the low-rank denoiser, as illustrated in Fig. \ref{tcrpca}.  We next briefly explain Algorithm \ref{algorithm1} in a line-by-line manner.  \par 
		We initialize $m_{\boldsymbol{L}\rightarrow\delta}^{(0)}$ and  $m_{\boldsymbol{S}\rightarrow\delta}^{(0)}$ with $\GN({\boldsymbol{L};\mathbf{0},v_{\boldsymbol{L}}\boldsymbol{I}})$ and $\GN({\boldsymbol{S};\mathbf{0},v_{\S}\boldsymbol{I}})$, where $v_{\boldsymbol{L}}$ and $v_{\S}$ are respectively the normalized powers of $\boldsymbol{L}$ and $\S$. In lines 2, we calculate the message from $\X$ to $\y$ (which is equal to the message from $\delta$ to $\X$). In line 3, $\X$ is estimated by combining the updated message from $\delta$ and the knowledge of $\y$. The linear denoiser $\D_{\X}$  has a variety of choices, such as  the LMMSE denoiser  described in the preceding subsection. In lines 4 and 6, the calculations methods of  $a_{\X}^{(t)}$, $c_{\X}^{(t)}$ and  $v_{\y\rightarrow \X}^{(t+1)}$ are given in Lemma \ref{lemmaX} and its corollaries.
		In lines 7-8, the message from $\delta$ to $\boldsymbol{L}$ and the message from $\delta$ to $\S$ are  updated based on the sum-product rule. In line 10-13, the messages associated with the sparsity denoiser are updated.
		For the sparsity denoiser $\D_{\S}$ in line 13, we choose the point-wise MMSE denoiser when the elements of $\S$ are independently and identically drawn from a known distribution; we choose the soft-thresholding denoiser when the prior distribution of $\S$ is unknown. In lines 11 and 13, the updating methods of $a_{\S}^{(t)}$, $c_{\S}^{(t)}$, and $v_{p(\S)\rightarrow \S}^{(t)}$ can be found in Subsection \ref{Sec:sparsity_denoiser}. In lines 14-17, the messages associated with the low-rank denoiser are updated.
		For the low-rank denoiser in line 14, we choose the best-rank-$r$ denoiser when the rank of $\boldsymbol{L}$ is known, and choose the singular-value soft-thresholding  otherwise. In lines 15 and 17, the updating methods of $a_{\boldsymbol{L}}^{(t)}$, $c_{\boldsymbol{L}}^{(t)}$ and $v_{p(\boldsymbol{L})\rightarrow \boldsymbol{L}}^{(t)}$ can be found in Subsection \ref{Sec:low_rank_denoiser}. 
		\section{Asymptotic Performance Analysis}\label{Sec4}
		\subsection{Preliminaries}\label{Sec4-A}
		The dynamic processes for message-passing algorithms \cite{bayati2011dynamics,berthier2020state} can be tracked by scalar functions that delineate the progression of variances, commonly known as state evolution. The state evolution is proved to be  exact asymptotically in the large system when the non-linear function is separable\cite{bayati2011dynamics}. Moreover,  the state evolution of AMP with non-separable functions  is also successful in many experimental observations.  While, the corresponding theoretical proof is  a challenging. In \cite{berthier2020state}, the authors generalize state evolution to  non-separable non-linear functions under $i.i.d.$ Gaussian sensing matrix.  In \cite{rangan2019vector}, the authors extend the Gaussian sensing matrix to the right-orthogonally invariant matrix and provide theoretical proofs to the corresponding analysis. While, some technical conditions should be met in their theory\cite{rangan2019vector} that is not satisfied in our case. For this, instead of providing rigorous analysis of the overall algorithm,  we present some assumptions for the state evolution.

		Assumptions \ref{as1} and \ref{as2} allow  for decoupling the probability spaces  from iteration to iteration. Similar assumptions are also used in the prior work such as AMP\cite{donoho2009message},  EP\cite{minka2013expectation}, and other message-passing algorithms \cite{Ma2014Turbo,xue2019tarm, 8502093}. With Assumptions \ref{as1} and \ref{as2}, we aim to establish the mean square error (MSE) transfer functions of the denoisers  and analyze the convergence and the fixed point of the algorithm based on the derived transfer functions as $n\to\infty$ with $\frac{m}{n} = \alpha \in(0,1]$, $\frac{n_{1}}{n_{2}} = \beta \in(0,1]$, $\frac{r}{n_{1}} = \gamma \in(0,1]$ and 
		$\rho \in(0,1]$. For the sake of simplicity, we henceforth abbreviate this limit process as ``$n\to\infty$'' or ``$n_1\to\infty$'' by ommiting the other conditions.
		
		For ease of analysis, we introduce some definitions and regularity conditions for the three denoisers involved in the ITMP algorithm.
		\begin{definition}[Pseudo-Lipschitz\cite{5695122}]
			For $k\in\mathbb{N}_{>0}$ and any $n,m \in\mathbb{N}_{>0}$, a function $f:\mathbb{R}^{n}\to\mathbb{R}^{m}$ is called pseudo-Lipschitz of order $k$ if there exists a constant $L$ such that for any $\x,\y\in\mathbb{R}^{n}$,
			\BE
			\begin{aligned}
				\frac{\|f(\boldsymbol{x})\!-\!f(\boldsymbol{y})\|_{2}}{\sqrt{m}}\! \leq\! L\!\left(\!1\!+\!\left(\!\frac{\|\boldsymbol{x}\|_{2}}{\sqrt{n}}\!\right)^{k\!-\!1}\!\!\!\!+\!\left(\!\frac{\|\boldsymbol{y}\|_{2}}{\sqrt{n}}\!\right)^{k\!-\!1}\!\right) \frac{\|\boldsymbol{x}\!-\!\boldsymbol{y}\|_{2}}{\sqrt{n}}.
			\end{aligned}
			\EE
			$L$ is then called the pseudo-Lipschitz constant of $f$. Furthermore, denote by $\{L_{n}\}_{n\in\mathbb{N}_{>0}}$ the Lipschitz constants of a sequence of pseudo-Lipschitz functions $\{f_{n}\}_{n\in\mathbb{N}_{>0}}$ of order $k$. If  $L_n <\infty$ for each $n$ and $\lim sup_{n\to \infty} L_n <\infty$, the sequence $\{f_{n}\}_{n\in\mathbb{N}_{>0}}$  is called uniformly pseudo-Lipachitz of order $k$. We call any $L>\lim sup_{n\to \infty} L_n <\infty$ a pseudo-Lipschitz constant of the sequence.
		\end{definition}
		For $k=1$, the above definition reduces to the standard definition of a Lipschitz function.
		
		\begin{definition}[Weak Convergence]\label{defi_weakcon} Let $\{p_n\}_{n\in\mathbb{N}}$  be a sequence of probability measures on $(\mathbb{S},\mathbf{\Omega})$, where $\mathbb{S}$ is a metric space and $\mathbf{\Omega}$ is the Borel $\sigma$-algebra on $\mathbb{S}$.
			We say that a sequence of probability measure $\{p_n\}_{n\in\mathbb{N}}$  converges weakly to a probability measure $p_{\infty}$  and write $p_n\overset{w}{=}p_{\infty}$ if
			\BE
			\begin{aligned}
				\int f(z) d p_n(z)\to \int f(z) d p_{\infty}(z)
			\end{aligned}
			\EE
			for all continuous and bounded function $f:\mathbb{S}\to\mathbb{R}$.
		\end{definition}
	This article focuses solely on random variables whose domains are restricted to the real number field. Consequently, the sequence  $\{p_n\}_{n\in\mathbb{N}}$  converges weakly to a probability measure $p_{\infty}$ if and only if it converges to the same probability measure $p_{\infty}$ in distribution.

		We first present the assumptions we make for our state evolution analysis.
		\begin{assumption}\label{as3}
			\
			\begin{enumerate}[ (3.i)]
				\item  The entries of $\S$ follow an i.i.d. distribution $p_{\S}(S)$  with finite $2k$-th moment for some $k\geq1$;
				\item $\mathcal{Q}^{(t)}_{N_{\S}}(N_{\S})$ in Assumption \ref{as1} has finite $2k$-th moment for some $k\geq1$ for each iteration;
				\item 	 $\D_{\S}:\mathbb{R}^{n}\to\mathbb{R}^{n}$ is a separable denoiser and  a pseudo-Lipschitz function of order $k$  for each components. Here ``separable” means that we  evaluate the elements of $\S$ separately  by
				$\hat{S}^{(t)}_{i,j}=\D_{\S}\left( S_{i,j} +\sqrt{v_{\delta\to\S}^{(t)}} N^{(t)}_{\S,i,j} \right)  \forall i,j .$;
				\item $\boldsymbol{L}\in\mathbb{R}^{n_1\times n_2}$ is rotationally invariant;
				\item 	 $\mathcal{Q}^{(t)}_{N_{\boldsymbol{L}}}(N_{\boldsymbol{L}})$ in Assumption \ref{as2} has finite $4$-th moment for each iteration. Let $\sigma_{i,\boldsymbol{L}}$ be the  $i$-th singular value of $\boldsymbol{L}$, and $\mu_{\boldsymbol{L}_{n_1,n_2}}(\sigma)$ be the empirical distribution of $\{\frac{\sigma_{i,\boldsymbol{L}}}{\sqrt{n_2}}\}_{i}^{n_1}$. Let $\frac{\sigma_{i,\boldsymbol{L}}}{\sqrt{n_2}}$ be bounded. The empirical singular value distribution $\mu_{\boldsymbol{L}_{n_1,n_2}}(\sigma)$  converges weakly to a deterministic distribution $\mu_{\boldsymbol{L}}(\sigma)$
			, and the $2$-nd moment of $\mu_{\boldsymbol{L}}(\sigma)$ is finite;
				\item $\D_{\boldsymbol{L}}:\mathbb{R}^{n_1\times n_2}\to\mathbb{R}^{n_1\times n_2}$ is a non-separable denoiser acting on the singular values of the noisy observation (e.g.  $\D_{\boldsymbol{L}}(\boldsymbol{L}+\sqrt{v}\N)=\sqrt{n_2v}\U \D_{\boldsymbol{L}}(\frac{\boldsymbol{\Sigma}}{\sqrt{n_2v}})\V^T$), where $\D_{\boldsymbol{L}}\left( \frac{\boldsymbol{\Sigma}}{\sqrt{n_2v}}\right) $ is a diagonal matrix  with the $(i,i)$-th entry $ f_{\boldsymbol{L}}(\frac{\sigma_i}{\sqrt{n_2v}})  $.  And $f_{\boldsymbol{L}}: [0,+\infty)\rightarrow[0,+\infty)$ is a monotonically  increasing  and  continuously differentiable function, where $f_{\boldsymbol{L}}(0)=0$ and $f'_{\boldsymbol{L}}(0)=0$ . 
				\item The empirical eigenvalue distribution of $\A\A^{T}$ converges weakly    to a deterministic distribution $p(\theta_{\A\A^T})$ with compact support $[\theta_{\A\A^T,\min},\theta_{\A\A^T,\max}]$, where $0<\theta_{\A\A^T,\min}<\theta_{\A\A^T,\max}<\infty$.
			\end{enumerate}
		\end{assumption}
		Assumption (\ref{as3}.i)-(\ref{as3}.iii) for the sparsity denoiser was previously introduced in \cite{bayati2011dynamics} to analyze the asymptotic property for separable Lipschitz functions in AMP. Here, we characterize the asymptotic behavior of the parameters $a_{\S}^{(t)}$ , $c_{\S}^{(t)}$ and the extrinsic mean $\S_{\S \rightarrow p(\S)}^{(t)}$ based on this assumption.
		\begin{figure}[thbp!]
			\centering
			\begin{minipage}[t]{1.0\linewidth}
				\centering
				\begin{tabular}{@{\extracolsep{\fill}}c@{}c@{}c@{}@{\extracolsep{\fill}}}
					\includegraphics[width=0.33\linewidth]{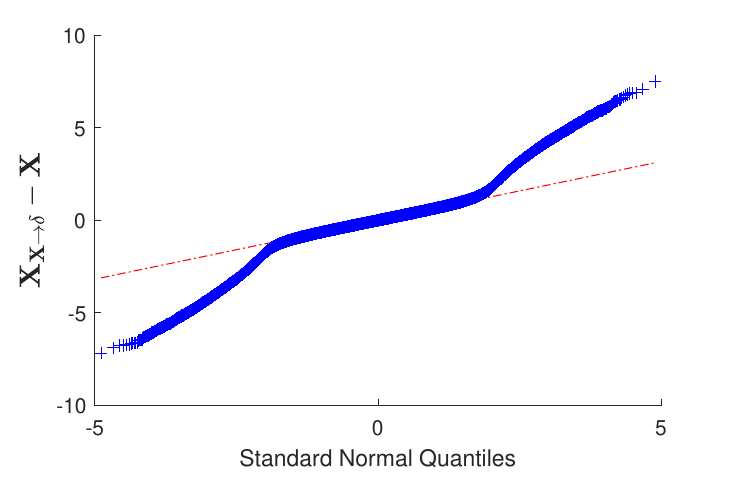}  &
                    \includegraphics[width=0.33\linewidth]{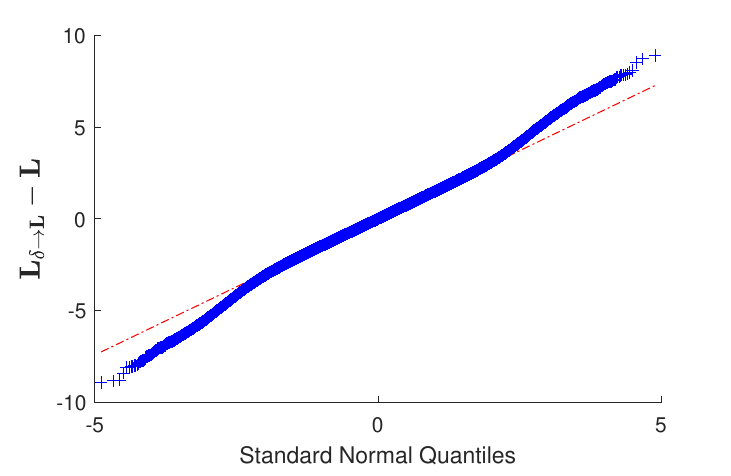} &
				\includegraphics[width=0.33\linewidth]{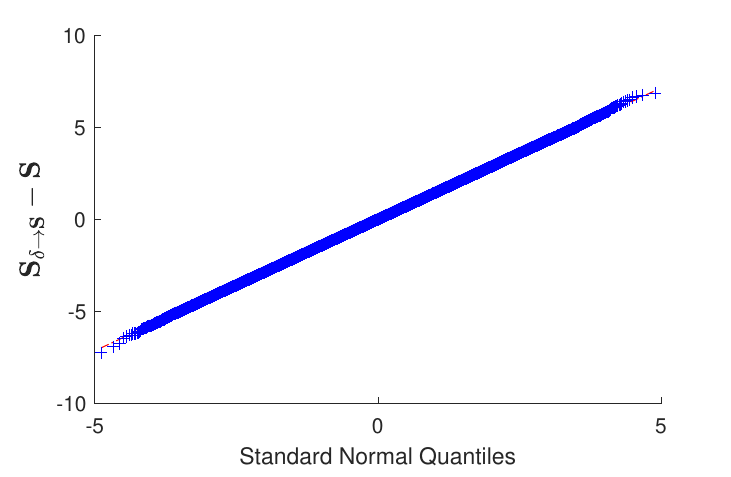} \\
					(a) input error of module $\X$ & (b) input error of module $\boldsymbol{L}$& (c) input error of module $\S$\\
				\end{tabular}
			\end{minipage}
			\begin{minipage}[t]{1.0\linewidth}
				\centering
				\begin{tabular}{@{\extracolsep{\fill}}c@{}c@{}c@{}@{\extracolsep{\fill}}}
								\includegraphics[width=0.33\linewidth]{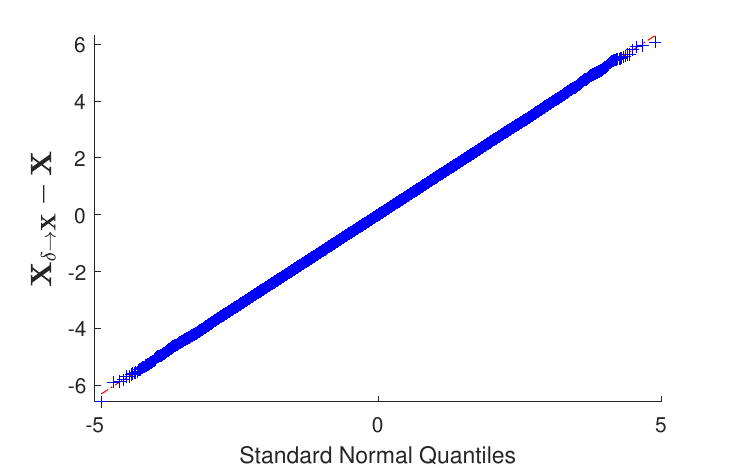} &		\includegraphics[width=0.33\linewidth]{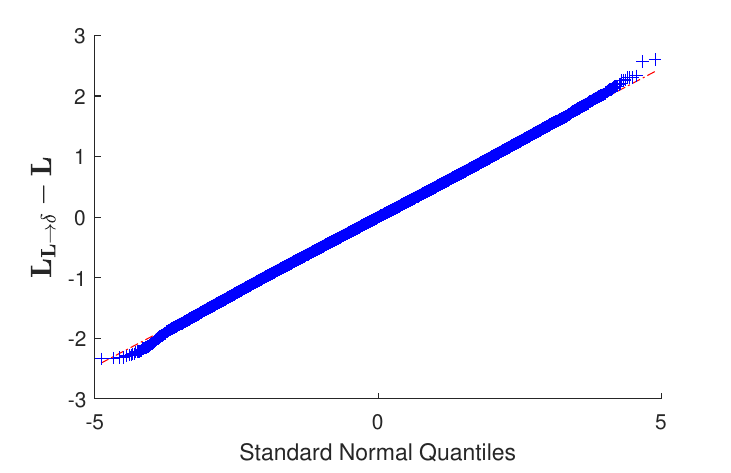} &		\includegraphics[width=0.33\linewidth]{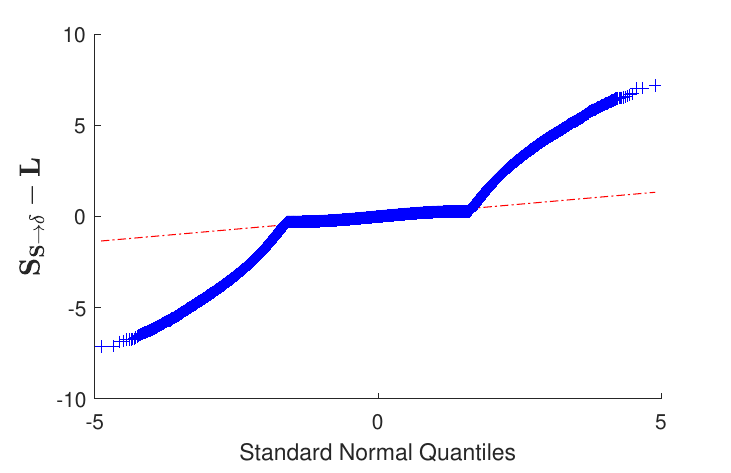}\\
								(d) output error of module $\X$ & (e) output error of module $\boldsymbol{L}$ & (f) output error of module $\S$\\
				\end{tabular}
			\end{minipage}
			\caption{The QQplots of the input and output error in the second iteration of the proposed algorithm. The first row are the QQplots of the estimation errors \textbf{input} each modules. The second row are the QQplots of the estimation errors \textbf{output} each modules. The entries of sparse matrix $\S$ are independent drawn from the Gaussian-Bernoulli distribution with $\rho=0.1$ and the low-rank matrix $\boldsymbol{L}$ with rank-$r$. We use the point-wise MMSE denoiser to estimate $\S$ and best-rank-$r$ denoiser to estimate $\boldsymbol{L}$. We use the  LMMSE denoiser to estimate $\X$. Other parameter settings are $n=n_1\times n_2=1000,\beta=1,\alpha=0.4,\sigma_{\boldsymbol{n}}^2=10^{-5},\gamma=0.05.$} 	
\label{qqplot1111}  
		\end{figure}

		For low-rank matrix recovery, the existing work mostly focused on the estimation of a low-rank matrix given its noisy observation under an additive white Gaussian noise; see, e.g.,  \cite{gavish2017optimal,candes2013unbiased,gavish2014optimal}. 
		The results therein are, however, inapplicable here since numerical experiments reveal that empirical distribution of the input error $\boldsymbol{L}_{\boldsymbol{L} \rightarrow p(\boldsymbol{L})}^{(t)}-\boldsymbol{L}$ of low-rank denoiser is not Gaussian; see, e.g., Fig.\ref{qqplot1111}. Therefore, we need to consider the situation that the entries of $\boldsymbol{L}_{\boldsymbol{L} \rightarrow p(\boldsymbol{L})}^{(t)}-\boldsymbol{L}$ are $i.i.d.$ non-Gaussian. To facilitate the analysis, we restrict $\boldsymbol{L}$ to satisfy (iv) and (v) of Assumption \ref{as3}.
		
		Part (vi) of Assumption \ref{as3} gives some regularity conditions on the spectral denoiser. 
		More concretely, instead of giving a mapping $f_{\boldsymbol{L}}$ related with the matrix size $n_2$ and $v_{\delta\to\boldsymbol{L}}^{(t)}$, we calibrate $\D_{\boldsymbol{L}}$ to  $\frac{\boldsymbol{L}_{\delta\to\boldsymbol{L}}^{(t)}}{\sqrt{n_2v_{\delta\to\boldsymbol{L}}^{(t)}}}=\frac{\boldsymbol{L}}{\sqrt{n_2v_{\delta\to\boldsymbol{L}}^{(t)}}}+\frac{\N^{(t)}_{\boldsymbol{L}}}{\sqrt{n_2}}$  similarly to the expression in \cite{gavish2017optimal}. In this paper, we consider the most general form of an invariant estimator\footnote{We say that an estimator $D_{\boldsymbol{L}}( \boldsymbol{L}_{\delta\to\boldsymbol{L}}^{(t)})$ is orthogonally invariant if 
			$\U\D_{\boldsymbol{L}}( \boldsymbol{L}_{\delta\to\boldsymbol{L}}^{(t)})\V=\D_{\boldsymbol{L}}(\U\boldsymbol{L}_{\delta\to\boldsymbol{L}}^{(t)}\V)$
			where  $\U\in\mathcal{O}_{n_1}$, $\V\in\mathcal{O}_{n_2}$. }, which is given by 
		\BE\label{indl}
		\begin{aligned}
			\D_{\boldsymbol{L}}\left( \boldsymbol{L}_{\delta\to\boldsymbol{L}}^{(t)}\right) =\sqrt{n_2v_{\delta\to\boldsymbol{L}}^{(t)}}\U^{(t)} \D_{\boldsymbol{L}}\left( \frac{\boldsymbol{\Sigma}^{(t)}}{\sqrt{n_2v_{\delta\to\boldsymbol{L}}^{(t)}}}\right)\left( \V^{(t)}\right) ^{T}. 
		\end{aligned}
		\EE
		where $\boldsymbol{U}^{(t)}\boldsymbol{\Sigma}^{(t)}\left( \boldsymbol{V}^{(t)}\right) ^T$ be the SVD of  $\boldsymbol{L}_{\delta\to\boldsymbol{L}}^{(t)}$, $\boldsymbol{U}^{(t)}\in \mathbb{R}^{n_1\times n_1}$ and $\boldsymbol{V}^{(t)}\in \mathbb{R}^{n_2\times n_2}$ are orthogonal matrices, and $\boldsymbol{\Sigma}^{(t)}=\text{diag}\{\sigma_1^{(t)},\sigma_2^{(t)},\cdots, \sigma_{n_1}^{(t)}\}\in \mathbb{R}^{n_1\times n_2}$ is a diagonal matrix with the diagnal elements arranged in the descending order. Let $\tilde{\sigma}_{i}^{(t)}=\frac{\sigma_{i}^{(t)}}{\sqrt{n_2v^{(t)}_{\boldsymbol{L} \rightarrow p(\boldsymbol{L})}}}$.
		Thus, the invariant estimator is equivalent to a vector map $\mathbb{R}^{n_1}\to\mathbb{R}^{n_1}$ acting on the normalized singular value $\tilde{\sigma}_{i}^{(t)}$. Each normalized singular values $\tilde{\sigma}_{i}^{(t)}$ are mapped by an identical scaling function $f_{\boldsymbol{L}}$.  And $f_{\boldsymbol{L}}(\tilde{\sigma}_{i}^{(t)})$ is the $(i,i)$-th entry of the diagonal matrix $\D_{\boldsymbol{L}}\left( \frac{\boldsymbol{\Sigma}^{(t)}}{\sqrt{n_2v_{\delta\to\boldsymbol{L}}^{(t)}}}\right)$. 
		With the smoothing technique, the $\D_{\boldsymbol{L}}$ assumed in Assumption (\ref{as3}.vi) can be made arbitrarily close to commonly used spectral denoisers, such as the best-rank-$r$ denoiser, the singular value soft-thresholding denoiser, and the singular value hard-thresholding denoiser.

		Assumption (\ref{as3}.vii) for linear denoiser was introduced  to analyze the asymptotic behavior of the ROIL operator $\mathcal{A}$ in \cite{ma2017orthogonal} and  \cite{2019Rigorous}. Under Assumptions 1-3, we derive the transfer functions of the denoisers in the following three subsections.	
		\subsection{Transfer Function of the Sparsity Denoiser}
		Denote by $\hat{v}_{\S\to p(\S)}^{(t)}$ the empirical MSE of $\S^{(t)}_{\S\to p(\S)}$ (as an estimate of $\S$), i.e.,
		\BE   
		\begin{aligned}   	
			\hat{v}_{\S\to p(\S)}^{(t)}\triangleq \frac{1}{n}\left\| \S_{\S\to p(\S)}^{(t)}-\S\right\|_{F}^2.
		\end{aligned} 
		\EE 
		Similarly, $\hat{v}_{p(\S)\to\S}^{(t)}$ the empirical MSE of $\S_{p(\S)\to\S}^{(t)}$ respectively by 
		\BE    
		\begin{aligned}   	 	 	 	
			\hat{v}_{p(\S)\to\S}^{(t)}\triangleq \frac{1}{n}\left\| \S_{p(\S)\to\S}^{(t)}-\S\right\|_{F}^2.
		\end{aligned}    
		\EE  
		
		Furthermore, denote the empirical values of $a^{(t)}_{\S}$ and $c^{(t)}_{\S}$ respectively by  
		\BS 
		\begin{align} \label{empSac}
			\hat{a}^{(t)}_{\S} &\triangleq\frac{\left \langle \S^{(t)}_{\S\to p(\S)} -\S, \S^{(t)}\right \rangle}{ \left\langle \S^{(t)}_{\S\to p(\S)} -\S, \S^{(t)}_{\S\to p(\S)}\right \rangle}\\ 
			\hat{c}^{(t)}_{\S}&\triangleq\frac{\left\langle\left(\S^{(t)}-\hat{a}_{\S}^{(t)} \S^{(t)}_{\S\to p(\S)}\right) ,\S_{\S\to p(\S)}^{(t)}\right\rangle}{\left\lVert \S^{(t)}-\hat{a}_{\S}^{(t)} \S^{(t)}_{\S\to p(\S)}\right\rVert^2_{F}} .
		\end{align} 
		\ES 
		
		As the system size approaches infinity, the randomness observed in the empirical Mean Squared Errors (MSEs) vanishes, and these MSEs converge towards constant values. As a result, the output empirical MSE of a denoiser can be expressed as a deterministic function of the input empirical MSE in the large system limit. This deterministic function is called the transfer function of the corresponding denoiser. 
		
		From Assumption \ref{as1}, we can model  $\S_{\delta\to\S}^{(t)}$ as
		\BE
		\begin{aligned}
			\S_{\delta\to\S}^{(t)}=\S +\sqrt{v_{\delta\to\S}^{(t)}} \N_{\S}^{(t)}
		\end{aligned}
		\EE
		where the elements of $\N^{(t)}_{\S}$ are $i.i.d.$ drawn from the distribution $\mathcal{Q}^{(t)}_{N_{\S}}(N_{\S})$  with  zero mean  and unit variance. 
		\subsubsection{Asymptotic MSE of a general separable denoiser}	
		Denote $\mathcal{D}_{\S}$ a general separable denoiser  without knowing the exact distribution of $\S$ and $\N_{\S}^{(t)}$. The asymptotic behaviors of separable denoisers under Gaussian noise have been studied in prior works  \cite{bayati2011dynamics}. In the following lemma, we present the asymptotics of a general separable denoiser under $i.i.d.$ noise, which is not necessarily Gaussian.
		\begin{lemma}[Asymptotics for sparse denoiser]\label{asymS}
			Under Assumptions 1 and 3, as $n\to\infty$, the following holds: 
			\begin{enumerate}[1)]  	
				\item The empirical posteriori variance asymptotically satisfies
				\BE\label{lemma_asys_eq1}
				\begin{aligned}
					\frac{1}{n}\left\| \D_{\S}\left( \S_{\S\to p(\S)}^{(t)}\right) -\S\right\|_{F}^2\overset{p}{=}\operatorname{MSE}_{\S}(v_{\delta\to\S}^{(t)})
				\end{aligned}
				\EE 
				where 	
				\BE
				\begin{aligned}
					\operatorname{MSE}_{\S}(v_{\delta\to\S}^{(t)})=\mathbb{E}\left[ \left(\D_{\S}(S+\sqrt{v^{(t)}_{\delta\to\S}}N^{(t)}_{\S})-S\right)^{2}\right]. 
				\end{aligned} 
				\EE
				Here, $S$ follows the distribution of $p_{\S}(S)$, and $N^{(t)}_{\S}$ is independently drawn from $\mathcal{Q}_{N_{\S}}^{(t)}(N_{\S})$.
				\item \BE\label{lemma_asys_eq2}
				\begin{aligned}
					\hat{a}_{\S}^{(t)}\overset{p}{=} a_{\S,\infty}^{(t)} \quad\text{and} \quad  \hat{c}_{\S}^{(t)}\overset{p}{=}c_{\S,\infty}^{(t)}			
				\end{aligned}
				\EE
				where $\hat{a}_{\S}^{(t)}$ and $\hat{c}_{\S}^{(t)}$ are defined in (\ref{empSac}), and 
				\BS
				\begin{align}
					&a_{\S,\infty}^{(t)}=	\frac{\mathbb{E}[\sqrt{v_{\delta\to\S}^{(t)}}N^{(t)}_{\S}\D_{\S} \left( S+\sqrt{v^{(t)}_{\delta\to\S}} N^{(t)}_{\S} \right)]}{v_{\delta\to\S}^{(t)}}\\
					& c_{\S,\infty}^{(t)}=\frac{\mathbb{E}\left[\!\left(\!S\! -\!\frac{\mathbb{E}[S^2]}{\sqrt{v_{\delta\to\S}^{(t)}}}N_{\S}^{(t)}\!\right)\! \D_{\S}\! \left(\!S\!+\!\sqrt{v^{(t)}_{\delta\to\S}} N^{(t)}_{\S}\!\! \right)\!\right] }{\mathbb{E}\!\left[\!\left|\!  \D_{\S}\!\!\left(\! S\!\!+\!\!\sqrt{\!v^{(t)}_{\delta\to\S}\!} N^{(t)}_{\S}\! \right)\!\!\!-\!a_{\S,\infty}^{(t)}\!\! \left( \!\!S\!\!+\!\!\sqrt{\!v^{(\!t\!)}_{\delta\to\S}\!} N^{(\!t\!)}_{\S}\! \right)\!\right|^2\!\right]}.
				\end{align}
				\ES 
				\item
				The transfer function for sparsity denoiser is obtained by 
				\BE\label{sparseasy}
				\begin{aligned}
					\hat{v}_{p(\S)\to\S}^{(t)}\overset{p}{=}\varphi(v_{\delta\to\S}^{(t)})
				\end{aligned}
				\EE	
				where 
				\BE\label{SparseSE}
				\begin{aligned}
					&\varphi(v_{\delta\to\S}^{(t)})=			\mathbb(1+c_{\S,\infty}^{(t)} a_{\S,\infty}^{(t)})\mathbb{E}[S^2]-c_{\S,\infty}^{(t)}\mathbb{E} \left[S\D_{\S} \left( S+\sqrt{v^{(t)}_{\delta\to\S}} N^{(t)}_{\S} \right)\right].
				\end{aligned}
				\EE
				\item 
				\BS\label{as3ext}
				\begin{align}
					&\left \langle \S^{(t)}_{p(\S)\to\S}-\S,\S^{(t)}_{\S\to p(\S)}-\S\right \rangle\overset{p}{=}0\\
					&\left \langle \S^{(t)}_{p(\S)\to\S}-\S,\S\right \rangle\overset{p}{=}-\varphi(v_{\delta\to\S}^{(t)})
				\end{align}
				\ES
			\end{enumerate}
			where the expectation $\mathbb{E}$ is taken over the joint distribution $p_{\S}(S)\mathcal{Q}_{N_{\S}}(\N_{\S})$, $a_{\S,\infty}^{(t)}$ and $c_{\S,\infty}^{(t)}$ are respectively the large system limits of $a_{\S}^{(t)}$ and $c_{\S}^{(t)}$ in (\ref{Sac}).
		\end{lemma}
		\begin{proof}
			Proof see Appendix \ref{proofsse}.
		\end{proof}

		\subsection{Transfer Function of the Low-Rank Denoiser}\label{LASYM}
		Following the discussions on to the sparse denoiser, the empirical MSEs of in the low-rank denoiser are given by
		
		\BS 
		\begin{align}   	
			\hat{v}_{\boldsymbol{L}\to p(\boldsymbol{L})}^{(t)}&\triangleq \frac{1}{n}\left\| \boldsymbol{L}_{\boldsymbol{L}\to p(\boldsymbol{L})}^{(t)}-\boldsymbol{L}\right\|_{F}^2,\\
			\hat{v} _{p(\boldsymbol{L})\to\boldsymbol{L}}^{(t)}&\triangleq \frac{1}{n}\left\| \boldsymbol{L}_{p(\boldsymbol{L})\to\boldsymbol{L}}^{(t)}-\boldsymbol{L}\right\|_{F}^2
			.
		\end{align} 
		\ES
		The empirical values of $a_{\boldsymbol{L}}^{(t)}$ and  $c_{\boldsymbol{L}}^{(t)}$  are given by
		\BS\label{empLac}
		\begin{align}
			\hat{a}^{(t)}_{\boldsymbol{L}} &\triangleq\frac{\left \langle \boldsymbol{L}^{(t)}_{\boldsymbol{L}\to p(\boldsymbol{L})} -\boldsymbol{L}, \boldsymbol{L}^{(t)}\right \rangle}{ \left\langle \boldsymbol{L}^{(t)}_{\boldsymbol{L}\to p(\boldsymbol{L})} -\boldsymbol{L}, \boldsymbol{L}^{(t)}_{\boldsymbol{L}\to p(\boldsymbol{L})}\right \rangle}\\ 
			\hat{c}^{(t)}_{\boldsymbol{L}}&\triangleq\frac{\left\langle\left(\boldsymbol{L}^{(t)}-\hat{a}_{\boldsymbol{L}}^{(t)} \boldsymbol{L}^{(t)}_{\boldsymbol{L}\to p(\boldsymbol{L})}\right) ,\boldsymbol{L}_{\boldsymbol{L}\to p(\boldsymbol{L})}^{(t)}\right\rangle}{\left\lVert \boldsymbol{L}^{(t)}-\hat{a}_{\boldsymbol{L}}^{(t)} \boldsymbol{L}^{(t)}_{\boldsymbol{L}\to p(\boldsymbol{L})}\right\rVert^2_{F}}.
		\end{align}
		\ES
		From Assumption \ref{as2}, we have  
		\BE\label{modL}
		\begin{aligned}
			\boldsymbol{L}_{\delta\to\boldsymbol{L}}^{(t)}=\boldsymbol{L}+\sqrt{v_{\delta\to\boldsymbol{L}}^{(t)}} \N_{\boldsymbol{L}}^{(t)}
		\end{aligned}
		\EE
		where the entries of $\N_{\boldsymbol{L}}^{(t)}\in\mathbb{R}^{n_1\times n_2}$ are i.i.d. with zero mean and unit variance and $\boldsymbol{L}\in\mathbb{R}^{n_1\times n_2}$ is a matrix with rank $r$. $v_{\delta\to\boldsymbol{L}}^{(t)}$ is a finite positive real number.
		In the following, we give the asymptotic analysis of the low-rank denoisers, as $n_1,n_2\to \infty$ with $\frac{n1}{n_2}=\beta$.
		\begin{lemma}[Asymptotics for low-rank denoiser]\label{LMSEL}
			Under Assumptions \ref{as2} and \ref{as3}, as $n_1\to\infty$, the following holds:
			\begin{enumerate}[1)]  	
				\item
				The empirical posteriori variance asymptotically satisfies
				
				\BE
				\begin{aligned}\label{empvL}
					\frac{1}{n}\left\| \D_{\boldsymbol{L}}\left( \boldsymbol{L}_{\boldsymbol{L}\to p(\boldsymbol{L})}^{(t)}\right) -\boldsymbol{L}\right\|_{F}^2\overset{p}{=}\operatorname{MSE}_{\boldsymbol{L}}(v_{\delta\to\boldsymbol{L}}^{(t)})
				\end{aligned}
				\EE
				where $\operatorname{MSE}_{\boldsymbol{L}}(v_{\delta\to\boldsymbol{L}}^{(t)})$ is given in the following
			    \BE
				\begin{aligned}\label{MSEL}	\operatorname{MSE}_{\boldsymbol{L}}(v_{\delta\to\boldsymbol{L}}^{(t)})=v_{\delta\to\boldsymbol{L}}^{(t)}\left(  \mathbb{E} \left[ \left( f_{\boldsymbol{L}}(\!\sigma\!)-\!\sigma\!\right) ^2\right] +2(1-\beta)\mathbb{E}\left[  g_{\boldsymbol{L}}(\sigma)\right] + 2\beta \mathbb{E}_{\sigma,\hat{\sigma}}\left[ H_{\boldsymbol{L}}(\sigma,\hat{\sigma})\right]-1\right)  .
				\end{aligned}
				\EE
			 Here, $g_{\boldsymbol{L}}$:  $[0,+\infty)\rightarrow[0,+\infty)$ with $g_{\boldsymbol{L}}(x)=\frac{f_{\boldsymbol{L}}(x)}{x}$ for $x>0$ and $g_{\boldsymbol{L}}(0)=0$. $H_{\boldsymbol{L}}(x,y)$ is defined as 	
				\BE\label{func_h}
				\begin{aligned}
					H_{\boldsymbol{L}}(x,y)\triangleq\left\{
					\begin{array}{cl}
						0                           & (x,y)=(0,0)\\
						\frac{xf_{\boldsymbol{L}}'(x)+f_{\boldsymbol{L}}(x)}{2x}      & x=y\neq 0\\
						\frac{xf_{\boldsymbol{L}}(x)-yf_{\boldsymbol{L}}(y)}{x^2-y^2}. & x \neq y\\
					\end{array}\right.
				\end{aligned}
				\EE

				\item
				\BE
				\begin{aligned}
					\hat{a}^{(t)}_{\boldsymbol{L}}
					\overset{p}{=}a_{\boldsymbol{L},\infty}^{(t)}\quad \text{and}\quad
					\hat{c}^{(t)}_{\boldsymbol{L}}
					\overset{p}{=} c_{\boldsymbol{L},\infty}^{(t)}	
				\end{aligned}
				\EE
				where $\hat{a}^{(t)}_{\boldsymbol{L}}$ and $\hat{c}^{(t)}_{\boldsymbol{L}}$ are defined in (\ref{empLac}), and
				\BS\label{asyAC}
				\begin{align}
					a_{\boldsymbol{L},\infty}^{(t)}&=(1-\beta)\mathbb{E} g_{\boldsymbol{L}}(\sigma) +\beta\mathbb{E}\left[ H_{\boldsymbol{L}}(\sigma,\hat{\sigma}) \right]  \\
					c_{\boldsymbol{L},\infty}^{(t)}	&=\frac{\mathbb{E}\left[ \left(  f_{\boldsymbol{L}}(\sigma)-a_{\boldsymbol{L},\infty}^{(t)}\sigma\right)\sigma\right]  }{\mathbb{E}\left[ \left(  f_{\boldsymbol{L}}(\sigma)-a_{\boldsymbol{L},\infty}^{(t)}\sigma\right)^2 \right] }.
				\end{align}
				\ES
				\item The transfer function of the low-rank denosier is given by
				\BE\label{lowrankasy}
				\begin{aligned}
					\hat{v}_{p(\boldsymbol{L})\to\boldsymbol{L}}^{(t)}  \overset{p}{=} \phi(v_{\delta\to\boldsymbol{L}}^{(t)})
				\end{aligned}
				\EE
				where
				\BE\label{LowrankSE}
				\begin{aligned}
					\phi(v_{\delta\to\boldsymbol{L}}^{(t)})&\overset{p}{=}\frac{\mathbb{E}\left[  \sigma^2\right] \mathbb{E}\left[  f_{\boldsymbol{L}}^2(\sigma)\right] } {\mathbb{E}\left[ \left( f_{\boldsymbol{L}}(\sigma)-a_{\boldsymbol{L},\infty}^{(t)}\sigma\right)^2 \right] }\! v_{\delta\to\boldsymbol{L}}^{\!(t)\!}
				-\frac{
						\!\left(\mathbb{E}\left[ \sigma f_{\boldsymbol{L}}(\sigma)\right] \right)^2} {\mathbb{E}\left[ \left( f_{\boldsymbol{L}}(\sigma)-a_{\boldsymbol{L},\infty}^{(t)}\sigma\right)^2\right]  \!}\! v_{\delta\to\boldsymbol{L}}^{\!(t)\!}
				-v_{\delta\to\boldsymbol{L}}^{(t)}.
				\end{aligned}
				\EE
				\item\BS\label{asLext}
				\begin{align}
					&\left \langle \boldsymbol{L}^{(t)}_{p(\boldsymbol{L})\to\boldsymbol{L}}-\boldsymbol{L},\boldsymbol{L}^{(t)}_{\boldsymbol{L}\to p(\boldsymbol{L})}-\boldsymbol{L}\right \rangle\overset{p}{=}0\\
					&\left \langle \boldsymbol{L}^{(t)}_{p(\boldsymbol{L})\to\boldsymbol{L}}-\boldsymbol{L},\boldsymbol{L}\right \rangle\overset{p}{=}- \phi(v_{\delta\to\boldsymbol{L}}^{(t)})
				\end{align}
				\ES
				where the expectation is take over two independent and identical distributed variables,  $\sigma$ and $\hat{\sigma}$.
				Here $\sigma\sim \mu_{\boldsymbol{L},v_{\boldsymbol{L}\to p(\boldsymbol{L})}^{(t)}}$ with $\mu_{\boldsymbol{L},v_{\boldsymbol{L}\to p(\boldsymbol{L})}^{(t)}}(\sigma)=2v_{\boldsymbol{L}\to p(\boldsymbol{L})}^{(t)}\sigma\mu_{\M}(v_{\boldsymbol{L}\to p(\boldsymbol{L})}^{(t)}\sigma^2)$. And $\mu_{\M}$ is specified by the following Stieljes transform 
				\BE
				\begin{aligned}
					m(z):=\int\frac{1}{\lambda-z}d\mu_{\M}(\lambda)\notag
				\end{aligned}
				\EE
				where $m(z)$ is given by 
				\BE
				\begin{aligned}
				 m(z)
					=\int \frac{d\mu_{\boldsymbol{L}\boldsymbol{L}^T}(t)}{\frac{ t}{1+v_{\boldsymbol{L}\to p(\boldsymbol{L})}^{(t)}\beta m(z)}-(1+v_{\boldsymbol{L}\to p(\boldsymbol{L})}^{(t)}\beta m(z)) z+v_{\boldsymbol{L}\to p(\boldsymbol{L})}^{(t)}(1-\beta)}.
				\end{aligned}
				\EE
				Here, $\mu_{\boldsymbol{L}\boldsymbol{L}^T}$  is the limiting  empirical eigenvalue distribution of $\frac{1}{n_2}\boldsymbol{L}\boldsymbol{L}^{T}$.
			\end{enumerate}
		\end{lemma}
		\begin{proof}
			Proof see Appendix \ref{proof_LMSEL}.
		\end{proof}
		Note that SURE is also an unbiased estimator for the MSE of the singular value soft-thresholding denoiser, as demonstrated in \cite{candes2013unbiased}, even though the corresponding function $f_{\boldsymbol{L}}$ is Lipschitz but not continuously differentiable.		
  Furthermore, Lemmas \ref{Lip_lemma_f} and \ref{fLasy} also hold for SVST. Consequently, the proof of Lemma \ref{main_lemma_L} also holds for SVST, and the aforementioned results remain valid.

		To prove Lemma \ref{LMSEL}, we first obtain the asymptotic MSE transfer function of the low-rank denoiser with i.i.d. Gaussian input errors in the large-system limit based on the random matrix theory \cite{dozier2007empirical} in Lemma \ref{main_lemma_L}.a. For the non-Gaussian i.i.d. noise, we assume that the singular matrices of the low-rank matrix $\boldsymbol{L}$ follow the Haar distribution\footnote{Without imposing this assumption on the singular matrices, decomposing $\S$ and $\boldsymbol{L}$ from compressed measurements $\y$ is fundamentally ill-posed. In fact, there are various scenarios where the existence of a unique splitting of $\X$ into ``low-rank” and ``sparse” parts is not guaranteed. For instance, the low-rank matrix itself may exhibit significant sparsity, giving rise to challenges in identifiability.}, and impose some mild conditions on the empirical eigenvalue distribution of $\boldsymbol{L}$. We further show the uniformly Lipschitz property of the normalized MSE with respect to the singular values of the noise matrix in Lemma \ref{Lip_lemma_f}. Then, in  Lemma \ref{main_lemma_L}.c, we show that the asymptotic MSE under i.i.d. non-Gaussian noise is equal to that under i.i.d. Gaussian noise, by using the Marchenko–Pastur law \cite{marchenko1967distribution} from the random matrix theory and Lemma \ref{fLasy}. The detailed proof of Lemma 5 can be found in Appendix \ref{proof_LMSEL}.
		
		\subsection{Transfer Function of the Linear Denoiser}
		Before describing the dynamics of the linear denoiser in the large system, we give the definitions of the empirical MSEs:
		\BS 
		\begin{align}
			\hat{v}_{\X\to\y}^{(t)}&\triangleq  \frac{1}{n}\left\| \X_{\X\to\y}^{(t)}-\X \right\|_{F}^2\\
			\hat{v}_{\y\to\X}^{(t+1)}&\triangleq  \frac{1}{n}\left\| \X_{\y\to\X}^{(t+1)}-\X\right\|_{F}^2.		\end{align} 
		\ES
		
		The empirical values $a_{\X}^{(t)}$ and  $c_{\X}^{(t)}$  are given by
		\BS
		\begin{align}
			&\hat{a}_{\X}^{(t)}\!=\!\frac{\!\!\left\| \X\right\|_{F}^2\!\left \langle\! \X^{(t)}\!,\!\X^{(t)}_{\X\to \y}\!\right \rangle\!-\!\left \langle\! \X^{(t)}\!,\!\X\!\right \rangle\left \langle\! \X^{(t)}_{\X\to \y}\!,\!\X\!\right \rangle}{n	v_{\X\to\y}^{(t)}\left\| \X\right\|_{F}^2\!-\!\left( \left\| \X\right\|_{F}^2-\left \langle \X^{(t)}_{\X\to \y},\X\right\rangle\right)^2 }\\  
			&\hat{c}_{\X}^{(t)}\!=\!\left( \frac{ \left\langle  \X^{(t)}, \X\right \rangle-\hat{a}_{\X}^{(t)}\left \langle \X^{(t)}_{\X\to \y},\X\right \rangle}{\left\| \X\right\|_{F}^2}\right) ^{-1}.
		\end{align}
		\ES 
		Note that $\hat{a}_{\X}^{(t)}$ and $\hat{c}_{\X}^{(t)}$ are intractable in practice since both of them involve $\X$. However,  $\hat{a}_{\X}^{(t)}$ and $\hat{c}_{\X}^{(t)}$ are proved almost surely convergence to a finite number as $n\to\infty$. Hence, we can approximately evaluate the parameters by the asymptotic results in the iterations.
		\begin{lemma}[Asymptotics for linear denoiser]\label{L8}
			Under Assumptions \ref{as1}-5, the following relations hold  for each iteration t:
			\begin{enumerate}[1)] 
				\item For the LMMSE denoiser,
				\BS   	 \label{LMMSEasy} 
				\begin{align} 			  			
					\hat{a}_{\X}^{(t)}&\overset{p}{=}1 \!-v^{(t)}_{\!\boldsymbol{X} \rightarrow \boldsymbol{y}\!}\mathbb{E}\left[ \!\frac{\!\theta_{\A\A^T}\!}{\!\theta_{\A\A^T}\!v^{(t)}_{\!\boldsymbol{X} \rightarrow \boldsymbol{y}\!}\!+\!\sigma_{\boldsymbol{n}}^2\!} \right] \\  		\hat{c}_{\X}^{(t)}&\overset{p}{=}\frac{1}{v^{(t)}_{\!\boldsymbol{X} \rightarrow \boldsymbol{y}\!}\mathbb{E}\left[ \!\frac{\!\theta_{\A\A^T}\!}{\!\theta_{\A\A^T}\!v^{(t)}_{\!\boldsymbol{X} \rightarrow \boldsymbol{y}\!}\!+\!\sigma_{\boldsymbol{n}}^2\!} \right] }\\    			\hat{v}_{\y\to\X}^{(t+1)}&\overset{p}{=}\psi(v^{(t)}_{\boldsymbol{X} \rightarrow \boldsymbol{y}})
				\end{align}  		
				\ES
				where
				\BE\label{LMMSESE}
				\begin{aligned}
					\psi(v^{(t)}_{\boldsymbol{X} \rightarrow \boldsymbol{y}})=\frac{1}{\alpha\mathbb{E}\left[ \!\frac{\!\theta_{\A\A^T}\!}{\!\theta_{\A\A^T}\!v^{(t)}_{\!\boldsymbol{X} \rightarrow \boldsymbol{y}\!}\!+\!\sigma_{\boldsymbol{n}}^2\!} \right] }-v^{(t)}_{\boldsymbol{X} \rightarrow \boldsymbol{y}}.\\
				\end{aligned}
				\EE
				Here, the expectation is taken over the  limiting empirical eigenvalue distribution  $p(\theta_{\A\A^T})$.
				\item 
				\BS
				\begin{align}
					&\left \langle \X^{(t+1)}_{\y\to\X}-\X,\X^{(t)}_{\X\to \y}-\X\right \rangle\overset{p}{=}0\\
					&\left \langle \X_{\y\to\X}^{(t)} -\X,\X\right \rangle\overset{p}{=}0.
				\end{align}
				\ES
			\end{enumerate}	
		\end{lemma}
		\begin{proof}
			Proof see Appendix \ref{proof_L8}.
		\end{proof}

		\begin{corollary}
			Under Assumptions \ref{as1}-5,	suppose $\mathcal{A}$ is a partial Haar operator. Then   
			\BE  	
			\begin{aligned} 	\hat{v}_{\y\to\X}^{(t+1)}\overset{p}{=}\left( \frac{1}{\alpha}-1\right) v^{(t)}_{\boldsymbol{X} \rightarrow \boldsymbol{y}}+\frac{1}{\alpha}\sigma_{\boldsymbol{n}}^2.
			\end{aligned}
			\EE
			
		\end{corollary}
In practice,  the discrete Fourier transform (DFT) and the discrete cosine transform (DCT) are commonly used in the construction of the compressing matrix $\A$.  Despite the fact that the DFT and DCT matrices are subsamples of special orthogonal matrices, the transfer function derived above for the partial Haar operator is also empirically applicable to the random compressing matrix $\boldsymbol{A}$ generated by uniformly and randomly selecting rows from a DCT or DFT matrix in  the large system limits.

		\subsection{SE for the Overall Algorithm}
		In the last three subsections, we have analyzed the state evolution for each parts separately. In the following, we combine the SE of the three parts and analysis the  SE of the overall algorithm based on Assumptions 1-5. We will replace $v$ with $\tau$ to represent the variance of the messages in the following. 
		\begin{theorem}[State evolution]\label{the1}
			Under Assumptions 1-5, as $ n\to\infty$, the asymptotic MSE evolution of the ITMP satisfies the following equations:
			\BS\label{ov}
			\begin{align}
				\tau^{(t )}_{\y\to\X}\ &=\ \psi(\tau^{(t-1)}_{p(\S)\to\S}+\tau^{(t-1)}_{p(\boldsymbol{L})\to\boldsymbol{L}})\\
				\tau^{(t)}_{p(\S)\to\S}\ &=\ \varphi(\tau^{(t)}_{\y\to\X}+\tau^{(t-1)}_{p(\boldsymbol{L})\to\boldsymbol{L}})\\
				\tau^{(t)}_{p(\boldsymbol{L})\to\boldsymbol{L}}\ &=\ \phi(\tau^{(t)}_{\y\to\X}+\tau^{(t-1)}_{p(\S)\to\S}).
			\end{align}
			\ES
			where the general form of $\varphi$ and $\phi$ are respectively defined in (\ref{SparseSE}) and (\ref{LowrankSE}) and  $\psi$ is specific to the form (\ref{LMMSESE}) for LMMSE denosier .
		\end{theorem}
		\begin{proof}
			Under Assumptions 1 and 3, as $n\to\infty$, we obtain \eqref{sparseasy} by Lemma \ref{asymS}. Under Assumptions 2 and 4, as $n\to\infty$, we obtain \eqref{lowrankasy} by Lemma \ref{LMSEL}. We rewrite them as
			\BS
			\begin{align}
				\hat{\tau}_{p(\S)\to\S}^{(t)}\overset{p}{=}\varphi(\tau_{\delta\to\S}^{(t)})\\
				\hat{\tau}_{p(\boldsymbol{L})\to\boldsymbol{L}}^{(t)}  \overset{p}{=} \phi(\tau_{\delta\to\boldsymbol{L}}^{(t)})			
			\end{align} 
			\ES
			Recall that the message from $\delta$ to node $\X$ is constructed as follows:
			\BE
			\begin{aligned}
				\hat{\X}^{(t)}_{\delta\to \X}=\hat{\boldsymbol{L}}^{(t)}_{\boldsymbol{L}\to \delta}+\hat{\S}^{(t)}_{\S\to \delta},\quad
				\hat{\tau}^{(t)}_{\delta\to \X}=\hat{\tau}^{(t)}_{\boldsymbol{L}\to \delta}+\hat{\tau}^{(t)}_{\S\to \delta}.
			\end{aligned} 
			\EE 
			Thus, we obtain  
			\BE
			\begin{aligned}
			\left \langle \X^{(t)}_{\X\to \y} -\X, \X\right \rangle
				=\left\langle \boldsymbol{L}_{p(\boldsymbol{L})\to\boldsymbol{L}}^{(t)} -\boldsymbol{L}, \boldsymbol{L}+\S\right \rangle+\left\langle \S_{p(\S)\to\S}^{(t)} -\S, \S+\boldsymbol{L}\right \rangle\overset{p}{=} -\phi(\tau_{\delta\to\boldsymbol{L}}^{(t)})-\varphi(\tau_{\delta\to\S}^{(t)}),
			\end{aligned}
			\EE
			by (\ref{as3ext}b) of Lemma \ref{asymS} and (\ref{asLext}b)  Lemma \ref{LMSEL}, where $\boldsymbol{L}_{p(\boldsymbol{L})\to\boldsymbol{L}}^{(t)} -\boldsymbol{L}$ and $\S_{p(\S)\to\S}^{(t)} -\S$ are separately independent with $\S$ and $\boldsymbol{L}$.
			Then 
			\BE
			\begin{aligned}
				\hat{\tau}_{\X\to\y}^{(t)}\overset{p}{=}\phi(\tau_{\delta\to\boldsymbol{L}}^{(t)})+\varphi(\tau_{\delta\to\S}^{(t)}).
			\end{aligned}
			\EE
			Furthermore, by Lemma \ref{L8}, we obtain 
			\BE
			\begin{aligned}
				\hat{\tau}_{\y\to\X}^{(t+1)}\overset{p}{=}\psi(\phi(\tau_{\delta\to\boldsymbol{L}}^{(t)})+\varphi(\tau_{\delta\to\S}^{(t)})).
			\end{aligned}
			\EE
			Note that the messages from $\delta$ to node $\boldsymbol{L}$ and $\S$ are respectively constructed as 
			\BS
			\begin{align}
				&\hat{\S}^{(t+1)}_{\delta\to \S}=\hat{\X}^{(t+1)}_{\X\to \delta}-\hat{\boldsymbol{L}}^{(t)}_{\boldsymbol{L}\to \delta},\quad
				\hat{\tau}^{(t+1)}_{\delta\to \S}=\hat{\tau}^{(t+1)}_{\X\to \delta}+\hat{\tau}^{(t)}_{\boldsymbol{L}\to \delta}\\
				&\hat{\boldsymbol{L}}^{(t+1)}_{\delta\to \boldsymbol{L}}=\hat{\X}^{(t+1)}_{\X\to \delta}-\hat{\S}^{(t)}_{\S\to \delta},\quad
				\hat{\tau}^{(t+1)}_{\delta\to \boldsymbol{L}}=\hat{\tau}^{(t+1)}_{\X\to \delta}+\hat{\tau}^{(t)}_{\S\to \delta}.
			\end{align}
			\ES
			Thus, we obtain 
			\BE
			\begin{aligned}
				\hat{\tau}^{(t+1)}_{\delta\to \boldsymbol{L}}\overset{p}{=}\psi(\phi(\tau_{\delta\to\boldsymbol{L}}^{(t)})+\varphi(\tau_{\delta\to\S}^{(t)}))+\varphi(\tau_{\delta\to\S}^{(t)})\\
				\hat{\tau}^{(t+1)}_{\delta\to \S}\overset{p}{=}\psi(\phi(\tau_{\delta\to\boldsymbol{L}}^{(t)})+\varphi(\tau_{\delta\to\S}^{(t)}))+\varphi(\tau_{\delta\to\boldsymbol{L}}^{(t)}).
			\end{aligned}
			\EE
			So far, we complete the proof.
		\end{proof}
		Note that the proposed SE of the algorithm is developed  based on a semi-closed loop proof with Assumption \ref{as1} and \ref{as2}, comparing to the analysis in \cite{bayati2011dynamics,berthier2020state}.
		Nevertheless, the SE precisely forecasts the evolution of MSE as shown in Fig. \ref{simMSE}. The analysis of MSE provides meaningful insights, particularly for non-separable spectral functions in the presence of non-Gaussian noise.	 
		Moreover, the convergence analysis will be developed in the next section based on the proposed SE iteration.
		
		As $t\to\infty$,  the fixed point $(\tau_{\boldsymbol{S}}^{*},\tau_{\boldsymbol{L}}^{*})$ is given by
		\BS\label{ov1}
		\begin{align}
			&\tau_{\boldsymbol{S}}^{*}=\varphi\left(\psi(\tau_{\boldsymbol{L}}^{*}+\tau_{\boldsymbol{S}}^{*})+ \tau_{\boldsymbol{L}}^{*}\right) \\
			&\tau_{\boldsymbol{L}}^{*}=\phi\left( \psi(\tau_{\boldsymbol{L}}^{*}+\tau_{\boldsymbol{S}}^{*})+ \tau_{\boldsymbol{S}}^{*}\right).
		\end{align}
		\ES
		
		Clearly, the fixed point is a function of $\{\alpha, \beta, \gamma, \rho, p_{\S}(S), \mu_{\boldsymbol{L}}(\sigma), p(\theta_{\A\A^T}), \sigma_{\boldsymbol{n}}^2, \mathcal{D}_{\S},\mathcal{D}_{\boldsymbol{L}},\mathcal{D}_{\X}\}$. For given $\{\alpha, \beta, \gamma, \rho, p_{\S}(S)$, $ \mu_{\boldsymbol{L}}(\sigma), p(\theta_{\A\A^T})$, $\sigma_{\boldsymbol{n}}^2$, $\mathcal{D}_{\S},\mathcal{D}_{\boldsymbol{L}},\mathcal{D}_{\X}\}$, we say that the algorithm successfully recover $\boldsymbol{L}$ and $\S$, when $\operatorname{MSE}_{\S}\left(\psi(\tau_{\boldsymbol{L}}^{*}+\tau_{\boldsymbol{S}}^{*})+ \tau_{\boldsymbol{L}}^{*}\right)$ and $\operatorname{MSE}_{\boldsymbol{L}}\left( \psi(\tau_{\boldsymbol{L}}^{*}+\tau_{\boldsymbol{S}}^{*})+ \tau_{\boldsymbol{S}}^{*}\right)$ are smaller than a predetermined value.
		For some fixed parameters of $\{p_{\S}(s)$,$ p_{\boldsymbol{L}}(\theta)$, $p(\theta_{\A\A^T})$, $\sigma_{\boldsymbol{n}}^2$, $\mathcal{D}_{\S}$,$\mathcal{D}_{\boldsymbol{L}}$,$\mathcal{D}_{\X}\}$, if we choose some proper value of $\alpha$, $\beta$, $\gamma$ and $\rho$, the algorithm will stably recover $\S$ and $\boldsymbol{L}$. This is verified in the next section.
		
		\section{Convergence Analysis of State Evolution}\label{Sec5}
		Recall that in Theorem \ref{the1}, we obtain the SE  for the proposed algorithm. In this section, We follow the analysis in \cite{ma2019optimization} and characterize the convergence the SE.
		To this end, we formally define the variance evolution sequence as follows. 
		
		\begin{definition} Let $\tau_{\boldsymbol{L}}^{(t)}=v^{(t)}_{p(\boldsymbol{L})\to\boldsymbol{L}}$ , $\tau_{\S}^{(t)}=v^{(t)}_{p(\S)\to\S}$ and $\tau_{\X}^{(t)}=v^{(t)}_{\boldsymbol{X} \rightarrow \boldsymbol{y}}$.  The recursion starts from  $(\tau_{\boldsymbol{\boldsymbol{L}}}^{(0)},\tau_{\boldsymbol{\S}}^{(0)})\in\mathbb{R}_{+}\times\mathbb{R}_{+}$, and we get the sequences $\left\{\tau_{\boldsymbol{L}}^{(t)}\right\}_{t\geq 1}$ and $\left\{\tau_{\S}^{(t)}\right\}_{t\geq 1}$ by: 
			\BS\label{Eqn:SE_original}
			\begin{align}
				&\tau_{\boldsymbol{S}}^{(t)}=\varphi\left(\psi(\tau_{\boldsymbol{L}}^{(t-1)}+\tau_{\boldsymbol{S}}^{(t-1)})+ \tau_{\boldsymbol{L}}^{(t-1)}\right) \\
				&\tau_{\boldsymbol{L}}^{(t)}=\phi\left( \psi(\tau_{\boldsymbol{L}}^{(t-1)}+\tau_{\boldsymbol{S}}^{(t-1)})+ \tau_{\boldsymbol{S}}^{(t-1)}\right),
			\end{align}
			\ES
			where $\psi$, $\varphi$ and $\phi$ follow the definitions  in Theorem \ref{the1}.
		\end{definition} 
		For the purpose of convergence analysis, we focus on situation that the linear operator $\mathcal{A}$ is a partial Haar operator. Then, the transfer function $\psi$  for LMMSE denoiser  is given by 
		\BE\label{linSE}
		\begin{aligned}
			\psi(v_{\X\to\y}^{(t)})=v_{\X\to\y}^{(t)}\left( \frac{1}{\alpha}-1\right) +\frac{1}{\alpha}\sigma_{\boldsymbol{n}}^2
		\end{aligned} 
		\EE
		where $0<\alpha \leq 1$ is the undersampling ratio.

		Moreover, we make a few assumptions about the two scalar functions $\varphi(x)$ and $\phi(x)$. Our numerical results in Section \ref{Sec6} suggest these assumptions are mild and are satisfied for commonly used denoisers.
		
		\begin{assumption}\label{asphi}
			$\varphi(x)$ and $\phi(x)$ are continuous and strictly monotonically increasing on $[0,\infty)$ with $\varphi(0)=0$ and $\phi(0)=0$. $\varphi'(x)$ and $\phi'(x)$ are continuous on $[0,\infty).$
		\end{assumption}

		Our goal is to understand the conditions under which the sequences $\{\tau_{\boldsymbol{L}}^{(t)}\}_{t\geq 1}$ and $\{\tau_{\S}^{(t)}\}_{t\geq 1}$ in \eqref{Eqn:SE_original} converge to zero, namely, the matrices $\S$ and $\boldsymbol{L}$ are accurately recovered. We first give a necessary condition for perfect reconstruction.
		
		\subsection{Necessary Condition for Global Convergence}    
		
		\begin{lemma}[Necessary condition  for global convergence]\label{neccon}
			Suppose that  Assumption \ref{asphi} holds. Let $\{\tau_{\boldsymbol{L}}^{(t)}\}_{t\geq 1}$ and $\{\tau_{\S}^{(t)}\}_{t\geq 1}$ be generated as in \eqref{Eqn:SE_original}. If 
			\[
			\lim_{t\to0}\tau_{\boldsymbol{L}}^{(t)}=0\quad\text{and}\quad\lim_{t\to0}\tau_{\S}^{(t)}=0,
			\]
			for any $(\tau_{\boldsymbol{L}}^{(0)},\tau_{\boldsymbol{S}}^{(0)})\in\mathbb{R}_{+}\times\mathbb{R}_{+}$, then necessarily
			\BE\label{Eqn:necessary}
			\alpha > \alpha_{nec}:=\sup_{x >0}\ 	\frac{\varphi(x)+\phi(x)}{\varphi(x)+\phi(x)+x}.
			\EE
		\end{lemma}
		\begin{proof}
			See Appendix \ref{pneccon}.
		\end{proof}

		Note that both $\varphi(x)$ and $\phi(x)$ are nonnegative. Hence, $\alpha_{nec}\le1$ and the necessary condition in \eqref{Eqn:necessary} is not vacuous.
		Furthermore, $\alpha_{nec}$ can be lower bounded by
		\BE\label{Eqn:nec_lower}
		\begin{aligned}
			\alpha_{nec}\geq\max\{\alpha_{\S},\alpha_{\boldsymbol{L}}\}
		\end{aligned}
		\EE
		where 
		\BS
		\begin{align}
			&\alpha_{\S}:= \sup_{x>0}\ \frac{\varphi(x)}{\varphi(x)+x},\\
			&\alpha_{\boldsymbol{L}}:= \sup_{x>0}\ \frac{\phi(x)}{\phi(x)+x}.
		\end{align}
		\ES
		It is straightforward to show that $\alpha_{\S}$ is exactly the threshold of $\alpha$ for the global convergence of the iteration $\tau_{\boldsymbol{S}}^{(t)}=\varphi\left(\left( \frac{1}{\alpha}-1\right)\tau_{\boldsymbol{S}}^{(t-1)}\right)$, and $\alpha_{\boldsymbol{L}}$ is the threshold for the global convergence of the iteration $\tau_{\boldsymbol{L}}^{(t)}=\phi\left(\left( \frac{1}{\alpha}-1\right)\tau_{\boldsymbol{L}}^{(t-1)}\right)$. Hence, \eqref{Eqn:nec_lower} means that the number of measurements needed for exactly recovering both $\boldsymbol{L}$ and $\S$ is larger than that needed for recovering $\boldsymbol{L}$ or $\S$ individually, which is consistent with our intuition.

		\subsection{Sufficient Condition for Global Convergence}    
		
		Before we analyze the convergence of SE, we first define three thresholds of $\alpha$ which are useful to understand the fixed points of \eqref{Eqn:SE_original}: 

		\BE
		\begin{split}
			\alpha_1&:= \sup_{x>0}\ \frac{\varphi'(x)}{\varphi'(x)+1},\\[3pt]
			\alpha_2&:= \sup_{x>0}\ \frac{\phi'(x)}{\phi'(x)+1},\\[3pt]
			\alpha_3&:=\inf\{\alpha\in(0,1); x\geq\Psi_1(\Psi_2(x,\alpha),\alpha), \forall 0<x<\sup\varphi\},
		\end{split}
		\EE
		where $\Psi_1$ and $\Psi_2$ are defined in \eqref{dePsi}.  By Assumption \ref{asphi}, there exist a $0<\hat{x}<x$ such that $\phi'(\hat{x})x=\phi(x)$. Then, we obtain  $\alpha_1\geq\alpha_{\S}$ by definition. Similarly, we obtain $ \alpha_2\geq\alpha_{\boldsymbol{L}}$.  As it will be clear in our later discussions, these thresholds are related to the uniqueness of the fixed point for the SE equations \eqref{Eqn:SE_original}; see Lemma \ref{propertyvarphi}-Lemma \ref{alphaglobal}.

		We next present a sufficient condition for the global convergence of $\{\tau_{\boldsymbol{L}}^{(t)}\}_{t\geq 1}$ and $\{\tau_{\S}^{(t)}\}_{t\geq 1}$ based on $\alpha_1$, $\alpha_2$ and $\alpha_3$.

		\begin{theorem}[Sufficient condition for global convergence]\label{TH3}  Suppose that Assumption \ref{asphi} holds. Let $\{\tau_{\boldsymbol{L}}^{(t)}\}_{t\geq 1}$ and $\{\tau_{\S}^{(t)}\}_{t\geq 1}$ be generated as in \eqref{Eqn:SE_original}. Consider the noiseless model where $\sigma_{\boldsymbol{n}}^2=0$. 
			If $\alpha>\max\{\alpha_1,\alpha_2,\alpha_3\}$, then the following holds for any $\tau_{\boldsymbol{L}}^{(0)}\ge0$ and $\tau_{\boldsymbol{S}}^{(0)}\ge0$:
			\BE
			\begin{aligned}
				\lim _{t \rightarrow \infty}\tau_{\S}^{(t)}=0 \quad \text { and } \quad \lim _{t \rightarrow \infty} \tau_{\boldsymbol{L}}^{(t)}=0.
			\end{aligned}
			\EE
		\end{theorem}
		
		\begin{proof}
			Proof see Appendix \ref{proof_TH3}.
		\end{proof}

		\subsection{Proof of  Theorem \ref{TH3} }
		Our main goal is to study the dynamics of the  iterations 
		\BS\label{SEoverall}
		\begin{align}
			&\tau_{\boldsymbol{S}}^{(t)}=\varphi\left(\left( \frac{1}{\alpha}-1\right)\tau_{\boldsymbol{S}}^{(t-1)}+ \frac{1}{\alpha}\tau_{\boldsymbol{L}}^{(t-1)}\right) \\
			&\tau_{\boldsymbol{L}}^{(t)}=\phi\left( \left( \frac{1}{\alpha}-1\right)\tau_{\boldsymbol{L}}^{(t-1)}+ \frac{1}{\alpha}\tau_{\boldsymbol{S}}^{(t-1)}\right).
		\end{align}
		\ES
		
		First of all, we hope the dynamic system will converge to the solutions of the following equations:
		
		\BS\label{pfSE}
		\begin{align}
			&\tau_{\boldsymbol{S}}^{*}=\varphi\left(\left( \frac{1}{\alpha}-1\right)\tau_{\boldsymbol{S}}^{*}+ \frac{1}{\alpha}\tau_{\boldsymbol{L}}^{*}\right) \\
			&\tau_{\boldsymbol{L}}^{*}=\phi\left( \left( \frac{1}{\alpha}-1\right)\tau_{\boldsymbol{L}}^{*}+ \frac{1}{\alpha}\tau_{\boldsymbol{S}}^{*}\right).
		\end{align}
		\ES
		Before studying the dynamics of the SE, we first analyze the fixed point of (\ref{pfSE}a) and (\ref{pfSE}b). 
		
		When $\alpha$ is small, there may exist multiple fixed points. As $\alpha$ is large enough, both  (\ref{pfSE}a) and   (\ref{pfSE}b)  only have one globally attractive fixed point. This will be shown in  Lemmas \ref{propertyvarphi}.
		
	Note that based on Assumption \ref{asphi},  $\varphi$ and $\phi$ are strictly monotonically increasing functions on $(0,+\infty)$. Hence, there inverse functions exist. For convenience, we define the inverse of function $\Psi_{1}$ and $\Psi_{2}$ as 
		\BE\label{dePsi}
		\begin{aligned}
			&\Psi_{1}^{-1}(\tau_{\S},\alpha)\triangleq\alpha \varphi^{-1}(\tau_{\S})-(1-\alpha)\tau_{\S}\\			
			&\Psi^{-1}_{2}(\tau_{\boldsymbol{L}},\alpha) \triangleq\alpha \phi^{-1}(\tau_{\boldsymbol{L}})-(1-\alpha)\tau_{\boldsymbol{L}}.
		\end{aligned}
		\EE
		where $\varphi^{-1}(\tau_{\S})$ and  $\phi^{-1}(\tau_{\boldsymbol{L}})$ are defined in a certain domain. Specifically,
		\BE
		\begin{aligned}
			&\operatorname{dom} \varphi^{-1} \quad= \quad [0, \sup \varphi]\\
			&\operatorname{dom} \phi^{-1} \quad= \quad [0, \sup \phi].
		\end{aligned}
		\EE
		We need to restrict to 
		\BE
		\begin{aligned}
			\mathcal{R} \triangleq\left\{\left(\tau_{\S}, \tau_{\boldsymbol{L}}\right) : 0<\tau_{\S}\leq\sup\varphi,0<\tau_{\boldsymbol{L}}\leq\sup\phi\right\}.
		\end{aligned}
		\EE
		After a single iteration, for any initialization, 
		\BE
		\begin{aligned}
			\left(\tau_{\S}^{(t)}, \tau_{\boldsymbol{L}}^{(t)}\right)\in \mathcal{R}, \quad \forall t\geq 1
		\end{aligned}
		\EE
		With these definitions, we can equivalently write \eqref{pfSE} as
		\BS \label{fptSL}
		\begin{align} 		
			\tau_{\S}^{*}&=\Psi_{1}(\tau_{\boldsymbol{L}}^*,\alpha),\\ 		
			\tau_{\boldsymbol{L}}^{*}&=\Psi_{2}(\tau_{\S}^*,\alpha).
		\end{align} 	
		\ES 
		We will sometimes write $\Psi_{1}(\tau_{\S},\alpha)$ and $\Psi_{2}(\tau_{\boldsymbol{L}},\alpha)$ as $\Psi_{1}(\tau_{\S})$ and $\Psi_{2}(\tau_{\boldsymbol{L}})$ when $\alpha$ is fixed.

		\begin{figure}[htbp] 
		\centering	\includegraphics[width=0.5\linewidth]{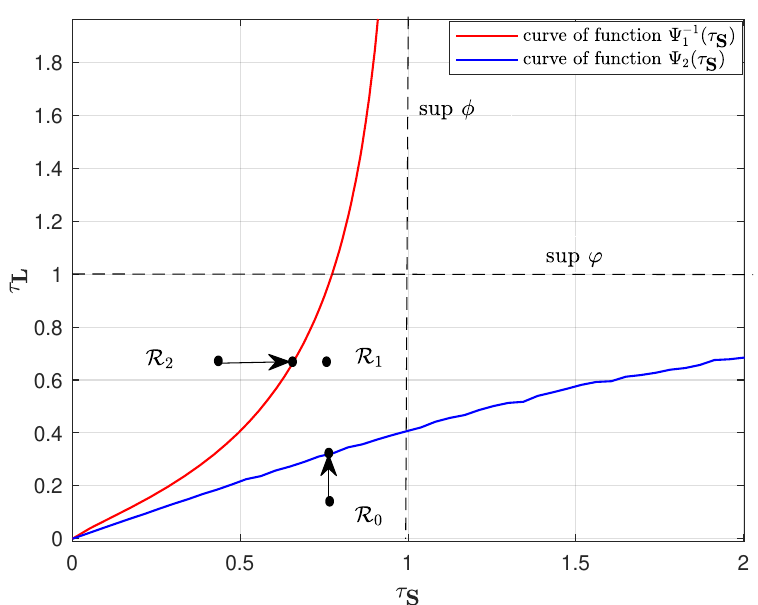} 
			\caption{Illustration of the three regions in (\ref{regdi}) with   $\alpha>\max\{\alpha_1,\alpha_2,\alpha_3\}$ and the generated $(\hat{\tau_{\S}}^{(t_{0})},\hat{\tau_{\boldsymbol{L}}}^{(t_{0})})$ from $(\tau_{\S}^{(t_{0})},\tau_{\boldsymbol{L}}^{(t_{0})})$ in the three regions $\mathcal{R}_0$, $\mathcal{R}_1$ and $\mathcal{R}_2$.   The \textbf{arrows}  represent the mapping defined in (\ref{map_109}). This function maps any initialization $\left(\tau_{\S}^{(t_0)}, \tau_{\boldsymbol{L}}^{(t_0)}\right)$ in $\mathcal{R}$ to a new initialization $\left(\hat{\tau}_{\S}^{(t_{0})}, \hat{\tau}_{\boldsymbol{L}}^{(t_{0})}\right)$ in $\mathcal{R}$.
				 Here, we choose MMSE estimator for the sparse denoiser and best-rank-$r$ estimator for the low rank denoiser. The entries of $\S$ are independently generated by following a Gaussian-Bernoulli distribution with unit variance and zero mean with $\rho=0.35$. The matrix $\boldsymbol{L}$ is generated by the multiplication of two independent zero mean Gaussian matrix of sizes $n_1\times r$ and $r\times n_2$ with $\gamma=0.1$. Then, $\boldsymbol{L}$ is normalized by $\frac{1}{n}\|\boldsymbol{L}\|^2_F=1$. $\alpha=0.8$.} 
			\label{Psifi1}
		\end{figure}

		The importance of the condition $\alpha>\max\{\alpha_1,\alpha_2,\alpha_3\}$ in Theorem \ref{TH3} can be explained as follows. First, we prove in Lemma \ref{propertyvarphi}  that both  (\ref{pfSE}a) and (\ref{pfSE}d)
		only exist one globally attractive fixed point when $\alpha>\max\{\alpha_1,\alpha_2\}$; see Lemma \ref{propertyvarphi} for definitions of attractive fixed point. Furthermore, (\ref{fptSL}a) and (\ref{fptSL}b) are all strictly increasing functions by Lemma \ref{propertyvarphi}. We further show in Lemma \ref{alphaglobal} that $(\tau_{\S},\tau_{\boldsymbol{L}})=(0,0)$ is the only fixed points of (\ref{pfSE}) when $\alpha>\max\{\alpha_1,\alpha_2,\alpha_3\}$.

		It turns out that, under the same condition, the SE dynamics converges to this fixed point from any starting point, which we shall prove in the rest of this section.
		
		Define, for given $t_{0}\geq1$,
		\BE\label{map_109}
		\begin{aligned}
			&\hat{\tau}_{\S}^{(t_{0})}\triangleq \max \left( \tau_{\S}^{(t_{0})}, \Psi_{1}(\tau_{\boldsymbol{L}}^{(t_{0})})\right) \\
			&\hat{\tau}_{\boldsymbol{L}}^{(t_{0})}\triangleq \max \left( \tau_{\boldsymbol{L}}^{(t_{0})}, \Psi_{2}(\tau_{\S}^{(t_{0})})\right).
		\end{aligned}
		\EE
		We generated $\left\{\hat{\tau}_{\boldsymbol{L}}^{(t)}\right\}_{t\geq t_{0}}$ and $\left\{\hat{\tau}_{\S}^{(t)}\right\}_{t\geq t_{0}}$ by (\ref{Eqn:SE_original}) from initialization $(\hat{\tau}_{\S}^{(t_{0})},\hat{\tau}_{\boldsymbol{L}}^{(t_{0})})$. {Our goal is to prove} 
		\BE
		\begin{aligned}
			&\hat{\tau}_{\S}^{(t_{0})}\geq 	\hat{\tau}_{\S}^{(t_{0}+1)}\geq	\hat{\tau}_{\S}^{(t_{0}+2)}\cdots\\
			&\hat{\tau}_{\boldsymbol{L}}^{(t_{0})}\geq 	\hat{\tau}_{\boldsymbol{L}}^{(t_{0}+1)}\geq	\hat{\tau}_{\boldsymbol{L}}^{(t_{0}+2)}\cdots.
		\end{aligned}
		\EE
		Namely,
		\BE
		\begin{aligned}
			\lim_{t\to\infty}\hat{\tau}_{\S}^{(t)}=\hat{\tau}_{\S}^{*}\\
			\lim_{t\to\infty}\hat{\tau}_{\boldsymbol{L}}^{(t)}=\hat{\tau}_{\boldsymbol{L}}^{*}
		\end{aligned}
		\EE
	Note that by Assumption \ref{asphi}, 
		\BS
		\BE
		\begin{aligned}
			\frac{ \partial \varphi\left(\left( \frac{1}{\alpha}-1\right)\tau_{\boldsymbol{S}}+ \frac{1}{\alpha}\tau_{\boldsymbol{L}}\right) }{ \partial \tau_{\boldsymbol{S}} }	>0,	  \frac{ \partial \varphi\left(\left( \frac{1}{\alpha}-1\right)\tau_{\boldsymbol{S}}+ \frac{1}{\alpha}\tau_{\boldsymbol{L}}\right) }{ \partial \tau_{\boldsymbol{L}} }	>0
		\end{aligned}
		\EE
		and 
		\BE
		\begin{aligned}
			\frac{ \partial \phi\left(\left( \frac{1}{\alpha}-1\right)\tau_{\boldsymbol{S}}+ \frac{1}{\alpha}\tau_{\boldsymbol{L}}\right) }{ \partial \tau_{\boldsymbol{S}} }	>0,	  \frac{ \partial \phi\left(\left( \frac{1}{\alpha}-1\right)\tau_{\boldsymbol{S}}+ \frac{1}{\alpha}\tau_{\boldsymbol{L}}\right) }{ \partial \tau_{\boldsymbol{L}} }	>0.
		\end{aligned}
		\EE
		\ES
		Then, for any $t\geq t_0$
		\BE
		\begin{aligned}
			&\hat{\tau}_{\S}^{(t)}\geq 	\tau_{\S}^{(t)}\\
			&\hat{\tau}_{\boldsymbol{L}}^{(t)}\geq 	\tau_{\boldsymbol{L}}^{(t)}.
		\end{aligned}
		\EE
		So far, we only need to prove $\hat{\tau}_{\S}^{*}=\hat{\tau}_{\boldsymbol{L}}^{*}=0$, which implies
		\BE
		\begin{aligned}
			&\lim_{t\to\infty}\tau_{\S}^{(t)}=0\\
			&\lim_{t\to\infty}\tau_{\boldsymbol{L}}^{(t)}=0.
		\end{aligned}
		\EE
		
		For this, we divide $\mathcal{R}$ into the following three sub-regions:
		\BE\label{regdi}	
		\begin{aligned} 	
			&\mathcal{R}_{0} \triangleq\left\{\left(\tau_{\S}, \tau_{\boldsymbol{L}}\right) | 0 < \tau_{\S}<\sup \phi, 0<\tau_{\boldsymbol{L}}<\Psi_{2}(\tau_{\S})\right\},\\ 		
			&\mathcal{R}_{1} \triangleq\left\{\left(\tau_{\S}, \tau_{\boldsymbol{L}}\right) | 0 < \tau_{\S}<\sup \phi  ,\Psi_{2}(\tau_{\S})\leq\tau_{\boldsymbol{L}}\leq\min(\Psi^{-1}_{1}(\tau_{\S}),\sup\varphi)\right\},\\ 		
			&\mathcal{R}_{2} \triangleq\left\{\left(\tau_{\S}, \tau_{\boldsymbol{L}}\right) | 0 < \tau_{\S}<\Psi_{1}(\tau_{\boldsymbol{L}}), \Psi_{1}^{-1}(\tau_{\S})<\tau_{\boldsymbol{L}}<\sup \phi\right\}.		 	
		\end{aligned}     
		\EE 
		Note that for any $	\left(\tau_{\S}^{(t_0)}, \tau_{\boldsymbol{L}}^{(t_0)}\right)\in \mathcal{R}$, the corresponding $(\hat{\tau}_{\S}^{(t_{0})},\hat{\tau}_{\boldsymbol{L}}^{(t_{0})})$ is in $\mathcal{R}_1$; see illustration in Fig. \ref{Psifi1}.
		Therefore, to conclude Theorem \ref{TH3}, it remains to show that  for any  $(\hat{\tau}_{\S}^{(t_0)},\hat{\tau}_{\boldsymbol{L}}^{(t_0)})\in\mathcal{R}_{1}$, the sequences $\left\{\hat{\tau}_{\boldsymbol{L}}^{(t)}\right\}_{t\geq t_{0}}$ and $\left\{\hat{\tau}_{\S}^{(t)}\right\}_{t\geq t_{0}}$ converge to $(0,0)$, as guaranteed by Lemma \ref{SECON1} below.
		
		\begin{lemma} \label{SECON1}
			Suppose that Assumption \ref{asphi} holds. Let $\left\{\hat{\tau_{\boldsymbol{L}}}^{(t)}\right\}_{t\geq t_{0}}$ and $\left\{\hat{\tau_{\S}}^{(t)}\right\}_{t\geq t_{0}}$ be generated as in \eqref{Eqn:SE_original}. Let $\alpha>\max\{\alpha_1,\alpha_2,\alpha_3\}$. Then 
			\begin{enumerate}[(i)] 	 	
				\item $(\hat{\tau}_{\S}^{(t)},\hat{\tau}_{\boldsymbol{L}}^{(t)})$ remains in $\mathcal{R}_{1}$ for all $t>t_{0}$ 	
				\item $(\hat{\tau}_{\S}^{(t)},\hat{\tau}_{\boldsymbol{L}}^{(t)})$ converges: 		
				\BE 	
				\begin{aligned} 		
					\lim _{t \rightarrow \infty} \hat{\tau}_{\S}^{(t)}=0 \quad \text { and } \quad \lim _{t \rightarrow \infty} \hat{\tau}_{\boldsymbol{L}}^{(t)}=0 .
				\end{aligned} 	
				\EE 
			\end{enumerate}  
		\end{lemma}  
		
		The proof of Lemma \ref{SECON1} can be found in Appendix \ref{Sec:Proof_Lemma9}.
		
		\section{Numerical Results }\label{Sec6}
		In this section, we present numerical experiments corroborating our main results. 
		The rank-$r$ matrix $\boldsymbol{L}$ is generated by the multiplication of two zero mean random Gaussian matrices of size $n_1\times r$ and $r\times n_2$, with $n = n_1 \times n_2$. Each entry of the sparse matrix $\S$ is independently generated by following a Gaussian-Bernoulli distribution with zero mean and unit variance, where $1-\rho$ is the probability of being non-zero in the distribution. We normalize $\boldsymbol{L}$ and $\S$ such that $\|\S\|_F^2=n$ and $\|\boldsymbol{L}\|_F^2=n$.
The parameter $\alpha$ is defined as $\alpha = \frac{m}{n}$, where $m$ represents the number of observations, i.e., the length of $\boldsymbol{y}$. $\alpha$ is the measurement rate of the compression matrix $\A$, which is an $m \times n$ matrix. 
		The matrix form $\A$ of linear operator $\mathcal{A}$ is generated by $\A=\boldsymbol{P}\boldsymbol{W}$,  where  $\boldsymbol{P}\in\mathbb{R}^{m\times n}$ is a random selection matrix that selects rows
		randomly from $\boldsymbol{W}$ and  $\boldsymbol{W}\in\mathbb{R}^{n\times n}$ is a  discrete cosine transform (DCT) matrix. 
		For the linear denoiser, we choose the LMMSE denoisr. Recall that the transfer function $\Psi(\cdot)$ is given by
		\BE
		\begin{aligned}
			\psi(v^{(t)}_{\boldsymbol{X} \rightarrow \boldsymbol{y}})=\left( \frac{1}{\alpha}-1\right) v^{(t)}_{\boldsymbol{X} \rightarrow \boldsymbol{y}}+\frac{1}{\alpha}\sigma_{\boldsymbol{n}}^2.
		\end{aligned}
		\EE
		For the sparse denoiser, we choose the  MMSE denoiser. The transfer function $\psi(\cdot)$ is given by
		\BE\label{transS}
		\begin{aligned}
			\varphi(v_{\delta\to\S}^{(t)})=	\left( \frac{1}{v_{\delta\to\S}^{(t)}} -\frac{1}{\operatorname{MSE}_{\S}(v_{\delta\to\S}^{(t)})}\right)^{-1}
		\end{aligned}
		\EE
		with
		\BE\label{sim_s_mmse}
		\begin{aligned}
			\operatorname{MSE}_{\S}(v_{\delta\to\S}^{(t)})=	\mathbb{E}\left[ \left(\mathbb{E}_S[S|S+\sqrt{v_{\delta\to\S}^{(t)}}N_{\S}] -S\right)^2 \right].
		\end{aligned}
		\EE
		For the low-rank denoiser, we choose best-rank-$r$ denoiser. Note that for the best-rank-$r$ denoiser, $f_{\boldsymbol{L}}$ can be regard as the singular value hard thresholding truncated at the normalized $r$-th singular value $\tilde{\sigma}_{r}^{(t)}$, i.e.,
		\BE\label{hard_exp}
		\begin{aligned}
			f_{\boldsymbol{L}}(x)=\operatorname{h}(x; \tilde{\sigma}_{r}^{(t)})= \begin{cases} 0 & \text { if } 0\leq x<\tilde{\sigma}_{r}^{(t)} \\ 
				x & \text {  if  } \tilde{\sigma}_{r}^{(t)}\leq x \end{cases}
		\end{aligned}
		\EE
		by identical functions $f_{\boldsymbol{L}}$ acting on all singular values $\sigma_{i}$. Let the cumulative distribution function of the empirical distribution of $\frac{1}{\sqrt{n_2v_{\boldsymbol{L}\rightarrow p(\boldsymbol{L})}^{(t)}}}\boldsymbol{L}_{\boldsymbol{L}\rightarrow p(\boldsymbol{L})}^{(t)}$ be denoted by $F_{\sigma,n_1}(x)$ 
		\BE
		\begin{aligned}F_{\sigma,n_1}(x)=\frac{1}{n_1} \sum_{i=1}^{n_1} \boldsymbol{1}_{\tilde{\sigma}_{r}^{(t)}\leq x}.
		\end{aligned}
		\EE
		As $n_1\to\infty$, $F_{\sigma,n_1}(x)$ convergence weakly to a deterministic function $F_{\sigma}(x)$ whenever $F_{\sigma}(x)$ is continuous at $x$. Then, let $\tilde{\sigma}^*$ be the largest value such that $F_{\sigma}(\tilde{\sigma}^*)=1-\gamma$.
		Note that the best-rank-r denoiser  is not Lipschitz continuous and has discontinuities. We employ the smoothing technique to construct a smooth denoiser   satisfies the requirements in Assumption \ref{as3}, which is proved closely enough to the original one in Lemma \ref{smooth_bestrankr}. Concretely, we smooth $\operatorname{h}(x; \tilde{\sigma}^*)$ by  convolving it with a bump function.
		For arbitrary radii $\epsilon>0$, define
		\BE
		\begin{aligned}
			B_{\epsilon}(x)=	\frac{B(\frac{x}{\epsilon})}{\int B(\frac{x}{\epsilon}) dx}.
		\end{aligned}
		\EE
		Let 
		\BE\label{def_smooth_h}
		\begin{aligned}
			\operatorname{h}_{\epsilon}(x, \tilde{\sigma}^*)=\operatorname{h}(x, \tilde{\sigma}^*)\star 	B_{\epsilon}(x).
		\end{aligned}
		\EE
		Similar technique have been adapted in \cite{10226297}\cite{DoesLZheng} with Gaussian kernel to smooth the point-wise denoisers in the message passing algorithms. In Lemma \ref{smooth_lemma}, we show that $\operatorname{h}_{\epsilon}(x, \tilde{\sigma}^*)$  satisfies the requirements in Assumption (\ref{as3}.vi). We show that the corresponding low-rank denoiser $D_{\boldsymbol{L},\epsilon}$ is closed to the best-rank-$r$ denoiser in Lemma \ref{smooth_bestrankr}. Its proof is shown in Appendix \ref{proof_smooth_bestrankr}.
		\begin{lemma}\label{smooth_bestrankr}
			Let $x,\sigma^*,\epsilon,v>0$ and $\epsilon<\sigma^*$. Let $\boldsymbol{L}_{\N}=\boldsymbol{L}+\sqrt{v}\N$.
			Suppose Assumption \ref{as2} hold.
			As $n\to\infty$,
			\BE
			\begin{aligned}
				\frac{1}{n}\left\| \D_{\boldsymbol{L},\epsilon}(\boldsymbol{L}_{\N})-\boldsymbol{L}\right\|^2_F\overset{p}{=} \operatorname{MSE}_{\boldsymbol{L},\epsilon}(v)
			\end{aligned}
			\EE
			where $\operatorname{MSE}_{\boldsymbol{L},\epsilon}(v)$ is given in (\ref{MSEL}) with $f_{\boldsymbol{L}}(x)=\operatorname{h}_{\epsilon}(x, \tilde{\sigma}^*)$. Here, $\tilde{\sigma}^*\triangleq \sup \{x: F_{\sigma}(x)=1-\gamma\} $ and $ F_{\sigma}(x)$ is cumulative dsitribution function of the limiting empirical singular value distribution of $\frac{1}{\sqrt{vn_2}}\boldsymbol{L}_{\N}$. Denote   the best-rank-$r$ denoiser by $\D_{\boldsymbol{L},b}$. Let $\epsilon_1,\epsilon_2 \cdots $ denote a decreasing sequence of numbers such that $\epsilon_i>0$ and $\epsilon_i\to\infty$. Then, $    	\lim_{i\to\infty}	\lim_{n\to\infty} \frac{1}{n}\left\| \D_{\boldsymbol{L},\epsilon_i}(\boldsymbol{L}_{\N})-\boldsymbol{L}\right\|^2_F$ weakly converge and 
			\BE\label{lemma_smooth_eqlim}
			\begin{aligned}
				\lim_{i\to\infty}	\lim_{n\to\infty} \frac{1}{n}\left\| \D_{\boldsymbol{L},\epsilon_i}(\boldsymbol{L}_{\N})-\boldsymbol{L}\right\|^2_F\overset{p}{=} \frac{1}{n}\left\|\D_{\boldsymbol{L},b}(\boldsymbol{L}_{\N})-\boldsymbol{L} \right\|
				^2_F.
			\end{aligned}
			\EE
			Moreover, the empirical parameters $\hat{a}^{(t)}_{\boldsymbol{L}} $  and $	\hat{c}^{(t)}_{\boldsymbol{L}} $  defined in  (\ref{empLac}) and $ \frac{1}{n}\left\| \boldsymbol{L}_{p(\boldsymbol{L})\to\boldsymbol{L}}^{(t)}-\boldsymbol{L}\right\|_{F}^2$ for best-rank-$r$ denoiser are all converge to their limits, which can be arbitrary approximated with the smooth denoiser sequence $\{\D_{\boldsymbol{L},\epsilon_i}\}$ as in (\ref{lemma_smooth_eqlim}).
		\end{lemma}
		
		Note that the analytical form obtain in (\ref{MSEL}) still requires the continuously differentiable property of $f_{\boldsymbol{L}}$. Consequently, (\ref{MSEL}) is not applicable to the best-rank-$r$ denoiser or the hard thresholding denoiser, despite their empirical MSE  proven convergence. However,  we can still employ Monte Carlo method to approximate the limits of $\frac{1}{n}\left\|\D_{\boldsymbol{L},b}(\boldsymbol{L}_{\N})-\boldsymbol{L} \right\|
		^2_F$ by SURE, which is an unbiased estimator proven in \cite{hansen2018stein}.

		\subsubsection{State evolution simulation}
		In this subsection, we verify the proposed state evolution of our algorithm. For this, we track $\frac{1}{n}\|\boldsymbol{L}^{(t)}_{\boldsymbol{L}\to\delta}-\boldsymbol{L}\|^2_F$ and $\frac{1}{n}\|\S_{\S\to\delta}^{(t)}-\S\|^2_F$ (referred to as ``simulation'') and the predicted MSEs by  (\ref{Svar}) and (\ref{varL}) (referred to as ``evolution''). 
		
		We choose the LMMSE denoiser given in Corollary \ref{coro_le1} for the linear module.  For the  sparse module, We observe that the empirical distribution of the error $\S_{\delta\to\S }^{(t)}-\S$ shown in Fig. \ref{qqplot1111} is close to Gaussian distribution. Then,
		we further assume that $N_{\S}$ follows Gaussian distribution in (\ref{sim_s_mmse}). 
		In our experiments,  $\hat{a}_{\S}^{(t)}$ for MMSE denoiser is approximated by $\frac{\hat{v}^{(t)}_{\S}}{v^{(t)}_{\S\to p(\S)}}$, with 
		\BE
		\begin{aligned}
			\hat{v}^{(t)}_{\S}=\frac{1}{n}\sum_{i,j}\mathop{\mathbb{E}}\limits_{S_{i,j}}\left\| S_{i,j}-\mathop{\mathbb{E}}\limits_{S_{i,j}}\left\lbrace S_{i,j}|\left( S_{\S\to p(\S)}^{(t)}\right)_{i,j} \right\rbrace  \right\|^2_F.
		\end{aligned}
		\EE
		According to Lemma \ref{smooth_bestrankr}, $\hat{a}_{\boldsymbol{L}}^{(t)}$ for the best-rank-$r$ denoiser is approximated by the Monte Carlo method in \cite{2008Monte}, where  
		\BE
		\begin{aligned}
			\hat{a}^{(t)}_{\boldsymbol{L}}=\frac{1}{n\boldsymbol{v}^{(t)}_{\boldsymbol{L}\to p(\boldsymbol{L})}}\left\langle\!\frac{\mathcal{D}\left(\!\boldsymbol{L}^{(t)}_{\boldsymbol{L}\to p(\boldsymbol{L})}\!+\!\epsilon \boldsymbol{N}\!\right)\!-\!\mathcal{D}\left(\!\boldsymbol{L}^{(t)}_{\boldsymbol{L}\to p(\boldsymbol{L})}\!\right)}{\epsilon}, \boldsymbol{N}\right\rangle
		\end{aligned}
		\EE
		Here, $\N$ is a random Gaussian matrix with zero
		mean and unit variance entries. $\epsilon$ is a small positive real number, which is take as $10^{-4}$ in our experiments.
		
		For the state evolution iteration, $\Psi(\cdot)$ is given by
		\BE
		\begin{aligned}
			\psi(v^{(t)}_{\boldsymbol{X} \rightarrow \boldsymbol{y}})=\left( \frac{1}{\alpha}-1\right) v^{(t)}_{\boldsymbol{X} \rightarrow \boldsymbol{y}}+\frac{1}{\alpha}\sigma_{\boldsymbol{n}}^2.
		\end{aligned}
		\EE
		The transfer function $\psi(\cdot)$ is given by
		\BE\label{transS}
		\begin{aligned}
			\varphi(v_{\delta\to\S}^{(t)})=	\left( \frac{1}{v_{\delta\to\S}^{(t)}} -\frac{1}{\operatorname{MSE}_{\S}(v_{\delta\to\S}^{(t)})}\right)^{-1}
		\end{aligned}
		\EE
		with
		\BE
		\begin{aligned}
			\operatorname{MSE}_{\S}(v_{\delta\to\S}^{(t)})=	\mathbb{E}\left[ \left(\mathbb{E}_S[S|S+\sqrt{v_{\delta\to\S}^{(t)}}N_{\S}] -S\right)^2 \right].
		\end{aligned}
		\EE
		The transfer function $\phi(v_{\delta\to\boldsymbol{L}}^{(t)})$ for best-rank-r denoiser is approximated   by using Monte Carlo method.
		Consider the following model
		\BE
		\begin{aligned}
			\boldsymbol{L}_{\boldsymbol{L}\to p(\boldsymbol{L})}=\boldsymbol{L}+\sqrt{v_{\delta\to\boldsymbol{L}}}\N
		\end{aligned}
		\EE
		where $\N$ is a random Gaussian matrix with zero
		mean and unit variance entries, and  $v_{\delta\to\boldsymbol{L}}$ is a positive real number.  
		\begin{figure}[htbp]
		\centering	\includegraphics[width=0.5\linewidth]{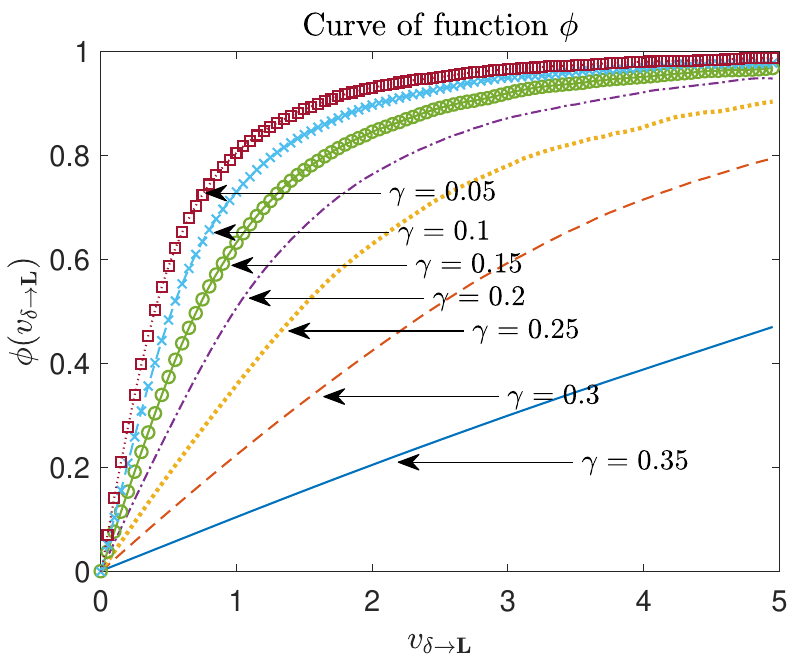}
			\caption{The curve $\phi(v_{\delta\to\boldsymbol{L}})$ obtained by averaging 50 runs with a precision of $0.01$ for the variance. The other settings are $n_1=n_2=2000$.}	
			\label{simvarphi}
		\end{figure}
		We obtain the curve of $\phi(v_{\delta\to\boldsymbol{L}})$ in the region $[0,5]$ with  $ \frac{1}{n}\left\| \hat{c}_{\boldsymbol{L}}\left( \mathcal{D}\left(\boldsymbol{L}_{\boldsymbol{L}\to p(\boldsymbol{L})}\right)-\hat{\alpha}_{\boldsymbol{L}}\boldsymbol{L}_{\boldsymbol{L}\to p(\boldsymbol{L})}\right) -\boldsymbol{L}\right\|^2_F$ by using the Monte Carlo method, which is shown in Fig. \ref{simvarphi}.
		Then, we are able to approximate $\phi(v_{\delta\to\boldsymbol{L}})$ by linear interpolation of the obtained data with arbitrary $v_{\delta\to\boldsymbol{L}}$.
		\begin{figure}[htbp]\begin{centering}
					\includegraphics[width=0.35\linewidth,height=3.2in]{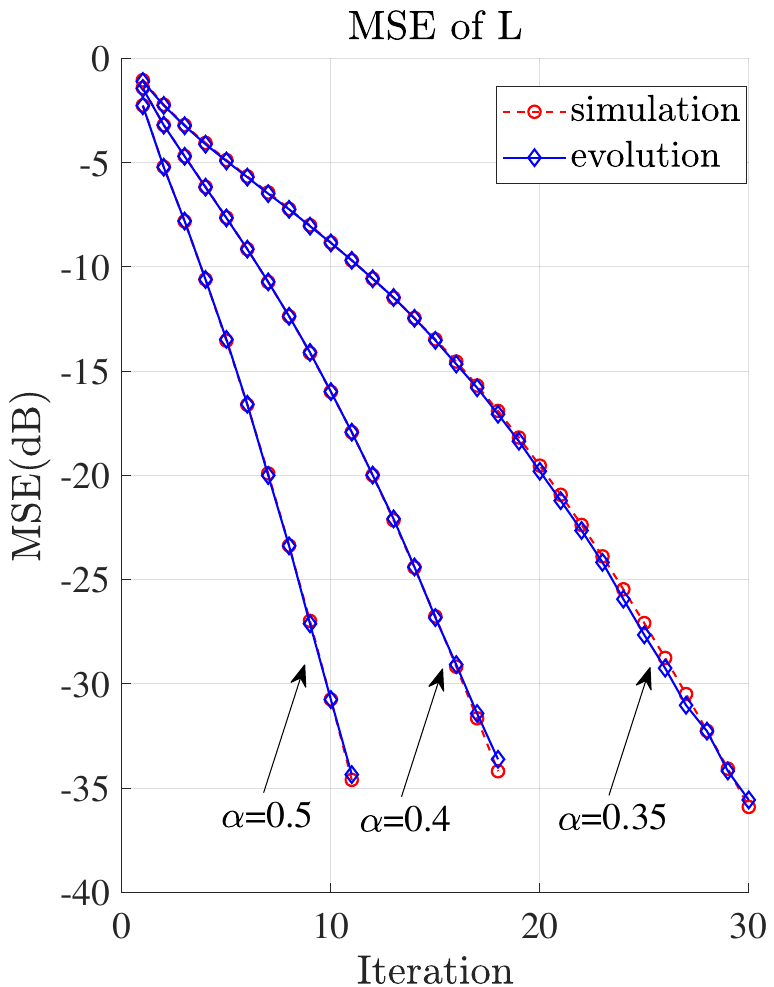}
				\includegraphics[width=0.35\linewidth,height=3.2in]{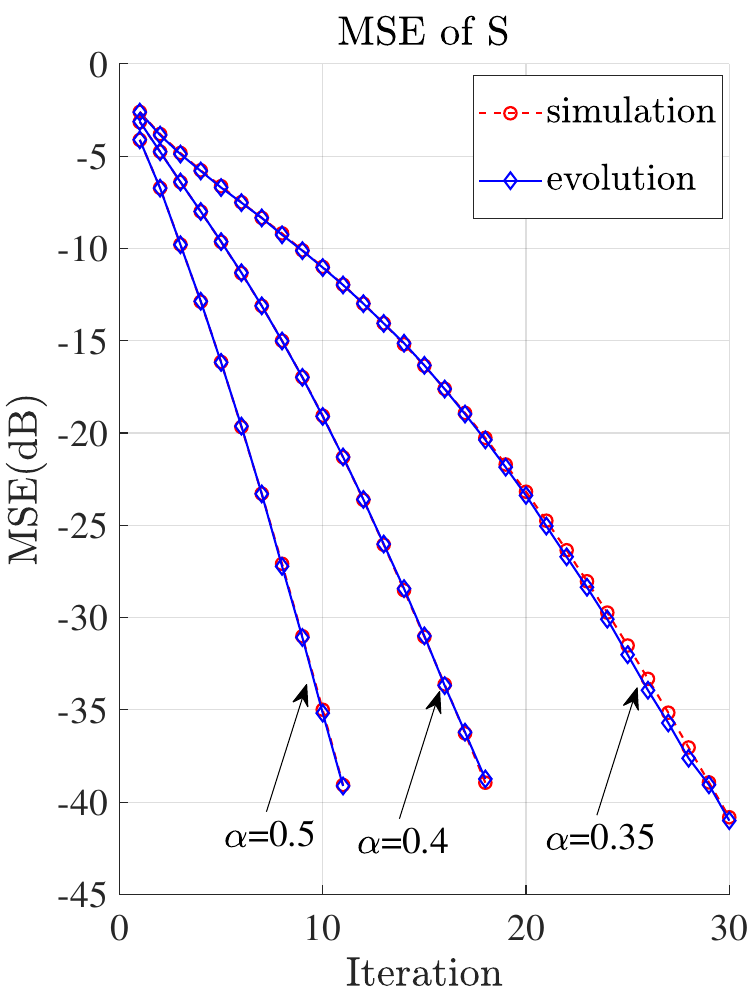}
				\caption{The left figure shows the dynamic changes of the estimation error of message $m_{\boldsymbol{L}\to\delta}^{(t)}$. The right figure shows the dynamic changes of the estimation error of message $m_{\S\to\delta}^{(t)}$. The settings are $n_1=n_2=1000, \gamma=0.05,\sigma_{\boldsymbol{n}}^2=0, \rho=0.05$ and different $\alpha$.}\label{simMSE}
			\end{centering}
		\end{figure}
		From Fig. \ref{simMSE}, we see that
		the state evolution iteration accurately predicts the variance in practical algorithm iteration.
		\subsubsection{Phase transition}
		Recall that Theorem \ref{TH3} proposes a sufficient condition for the global convergence of $\alpha$. The following experiments verify the accuracy of the boundary presented in Theorem \ref{TH3}.  Fig. \ref{img3} illustrates the phase transition phenomenon generated by varying $\rho$ and $\gamma$ under a given measurement rate $\alpha$, representing the boundaries between the \textit{success} and \textit{failure} of the algorithm. Both the grid intervals of the horizontal and vertical axes in the figure are set to $0.01$.
		 To numerically verify Theorem \ref{TH3}, we obtain the corresponding transfer functions ($\varphi(x)$ and $\phi(x)$) and the corresponding first-order derivative curves ($\varphi'(x)$ and $\phi'(x)$) for each pair $(\rho, \gamma)$ through numerical simulations:
		\BE
		\begin{split}
			\alpha_1&:= \sup_{x>0}\ \frac{\varphi'(x)}{\varphi'(x)+1},\\[3pt]
			\alpha_2&:= \sup_{x>0}\ \frac{\phi'(x)}{\phi'(x)+1},\\[3pt]
			\alpha_3&:=\inf\{\alpha\in(0,1); x\geq\Psi_1(\Psi_2(x,\alpha),\alpha), \forall 0<x<\sup\varphi\}.
		\end{split}
		\EE
		Furthermore, for each pair $(\rho, \gamma)$, we calculate $	\alpha_1, \alpha_2$ and $\alpha_3$ based on the numerical values of the transfer functions. If $\max\{\alpha_1, \alpha_2, \alpha_3\}$ is lower than the sampling rate $\alpha$, Theorem \ref{TH3} shows that the  recovery can be achieved.  For example, in the upper-left plot of Fig. \ref{img3}, if for a given pair $(\rho, \gamma)$ the corresponding $\max\{\alpha_1, \alpha_2, \alpha_3\}$ obtained through numerical simulations is less than $0.4$, then Theorem \ref{TH3} asserts that the algorithm will succeed.
		The \textit{blue} region in the figure represents the theoretically predicted successful recovery area, while the \textit{yellow} region represents the area represents the area where recovery may or may not be possible. To empirically verify whether the theory is correct, we run experiments with a set of $(\rho, \gamma)$ and check for successful recovery. To obtain the boundary points of the empirical phase transition, for a given $\rho$,  we conduct experiments to explore the maximum value of $\gamma$ that allows the algorithm to successfully recover.  The maximum parameters  $(\rho, \gamma)$ that can enable  successful experiments are indicated by the $*$ in the figure.
		
	The settings of Monte Carlo experiments are given as follows. The rank-$r$ matrix $\boldsymbol{L}$ is generated by the multiplication of two zero mean random Gaussian matrices of size $n_1\times r$ and $r\times n_2$.   Each entry of the sparse matrix $\S$ is independently generated by following a Gaussian-Bernoulli distribution with zero mean and unit variance, where $1-\rho$ is the probability of being non-zero in the distribution. We normalize $\boldsymbol{L}$ and $\S$ such that $\|\S\|_F^2=n$ and $\|\boldsymbol{L}\|_F^2=n$. 
		The matrix form $\A$ of linear operator $\mathcal{A}$ is generated by $\A=\boldsymbol{P}\boldsymbol{W}$,  where  $\boldsymbol{P}\in\mathbb{R}^{n\times n}$ is a random permutation matrix and  $\boldsymbol{W}\in\mathbb{R}^{n\times n}$ is a  discrete cosine transform (DCT) matrix. We consider the noiseless case, where $\sigma^2_{\boldsymbol{n}}=0$.
		We consider the algorithm to have successfully recovered $\boldsymbol{L}$ and $\S$ when the following conditions are met:
		1) the normalized mean square error (NMSE) for $\boldsymbol{L}$ and $\S$ is defined as $\frac{\left\| \boldsymbol{L}^{(t)} - \boldsymbol{L}\right\|^2_F}{\left\| \boldsymbol{L}\right\|^2_F} \leq 10^{-3}$ and $\frac{\left\| \S^{(t)} - \S\right\|^2_F}{\left\| \boldsymbol{L}\right\|^2_F} \leq 10^{-3}$, respectively;
		2) the iteration count $t$ is less than 100. For each grid point, the numerical values were tested in ten    Monte Carlo experiments. If all the ten trials are successful, we consider that the values at this grid point can successfully recover the desired signal.

		We illustrate the phase transition of the algorithm in Fig. \ref{img3}.
		\begin{figure}[htbp]
			\centering\includegraphics[width=0.6\linewidth]{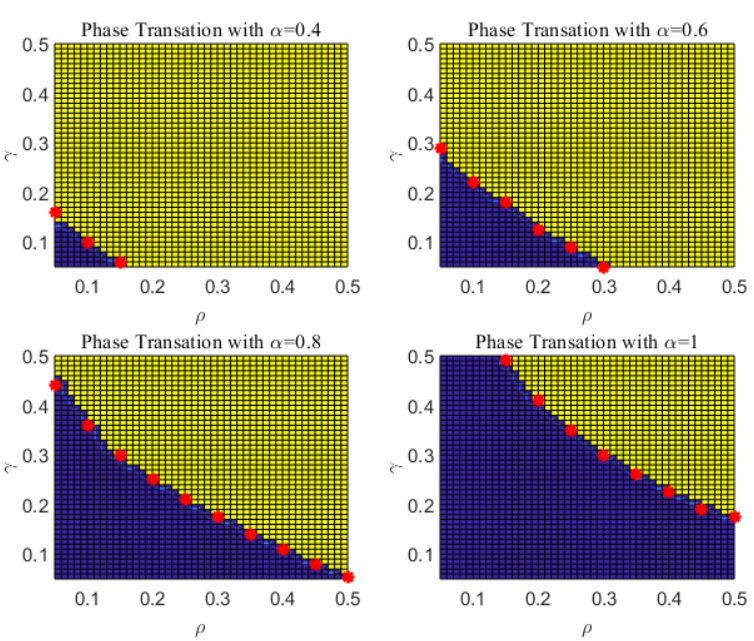}
			\caption{The phase transition curves of CRPCA algorithms with different $\alpha$. The  blue region corresponds to the situation that the algorithm successfully recovers $\boldsymbol{L}$ and $\S$   predicted by Theorem \ref{TH3}.  The 2D red coordinates $*$ represent the limit situation where the algorithm successfully recovers $\boldsymbol{L}$ and $\S$ through practical experiments.  
				At the 2D red coordinates $*$, increasing $\rho$ or $\gamma$ will lead to the failure of the algorithm.  }\label{img3}
		\end{figure}
	We observe that the sufficient boundary condition proposed in Theorem \ref{TH3} closely approximates the empirical phase transition observed in the experiments.

		\begin{figure}[htbp]
			\centering\includegraphics[width=0.6\linewidth]{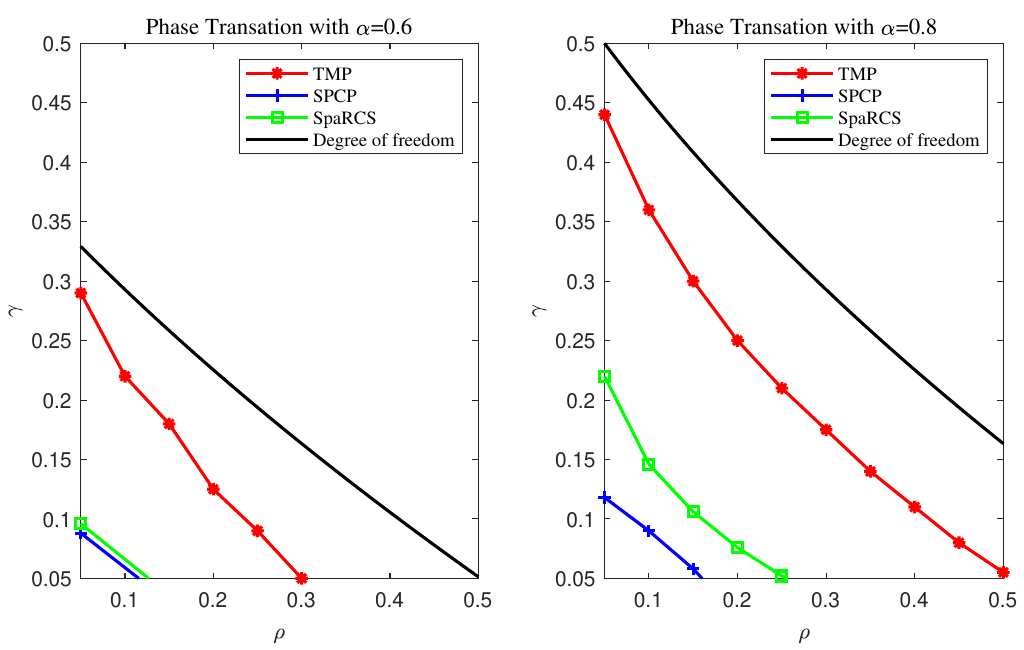}
			\caption{ The empirical phase transition curves of various CRPCA algorithm. $n_1=n_2=500, \sigma_{\boldsymbol{n}}^2=0$. The region upon each phase transition curve corresponds to the situation that the corresponding algorithm
				successfully recovers $\boldsymbol{L}$ and $\S$.}\label{phasecomp}
		\end{figure}
	
	Furthermore, we compare the phase transition of our algorithm with that of SPCP \cite{aravkin2014variational} and SpaRCS \cite{waters2011sparcs}, as illustrated in Fig. \ref{phasecomp}. According to discussions on the degrees of freedom of $(\boldsymbol{L},\S)$ in \cite{2013Compressive}, the minimal sampling ratio $\alpha$ required for successful recovery must satisfy $\alpha \geq \rho + (2-\gamma)\gamma$.
			We see that the phase transition curve of our algorithm is the closest to
	the upper bound specified by the degrees of freedom and considerably higher than those  of
	SpaRCS and SPCP.  
	
 Indeed, since the recovery of $\S$ necessitates the simultaneous reconstruction of its support set and the values on this support, a higher sampling rate is required than merely the sparse level $\rho$.	
	Consequently, it is acknowledged that there is a discernible gap between the theoretical degrees of freedom and the empirical phase transition curve of the CRPCA problem.
	
		\section{Conclusions}
		In this paper, we presented a new formulation for turbo message passing in the context of  the CRPCA problem. We  derived an asymptotic MSE transfer function of the low-rank denoiser with i.i.d. non-Gaussian input noise under certain regularity conditions.  Leveraging these results, we established the SE of the algorithm under some mild assumptions.  By analyzing the SE, we investigated the global convergence properties and identified the boundary of measurements required for reliable signal recovery. Numerical results demonstrated that the SE accurately tracks the MSEs of the denoisers in each iteration. Additionally, the predicted boundary based on the SE closely matches the empirical phase transition curve. 
		
		In conclusion, our proposed algorithm provides an effective approach for the  CRPCA problem, leveraging  turbo message passing and the SE framework. The SE analysis offers valuable insights into the algorithm's behavior and convergence properties. The phase transition curve comparison further highlights the efficacy of our algorithm.

		\section*{Acknowledgements}
		We would like to thank Rishabh Dudeja for his suggestion about the proof of Lemma \ref{fLasy}.

		\begin{appendices}

			\section{Proofs Related to the Algorithm Design}
			\subsection{Useful Lemma}
			We first present the following useful lemma applied in the proofs.
			\begin{lemma}\label{explin}
				Let $\boldsymbol{a}\in\mathbb{R}^{n}$ denote a random vector satisfies $\frac{1}{n}\mathbb{E}\left\|\boldsymbol{a}\right\|_{F}^2=c_1$ and $\boldsymbol{b}\in\mathbb{R}^{n}$ denote a random vector satisfies $\frac{1}{n}\mathbb{E}\left[ \boldsymbol{a}^{T}\boldsymbol{b}\right] =c_2$. Suppose $\textbf{M}\in\mathbb{R}^{n\times n}$ is a deterministic matrix. Let  $\frac{\operatorname{tr}(\M)}{n}= d$.   Let $\V\in\mathcal{O}{n}$ be a Haar matrix independent of $\boldsymbol{a}$, $\boldsymbol{b}$ and $\M$. Then
				\BE\begin{aligned}
					\frac{1}{n}\mathbb{E}\left[ \boldsymbol{a}^{T}\V^{T}\M\V \boldsymbol{a}\right] =c_1d; \quad \frac{1}{n}\mathbb{E}\left[ \boldsymbol{a}^{T}\V^{T}\M\V \boldsymbol{b}\right] =c_2d
				\end{aligned}
				\EE 
			\end{lemma}
			\begin{proof}
				\BS
				\begin{align}
					\frac{1}{n}\mathbb{E}\left[ \boldsymbol{a}^{T}\V^{T}\M\V \boldsymbol{a}\right] 
					&=\frac{1}{n}\mathbb{E}\left[ \left( \sum_{i,j}M_{i,j}\left( \sum_{k}a_{k}V_{i,k}\right)\left( \sum_{l}a_{l}V_{j,l}\right) \right)  \right]\\
					&=\frac{1}{n}\mathbb{E}\left[ \sum_{i=j}M_{i,j}\left( \sum_{k}a^2_{k}V^2_{i,k}\right) \right]=\frac{c_1}{n} \sum_{i=j}M_{i,j}
					=c_1 d.
				\end{align}
				\ES
				Similarly,
				\BE
				\begin{aligned}
					&\quad\frac{1}{n}\mathbb{E}\left[ \boldsymbol{b}^{T}\V^{T}\M\V \boldsymbol{a}\right] =c_2 d
				\end{aligned}
				\EE
			\end{proof}
			In the rest of this appendix, we provide the proofs  needed in Section \ref{Sec3}.
			\subsection{Proof of (\ref{Svar})}\label{proof_eq16}
			To start with, we have
			\BS\label{prle1}
			\begin{align}
			\mathbb{E}\left\langle \S_{p(\S)\rightarrow \S}^{(t)} ,\S_{\S\to p(\S)}^{(t)}\right\rangle
				&=\mathbb{E}\left\langle c^{(t)}_{\S}\left(\S^{(t)}-a_{\S}^{(t)} \S^{(t)}_{\S\to p(\S)}\right) ,\S_{\S\to p(\S)}^{(t)}\right\rangle\\
				&=\mathbb{E}\left\lVert c^{(t)}_{\S}\left( \S^{(t)}-a_{\S}^{(t)} \S^{(t)}_{\S\to p(\S)}\right) \right\rVert^2_{F}\\
				&=\mathbb{E}\left\lVert \S_{p(\S)\rightarrow \S}^{(t)} \right\rVert^2_{F}
			\end{align}
			\ES
			where (\ref{prle1}a) follows from the definition of $\S_{p(\S)\rightarrow \S}^{(t)}$ in (\ref{formS}), (\ref{prle1}b) follows from the definition of $c_{\S}^{(t)}$ in (\ref{Sac}b)  and (\ref{prle1}c) follows from the definition of $\S_{p(\S)\rightarrow \S}^{(t)}$ in (\ref{formS}).  From (\ref{turboprinciple}b) and Assumption \ref{as1}, we obtain 
			\BS\label{eq_pr_16}
			\begin{align}
			\mathbb{E}\left\langle \S_{p(\S)\rightarrow \S}^{(t)} ,\S_{\S\to p(\S)}^{(t)}\right\rangle
				&= \mathbb{E}\left\langle \S,\S_{\S\to p(\S)}^{(t)} \right\rangle+ \mathbb{E}\left\langle\S_{p(\S)\rightarrow \S}^{(t)},\S \right\rangle-\mathbb{E}\left\lVert \S \right\rVert^2_{F}\\
				&=\mathbb{E}\left\langle\S_{p(\S)\rightarrow \S}^{(t)},\S \right\rangle.
			\end{align}
			\ES
			Then,
			\BS\label{prle2}
			\begin{align}
		\mathbb{E}\left\langle \S_{p(\S)\rightarrow \S}^{(t)}-\S ,\S\right\rangle
				&=\mathbb{E}\left\langle \S_{p(\S)\rightarrow \S}^{(t)}-\S ,\S+\S_{\S\to p(\S)}^{(t)}-\S\right\rangle\\
				&=\mathbb{E}\left\langle \S_{p(\S)\rightarrow \S}^{(t)} ,\S_{\S\to p(\S)}^{(t)}\right\rangle-\mathbb{E}\left\lVert \S \right\rVert^2_{F}\\
				&=-\mathbb{E}\left\lVert \S_{p(\S)\rightarrow \S}^{(t)} \right\rVert^2_{F}+\!2\mathbb{E}\left\langle\! \S_{p(\S)\rightarrow \S}^{(t)} ,\S\!\right\rangle\!-\!\mathbb{E}\left\lVert \S \right\rVert^2_{F}\\
				&=-\mathbb{E}\left\lVert  \S_{p(\S)\rightarrow \S}^{(t)}-\S \right\rVert^2_{F}
			\end{align}
			\ES
			where (\ref{prle2}a) is from (\ref{turboprinciple}b), (\ref{prle2}b) is from  Assumption \ref{as1},  (\ref{prle2}c) is from (\ref{turboprinciple}b) and Assumption 1, (\ref{prle2}d) is from (\ref{prle1}) and (\ref{eq_pr_16}).  
			We are now ready to prove (\ref{Svar}). To this end, we have
			\BS\label{17pr1}
			\begin{align}
			\frac{1}{n}\mathbb{E}\left\|\S_{p(\S)\rightarrow \S}^{(t)}-\S \right\|_{F}^2
				&=\frac{1}{n}\mathbb{E}\left\|\S \right\| _{F}^2-\frac{1}{n}\mathbb{E}\left\langle \S^{(t)}_{\S\to p(\S)},\S^{(t)}_{p(\S)\rightarrow \S}\right\rangle\\
				&=\frac{1}{n}\mathbb{E}\left\|\S \right\| _{F}^2-c_{\S}^{(t)}\frac{1}{n}\mathbb{E} \left\langle \S^{(t)},\S^{(t)}_{\S\to p(\S)} \right\rangle
		+c_{\S}^{(t)} a_{\S}^{(t)}\left(\frac{1}{n}\mathbb{E}\left\|\S \right\| _{F}^2+v_{\S\to p(\S)}^{(t)} \right)
			\end{align}
			\ES
			where (\ref{17pr1}a) follows from (\ref{eq_pr_16}) and (\ref{prle2}), and (\ref{17pr1}b)  follows by substituting $\S^{(t)}_{p(\S)\rightarrow \S}$  in (\ref{formS}) and Assumption 1.
			\subsection{Proof of (\ref{varL})}\label{proof_var_exp}
			To start with, we obtain the following results for the low rank denoiser similarly to the derivation in \ref{proof_eq16}.
			\BS\label{relaL}
			\begin{align}
				&\mathbb{E}\left\langle \boldsymbol{L}_{p(\boldsymbol{L})\rightarrow \boldsymbol{L}}^{(t)} ,\boldsymbol{L}_{\boldsymbol{L}\to p(\boldsymbol{L})}^{(t)}\right\rangle=\mathbb{E}\left\lVert \boldsymbol{L}_{p(\boldsymbol{L})\rightarrow \boldsymbol{L}}^{(t)} \right\rVert^2_{F}\\
				&\mathbb{E}\left\langle \boldsymbol{L}_{p(\boldsymbol{L})\rightarrow \boldsymbol{L}}^{(t)}-\boldsymbol{L} ,\boldsymbol{L}\right\rangle=-\mathbb{E}\left\lVert  \boldsymbol{L}_{p(\boldsymbol{L})\rightarrow \boldsymbol{L}}^{(t)}-\boldsymbol{L}\right\rVert^2_{F}.
			\end{align}
			\ES
			Then, we start to prove (\ref{varL}).
			\BS\label{pr26}
			\begin{align}
			\frac{1}{n}\mathbb{E}\left\|\boldsymbol{L}_{p(\boldsymbol{L})\rightarrow \boldsymbol{L}}^{(t)}-\boldsymbol{L} \right\|_{F}^2
				&=\frac{1}{n}\mathbb{E}\left\|\boldsymbol{L}_{p(\boldsymbol{L})\rightarrow \boldsymbol{L}}^{(t)}-\boldsymbol{L}_{\boldsymbol{L}\rightarrow p(\boldsymbol{L})}^{(t)} \right\|_{F}^2-\frac{1}{n}\mathbb{E}\left\|\boldsymbol{L}_{\boldsymbol{L}\rightarrow p(\boldsymbol{L})}^{(t)}-\boldsymbol{L}\right\|_{F}^2\\
				&=\frac{1}{n}\mathbb{E}\left\|\boldsymbol{L}_{\boldsymbol{L}\rightarrow p(\boldsymbol{L})}^{(t)} \right\|_{F}^2-\frac{1}{n}\mathbb{E}\left\|\boldsymbol{L}_{p(\boldsymbol{L})\rightarrow \boldsymbol{L}}^{(t)} \right\|_{F}^2-v_{\boldsymbol{L}\rightarrow p(\boldsymbol{L})}^{(t)}\\
				&=-\frac{\mathbb{E}\left\langle\left(\boldsymbol{L}^{(t)}-a_{\boldsymbol{L}}^{(t)} \boldsymbol{L}^{(t)}_{\boldsymbol{L}\to p(\boldsymbol{L})}\right) ,\boldsymbol{L}_{\boldsymbol{L}\to p(\boldsymbol{L})}^{(t)}\right\rangle^2}{\mathbb{E}\left\lVert \boldsymbol{L}^{(t)}-a_{\boldsymbol{L}}^{(t)} \boldsymbol{L}^{(t)}_{\boldsymbol{L}\to p(\boldsymbol{L})}\right\rVert^2_{F}}+\frac{1}{n}\mathbb{E}\left\|\boldsymbol{L}_{\boldsymbol{L}\rightarrow p(\boldsymbol{L})}^{(t)} \right\|_{F}^2 -v_{\boldsymbol{L}\rightarrow p(\boldsymbol{L})}^{(t)}\\
				&=\frac{\mathbb{E}\!\left[ \left\|\!\boldsymbol{L}^{(t)}_{\boldsymbol{L}\to p(\boldsymbol{L})}\!\right\|_{F}^{2}\!\left\|\!\boldsymbol{L}^{(t)}\!\right\|_{F}^{2}\!\!-\!\left\langle\! \boldsymbol{L}^{(t)}_{\boldsymbol{L}\to p(\boldsymbol{L})},\boldsymbol{L}^{(t)}\!\right\rangle^2 \!\right] }{n\mathbb{E}\left[\left\| \boldsymbol{L}^{(t)}- a_{\boldsymbol{L}}^{(t)} \boldsymbol{L}^{(t)}_{\boldsymbol{L}\to p(\boldsymbol{L})} \right\|_{F}^2\right] } -\! v_{\boldsymbol{L}\to p(\boldsymbol{L})}^{(t)}
			\end{align}
			\ES
			where (\ref{pr26}a) follows from (\ref{relaL}b) and Assumption \ref{as2}, (\ref{pr26}b)   from the fact of (\ref{relaL}a), (\ref{pr26}c) follows from (\ref{Sac}b), and (\ref{pr26}d) is obtained by some straightforward manipulations.
			\subsection{Proof of Lemma 1}\label{prl1}
			We first prove (\ref{lem11}). 
			From Assumptions 1 and 2, we obtain
			\BS\label{orthlemma1}
			\begin{align}
				&\mathbb{E}\left\langle \S_{p(\boldsymbol{L})\rightarrow \S}^{(t)}-\S ,\boldsymbol{L}\right\rangle=0\\
				&\mathbb{E}\left\langle \boldsymbol{L}_{p(\boldsymbol{L})\rightarrow \boldsymbol{L}}^{(t)}-\boldsymbol{L} ,\S\right\rangle=0.
			\end{align}
			\ES
			Then, 
			\BS\label{prooflemma11}
			\begin{align}
			\mathbb{E}\left\langle \X_{\X\rightarrow \y}^{(t)}-\X ,\X\right\rangle
				&=\mathbb{E}\left\langle \S_{p(\S)\rightarrow \S}^{(t)}-\S ,\S\right\rangle+\mathbb{E}\left\langle \boldsymbol{L}_{p(\boldsymbol{L})\rightarrow \boldsymbol{L}}^{(t)}-\boldsymbol{L} ,\boldsymbol{L}\right\rangle\\
				&=-\mathbb{E}\left\lVert  \S_{p(\S)\rightarrow \S}^{(t)}-\S \right\rVert^2_{F}-\mathbb{E}\left\lVert  \boldsymbol{L}_{p(\boldsymbol{L})\rightarrow \boldsymbol{L}}^{(t)}-\boldsymbol{L}\right\rVert^2_{F}\\
				&=-v_{\delta\to\X}^{(t)}
			\end{align}
			\ES
			where (\ref{prooflemma11}a) follows from the definition of $\X_{\X\rightarrow \y}^{(t)}$ in (\ref{mx}a) and (\ref{orthlemma1}),  (\ref{prooflemma11}b) follows from (\ref{prle2}) and (\ref{relaL}b), and (\ref{prooflemma11}c) follows from (\ref{mx}b).
			We next prove (\ref{lem12}). To this end, we have 
			\BS\label{orthin}
			\begin{align}
				\quad\mathbb{E}\left\langle\S_{\S\rightarrow p(\S)}^{(t+1)}-\S ,\S\right\rangle
				&=\mathbb{E}\left\langle\S^{(t)}_{\S\to \delta}-\S ,\S\right\rangle+\mathbb{E}\left\langle \M\left( \A\operatorname{vec}\left( \boldsymbol{n} \right)\right),\operatorname{vec}(\S) \right\rangle
			-\mathbb{E}\left\langle \M\left( \A\operatorname{vec}\left(  \S^{(t)}_{\S\to \delta}-\S\right)\right),\operatorname{vec}(\S) \right\rangle\notag\\
				&\quad -\mathbb{E}\left\langle \M\left( \A\operatorname{vec}\left(  \boldsymbol{L}^{(t)}_{\boldsymbol{L}\to \delta}-\boldsymbol{L}\right)\right),\operatorname{vec}(\S) \right\rangle\\
				&=0
			\end{align}
			\ES
			where (\ref{orthin}b) is obtained by substituting (\ref{LE1}), and (\ref{orthin}c) is from Assumptions 1 and 2 as well as the independence of between $\S$ and $\boldsymbol{n}$. 
			Similarly, we obtain 
			$\mathbb{E}\left\langle\boldsymbol{L}_{\boldsymbol{L}\to p(\boldsymbol{L})}^{(t+1)}-\boldsymbol{L} ,\boldsymbol{L}\right\rangle=0$, which completes the proof.
			\subsection{Proof of Lemma \ref{lemss}}\label{prl2}
			We first present some useful relations in MMSE denoiser.
			The correlations between $\S_{\S \rightarrow p(\S)}^{(t)}$, $\S^{(t)}$ and $\S$ are given by 
			\BS
			\begin{align}
				&\mathbb{E}\left\langle\S_{\S \rightarrow p(\S)}^{(t)}, \S^{(t)}\right\rangle 
				=\mathbb{E}\left\| \S\right\|_{F}^2\label{espripost}\\
				&\mathbb{E}\left\langle\S, \S^{(t)}\right\rangle=\sum_{i, j} \mathop{\mathbb{E}} \left[ \mathop{\mathbb{E}}\limits_{S_{i,j}} \left[ S_{i,j} | (S_{\S \rightarrow p(\S)}^{(t)})_{i,j}\right]^2 \right] \label{esspost}.
			\end{align}
			\ES
			And, the MSE of the MMSE denoiser is given by 
			\BE\label{mmsevar1}
			\begin{aligned}\mathbb{E}\left[ \left\|  \S-\S^{(t)} \right\|  ^2\right] =\sum_{i, j}\left[ \mathop{\mathbb{E}} \left[ \left( S_{i,j}\right) ^2\right]
				\!-\!\mathop{\mathbb{E}}\left[\mathop{\mathbb{E}}\limits_{S_{i,j}} \left[ S_{i,j} | (S_{\S \rightarrow p(\S)}^{(t)})_{i,j}\right]^2\right]\right] 	
			\end{aligned}
			\EE 
			where (\ref{mmsevar1})  follows from (\ref{esspost}).  
			Combining the above relations, we obtain 
			\BE\label{lemma2orth1}
			\begin{aligned}
				\mathbb{E}\left \langle \S^{(t)}_{\S\to p(\S)} -\S, \S^{(t)}\right \rangle=\mathbb{E}\left\|\S^{(t)}-\S \right\| _{F}^2.
			\end{aligned} 
			\EE
			What remains is to  prove (\ref{lemsse}).
			\BE \label{lemma2a}
			\begin{aligned} 
				a_{\S}^{(t)}    
				=\frac{\mathbb{E}\left \langle \S^{(t)}_{\S\to p(\S)} -\S, \S^{(t)}\right \rangle}{n v_{\S_{\S \rightarrow p(\S)}}^{(t)}}= \frac{v_{\S}^{(t)}}{v^{(t)}_{\S \rightarrow p(\S)}}
			\end{aligned} 
			\EE
			where (\ref{lemma2a}) follows from Assumption 1 and (\ref{lemma2orth1}).
			Then, (\ref{lemsse}b) is given by 
			\BS\label{lemma2c} 
			\begin{align}  	
				\quad c_{\S}^{(t)}
				=\frac{\left(1-a_{\S}^{(t)}\right)\mathbb{E}\left\| \S\right\|_{F}^2 -nv_{\S}^{(t)} }{\left(1-a_{\S}^{(t)}\right)^2\mathbb{E}\left\| \S\right\|_{F}^2-nv_{\S}^{(t)}\left( 1-a_{\S}^{(t)}\right) }&=\frac{1}{1-a_{\S}^{(t)}}\\
				&=\frac{v^{(t)}_{\S \rightarrow p(\S)}}{v^{(t)}_{\S \rightarrow p(\S)}-v_{\S}^{(t)}}
			\end{align} \ES
			where (\ref{lemma2c}a) follows from (\ref{espripost}) and Assumption 1, (\ref{lemma2c}b) follows from (\ref{lemma2a}).
			The rest is to prove (\ref{lemss2}). By  substituting $a_{\S}^{(t)} $ and $c_{\S}^{(t)}$ in (\ref{17pr1}) with the reduced form in (\ref{lemma2a}) and (\ref{lemma2c}), we obtain
			\BE\label{lemma2mse}
			\begin{aligned}
				\quad v_{p(\S)\to\S}^{(t)}=\frac{v^{(t)}_{\S \rightarrow p(\S)}v_{\S}^{(t)}}{v^{(t)}_{\S \rightarrow p(\S)}-v_{\S}^{(t)}}.
			\end{aligned} \EE
			This completes the proof of Lemma \ref{lemss}.

			\subsection{Proof of Lemma 3}\label{prl3}
			Recall that the SVD form of the matrix form $\A$ of ROIL operator $\mathcal{A}$ is $\A=\U_{\A}\Sigma_{\A}\V_{\A}^{T}$, where $\U_{A} \in \mathbb{R}^{m \times m}$ and $\V_{A} \in \mathbb{R}^{n \times n}$ are orthogonal matrices and $\Sigma_{\A}\in \mathbb{R}^{m \times n}$ is a diagonal matrix. Moreover,  $\V_{A} \in \mathbb{R}^{n \times n}$ is a Haar random matrix. Based on  Assumptions \ref{as1} and \ref{as2}, for each iteration, $\X_{\X\to \y}^{(t)}$ is independent of the linear operator  $\mathcal{A}$, specifically, is independent of the haar matrix $\V_{A}$.
			To start with, we have
			\BS\label{correl1}
			\begin{align}
			\mathbb{E}\left\|\A vec\left( \X_{\X\to \y}^{(t)}-\X\right)  \right\|_{F}^2
				&=\mathbb{E}\left[ vec\left(\!\X_{\X\to \y}^{(t)}\!-\!\X\right)^{T}\V_{\A}\Sigma_{\A}^{T}\Sigma_{\A}\V_{\A}^{T}vec\left(\!\X_{\X\to \y}^{(t)}\!-\!\X\!\right)\right] \\
				&=\frac{\left\|\A\right\|_{F}^2}{n}\mathbb{E}\left\|\X_{\X\to \y}^{(t)}-\X\right\|_{F}^2
			\end{align}
			\ES
			where (\ref{correl1}b) follows from Lemma \ref{explin}.  	
			
			Similarly, we first present the following relations before determining the parameters $a_{\X}^{(t)}$ and $c_{\X}^{(t)}$. Denote  $\X_{L}^{(t)}$ the estimation under LMMSE denoiser, which is given in (\ref{LMMSE}). Then,
			\BS\label{orx}
			\begin{align}
				\frac{1}{n}\mathbb{E}\left\langle \X_{\X\to \y}^{(t)},\X \right\rangle &= \frac{1}{n}\mathbb{E}\|\X\|_2^F-v_{\X\to \y}^{(t)}\\
				\frac{1}{n}\mathbb{E}\left\langle \X_{L}^{(t)},\X \right\rangle &= \frac{1}{n}\mathbb{E}\|\X\|_2^F-v_{\X_{L}}^{(t)}\\
				\frac{1}{n}\mathbb{E}\left\langle \X_{\X\to \y}^{(t)},\X_{L}^{(t)} \right\rangle &= \frac{1}{n}\mathbb{E}\|\X\|_2^F-v_{\X\to \y}^{(t)}.
			\end{align}
			\ES
			where (\ref{orx}a) follows from Lemma \ref{lemma1}, (\ref{orx}b) and (\ref{orx}c) follows from Lemma \ref{explin} and (\ref{orx}a), and $v_{\X_{L}}^{(t)}$ is given by
			\BE
			\begin{aligned}
				v_{\X_{L}}^{(t)}=v^{(t)}_{\boldsymbol{X} \rightarrow \boldsymbol{y}} \!-\!\frac{(v^{(t)}_{\boldsymbol{X} \rightarrow \boldsymbol{y}})^{2}}{n_{1} n_{2}} \operatorname{tr}\left(\! 
				\boldsymbol{A}^{T}\!\left(v^{(t)}_{\boldsymbol{X} \rightarrow \boldsymbol{y}
					} \boldsymbol{A} \boldsymbol{A}^{T}\!+\!\sigma_{\boldsymbol{n}}^{2} \boldsymbol{I}\right)^{-1} \!\boldsymbol{A}\!\right).
			\end{aligned}
			\EE
			Then, we take (\ref{orx}) into (\ref{Xformac}) and obtain 
			\BS	\setlength{\abovedisplayskip}{6pt}
			\setlength{\belowdisplayskip}{6pt}
			\begin{align}
				a_{\X}^{(t)}&=\!\frac{\!\frac{\!\mathbb{E}\!\|\X\|_2^F}{n}\!\!\left(\!\frac{\mathbb{E}\|\X\|_2^F}{n}\!-\!\! v_{\X\to \y}^{(t)}\!\right)\!-\!\left(\!\frac{\mathbb{E}\|\X\|_2^F}{n}\!\!-\!\!v_{\X_L}^{(t)}\!\right)\!\left(\!\frac{\mathbb{E}\|\X\|_2^F}{n}\!\!-\!\!v_{\X\to \y}^{(t)}\!\!\right)}{\frac{\mathbb{E}\|\!\X\!\|_2^F}{n}v_{\X\to \y}^{(t)}\!-\!\left(\!v_{\X\to \y}^{(t)}\! \right)^2 }=\frac{v_{\X_L}^{(t)}}{v_{\X\to \y}^{(t)}}\\
				c_{\X}^{(t)}&=\left( \frac{\left(\frac{\mathbb{E}\|\X\|_2^F}{n}-v_{\X_L}^{(t)}\right) -\frac{v_{\X_L}^{(t)}}{v_{\X\to \y}^{(t)}}(\frac{\mathbb{E}\|\X\|_2^F}{n}-v_{\X\to \y}^{(t)})}{\frac{\|\X\|_F^2}{n}}\right) ^{-1}=\frac{v_{\X\to \y}^{(t)}}{v_{\X\to \y}^{(t)}-v_{\X_L}^{(t)}}.
			\end{align}
			\ES
			Then, we obtain
			\BE
			\begin{aligned}
				\X^{(t+1)}_{\boldsymbol{y} \!\rightarrow \boldsymbol{X}}\!=\!\operatorname{vec}\!\left(\! \X^{(t)}_{\X\to\y}\!\right)\! +\!\!\frac{(v_{\X\to \y}^{(t)})^2}{v_{\X\to \y}^{(t)}\!\!-\!v_{\X_L}^{(t)}}\!\A^T \!\!\left(\! \y\!-\!\!\A\!\operatorname{vec}\left(\! \X^{(t)}_{\X\to\y}\right)\!\right) .
			\end{aligned}
			\EE
			
			And the corresponding variance $v_{\y\to\X}^{(t+1)}$ for LMMSE denoiser is given by     
			\BE\label{proof_eq_extvar}
			\begin{aligned}
				v_{\y\to\X}^{(t+1)}
				=\frac{v_{\X_L}^{(t)}v_{\X\to\y}^{(t)}}{v_{\X\to\y}^{(t)}-v_{\X_L}^{(t)}}
				=\frac{n}{\operatorname{tr}\left(\! 
					\boldsymbol{A}^{T}\!\left(v^{(t)}_{\boldsymbol{X} \rightarrow 
						\boldsymbol{y}} \boldsymbol{A} \boldsymbol{A}^{T}\!+\!\sigma^{2}_{\boldsymbol{n}} \boldsymbol{I}\right)^{-1} \!\boldsymbol{A}\!\right)}-v_{\X\to\y}^{(t)}.
			\end{aligned} 
			\EE
			\subsection{Proof of Corollary \ref{coro_le1}}\label{proof_coro_le1}
			If $\mathcal{A}$ is a partial Haar operator,   we have $\A\A^{T}=\I_m$.  Substituting  this into (\ref{proof_eq_extvar}), we obtain  
			\BS
			\begin{align}
				&v_{\X}^{(t)}=v^{(t)}_{\boldsymbol{X} \rightarrow \boldsymbol{y}} \!-\!\frac{(v^{(t)}_{\X\to\y})^{2}}{n_{1} n_{2}} \frac{m}{(v^{(t)}_{\X\to\y})^{2}+\sigma_{\boldsymbol{n}}^2}\\
				&v_{\y\to\X}^{(t+1)}=\frac{n-m}{m}v^{(t)}_{\boldsymbol{X} \rightarrow \boldsymbol{y}}+\frac{n}{m}\sigma_{\boldsymbol{n}}^2.
			\end{align}
			\ES

			\section{Proofs Related to The State Evolution in Sparsity Model}
			\subsection{Proof of Lemma \ref{asymS}}\label{proofsse}
			In this lemma,  we first prove the convergence of a general  pseudo-Lipschitz function $f(S_{i,j},N_{\S,i,j})$ with proper moments assumptions. Then, we show that $S_{i,j}\mathcal{D}_{\S}\left( S_{i,j} +\sqrt{v_{\delta\to\S}^{(t)}} N_{\S,i,j} \right)$ and $N_{\S,i,j}\mathcal{D}_{\S}\left( S_{i,j} +\sqrt{v_{\delta\to\S}^{(t)}} N_{\S,i,j} \right)$ are all pseudo-Lipschitz functions, which complete the proof.

			By the pseudo-Lipschitz property of $f$, we obtain 
			\BE
			\begin{aligned}
				f(S_{i,j},N_{\S,i,j})&\leq L\left( \left\| S_{i,j}\right\|^k+\left\| N_{\S,i,j}\right\|^k+1 \right).
			\end{aligned}
			\EE
			Then 
			\BS \label{spc}
			\begin{align}
				\mathbb{E}| 	f(S_{i,j},N_{\S,i,j}) |
				\leq L\mathbb{E}| \left\| S_{i,j}\right\|^k+\left\| N_{\S,i,j}\right\|^k+1|.
			\end{align}
			\ES
			By the bounded $k$-th assumption of $S$ and $N_S$, $\mathbb{E}| 	f(S_{i,j},N_{\S,i,j}) |$ is bounded.  By the Kolmogorov's strongly law of large number,  we obtain
			\BE
			\begin{aligned}
				\frac{1}{n}\sum_{i}f(S_{i},N_{\S,i})\overset{\text{a.s.}}{=}\mathbb{E}\left[ f(S,N)\right].
			\end{aligned}
			\EE
			with the $i.i.d.$ assumption of $S_{i,j}$ and $N_{\S,i,j}$ in Assumption \ref{as3}.
			Note that almost sure convergence implies convergence in probability. Then,
			\BE
			\begin{aligned}
				\frac{1}{n}\sum_{i}f(S_{i},N_{\S,i})\overset{p}{=}\mathbb{E}\left[ f(S,N)\right].
			\end{aligned}
			\EE
			
			By the $k$-th pseudo-Lipschitz assumption of a general pointwise denoiser $\mathcal{D}_{\S}$, we conclude that $S_{i,j}\mathcal{D}_{\S}\left( S_{i,j} +\sqrt{v_{\delta\to\S}^{(t)}} N_{\S,i,j} \right)$ and $N_{\S,i,j}\mathcal{D}_{\S}\left( S_{i,j} +\sqrt{v_{\delta\to\S}^{(t)}} N_{\S,i,j} \right)$ are pseudo-Lipschitz functions of $k+1$ and $\mathcal{D}_{\S}^2\left( S_{i,j} +\sqrt{v_{\delta\to\S}^{(t)}} N_{\S,i,j} \right)$ is pseudo-Lipschitz functions of $2k$. Then, we  conclude  $(\mathcal{D}_{\S}\left( S_{i,j} +\sqrt{v_{\delta\to\S}^{(t)}} N_{\S,i,j} \right)-S_{i,j})^2$ is pseudo-Lipschitz functions of $2k$. Then, (\ref{lemma_asys_eq1}) holds. Similarly, we conclude that $\frac{1}{n}\left \langle \S^{(t)}_{\S\to p(\S)} -\S, \S^{(t)}\right \rangle$ converge almost surely to $\mathbb{E}[\sqrt{v_{\delta\to\S}^{(t)}}N^{(t)}_{\S}\D_{\S} \left( S+\sqrt{v^{(t)}_{\delta\to\S}} N^{(t)}_{\S} \right)]$. And the denominator part of  $\hat{a}_{\S}^{(t)}$ almost surely converge to $v_{\delta\to\S}^{(t)}$ by i.i.d. Assumption. Then, $\hat{a}_{\S}^{(t)}$ converges.
			As for $\hat{c}_{\S}^{(t)}$, the denominator part is decomposed by three parts
			\BS\label{proof_leS_eq1}
			\begin{align}
				&	\frac{1}{n}\left\|  \D_{\S} \left( S+\sqrt{v^{(t)}_{\delta\to\S}} N^{(t)}_{\S} \right)\right\| ^2_F\\
				&		\frac{2}{n}\hat{a}_{\S}^{(t)}\left\langle \D_{\S} \left( S+\sqrt{v^{(t)}_{\delta\to\S}} N^{(t)}_{\S} \right), S+\sqrt{v^{(t)}_{\delta\to\S}} N^{(t)}_{\S} \right\rangle \\
				&	\frac{1}{n}(\hat{a}_{\S}^{(t)})^2	\left\|    S+\sqrt{v^{(t)}_{\delta\to\S}} N^{(t)}_{\S} \right\| ^2_F
			\end{align}
			\ES
			where (\ref{proof_leS_eq1}a) conveges due to the pseudo-Lipschitz property.
			Note that we have prove the almost sure convergence of $\hat{a}_{\S}^{(t)}$. Then, (\ref{proof_leS_eq1}b) and (\ref{proof_leS_eq1}c)  also converge. Then, (\ref{lemma_asys_eq2}) holds. Similarly, we conclude that (\ref{SparseSE}) holds. Recall that  $a_{\S}^{(t)}$ and  $c_{\S}^{(t)}$ are designed by (\ref{turboprinciple}).  With above asymptotic results for $\hat{a}_{\S}^{(t)}$ and $\hat{c}_{\S}^{(t)}$, (\ref{as3ext}) is obtained with some straightforward manipulations.

			\section{Proofs Related to the State Evolution in Low-Rank Model}
			\subsection{Useful Lemmas}

			For the convenience of the following proofs, we consider a general model with arbitrary finite noise level $\sqrt{v}$ by dropping the subscript $t$:
			\BE
				\begin{aligned}
					\boldsymbol{L}_{\N}=\boldsymbol{L}+\sqrt{v}\N,
				\end{aligned}
				\EE
				where $\boldsymbol{L}$ follows the requirements in Assumption (\ref{as3}.v). The entries of $\N$ are
				i.i.d. zero-mean and unit variance random variables. 
				Let  $\boldsymbol{U}\boldsymbol{\Sigma}\boldsymbol{V} ^T$ be the SVD of  $\boldsymbol{L}_{\N}$, where $\boldsymbol{U}\in \mathbb{R}^{n_1\times n_1}$ and $\boldsymbol{V}\in \mathbb{R}^{n_2\times n_2}$ are orthogonal matrices, and $\boldsymbol{\Sigma}=\text{diag}\{\sigma_1,\sigma_2,\cdots, \sigma_{n_1}\}\in \mathbb{R}^{n_1\times n_2}$ is a diagonal matrix with the diagonal elements arranged in the descending order. Denote by $\boldsymbol{u}_i$ and $\boldsymbol{v}_i$  respectively  the $i$-th columns of $\U$ and $\V$. 	Let  $\boldsymbol{U}_{\boldsymbol{L}}\boldsymbol{\Sigma}_{\boldsymbol{L}}\boldsymbol{V}_{\boldsymbol{L}}^T$ be the SVD of  $\boldsymbol{L}$, where  $\boldsymbol{\Sigma}_{\boldsymbol{L}}=\text{diag}\{\sigma_{1,\boldsymbol{L}},\sigma_{2,\boldsymbol{L}},\cdots, \sigma_{n_1,\boldsymbol{L}}\}\in \mathbb{R}^{n_1\times n_2}$ is a diagonal matrix with the diagonal elements arranged in the descending order. 	Let  $\boldsymbol{U}_{\N}\boldsymbol{\Sigma}_{\N}\boldsymbol{V}_{\N}^T$ be the SVD of  $\N$, where  $\boldsymbol{\Sigma}_{\N}=\text{diag}\{\sigma_{1,\N},\sigma_{2,\N},\cdots, \sigma_{n_1,\N}\}\in \mathbb{R}^{n_1\times n_2}$ is a diagonal matrix with the diagonal elements arranged in the descending order.   Let $\tilde{\sigma}_{i,\boldsymbol{L}}$ be the $i$-th singular value of $\frac{\boldsymbol{L}}{\sqrt{n_2v}}$ and $\tilde{\sigma}_{i,\N}$ be the $i$-th singular value of $\frac{\N}{\sqrt{n_2}}$. Let $\tilde{\boldsymbol{\sigma}}_{\boldsymbol{L}}=[\tilde{\sigma}_{1,\boldsymbol{L}}, \tilde{\sigma}_{2,\boldsymbol{L}},\cdots,\tilde{\sigma}_{n_1,\boldsymbol{L}}]$ and $\tilde{\boldsymbol{\sigma}}_{\N}=[\tilde{\sigma}_{1,\N}, \tilde{\sigma}_{2,\N},\cdots,\tilde{\sigma}_{n_1,\N}]$.

			First of all, we  present the following lemma concerning the differentiability of denoiser $\mathcal{D}_{\boldsymbol{L}}$. We restrict our analysis to the scenario in which $\boldsymbol{L}_{\N}$ is both simple and full-rank. This condition is ensured with probability 1 when $\N$ is an i.i.d. matrix, as explained in \cite{Onthedegrees}.
			 The first order differential form of spectral denoiser $\mathcal{D}_{\boldsymbol{L}}(\boldsymbol{L}_{\N})$ has been derived in \cite{candes2013unbiased}, where $f_i$ is continuous.   Here $f_i:\mathbb{R}_+\to\mathbb{R}_+$ is the scaling function acting on $\sigma_{i}$.  The derived form also can be applied to SVST\cite{candes2013unbiased}.    For the completeness of our proof, we give a briefly proof.
			\begin{lemma}[First order differential for spectral function]\label{difform} Suppose that $\boldsymbol{L}_{\N}$ is simple\footnote{We say that a matrix is simple if it has no repeated singular values besides zero.} and  full-rank.
				The first order differential form of spectral denoiser $\mathcal{D}_{\boldsymbol{L}}(\boldsymbol{L}_{\N})$ is given by 
				\BE
				\begin{aligned}
					\U\! \left(\!  [\boldsymbol{\Sigma}_{I}\boldsymbol{\Sigma}_{DI}]\!\circ(\U^{T}\boldsymbol{\Delta}\V) +[\boldsymbol{\Sigma}_{IJ}\circ\left[\! (\U^{T}\boldsymbol{\Delta}\V)_{1:n_1,1:n_1}^{T}\right] ,\boldsymbol{0}]\! \right)\!\V^{T}
				\end{aligned}
				\EE
				where $\boldsymbol{\Delta}\in\mathbb{R}^{n_1\times n_2}$ is an arbitrary perturbation matrix.
			\end{lemma}
			\begin{proof}
				Based on the standard rules of calculus, we obtain the differentiation  	
				\BE\label{diff}  	
				\begin{aligned} 		 		
					\mathrm{d} \D_{\boldsymbol{L}}[\boldsymbol{\Delta}]=\mathrm{d} \boldsymbol{U}[\boldsymbol{\Delta}] \D_{\boldsymbol{L}}(\boldsymbol{\Sigma}) \boldsymbol{V}^{T}+\boldsymbol{U}\! \mathrm{d}(\D_{\boldsymbol{L}} \circ \boldsymbol{\Sigma}\!)[\!\boldsymbol{\Delta}] \boldsymbol{V}^{T}\!+\!\boldsymbol{U} \D_{\boldsymbol{L}}(\boldsymbol{\Sigma}) \mathrm{d}\! \boldsymbol{V}\![\!\boldsymbol{\Delta}\!]^{T} 
				\end{aligned} 
				\EE	
				where $\mathrm{d} \boldsymbol{U}[\boldsymbol{\Delta}]$, $\mathrm{d}(\D_{\boldsymbol{L}} \circ \boldsymbol{\Sigma})[\boldsymbol{\Delta}]$ and $ \mathrm{d} \boldsymbol{V}[\boldsymbol{\Delta}]^{T} $ are respectively the differentiation of $\U$, $\D_{\boldsymbol{L}}(\boldsymbol{\Sigma})$ and $\V^{T}$  caused by the perturbation $\boldsymbol{\Delta}$.
				The perturbation matrix $\boldsymbol{\Delta}$ can be represented by: 
				\BE  \label{purdelat}	
				\begin{aligned} 	 	
					\boldsymbol{\Delta}=\mathrm{d} \boldsymbol{U}[\boldsymbol{\Delta}] \boldsymbol{\Sigma} \boldsymbol{V}^{T}+\boldsymbol{U} \mathrm{d}( \boldsymbol{\Sigma})[\boldsymbol{\Delta}] \boldsymbol{V}^{T}+\boldsymbol{U} \boldsymbol{\Sigma} \mathrm{d} \boldsymbol{V}[\boldsymbol{\Delta}]^{T}. 	 
				\end{aligned}  	
				\EE	 	
				We multiply $\boldsymbol{U}^{T}$ on the left side of (\ref{purdelat}) and $\boldsymbol{V}$ on the left side.
				\BE  		 	
				\begin{aligned} 	 				 		\boldsymbol{U}^{T}\boldsymbol{\Delta}\boldsymbol{V}=\boldsymbol{U}^{T}\mathrm{d} \boldsymbol{U}[\boldsymbol{\Delta}] \boldsymbol{\Sigma} + \mathrm{d}( \boldsymbol{\Sigma})[\boldsymbol{\Delta}] + \boldsymbol{\Sigma} \mathrm{d} \boldsymbol{V}[\boldsymbol{\Delta}]^{T} 	\boldsymbol{V} 	
				\end{aligned}  	 	
				\EE	
				We denote $\Omega_{\U}[\boldsymbol{\Delta}]=\U^{T} \mathrm{d} \U[\boldsymbol{\Delta}] \text { and } \Omega_{\V}[\boldsymbol{\Delta}]=\mathrm{d} \V[\boldsymbol{\Delta}]^{T} V$. $\Omega_{\U}[\boldsymbol{\Delta}]$ and $\Omega_{\V}[\boldsymbol{\Delta}]$ are anti-symmetric.
				\BE 	\begin{aligned} 		
					\Omega_{\boldsymbol{U}, i i}[\boldsymbol{\Delta}]=0 \quad \text { and } \quad \Omega_{\boldsymbol{V}, i i}[\boldsymbol{\Delta}]=0.
				\end{aligned} 	
				\EE 	
				Then, it's not hard to see that the diagonal entries of $\boldsymbol{U}^{T}\boldsymbol{\Delta}\boldsymbol{V}$ are respectively equal to those of $\mathrm{d}( \boldsymbol{\Sigma})[\boldsymbol{\Delta}]$, i.e., $\left( \boldsymbol{U}^{T}\boldsymbol{\Delta}\boldsymbol{V}\right)_{i,i}=\left(\mathrm{d}( \boldsymbol{\Sigma})[\boldsymbol{\Delta}] \right)_{i,i} $ for $1\leq i \leq n_1$.
				For the off-diagonal entries ($i\neq j$) of $\boldsymbol{U}^{T}\boldsymbol{\Delta}\boldsymbol{V}$,  we obtain the following the relations 
				\BS\label{duv} 	
				\begin{align} 
					\left(\boldsymbol{U}^{T} \boldsymbol{\Delta} \boldsymbol{V}\right)_{i j} &=\Omega_{\boldsymbol{U}, i j}[\boldsymbol{\Delta}] \sigma_{j}+\sigma_{i} \Omega_{\boldsymbol{V}, i j}[\boldsymbol{\Delta}] \\ 	-\left(\boldsymbol{U}^{T} \boldsymbol{\Delta} \boldsymbol{V}\right)_{j i} &=\Omega_{\boldsymbol{U}, i j}[\boldsymbol{\Delta}] \sigma_{i}+\sigma_{j} \Omega_{\boldsymbol{V}, i j}[\boldsymbol{\Delta}]. 	
				\end{align} 	\ES
				Based on the above facts, we next represent the differentiation ($\mathrm{d} \D_{\boldsymbol{L}}[\boldsymbol{\Delta}]$) of $\D_{\boldsymbol{L}}$ at  $\boldsymbol{L}_{\N}$ as a function of $\U$, $\boldsymbol{\Sigma}$, $\V$ and arbitrary perturbation $\boldsymbol{\Delta}$  in a closed form. Recall that the differentiation of the SVD form of $\D_{\boldsymbol{L}}$ is given in (\ref{diff}). We will replace the differentiation of $\U$, $\boldsymbol{\Sigma}$ and $\V$ with the linear combination of $\U$, $\boldsymbol{\Sigma}$, $\V$ and $\boldsymbol{\Delta}$.  
				From (\ref{duv}), for $i\neq j$, we obtain
				\BS
				\begin{align}
					\Omega_{\boldsymbol{U}, i j}[\boldsymbol{\Delta}] &=-\frac{\sigma_{j}(\boldsymbol{u}_{i}^{T}\boldsymbol{\Delta}\boldsymbol{v}_{j})+\sigma_{i}(\boldsymbol{u}_{j}^{T}\boldsymbol{\Delta}\boldsymbol{v}_{i})}{\sigma_{i}^2-\sigma_{j}^2}  \quad 1\leq i,j\leq n_1,  \\
					\Omega_{\boldsymbol{V}, i j}[\boldsymbol{\Delta}] &=\frac{\sigma_{i}(\boldsymbol{u}_{i}^{T}\boldsymbol{\Delta}\boldsymbol{v}_{j})+\sigma_{j}(\boldsymbol{u}_{j}^{T}\boldsymbol{\Delta}\boldsymbol{v}_{i})}{\sigma_{i}^2-\sigma_{j}^2} \quad 1\leq i,j\leq n_1,\\
					\Omega_{\boldsymbol{V}, i j}[\boldsymbol{\Delta}]
					&=\frac{\boldsymbol{u}_{i}^{T}\boldsymbol{\Delta}\boldsymbol{v}_{j}}{\sigma_i}, \quad i\leq n_1, n_1<j\leq n_2.
				\end{align}
				\ES
				In addition, the differentiation of the singular value $\boldsymbol{\Sigma}$ is given by $\left( \mathrm{d}( \boldsymbol{\Sigma})\right)_{i,i}=	\left(\boldsymbol{U}^{T} \boldsymbol{\Delta} \boldsymbol{V}\right)_{i i}$.
				
				For $1\leq i=j\leq n_1$, 
				\BE
				\begin{aligned}
					(\U^{T}\mathrm{d} \D_{\boldsymbol{L}}[\boldsymbol{\Delta}]\V)_{i,i}=  f_{i}^{'}(\sigma_{i})\boldsymbol{u}_{i}^{T}\boldsymbol{\Delta}\boldsymbol{v}_{i}.
				\end{aligned}
				\EE
				
				For $1\leq i,j\leq n_1$ and $i\neq j$, 
				\BE 
				\begin{aligned} 
					&(\U^{T}\mathrm{d}\D_{\boldsymbol{L}}[\boldsymbol{\Delta}]\V)_{i,j}\\
					&=\frac{f_{i}(\sigma_{i})\sigma_{i}\!-\!f_{j}(\sigma_{j})\sigma_{j}}{\sigma^2_{i}-\sigma^2_{j}}\boldsymbol{u}_{i}^{T}\boldsymbol{\Delta}\boldsymbol{v}_{j}\!+\!\frac{f_{i}(\sigma_{i})\sigma_{j}\!-\!f_{j}(\sigma_{j})\sigma_{i}}{\sigma^2_{i}-\sigma^2_{j}}\boldsymbol{u}_{j}^{T}\boldsymbol{\Delta}\boldsymbol{v}_{i}.
				\end{aligned} \EE
				
				For $1\leq i\leq n_1$ and $n_1< i\leq n_2 $, 
				\BE 
				\begin{aligned} 
					(\U^{T}\mathrm{d}f[\boldsymbol{\Delta}]\V)_{i,j}=\frac{f_{i}(\sigma_{i})}{\sigma_i}\boldsymbol{u}_{i}^{T}\boldsymbol{\Delta}\boldsymbol{v}_{j}.
				\end{aligned} \EE
				Denote $\frac{f_{i}(\sigma_{i})\sigma_{i}-f_{j}(\sigma_{j})\sigma_{j}}{\sigma^2_{i}-\sigma^2_{j}} $  by  the $(i,j)$-th entry of $\boldsymbol{\Sigma}_{I}$ for $1\leq i,j\leq n_1$, $i\neq j$ and $f_{i}^{'}(\sigma_{i})$ by the $(i,i)$-th entry of $\boldsymbol{\Sigma}_{I}$ for $1\leq i\leq n_1$. Denote $\frac{f_{i}(\sigma_{i})\sigma_{j}-f_{j}(\sigma_{j})\sigma_{i}}{\sigma^2_{i}-\sigma^2_{j}} $   by the $(i,j)$-th entry of $\boldsymbol{\Sigma}_{IJ}$ for $1\leq i,j\leq n_1$ and $i\neq j$. Denote the $(i,i)$-th entry of $\boldsymbol{\Sigma}_{IJ}$ by $0$ for $1\leq i,j\leq n_1$. Denote the $(i,j)$-th entry of $\boldsymbol{\Sigma}_{DI}$ by $\frac{f_i(\sigma_{i})}{\sigma_{i}}$ for $1\leq i\leq n_1$ and $n_1< j\leq n_2$. In summary, we  represent $\mathrm{d}\D_{\boldsymbol{L}}[\boldsymbol{\Delta}]$ as 
				\BE
				\begin{aligned}
					\U\! \left(\!  [\boldsymbol{\Sigma}_{I}\boldsymbol{\Sigma}_{DI}]\!\circ(\U^{T}\boldsymbol{\Delta}\V) \!+\![\boldsymbol{\Sigma}_{IJ}\circ\left[\! (\U^{T}\!\boldsymbol{\Delta}\V)_{1:n_1,1:n_1}^{T}\!\right] \!,\!\boldsymbol{0}]\! \right)\!\V^{T}
				\end{aligned}
				\EE
				where $\boldsymbol{\Sigma}_{I}$, $\boldsymbol{\Sigma}_{DI}$ and $\boldsymbol{\Sigma}_{IJ}$ are only related to the singular value $\boldsymbol{\Sigma}$ and spectral function $f_i$.
			\end{proof}

			With the  the first order differential for spectral denoiser $\D_{\boldsymbol{L}}$, the corresponding divergence can be derived.

			\begin{lemma}{\cite[Theorem IV.3]{candes2013unbiased}}\label{candes2013unbiased}
				Let $\D_{\boldsymbol{L}}$ be a low-rank denoiser and $f_{\boldsymbol{L}}:\mathbb{R}_+\to\mathbb{R}_+$ be differentiable. Suppose that $\boldsymbol{L}_{\N}$ is simple and  full-rank. Then $\operatorname{div}(\D_{\boldsymbol{L}}(\boldsymbol{L}_{\N}))$ is given by
				\BE
				\begin{aligned}
					\operatorname{div}(\D_{\boldsymbol{L}}(\boldsymbol{L}_{\N}))=&(n_2-n_1) \sum_{i=1}^{n_1} \frac{f_{\boldsymbol{L}}\left(\sigma_{i}\right)}{\sigma_{i}} +\sum_{i=1}^{n_1}  f_{\boldsymbol{L}}^{\prime}\left(\sigma_{i}\right)+ \sum_{i \neq j, i, j=1}^{n_1} \frac{\sigma_{i} f_{\boldsymbol{L}}\left(\sigma_{i}\right)-\sigma_{j} f_{\boldsymbol{L}}\left(\sigma_{j}\right)}{\sigma_{i}^{2}-\sigma_{j}^{2}}
				\end{aligned}
				\EE
				where $\sigma_{i}$ is the $i$-th singular value of $\boldsymbol{L}_{\N}$.
			\end{lemma}
			\begin{remark}\label{g_LH_L}
				Note that  $\boldsymbol{L}_{\N}$ is simple and has full-rank with probability 1 as shown in \cite{Onthedegrees}. For the other case, \cite[Theorem IV.6]{candes2013unbiased} also provides an extended version of $\operatorname{div}(\D_{\boldsymbol{L}}(\boldsymbol{L}_{\N}))$ where $\boldsymbol{L}_{\N}$ has repeated singular values and  has singular value $0$. Define $g_{\boldsymbol{L}}$:  $[0,+\infty)\rightarrow[0,+\infty)$ with $g_{\boldsymbol{L}}(x)=\frac{f_{\boldsymbol{L}}(x)}{x}$ for $x>0$ and $g_{\boldsymbol{L}}(0)=0$.
				Define $H_{\boldsymbol{L}}(x,y)$ by 	
				\BE
				\begin{aligned}
					H_{\boldsymbol{L}}(x,y)\triangleq\left\{
					\begin{array}{cl}
						0                           & (x,y)=(0,0)\\
						\frac{xf_{\boldsymbol{L}}'(x)+f_{\boldsymbol{L}}(x)}{2x}      & x=y\neq 0\\
						\frac{xf_{\boldsymbol{L}}(x)-yf_{\boldsymbol{L}}(y)}{x^2-y^2} & x \neq y.\\
					\end{array}\right.
				\end{aligned}
				\EE
				Then, a general expression for $\operatorname{div}(\D_{\boldsymbol{L}}(\boldsymbol{L}_{\N}))$ applicable to all matrices is given by:
				\BE
				\begin{aligned}
					\operatorname{div}(\D_{\boldsymbol{L}}(\boldsymbol{L}_{\N}))=&(n_2-n_1) \sum_{i=1}^{n_1} g_{\boldsymbol{L}}(x)+\sum_{i=1}^{n_1}  f_{\boldsymbol{L}}^{\prime}\left(\sigma_{i}\right)+ \sum_{i \neq j, i, j=1}^{n_1} 	H_{\boldsymbol{L}}(\sigma_i,\sigma_j).
				\end{aligned}
				\EE
			\end{remark}
			In the following lemma, we analyze the property of function $	H_{\boldsymbol{L}}(x,y)$.
			\begin{lemma}\label{lemmaH}
			
Suppose $f:[0,+\infty)\rightarrow[0,+\infty)$ is a monotonic increasing  and continuously differentiable   function defined on $[0,+\infty)$ and also satisfies the property of being locally Lipschitz over $[0,b]$ for any $b>0$. Let $M_{b}$ denote the Lipschitz constant on the interval $[0,b]$. Further assume that $f$  with $f(0)=0$ and $f'(0)=0$. Denote $H(x,y)$ and $G(x,y)$ by  
				\BE
				\begin{aligned}
					H(x,y)=\left\{
					\begin{array}{cl}
						0                           & (x,y)=(0,0)\\
						\frac{xf'(x)+f(x)}{2x}      & x=y\neq 0\\
						\frac{xf(x)-yf(y)}{x^2-y^2} & x \neq y\\
					\end{array}\right.
				\end{aligned}
				\EE
				and   
				\BE
				\begin{aligned}
					G(x,y)=\left\{
					\begin{array}{cl}
						0                           & (x,y)=(0,0)\\
						\frac{xf'(x)-f(x)}{2x}      & x=y\neq 0\\
						\frac{yf(x)-xf(y)}{x^2-y^2} & x \neq y.\\
					\end{array}\right.
				\end{aligned}
				\EE
Then, for any compact domain $D \subset [0, b] \times [0, b]$, where $b \geq 0$, $H(x,y)$ and $G(x,y)$ are both  continuous functions, where $H(x,y)$ is bounded by $\frac{3}{2}M_b$ and $G(x,y)$ is bounded by $M_b$.
			\end{lemma}
			\begin{proof} We first proof the boundedness.
				Without loss of generally, we first consider the situation $x>y$ in $[0, b] \times [0, b]$.  $H(x,y)$ and $G(x,y)$ are both continuous by assumption. We only need to show the boundedness. With the monotonic increasing property of $f$, we obtain
				\BE
				\begin{aligned}
					\frac{yf(y)-yf(x)}{x^2-y^2}\leq\frac{xf(y)-yf(x)}{x^2-y^2}\leq \frac{xf(y)-yf(y)}{x^2-y^2}.
				\end{aligned}
				\EE
				Note that $\frac{yf(y)-yf(x)}{x^2-y^2}<0$. Then
				\BE
				\begin{aligned}
				\left| \frac{xf(y)-yf(x)}{x^2-y^2}\right| \leq \operatorname{max}\left( \frac{f(y)}{x+y},\frac{y}{x+y}\frac{f(x)-f(y)}{x-y}\right) .
				\end{aligned}
				\EE
				With the Lipschitz property of $f$ and $f(0)=0$, we conclude $\frac{f(y)}{x+y}< \frac{M_b}{2}$ and $\frac{f(x)-f(y)}{x-y}< M_b$. Then 
				\BE
				\begin{aligned}
				|G(x,y)|=	\left| \frac{xf(y)-yf(x)}{x^2-y^2}\right| < \frac{M_b}{2}.
				\end{aligned}
				\EE
				Note that 
				\BE
				\begin{aligned}
				\frac{xf(x)-yf(y)}{x^2-y^2}-\frac{xf(y)-yf(x)}{x^2-y^2}
				=\frac{(x+y)\left( f(x)-f(y)\right) }{x^2-y^2}
				=\frac{ f(x)-f(y) }{x-y}
					<M_b.
				\end{aligned}
				\EE
				Then, by the triangle inequality,
				\BE
				\begin{aligned}
						|H(x,y)|=\left| \frac{xf(x)-yf(y)}{x^2-y^2}\right| &<\left| \frac{xf(y)-yf(x)}{x^2-y^2}\right| +\left|\frac{(x+y)\left( f(x)-f(y)\right) }{(x+y)(x-y)} \right| 
					<\frac{3M_b}{2}.
				\end{aligned}
				\EE				
Next, we establish the case for $x = y \neq 0$.
By the definition of $H(x,y)$ and $G(x,y)$  at $x=y\neq 0$, we obtain
				\BE
				\begin{aligned}
					&|H(x,y)|=\left|\frac{yf'(y)+f(y)}{2y}\right|\leq \left|\frac{f'(y)}{2}\right|+\left|\frac{f(y)}{2y}\right|<M_b \\
					&|G(x,y)|=\left|\frac{yf'(y)-f(y)}{2y}\right|\leq \left|\frac{f'(y)}{2}\right|+\left|\frac{f(y)}{2y}\right|<M_b.
				\end{aligned}
				\EE
			We next prove the continuity at $x = y \neq 0$. We need to show that for every $\epsilon > 0$, there exists $\delta_{\epsilon}$ such that for any $(x, y)$ satisfying $0 < (x - z)^2 + (y - z)^2 < \delta_{\epsilon}$, the following condition holds:
				\BE
				\begin{aligned}
					|H(x,y)-H(z,z)|< \epsilon.
				\end{aligned} 
				\EE
				Note that 
								\BE
				\begin{aligned}
					|H(x,y)-H(z,z)|=\left\{
					\begin{array}{cl}
						|\frac{xf'(x)+f(x)}{2x} - 	\frac{zf'(z)+f(z)}{2z}|    & x=y\neq 0\\
						|\frac{xf(x)-yf(y)}{x^2-y^2} - 	\frac{zf'(z)+f(z)}{2z}|  & x \neq y\\
					\end{array}\right.
				\end{aligned}
				\EE
				By assumption, $\frac{xf'(x)+f(x)}{2x}$ is continuous on $(0,+\infty)$. For the case $x=y\neq 0$, there exists $\delta_{\epsilon,1}>0$, such that for any $0<(x-z)^2<\delta_{\epsilon,1}$, $|\frac{xf'(x)+f(x)}{2x} - 	\frac{zf'(z)+f(z)}{2z}|<\epsilon$.
				 For the case $x\neq y$,  let $w=x^2$ and $k=y^2$. There exists $q$ such that
				\BE
				\begin{aligned}
					\frac{xf(x)-yf(y)}{x^2-y^2}=	\frac{\sqrt{w}f(\sqrt{w})-\sqrt{k}f(\sqrt{k})}{w-k}=(\sqrt{q}f(\sqrt{q}))'
				\end{aligned}
				\EE
				where $q\in(x^2,y^2)$ for $x< y$ or $q\in(y^2,x^2)$ for $y< x$ by the mean value theorem.
\!
				Note that $\min\{|x-z|,|y-z| \}<|\sqrt{q}-z|<\max\{|x-z|,|y-z| \}$.  By the continuity of $\frac{xf'(x)+f(x)}{2x}$, there exists $\delta_{\epsilon,2}>0$, such that for any $0<(x-z)^2+(y-z)^2<\delta_{\epsilon,2}$,
				\BE
				\begin{aligned}
					|\frac{\sqrt{q}f'(\sqrt{q})+f(\sqrt{q})}{2\sqrt{q}}  - 	\frac{zf'(z)+f(z)}{2z}|<\epsilon
				\end{aligned}
				\EE
				where $\sqrt{q}$ is bounded by $|\sqrt{q}-z|<\delta_{\epsilon,2}$. Denotes $\delta_{\epsilon}=\min\{\delta_{\epsilon,1},\delta_{\epsilon,2}\}$. Then, the obtained $\delta_{\epsilon}$ satisfies the requirement that $0 < (x-z)^2 + (y-z)^2 < \delta_{\epsilon}$, ensuring that $|H(x,y) - H(z,z)| < \epsilon$. Note that 
	\BE
	\begin{aligned}
			G(x,y)=\left\{
			\begin{array}{cl}	0                           & (x,y)=(0,0)\\
			f'(x)-H(x,y)    & x=y\neq 0\\	\frac{f(x)-f(y)}{x-y}- 	H(x,y) & x \neq y.\\	\end{array}\right.	\end{aligned}
	\EE
	Denotes 
	\BE
\begin{aligned}
	G_{H}(x,y)=\left\{
	\begin{array}{cl}	0                           & (x,y)=(0,0)\\
		f'(x)   & x=y\neq 0\\	\frac{f(x)-f(y)}{x-y} & x \neq y.\\	\end{array}\right.	\end{aligned}
\EE
Then, $G(x,y)=G_{H}(x,y)-H(x,y)$.
To show the continuity of $G_(x,y)$ at $x= y\neq 0$, we only need to show $G_{H}(x,y)$ is continuous at $x= y\neq 0$.  Similar to the method used to prove the continuity of $H(x, y)$ at $x = y \neq 0$, we can similarly demonstrate the continuity of $G_{H}(x, y)$.
 Then, $G(x,y)$ is also continuous at $x= y\neq 0$.
\!
				For the case $(x,y)=(0,0)$,
				 we first consider $H(x,y)$. We need to show $ {\forall} \epsilon>0$, ${\exists \delta_{\epsilon} }$, such that for any $(x,y)$ satisfy $ 0<x^2+y^2< \delta_{\epsilon}$, the following holds:
				 \BE
				 \begin{aligned}
				 	|H(x,y)-0|< \epsilon.
				 \end{aligned} 
				 \EE
\!				 
				 	By assumption, $\frac{xf'(x)+f(x)}{2x}$ is continuous on $(0,+\infty)$ and $\lim_{x\to 0^+}\frac{xf'(x)+f(x)}{2x}=0$.  For the case $x=y\neq 0$, there exists $\delta_{\epsilon,1}>0$, such that for any $0<x^2<\delta_{\epsilon,1}$, $|\frac{xf'(x)+f(x)}{2x}|<\epsilon$.
				 For the case $x\neq y$,  let $w=x^2$ and $k=y^2$. There exists $q$ such that
				 \BE
				 \begin{aligned}
				 	\frac{xf(x)-yf(y)}{x^2-y^2}=	\frac{\sqrt{w}f(\sqrt{w})-\sqrt{k}f(\sqrt{k})}{w-k}=(\sqrt{q}f(\sqrt{q}))'
				 \end{aligned}
				 \EE
				 where $q\in(x^2,y^2)$ for $x< y$ or $q\in(y^2,x^2)$ for $y< x$ by the mean value theorem. 	Note that $\min\{x,y \}<\sqrt{q}<\max\{x,y \}$.  By the property of $\frac{xf'(x)+f(x)}{2x}$, there exists $\delta_{\epsilon,2}>0$, such that for any $0<x^2+y^2<\delta_{\epsilon,2}$,
				 \BE
				 \begin{aligned}
				 	|\frac{\sqrt{q}f'(\sqrt{q})+f(\sqrt{q})}{2\sqrt{q}} -0|<\epsilon
				 \end{aligned}
				 \EE
				 where $\sqrt{q}$ is bounded by $|\sqrt{q}-z|<\delta_{\epsilon,2}$. Let $\delta_{\epsilon}=\min\{\delta_{\epsilon,1},\delta_{\epsilon,2}\}$. Then, the obtained $\delta_{\epsilon}$ 
				 satisfies the requirement.
\!				 
				 A similar analysis demonstrates the continuity of $G_{H}(x, y)$ at $(0,0)$, which in turn implies the continuity of $G(x, y)$ at $(0,0)$. With this, we complete the proof.
			\end{proof}
		Note that $f_{\boldsymbol{L}}$, required in Assumption (\ref{as3}.vi), is continuously differentiable. Therefore, its derivative is bounded in any closed and bounded domain, which implies the local Lipschitz property. Let $H_{\boldsymbol{L}}$ be defined by replacing $f$ with $f_{\boldsymbol{L}}$ in the expression for $H(x,y)$. Consequently, $H_{\boldsymbol{L}}(x,y)$ is also continuous and bounded in any closed and bounded domain.
		
			\begin{lemma}{\cite[Theorem 1.1]{dozier2007empirical}}\label{asysin} Denote $\mu_{\boldsymbol{L}\boldsymbol{L}^T,n_1}$ the empirical eigenvalue distribution of $\frac{1}{n_2}\boldsymbol{L}\boldsymbol{L}^{T}$. Assume that the empirical distribution converges weakly to a deterministic distribution $\mu_{\boldsymbol{L}\boldsymbol{L}^T}(\sigma)$. And the entries of $\N$ are i.i.d. zero-mean and unit variance random variables.   Let $\M_{\boldsymbol{L}+\sqrt{v}\N}=\frac{1}{n_2}(\boldsymbol{L}_{\N})(\boldsymbol{L}_{\N})^{T}$. Then the corresponding empirical eigenvalue distribution of $\mu_{\M,n_1 }$ converges weakly to a limiting distribution $\mu_{\M}$, whose Stieljes transform $m(z):=\int\frac{1}{\lambda-z}d\mu_{\M}(\lambda)$ satisfied the integral equation 
				\BE\label{proba}
				\begin{aligned}
					m(z)=\int \frac{d \mu_{\boldsymbol{L}\boldsymbol{L}^T}(t)}{\frac{ t}{1+v\beta m}-(1+v\beta m) z+v(1-\beta)}.
				\end{aligned}
				\EE
			\end{lemma}
			\begin{remark}
				Note that the singular value of $\frac{\boldsymbol{L}_{\N}}{\sqrt{vn_2}}$ is exactly the positive square root of eigenvalue of $\frac{1}{v}\M_{\boldsymbol{L}+\sqrt{v}\N}$. Thus, the empirical singular values distribution of $\frac{\boldsymbol{L}_{\N}}{\sqrt{n_2v}}$   weakly converges to a limiting distribution, denoted by $\mu_{\boldsymbol{L},v}$. 
				Then, the two limiting empirical distributions follow
				\BE
				\begin{aligned}
					\mu_{\boldsymbol{L},v}(\sigma)=2v\sigma\mu_{\M}(v\sigma^2).
				\end{aligned}
				\EE  
			\end{remark}
			
			\begin{lemma}\label{lemmau}
				Let $\{\sigma_{i}\}_{i=1}^{n_1}$ be a sequence of elements, where $\sigma_{i}\in \mathbb{R}^{+}$ is bounded.
				Write $\delta_{\sigma_{i}}$ for the unit mass at $\sigma_{i}$ and zero elsewhere, the empirical measure and the distribution function  are respectively defined by
				\BE
				\begin{aligned}
					&\mu_{\sigma,n_1}(x)=\frac{1}{n_1} \sum_{i=1}^{n_1} \delta_{\sigma_{i}}(x=\sigma_{i})\\
					&F^{\sigma}_{n_1}(x)=\frac{1}{n_1} \sum_{i=1}^{n_1} \boldsymbol{1}_{\sigma_{i}\leq x}.
				\end{aligned}
				\EE
				
				Assume that the empirical measure of $\mu_{\sigma,n_1}$  converges weakly to a limiting  measure $\mu_{\sigma}$ as $n_1\to\infty$. 
				Let $\{(\sigma_{i},\sigma_{j})\}=\{\sigma_{i}\}_{i=1}^{n_1}\times\{\sigma_{j}\}_{j=1}^{n_1}$. 
				Write $\delta_{(\sigma_{i},\sigma_{j})}$ for the unit mass at $(\sigma_{i},\sigma_{j})$ and zero elsewhere. The empirical measure and the  distribution function are denoted by
				\BE
				\begin{aligned}
					&\mu_{(\sigma,\sigma),n_1}=\frac{1}{n_1\times n_1} \sum_{i=1,j=1}^{n_1,n_1} \delta_{(\sigma_{i},\sigma_{j})}\\
					&F^{\sigma,\sigma}_{n_1}(x,y)=\frac{1}{n_1\times n_1} \sum_{i,j=1}^{n_1,n_1} \boldsymbol{1}_{\sigma_{i}\leq x, \sigma_{j}\leq y}.
				\end{aligned}
				\EE
				Then 
				\BE\label{weakc}
				\begin{aligned}
					\mu_{(\sigma,\sigma),n_1}\overset{p}{=}\mu_{\sigma}\times\mu_{\sigma}.
				\end{aligned}
				\EE
			\end{lemma}
			\begin{proof}
				Note that 
				\BE
				\begin{aligned}
					&F^{\sigma,\sigma}_{n_1}(x,y)=\frac{1}{n_1\times n_1} \sum_{i,j=1}^{n_1,n_1} \boldsymbol{1}_{\sigma_{i}\leq x, \sigma_{j}\leq y}\\
					&=\frac{\#\left\{1 \leq i \leq n_1,1 \leq j \leq n_1 \mid \sigma_{i}\leq x, \sigma_{j}\leq y\right\}}{n_1\times n_1}\\
					&=\frac{\#\left\{1 \leq i \leq n_1\mid  \sigma_{i}\leq x\right\}\times\#\left\{1 \leq j \leq n_1\mid  \sigma_{j}\leq y\right\}}{n_1\times n_1}\\
					&=	F^{\sigma}_{n_1}(x)\times	F^{\sigma}_{n_1}(y).
				\end{aligned}
				\EE
				We have 
				\BE
				\begin{aligned}
					\lim_{n_1\to\infty}	F^{\sigma}_{n_1}(x)\to F^{\sigma}(x),
				\end{aligned}
				\EE
				for every number $ x \in \mathbb{R}$ at where $F^{\sigma}(x)$ is continuous.  And, for every continuous point $ (x,y) \in \mathbb{R}\times\mathbb{R} $ of $F^{\sigma}(x)\times F^{\sigma}(y)$,  $ x $ and $ y $ are also the continuous points of function $F^{\sigma}$. Then, we obtain that 
				\BS
				\begin{align}
					\lim_{n_1\to\infty}	F^{\sigma}_{n_1}(x)\to F^{\sigma}(x),\\
					\lim_{n_1\to\infty}	F^{\sigma}_{n_1}(y)\to F^{\sigma}(y),
				\end{align}
				\ES where $ (x,y) \in \mathbb{R}\times\mathbb{R} $  is the continuous point of $F^{\sigma}(x)\times F^{\sigma}(y)$. Then
				\BE
				\begin{aligned}
					&\lim_{n_1\to\infty} F^{\sigma,\sigma}_{n_1}(x,y)
					\to	F^{\sigma}(x)\times F^{\sigma}(y),
				\end{aligned}
				\EE
				which completes the proof.
			\end{proof}

			\begin{lemma}\label{fLasy2}
				Let $\{\sigma_{i}\}_{i=1}^{n_1}$ be a sequence of elements, where $\sigma_{i}\in \mathbb{R}^{+}$ is bounded.	Write $\delta_{\sigma_{i}}$ for the unit mass at $\sigma_{i}$ and zero elsewhere, the empirical measure and the cumulative distribution function  are respectively defined by
				\BE
				\begin{aligned}
					&\mu_{\sigma,n_1}=\frac{1}{n_1} \sum_{i=1}^{n_1} \delta_{\sigma_{i}}\\
					&F^{\sigma}_{n_1}(x)=\frac{1}{n_1} \sum_{i=1}^{n_1} \boldsymbol{1}_{\sigma_{i}\leq x}.
				\end{aligned}
				\EE
				
				Assume that the empirical measure of $\mu_{\sigma,n_1}$  converges weakly to a limiting  measure $\mu_{\sigma}$ as $n_1\to\infty$. 
				Then, as $n_1\to\infty$,
				\BE
				\begin{aligned}
					\frac{1}{n_1^2}\sum^{n_1}_{i \neq j, i, j=1}H(\sigma_{i},\sigma_{j})\overset{p}{=}\mathbb{E}\left[ H(\sigma,\hat{\sigma})\right] 
				\end{aligned}
				\EE
				where the expectation is taken over two independent and identical variables $\sigma\sim \mu_{\sigma} $ and $\hat{\sigma}\sim \mu_{\sigma}$. 
			\end{lemma}
			\begin{proof}
				Note that
				\BE
				\begin{aligned}
					\frac{1}{n_1^2}\sum^{n_1}_{i \neq j, i, j=1}H(\sigma_{i},\sigma_{j})
					=\frac{1}{n_1^2}\sum^{n_1}_{i,j=1}H(\sigma_{i},\sigma_{j})-\frac{1}{n_1^2}\sum^{n_1}_{i}H(\sigma_{i},\sigma_{i}).
				\end{aligned}
				\EE
				From Lemma \ref{lemmaH}, $H(x,y)$ is continuous and bounded on any compact domain. And $\sigma_{i}$ is bounded.  Then, by Lemma \ref{lemmau}  and Definition $\ref{defi_weakcon}$, 
				\BE
				\begin{aligned}
					\frac{1}{n_1^2}\sum^{n_1}_{i, j=1} H(\sigma_{i},\sigma_{j})\overset{p}{=} \mathbb{E}\left[H(\sigma,\hat{\sigma})\right] ,
				\end{aligned}
				\EE
				where the expectation is taken over the  two independent and identical variables $\sigma\sim \mu_{\sigma} $ and $\hat{\sigma}\sim \mu_{\sigma}$. 
				Moreover, $H(\sigma_{i},\sigma_{i})$ is bounded on bounded domain from Lemma \ref{lemmaH}. Then, as $n_1\to \infty$,
				\BE
				\begin{aligned}
					\frac{1}{n_1^2}\sum^{n_1}_{i}H(\sigma_{i},\sigma_{i})\rightarrow 0 .
				\end{aligned}
				\EE
				Then, as $n_1\to\infty$, we obtain 
				\BE
				\begin{aligned}
					\frac{1}{n}\sum^{n_1}_{i \neq j, i, j=1}H(\sigma_{i},\sigma_{j})\overset{p}{=}\mathbb{E}\left[ H(\sigma,\hat{\sigma})\right].
				\end{aligned}
				\EE

			\end{proof}
			
			The following lemma presents a sufficient condition for the spectral denoiser $\D_{\boldsymbol{L}}$, whose MSE can be unbiasedly estimated using SURE.

			\begin{lemma}{\cite[Theorem 1]{hansen2018stein}}\label{spectralfunc}
				Consider a spectral function estimator
				\BE\label{eqx_1}
				\begin{aligned}
					\mathcal{D}_{\boldsymbol{L}}(\boldsymbol{L}_{\N})=\sum_{i=1}^{n_1}f_i(\sigma_i)\boldsymbol{u}_i\boldsymbol{v}^T_i
				\end{aligned}
				\EE
				with finite second moment and with spectral functions fulfilling that $f_1,\cdots,f_{n_1-1}$ are continuously differentiable on $(0,\infty)$, $f_{n_1}$ is continuously differentiable on $[0,\infty)$ with $f_{n_1}(0)=f'_{n_1}(0)=0$, $f_i \geq f_j$ for $i\leq j$ and $f'_i\geq 0$. Then, SURE (\ref{SURE}) is an unbiased estimator of the MSE. 
					\BE\label{SURE}
					\begin{aligned}
						\operatorname{SURE}=\|\mathcal{D}_{\boldsymbol{L}}(\boldsymbol{L}_{\N})-\boldsymbol{L}_{\N}\|^2_F+2v\operatorname{div}\left( \mathcal{D}_{\boldsymbol{L}}(\boldsymbol{L}_{\N})\right) -n_1n_2v.
					\end{aligned}
					\EE

			\end{lemma}
			\begin{lemma}[Lipschitz property of $\mathcal{D}_{\boldsymbol{L}}$]\label{Lip_lemma_f}
			Let  Assumption (\ref{as3}.v) and (\ref{as3}.vi) hold. Let $\tilde{{\sigma}}_{i,\boldsymbol{L}}$ be the $i$-th singular value of $\frac{\boldsymbol{L}}{\sqrt{n_2v}}$ and $\tilde{\sigma}_{i,\N}$ be the $i$-th singular of  $\frac{\N}{\sqrt{n_2}}$.
			Assume that  the largest singular value of $\frac{\boldsymbol{L}}{\sqrt{n_2}}$ is bounded by $C_1$. 
			Denote $F_1:\mathbb{R}^{n_1\times n_2 }\to\mathbb{R}^{n_1\times n_2 }$, i.e.,
			\BE
			\begin{aligned}
				F_1(\sqrt{v}\boldsymbol{N})\triangleq\sqrt{n_2 v}\mathcal{D}_{\boldsymbol{L}}\left( \frac{\boldsymbol{L}+\sqrt{v}\boldsymbol{N}}{\sqrt{n_2 v}}\right).
			\end{aligned}
			\EE			
			Then, $F_1(\sqrt{v}\boldsymbol{N})$ is a uniformly pseudo-Lipschitz function of order 1 (namely, Lipschitz) with pseudo-Lipschitz constant $\frac{\sqrt{26}}{2} M_{\frac{C_1}{\sqrt{v}}+1+\sqrt{\beta}}$,
			where $M_{\frac{C_1}{\sqrt{v}}+1+\sqrt{\beta}}$ is the Lipschitz constant of $f_{\boldsymbol{L}}$ in $[0,M_{\frac{C_1}{\sqrt{v}}+1+\sqrt{\beta}}]$.	Denote $F_2:\mathbb{R}^{n_1}\to\mathbb{R}$, i.e.,
			\BE
			\begin{aligned}
				F_2(\tilde{\sigma}_{1,\N}, \tilde{\sigma}_{2,\N},\cdots,\tilde{\sigma}_{n_1,\N})=\frac{1}{\sqrt{n_1}} \left\|  \D_{\boldsymbol{L}}\left(\frac{ \boldsymbol{L}}{\sqrt{n_2v}}+\U_{\N}\frac{\boldsymbol{\Sigma}_{\N}}{\sqrt{n_2}}\V^T_{\N}\right) - \frac{ \boldsymbol{L}}{\sqrt{n_2v}}\right\| ^2_F, 
			\end{aligned}
			\EE
			where $\operatorname{diag}(\frac{\boldsymbol{\Sigma}_{\N}}{\sqrt{n_2}})=\tilde{\boldsymbol{\sigma}}_{\N}=[\tilde{\sigma}_{1,\N}, \tilde{\sigma}_{2,\N},\cdots,\tilde{\sigma}_{n_1,\N}]$. Then, $	F_2(\tilde{\sigma}_{1,\N}, \tilde{\sigma}_{2,\N},\cdots,\tilde{\sigma}_{n_1,\N})$ is a uniformly pseudo-Lipschitz  function of order 1   with  pseudo-Lipschitz constant  $\sqrt{26}M_{\frac{C_1}{\sqrt{v}}+1+\sqrt{\beta}}\left( \frac{C_1}{\sqrt{v}}+M_{\frac{C_1}{\sqrt{v}}+1+\sqrt{\beta}}(\frac{C_1}{\sqrt{v}}+1+\sqrt{\beta})\right)$.
		\end{lemma}
	 \begin{proof}
	To prove the Lipschitz property, we bound the operator norm of the corresponding Jacobi matrix  . To do this, we present the first order differential form of $\mathcal{D}_{\boldsymbol{L}}$ and show the boundedness of  the operator norm of the Jacobi matrix  $\boldsymbol{J}\in\mathbb{R}^{n_1n_2\times n_1n_2}$ of $\operatorname{vec}(F_1(\sqrt{v}\boldsymbol{N})): \mathbb{R}^{n_1n_2}\to\mathbb{R}^{n_1n_2}$.  We prove the existence of a constant independent of the dimensions $(n_1,n_2)$. Let $\boldsymbol{\Delta}\in\mathbb{R}^{n_1\times n_2}$ with $\|\boldsymbol{\Delta}\|^2_F=1$.
		Consider the differential form of $\mathcal{D}_{\boldsymbol{L}}$ at $\frac{\boldsymbol{L}+\sqrt{v}\N}{\sqrt{n_2v}}$ under perturbation $\epsilon\boldsymbol{\Delta}$, where $\epsilon$ is small enough.
		Then, 
		\BE
		\begin{aligned}
			F_1(\sqrt{v}\boldsymbol{N}+\epsilon\boldsymbol{\Delta})\!-\! F_1(\sqrt{v}\boldsymbol{N})=&\sqrt{n_2v}\mathcal{D}_{\boldsymbol{L}}(\frac{\boldsymbol{L}+\sqrt{v}\N+\epsilon\boldsymbol{\Delta}}{\sqrt{n_2v}})-\sqrt{n_2v}	\mathcal{D}_{\boldsymbol{L}}(\frac{\boldsymbol{L}+\sqrt{v}\N}{\sqrt{n_2v}})\\
			=&\left(\! \U\! \left(\!  (\![\boldsymbol{\Sigma}_{I}\!\boldsymbol{\Sigma}_{DI}\!]\!)\!\circ(\!\U^{T}\!\boldsymbol{\Delta}\!\V\!) \! +\![\!\boldsymbol{\Sigma}_{IJ}\circ\left[\! (\U^{T}\!\boldsymbol{\Delta}\V)_{1:n_1,1:n_1}^{T}\!\right] \!,\boldsymbol{0}] \!\right)\!\V^{T}\!\right) \epsilon\!+o(\epsilon),
		\end{aligned}
		\EE
		where $(\U^{T}\boldsymbol{\Delta}\V)_{1:n_1,1:n_1}\in\mathbb{R}^{n_1\times n_1}$ represents the matrix formed by taking the first $n_1$ rows and columns of $(\U^{T}\boldsymbol{\Delta}\V)$ and the matrix $\boldsymbol{0}$ is an $n_1 \times (n_2 -n_1)$ matrix with all elements equal to zero.
		\\
		Note that 
		\BE
		\begin{aligned}
			\boldsymbol{J}\operatorname{vec}(\boldsymbol{\Delta})= \U\! \left(  (\![\!\boldsymbol{\Sigma}_{I}\!\boldsymbol{\Sigma}_{DI}\!]\!)\circ(\!\U^{T}\!\boldsymbol{\Delta}\!\V\!) \!\!+\![\boldsymbol{\Sigma}_{IJ}\circ\left[\! (\U^{T}\!\boldsymbol{\Delta}\V)_{1:n_1,1:n_1}^{T}\!\right] \!,\!\boldsymbol{0}] \!\right)\!\V^{T}	.\notag
		\end{aligned}
		\EE
		We only need to show the boundedness of 
		\BE
		\begin{aligned}
						\| \boldsymbol{J}\|_{op}=\underset{\|\boldsymbol{\Delta}\|^2_F=1}{\operatorname{sup}}
						\|\boldsymbol{J}\operatorname{vec}(\boldsymbol{\Delta})\|.
		\end{aligned}
		\EE
		Note that
		\BS
		\begin{align}
			\left\| \boldsymbol{J}\operatorname{vec}(\boldsymbol{\Delta})\right\| ^2_F&=\left\|  \U \left(  ([\boldsymbol{\Sigma}_{I}\boldsymbol{\Sigma}_{DI}])\circ(\U^{T}\boldsymbol{\Delta}\V) +[\boldsymbol{\Sigma}_{IJ}\circ\left[ (\U^{T}\boldsymbol{\Delta}\V)_{1:n_1,1:n_1}^{T}\right] ,\boldsymbol{0}] \right)\V^{T}\right\|^2_F\\
			&\leq2\left\|   ([\boldsymbol{\Sigma}_{I}\boldsymbol{\Sigma}_{DI}])\circ(\U^{T}\boldsymbol{\Delta}\V) \right\|^2_F+2\left\|  [\boldsymbol{\Sigma}_{IJ}\circ\left[ (\U^{T}\boldsymbol{\Delta}\V)_{1:n_1,1:n_1}^{T}\right] ,\boldsymbol{0}] \right\|^2_F\\
			&=2\sum_{i,j}[\boldsymbol{\Sigma}_{I}\boldsymbol{\Sigma}_{DI}]_{i,j}^2(\U^{T}\boldsymbol{\Delta}\V)_{i,j}^2 +2\left\|  [\boldsymbol{\Sigma}_{IJ}\circ\left[ (\U^{T}\boldsymbol{\Delta}\V)_{1:n_1,1:n_1}^{T}\right] ,\boldsymbol{0}] \right\|^2_F\\
			&\leq 2\left( \max_{i,j} |[\boldsymbol{\Sigma}_{I}\boldsymbol{\Sigma}_{DI}]|\right) ^2\sum_{i,j}(\U^{T}\boldsymbol{\Delta}\V)_{i,j}^2+2\left( \max_{i,j}|\boldsymbol{\Sigma}_{IJ}|\right) ^2\sum_{i,j}(\U^{T}\boldsymbol{\Delta}\V)_{i,j}^2\\
			&\leq 2\left( \max_{i,j}|[\boldsymbol{\Sigma}_{I}\boldsymbol{\Sigma}_{DI}]|\right) ^2+2\left( \max_{i,j}|\boldsymbol{\Sigma}_{IJ}|\right) ^2\leq \frac{13}{2} M_{\tilde{\sigma}_{1}}^2<\infty
		\end{align}
		\ES
		where $M_{\tilde{\sigma}_{1}}$ is the Lipschitz constant of $f_{\boldsymbol{L}}$ in $[0,\tilde{\sigma}_{1}]$,  $ \max_{i,j} |[\boldsymbol{\Sigma}_{I}\boldsymbol{\Sigma}_{DI}]|$ denotes the largest absolute value of the entries of $[\boldsymbol{\Sigma}_{I}\boldsymbol{\Sigma}_{DI}]$, and   $\max_{i,j}|\boldsymbol{\Sigma}_{IJ}|$ denotes the largest absolute value of the entries of $\boldsymbol{\Sigma}_{IJ}$. 
		Recall that the $(i,j)$-th entry of $\boldsymbol{\Sigma}_{I}$ is $H_{\boldsymbol{L}}(\tilde{\sigma}_{i},\tilde{\sigma}_{j})$
		for $1\leq i,j\leq n_1$, $i\neq j$ and the $(i,i)$-th entry of $\boldsymbol{\Sigma}_{I}$ is $f_{\boldsymbol{L}}^{'}(\tilde{\sigma}_{i})$   for $1\leq i\leq n_1$. The $(i,j)$-th entry of $\boldsymbol{\Sigma}_{IJ}$ is $G_{\boldsymbol{L}}(\tilde{\sigma}_{i},\tilde{\sigma}_{j}) $    for $1\leq i,j\leq n_1$ and $i\neq j$. The $(i,i)$-th entry of $\boldsymbol{\Sigma}_{IJ}$ is $0$ for $1\leq i,j\leq n_1$. The $(i,j)$-th entry of $\boldsymbol{\Sigma}_{DI}$ is $\frac{f_{\boldsymbol{L}}(\tilde{\sigma}_{i})}{\tilde{\sigma}_{i}}$ for $1\leq i\leq n_1$ and $n_1< j \leq n_2$.
		Then
		\BE
		\begin{aligned}
			\max_{i,j}|[\boldsymbol{\Sigma}_{I}\boldsymbol{\Sigma}_{DI}]|\leq \frac{3}{2} M_{\tilde{\sigma}_{1}} \ \ \text{and} \ \  \max_{i,j}|\boldsymbol{\Sigma}_{IJ}|\leq  M_{\tilde{\sigma}_{1}}.
		\end{aligned}
		\EE
		Note that
		\BE
		\begin{aligned}
				\limsup_{n_1\to \infty} \tilde{\sigma}_{1}\leq \frac{C_1}{\sqrt{v}}+1+\sqrt{\beta}\ \  a.s. .
		\end{aligned}
		\EE
		Then, we obtain the pseudo-Lipschitz constant $L_{F_1}$ of the uniformly Lipschitz function as $n_1\to\infty$
		\BE
		\begin{aligned}
			\| \boldsymbol{J}\|_{op} \leq \frac{\sqrt{26}}{2} M_{\frac{C_1}{\sqrt{v}}+1+\sqrt{\beta}}\triangleq L_{F_1}.
		\end{aligned}
		\EE		
		Consider the perturbation $\epsilon\boldsymbol{d}\in\mathbb{R}^{n_1}$ with $\|\boldsymbol{d}\|^2=1$. Let $\boldsymbol{\Delta}_1=\U_{\N}\operatorname{diag}(\boldsymbol{d})\V^T_{\N}$.
		Here, $\operatorname{diag}(\boldsymbol{d})\in\mathbb{R}^{n_1\times n_2}$, where $(\operatorname{diag}(\boldsymbol{d}))_{i,i}=\boldsymbol{d}_i$ denotes that the elements on the diagonal are given by the elements of vector $\boldsymbol{d}$, with all other entries being zero.
		We show that
		the differential form of $					 			F_2(\tilde{\sigma}_{1,\N}, \tilde{\sigma}_{2,\N},\cdots,\tilde{\sigma}_{n_1,\N})$ is given by
		\BE
		\begin{aligned}
			&\lim_{\epsilon\to 0}\frac{F_2(\tilde{\boldsymbol{\sigma}}_{\N}+\epsilon\boldsymbol{d})-F_2(\tilde{\boldsymbol{\sigma}}_{\N})}{\epsilon}\\
			=&\lim_{\epsilon\to 0}\frac{ \frac{1}{\sqrt{n_1}} \left\|  \D_{\boldsymbol{L}}\left(\frac{ \boldsymbol{L}}{\sqrt{n_2v}}+\U_{\N}\frac{\boldsymbol{\Sigma}_{\N}}{\sqrt{n_2}}\V^T_{\N}+\epsilon\U_{\N}\operatorname{diag}(\boldsymbol{d})\V^T_{\N}\right)\right\|^2_F -\frac{1}{\sqrt{n_1}} \left\| \D_{\boldsymbol{L}}\left(\frac{ \boldsymbol{L}}{\sqrt{n_2v}}+\U_{\N}\frac{\boldsymbol{\Sigma}_{\N}}{\sqrt{n_2}}\V^T_{\N}\right)\right\|^2_F }{\epsilon}\\
			&-\!\lim_{\epsilon\to 0}\frac{\!\frac{2}{\sqrt{n_1}}\! \left\langle\! \D_{\boldsymbol{L}}\!\!\left(\!\frac{\!\boldsymbol{L}\!}{\sqrt{n_2v}}\!+\!\U_{\N}\!\frac{\!\boldsymbol{\Sigma}_{\N}\!}{\sqrt{n_2}}\!\V^T_{\N}\!+\!\epsilon\U_{\N}\operatorname{diag}(\boldsymbol{d})\V^T_{\N}\!\right)\!,\! \frac{\!\boldsymbol{L}\!}{\sqrt{n_2v}}\!\right\rangle\!-\!\frac{1}{\sqrt{n_1}}\! \left\langle\! \D_{\boldsymbol{L}}\!\left(\!\frac{\!\boldsymbol{L}\!}{\sqrt{n_2v}}\!+\!\U_{\N}\frac{\!\boldsymbol{\Sigma}_{\N}\!}{\sqrt{n_2}}\!\V^T_{\N}\!\right)\!,\! \frac{\!\boldsymbol{L}\!}{\sqrt{n_2v}}\!\right\rangle}{\epsilon}
			\\
			=& \frac{2}{\sqrt{n_1}}\left\langle\left(\! \U\! \left( [\boldsymbol{\Sigma}_{I}\!\boldsymbol{\Sigma}_{DI}]\!\circ(\U^{T}\!\boldsymbol{\Delta}_1\!\V) \!+\![\boldsymbol{\Sigma}_{IJ}\circ\left[\! (\U^{T}\!\boldsymbol{\Delta}_1\V)_{1:n_1,1:n_1}^{T}\!\right] \!,\!\boldsymbol{0}] \right)\!\V^{T}\!\right),\D_{\boldsymbol{L}}\left(\frac{ \boldsymbol{L}}{\sqrt{n_2v}}+\U_{\N}\frac{\boldsymbol{\Sigma}_{\N}}{\sqrt{n_2}}\V^T_{\N}\right) \right\rangle\\
			& -\frac{2}{\sqrt{n_1}}\left\langle\left(\! \U\! \left(\ [\boldsymbol{\Sigma}_{I}\!\boldsymbol{\Sigma}_{DI}]\!\circ(\U^{T}\!\boldsymbol{\Delta}_1\!\V) +[\!\boldsymbol{\Sigma}_{IJ}\circ\left[\! (\U^{T}\!\boldsymbol{\Delta}_1\V)_{1:n_1,1:n_1}^{T}\!\right] \!,\!\boldsymbol{0}]\! \right)\!\V^{T}\right),\frac{ \boldsymbol{L}}{\sqrt{n_2v}} \right\rangle \\
			=&\frac{2}{\sqrt{n_1}}\operatorname{vec}(\mathcal{D}_{\boldsymbol{L}}\left( \frac{\boldsymbol{L}+\sqrt{v}\boldsymbol{N}}{\sqrt{n_2 v}}\right))^{T}\boldsymbol{J}\operatorname{vec}(\boldsymbol{\Delta}_1)-\frac{2}{\sqrt{vn_1n_2}}\operatorname{vec}(\boldsymbol{L})^{T}\boldsymbol{J}\operatorname{vec}(\boldsymbol{\Delta}_1)
		\end{aligned} 
		\EE
		Then, we conclude that
		\BE
		\begin{aligned}
			\left| \lim_{\epsilon\to 0}\frac{F_2(\tilde{\boldsymbol{\sigma}}_{\N}+\epsilon\boldsymbol{d})-F_2(\tilde{\boldsymbol{\sigma}}_{\N})}{\epsilon}\right| &\leq \frac{2\|\boldsymbol{L}\|_F}{\sqrt{vn_1n_2}}\|\boldsymbol{J}\|_{op}+\frac{2\|\mathcal{D}_{\boldsymbol{L}}\left( \frac{\boldsymbol{L}+\sqrt{v}\boldsymbol{N}}{\sqrt{n_2 v}}\right)\|_F}{\sqrt{n_1}}\|\boldsymbol{J}\|_{op}\\
			& \leq \left( 2\tilde{\sigma}_{1,\boldsymbol{L}}+2f_{\boldsymbol{L}}(\tilde{\sigma}_{1})\right) \|\boldsymbol{J}\|_{op}
		\end{aligned}
		\EE
		Similar to $L_{F_1}$, we obtain
		\BE
		\begin{aligned}
			\limsup_{n_1\to\infty}\left( 2\tilde{\sigma}_{1,\boldsymbol{L}}+2f_{\boldsymbol{L}}(\tilde{\sigma}_{1})\right) \|\boldsymbol{J}\|_{op}\leq \sqrt{26}M_{\frac{C_1}{\sqrt{v}}+1+\sqrt{\beta}}\left(\! \frac{C_1}{\sqrt{v}}+M_{\frac{C_1}{\sqrt{v}}+1+\sqrt{\beta}}(\frac{C_1}{\sqrt{v}}+1+\sqrt{\beta})\!\right)\triangleq L_{F_2},
		\end{aligned}
		\EE
		which completes the proof.
	\end{proof}

		\begin{definition}[Wasserstein distances] \label{wasserstein}
		Let $\mathcal{P}(r)$ be the set of all Borel probability measures $P$ on $\mathcal{R}$ with $\int_{\mathcal{R}}\|x\|^r\ dP(x)<\infty$, where $r\in[1,\infty)$. For $P,Q\in\mathcal{P}(r)$, the r-Wasserstein distance between $P$ and $Q$ is defined by
		\BE
		\begin{aligned}
			d_r(P, Q):=\inf _{(X, Y)} \mathbb{E}\left(\|X-Y\|^r\right)^{1 / r}
		\end{aligned}
		\EE
		where the infimum is taken over all pairs of random vectors $(X, Y )$ defined on a common probability space with $X \sim P$ and $Y \sim Q$.
	\end{definition}

		\begin{lemma}[Weak Convergence with uniformly Lipschitz function $F_2$]\label{fLasy}
		Let $\{\sigma_{i}\}_{i=1}^{n_1}$ and $\{\lambda_{i}\}_{i=1}^{n_1}$ be two sets, where $\sigma_{i}\in \mathbb{R}^{+}$ and $\lambda_{i}\in\mathbb{R}^{+}$ are bounded. $\{\sigma_{i}\}_{i=1}^{n_1}$ and $\{\lambda_{i}\}_{i=1}^{n_1}$ separately have no repeated elements. Let $\boldsymbol{\sigma}=[\sigma_{1},\sigma_{2},\cdots,\sigma_{n_1}]$ and $\boldsymbol{\lambda}=[\lambda_{1},\lambda_{2},\cdots,
		\lambda_{n_1}]$. Assume that the empirical distribution of $\{\sigma_{i}\}_{i=1}^{n_1}$ and $\{\lambda_{i}\}_{i=1}^{n_1}$ converge weakly to the same limiting distribution $\mu_{\N}$ as $n_1\to\infty$. $F_2(\sigma_{1},\sigma_{2},\cdots,\sigma_{n_1})$ is defined as in Lemma \ref{Lip_lemma_f}.   Then, there exits a permutation  $\{\pi_i\}$  of $\{1,2,\cdots,n_1\}$ such that, as $n_1\to\infty$,
		\BE
		\begin{aligned}
			\frac{1}{\sqrt{n_1}}F_2(\sigma_{1},\sigma_{2},\cdots,\sigma_{n_1})- 	\frac{1}{\sqrt{n_1}}F_2(\lambda_{\pi_1},\lambda_{\pi_2},\cdots,\lambda_{\pi_{n_1}})\to 0.
		\end{aligned}
		\EE
	\end{lemma}
	
	\begin{proof}
		Let $\mu_{\sigma}(n_1)$ denote the empirical distribution of $\{\sigma_{i}\}_{i=1}^{n_1}$  and $\mu_{\lambda}(n_1)$   denote the empirical distribution of $\{\lambda_{i}\}_{i=1}^{n_1}$.
	Note that by assumption, $\mu_{\sigma}(n_1)$ and $\mu_{\lambda}(n_1)$ converge weakly to the same limiting distribution $\mu_{\N}$. $\sigma_{i}\in \mathbb{R}^{+}$ and $\lambda_{i}\in\mathbb{R}^{+}$ are bounded. From Theorem 7.12 in \cite{villani2021topics}, we obtain 		
		\BE
		\begin{aligned}
			d_2(\mu_{\sigma}(n_1), \mu_{\N})\to 0 \quad \text{and} 		\quad			d_2(\mu_{\lambda}(n_1), \mu_{\N})\to 0,
		\end{aligned}
		\EE
		where $d_2$ is the $2$-Wasserstein distance defined in Definition \ref{wasserstein}.
		Then, we obtain 
		\BE\label{usconv}
		\begin{aligned}
			d_2(\mu_{\lambda}(n_1),\mu_{\sigma}(n_1))\to 0.
		\end{aligned}
		\EE
		Recall that the 2-Wasserstein distance for empirical measure is defined
		\BE\label{liW2}
		\begin{aligned}
			d_2(\mu_{\lambda}(n_1),\mu_{\sigma}(n_1))
			=\sqrt{\frac{1}{n_1}\min_{\pi}\sum_{i=1}^{n_1}\|\sigma_i-\lambda_{\pi_i} \|^2}=\sqrt{\frac{1}{n_1}\min_{\boldsymbol{P}}\left\|\boldsymbol{\sigma}-\boldsymbol{\lambda}\boldsymbol{P}\right\| ^2}
		\end{aligned}
		\EE
		where the minimize is taken over all permutations $\pi$.  
		This is a linear assignment problem, and can be solved by the Hungarian algorithm in cubic time\cite{phelps2001lectures,1986Matching}.		Denote the optimal permutation  by $\boldsymbol{\lambda}_{\pi^*}$ and the corresponding permutation matrix $\boldsymbol{P}^*$ such that $\boldsymbol{\lambda}_{\pi^*}=\boldsymbol{\lambda}\boldsymbol{P}^*$.			
		Recall that $F_2(\sigma_{1},\sigma_{2},\cdots,\sigma_{n_1})$ is a uniformly pseudo-Lipschitz function of order $1$, with pseudo-Lipschitz constant $L_{F_2}$.
		Then
		\BE
		\begin{aligned}
			\frac{1}{\sqrt{n_1}}\|F_2(\boldsymbol{\sigma})- F_2(\boldsymbol{\lambda}\boldsymbol{P}^*)\|_2
			\leq \frac{L_{F_2,n_1}}{\sqrt{n_1}}\|\boldsymbol{\sigma}-\boldsymbol{\lambda}\boldsymbol{P}^*\|_2=L_{F_2,n_1}d_2(\mu_{\lambda}(n_1),\mu_{\sigma}(n_1))
		\end{aligned}
		\EE
	Note that as $n_1\to\infty$, $		\limsup_{n_1\to \infty} L_{F_2,n_1}<L_{F_2}$. Together with (\ref{usconv}), we obtain
		\BE
		\begin{aligned}
			\frac{1}{\sqrt{n_1}}\|f(\boldsymbol{\sigma})- f(\boldsymbol{P}^*\boldsymbol{\lambda})\|_2\to 0,
		\end{aligned}
		\EE
		which completes the proof.
	\end{proof}

		In the following, we mainly concentrate the MSE function of the Low-rank denoiser. In this lemma, we first prove the convergence of the Stein's unbiased risk estimate for $\D_{\boldsymbol{L}}$ defined in  Assumption (\ref{as3}.vi). Then, we show that SURE is an unbiased estimator for the MSE of $\D_{\boldsymbol{L}}$ when $f_{\boldsymbol{L}}$ is continuously differentiable. Next, we obtain analytical form of the asymptotic MSE under Gaussian noise. By the property of the low-rank denoiser and the random matrix theorem, we finally conclude that the asymptotic MSE under Gaussian is identical to that under i.i.d. non Gaussian.
	\begin{lemma}\label{main_lemma_L}
				Denote by $\boldsymbol{L}=\U_{\boldsymbol{L}}\boldsymbol{\Sigma}_{\boldsymbol{L}}\V_{\boldsymbol{L}}^T$ the SVD of the low-rank matrix $\boldsymbol{L}$ and $\N=\U_{\N}\boldsymbol{\Sigma}_{\N}\V_{\N}^T$ the SVD of the noise matrix $\N$.  Let $\tilde{\sigma}_{i,\boldsymbol{L}}$ be the $i$-th singular value of $\frac{\boldsymbol{L}}{\sqrt{n_2v}}$ and $\tilde{\sigma}_{i,\N}$ be the $i$-th singular value of $\frac{\N}{\sqrt{n_2}}$. Let $\tilde{\boldsymbol{\sigma}}_{\boldsymbol{L}}=[\tilde{\sigma}_{1,\boldsymbol{L}}, \tilde{\sigma}_{2,\boldsymbol{L}},\cdots,\tilde{\sigma}_{n_1,\boldsymbol{L}}]$ and $\tilde{\boldsymbol{\sigma}}_{\N}=[\tilde{\sigma}_{1,\N}, \tilde{\sigma}_{2,\N},\cdots,\tilde{\sigma}_{n_1,\N}]$. Specifically, in this lemma, we use $\boldsymbol{G}$ to denote an i.i.d. random Gaussian matrix with zero mean and unit variance, and $\N$ to denote a general i.i.d. random matrix.			
				Define
				\BE\label{def_pr_F}
				\begin{aligned}
					F(\U_{\boldsymbol{L}},\U_{\N},\V_{\boldsymbol{L}},\V_{\N},\tilde{\boldsymbol{\sigma}}_{\boldsymbol{L}},\tilde{\boldsymbol{\sigma}}_{\N})=\ &\frac{1}{n_1n_2}\left\| \D_{\boldsymbol{L}}\left(\boldsymbol{L}+\sqrt{v}\N\right)-\boldsymbol{L}\right\| ^2_F\\
					=\ &\frac{v}{n_1}\left\| \D_{\boldsymbol{L}}\left(\U_{\boldsymbol{L}}\frac{\boldsymbol{\Sigma}_{\boldsymbol{L}}}{\sqrt{n_2v}}\V^T_{\boldsymbol{L}}+\U_{\N}\frac{\boldsymbol{\Sigma}_{\N}}{\sqrt{n_2}}\V^T_{\N}\right)-\U_{\boldsymbol{L}}\frac{\boldsymbol{\Sigma}_{\boldsymbol{L}}}{\sqrt{n_2v}}\V^T_{\boldsymbol{L}}\right\| ^2_F.
				\end{aligned}
				\EE
				where $\tilde{\boldsymbol{\sigma}}_{\boldsymbol{L}}=\operatorname{diag}(\frac{\boldsymbol{\Sigma}_{\boldsymbol{L}}}{\sqrt{n_2v}})$ and $\tilde{\boldsymbol{\sigma}}_{\N}=\operatorname{diag}(\frac{\boldsymbol{\Sigma}_{\N}}{\sqrt{n_2}})$. The following holds:
				\begin{enumerate}[ (\ref{main_lemma_L}.a)]
					\item Let Assumption 
					(\ref{as3}.v) and (\ref{as3}.vi) hold.  The SVD form of an i.i.d. random Gaussian $\boldsymbol{G}$ is given by $\boldsymbol{G}=\U_{\boldsymbol{G}}\boldsymbol{\Sigma}_{\boldsymbol{G}}\V_{\boldsymbol{G}}^T$. Let $\tilde{\boldsymbol{\sigma}}_{\boldsymbol{G}}=\operatorname{diag}(\frac{\boldsymbol{\Sigma}_{\boldsymbol{G}}}{\sqrt{n_2}})$. As $n\to\infty$
					\BE
					\begin{aligned}
						F(\U_{\boldsymbol{L}},\U_{\boldsymbol{G}},\V_{\boldsymbol{L}},\V_{\boldsymbol{G}},\tilde{\boldsymbol{\sigma}}_{\boldsymbol{L}},\tilde{\boldsymbol{\sigma}}_{\boldsymbol{G}})\overset{p}{=} \operatorname{AMSE}(v) 
					\end{aligned}
					\EE
					where $\operatorname{AMSE}(v) =v\left( \mathbb{E} \left[ \left( f_{\boldsymbol{L}}(\!\sigma\!)-\!\sigma\!\right) ^2\right] +2(1-\beta)\mathbb{E}\left[  g_{\boldsymbol{L}}(\sigma)\right]+ 2\beta \mathbb{E}_{\sigma,\hat{\sigma}}\left[ H_{\boldsymbol{L}}(\sigma,\hat{\sigma})\right]-1\right)$				
					Here, $g_{\boldsymbol{L}}$:  $[0,+\infty)\rightarrow[0,+\infty)$ with $g_{\boldsymbol{L}}(x)=\frac{f_{\boldsymbol{L}}(x)}{x}$ for $x>0$ and $g_{\boldsymbol{L}}(0)=0$. $H_{\boldsymbol{L}}(x,y)$ is defined as 	
					\BE
					\begin{aligned}
						H_{\boldsymbol{L}}(x,y)\triangleq\left\{
						\begin{array}{cl}
							0                           & (x,y)=(0,0)\\
							\frac{xf_{\boldsymbol{L}}'(x)+f_{\boldsymbol{L}}(x)}{2x}      & x=y\neq 0\\
							\frac{xf_{\boldsymbol{L}}(x)-yf_{\boldsymbol{L}}(y)}{x^2-y^2}. & x \neq y\\
						\end{array}\right.
					\end{aligned}
					\EE
					The expectation is taken over two independent and identical distributed variables,  $\sigma$ and $\hat{\sigma}$.  Let $\M_{\boldsymbol{L}+\sqrt{v}\N}=\frac{1}{n_2}(\boldsymbol{L}+\sqrt{v}\N)(\boldsymbol{L}+\sqrt{v}\N)^{T}$. Define the corresponding limiting  empirical eigenvalue distribution by $\mu_{\M_{\boldsymbol{L}+\sqrt{v}\N}}$. 
					Here $\sigma, \hat{\sigma} \sim \mu_{\boldsymbol{L},v}$ with $\mu_{\boldsymbol{L},v}(\sigma)=2v\sigma\mu_{\M_{\boldsymbol{L}+\sqrt{v}\N}}(v\sigma^2)$. And $\mu_{\M_{\boldsymbol{L}+\sqrt{v}\N}}$ is specified by the following Stieljes transform 
					\BE
					\begin{aligned}
						m(z):=\int\frac{1}{\lambda-z}d\mu_{\M_{\boldsymbol{L}+\sqrt{v}\N}}(\lambda)\notag
					\end{aligned}
					\EE
					where $m(z)$ is given by 
					\BE\label{eq_trans}
					\begin{aligned}
						m(z)
						=\int \frac{d\mu_{\boldsymbol{L}\boldsymbol{L}^T}(t)}{\frac{t}{1+v\beta m(z)}-(1+v\beta m(z)) z+v(1-\beta)}.
					\end{aligned}
					\EE
					Here, $\mu_{\boldsymbol{L}\boldsymbol{L}^T}(t)$  is the limiting  empirical eigenvalue distribution of $\frac{1}{n_2}\boldsymbol{L}\boldsymbol{L}^{T}$, with $\mu_{\boldsymbol{L}\boldsymbol{L}^T}(t)=\frac{\mu_{\boldsymbol{L}}(\sqrt{t})}{2\sqrt{t}}$.
					\item Let Assumption 
					(\ref{as3}.v) and (\ref{as3}.vi) hold.  Let $\tilde{\U}_{\boldsymbol{L}},\tilde{\U}_{\N},\tilde{\V}_{\boldsymbol{L}},\tilde{\V}_{\N} \in \mathcal{O}{n}$	be independent haar matrices, which are also independent with $\tilde{\boldsymbol{\sigma}}_{\boldsymbol{L}}, \tilde{\boldsymbol{\sigma}}_{\boldsymbol{G}}$. Then	  
					\BE
					\begin{aligned}
						F(\U_{\boldsymbol{L}},\U_{\boldsymbol{G}},\V_{\boldsymbol{L}},\V_{\boldsymbol{G}},\tilde{\boldsymbol{\sigma}}_{\boldsymbol{L}},\tilde{\boldsymbol{\sigma}}_{\boldsymbol{G}})\overset{p}{=}	F(\tilde{\U}_{\boldsymbol{L}},\tilde{\U}_{\N},\tilde{\V}_{\boldsymbol{L}},\tilde{\V}_{\N},\tilde{\boldsymbol{\sigma}}_{\boldsymbol{L}},\tilde{\boldsymbol{\sigma}}_{\boldsymbol{G}})
					\end{aligned}
					\EE			 
					\item Let Assumption 
					(\ref{as3}.iv-vi) hold. Then
					\BE
					\begin{aligned}
						F(\U_{\boldsymbol{L}},\U_{\N},\V_{\boldsymbol{L}},\V_{\N},\tilde{\boldsymbol{\sigma}}_{\boldsymbol{L}},\tilde{\boldsymbol{\sigma}}_{\N})\overset{p}{=}F(\U_{\boldsymbol{L}},\U_{\boldsymbol{G}},\V_{\boldsymbol{L}},\V_{\boldsymbol{G}   },\tilde{\boldsymbol{\sigma}}_{\boldsymbol{L}},\tilde{\boldsymbol{\sigma}}_{\boldsymbol{G}}),
					\end{aligned}
					\EE	
namely,
\BE
\begin{aligned}
	\ &\frac{1}{n_1n_2}\left\| \D_{\boldsymbol{L}}\left(\boldsymbol{L}+\sqrt{v}\N\right)-\boldsymbol{L}\right\| ^2_F\overset{p}{=} \operatorname{AMSE}(v). 
\end{aligned}
\EE
				\end{enumerate}
		\end{lemma}
				\begin{proof}
					\begin{enumerate}
						
         \item{Proof of (Conclusion:(\ref{main_lemma_L}.a))}
		With  Assumption \ref{as1}, the spectral denoiser $\D_{\boldsymbol{L}}$ satisfies the requirements in \cite[Theorem 1.1]{dozier2007empirical}. $\boldsymbol{L}+\sqrt{v}\N$ is simple and full rank with probably $1$. Then, under Gaussian noise, we conclude that  SURE is an unbiased estimator for the MSE of $\D_{\boldsymbol{L}}$ with continuously differentiable $f_{\boldsymbol{L}}$, where
				\BE\label{expsure}
				\begin{aligned}
					\mathbb{E}\left[\frac{v}{n_1}\left\| \D_{\boldsymbol{L}}\left( \frac{\boldsymbol{L}\!+\!\sqrt{v}\boldsymbol{G}}{\sqrt{n_2 v}}\right) \!-\!\frac{\boldsymbol{L}}{\sqrt{n_2 v}}\right\|_{F}^2\right]=	\mathbb{E}\left[\frac{v}{n_1} \left\| \mathcal{D}_{\boldsymbol{L}}(\frac{\boldsymbol{L}+\sqrt{v}\boldsymbol{G}}{\sqrt{n_2 v}})\!-\!\frac{\boldsymbol{L}\!+\!\sqrt{v}\boldsymbol{G}}{\sqrt{n_2 v}}\right\| ^2_F\!+\!\frac{v}{n_1}\frac{2}{n_2}\operatorname{div}\!\left( \D_{\boldsymbol{L}}\left(\! \frac{\boldsymbol{L}\!+\!\sqrt{v}\boldsymbol{G}}{\sqrt{vn_2 }}\!\right)\right)\! -\!v\right]. 
				\end{aligned}
				\EE
				In the above, $\operatorname{div}\left( \mathcal{D}_{\boldsymbol{L}}(\frac{\boldsymbol{L}+\sqrt{v}\boldsymbol{G}}{\sqrt{n_2 v}})\right)$ is defined by 
				\BE
				\begin{aligned}
					\operatorname{div}\left( \mathcal{D}_{\boldsymbol{L}}\left( \frac{\boldsymbol{L}+\sqrt{v}\boldsymbol{G}}{\sqrt{n_2 v}}\right)\right)  = (n_2-n_1)\sum_{i=1}^{n_1}g_{\boldsymbol{L}}(\tilde{\sigma}_{i})+\sum_{i=1}^{n_1}f'_{\boldsymbol{L}}(\tilde{\sigma}_{i})+\sum_{i,j,i \neq j}^{n_1}H_{\boldsymbol{L}}(\tilde{\sigma}_{i},\tilde{\sigma}_{j}),
				\end{aligned}
				\EE		
				where $\tilde{\sigma}_{i}$ is the $i$-th singular value of $ \frac{\boldsymbol{L}+\sqrt{v}\boldsymbol{G}}{\sqrt{n_2 v}}$.
				\begin{enumerate}[Step (1) :]
					\item To show the convergence of $\frac{v}{n_1}\left\| \D_{\boldsymbol{L}}\left( \frac{\boldsymbol{L}+\sqrt{v}\boldsymbol{G}}{\sqrt{n_2 v}}\right) -\frac{\boldsymbol{L}}{\sqrt{n_2 v}}\right\|_{F}^2$, we first show the convergence of the 	$\operatorname{SURE}$ in the large system even under i.i.d. noise $\N$. We reformulate $\operatorname{SURE}$  by replacing the $\operatorname{div}$ in Lemma \ref{candes2013unbiased}
					\BS\label{nproof_1}
					\begin{align}
						&\quad\frac{v}{n_1}\!\left\|\! \D_{\boldsymbol{L}}\left(\! \frac{\boldsymbol{L}+\sqrt{v}\N}{\sqrt{vn_2 }}\!\right)\! -\! \frac{\boldsymbol{L}+\sqrt{v}\N}{\sqrt{vn_2 }}\!\right\|_{F}^2\!+\!\frac{2v}{n_1n_2}\operatorname{div}\!\left(\D_{\boldsymbol{L}}\left(\! \frac{\boldsymbol{L}+\sqrt{v}\N}{\sqrt{vn_2 }}\!\right)\right)  -v \\
						&=v\left(\frac{1}{n_1}\sum_{i}^{n_1}(f_{\boldsymbol{L}}(\tilde{\sigma}_{i})-\tilde{\sigma}_{i})^2-1 \right)+\frac{2v}{n_1n_2}\sum_{i,j,i \neq j}^{n_1}H_{\boldsymbol{L}}(\tilde{\sigma}_{i},\tilde{\sigma}_{j})
						+2v\!\left(\! \frac{n_2\!-\!n_1}{n}\! \sum_{i=1}^{n_1}\!g_{\boldsymbol{L}}(\tilde{\sigma}_{i})\!+\!\frac{1}{n_1n_2}\sum_{i=1}^{n_1}\!f'_{\boldsymbol{L}}(\tilde{\sigma}_{i})\!\right)\\
						&\overset{p}{=} v \mathbb{E} \left[ \left( f_{\boldsymbol{L}}(\!\sigma\!)-\!\sigma\!\right) ^2\right] +2v(1-\beta)\mathbb{E}\left[  g_{\boldsymbol{L}}(\sigma)\right]
						+ 2v\beta \mathbb{E}_{\sigma,\hat{\sigma}}\left[ H_{\boldsymbol{L}}(\sigma,\hat{\sigma})\right]-v 
					\end{align}
					\ES	
					where $\tilde{\sigma}_{i}=\frac{\sigma_{i}}{\sqrt{n_2v}}$ and $\sigma_{i}$ is the $i$-th singular value of $\boldsymbol{L}+\sqrt{v}\N$. $g_{\boldsymbol{L}}$ and $H_{\boldsymbol{L}}$ are defined  in (\ref{func_h}).  The expectation is taken over the limiting empirical singular distribution of $\frac{\boldsymbol{L}+\sqrt{v}\N}{\sqrt{vn_2 }}$, denoted by $\sigma\sim \mu_{\boldsymbol{L},v}$; $\hat{\sigma}$ is an i.i.d. copy of $\sigma$; (\ref{nproof_1}b) follows the divergence of spectral estimators obtained in\cite[Theorem IV.3]{candes2013unbiased}. 		
					\\
					To prove (\ref{nproof_1}c), we first show the boundedness of the singular values $\{\tilde{\sigma}_{i}\}_1^{n_1}$.
					By \cite{1980A}, the largest singular value $\sigma_{1,\N}$ of $\frac{1}{\sqrt{n_2}} \N$
					\BE
					\begin{aligned}
						\lim\sup_{n_1 \rightarrow \infty} \sigma_{1,\N} \leqslant 1+\sqrt{\beta}  \quad a.s..
					\end{aligned}
					\EE
					Then, the largest singular value $\tilde{\sigma}_{1}$ is bounded by  $\frac{\sigma_{1,\boldsymbol{L}}}{\sqrt{v}}+1+\sqrt{\beta}$, where $\sigma_{1,\boldsymbol{L}}$ is the largest singular value of $\frac{\boldsymbol{L}}{\sqrt{n_2}}$.  By Lemma \ref{asysin}, we conclude that the empirical singular value distribution of $\frac{\boldsymbol{L}+\sqrt{v}\N}{\sqrt{vn_2 }}$ converges weakly to a deterministic distribution, denoted by $\mu_{\boldsymbol{L},v}(\sigma)$.
					Moreover, $\mu_{\boldsymbol{L},v}(\sigma)=2v\sigma\mu_{\M_{\boldsymbol{L}+\sqrt{v}\N}}(v\sigma^2)$. And $\mu_{\M_{\boldsymbol{L}+\sqrt{v}\N}}$ is specified by the following Stieljes transform 
					\BE
					\begin{aligned}
						m(z):=\int\frac{1}{\lambda-z}d\mu_{\M_{\boldsymbol{L}+\sqrt{v}\N}}(\lambda)\notag
					\end{aligned}
					\EE
					where $m(z)$ is given by 
					\BE
					\begin{aligned}
						m(z)=\int \!\frac{d\mu_{\boldsymbol{L}\boldsymbol{L}^T}(t)}{\frac{ t}{1\!+v\beta m(z)}\!-\!(\!1+v\beta m(z)\!) z\!\!+\!\!v(\!1\!-\!\beta\!)}.
					\end{aligned}
					\EE
					Furthermore, we prove the weak convergence of the product  empirical measure of  $\{(\tilde{\sigma}_{i},\tilde{\sigma}_{j})\}=\{ \tilde{\sigma}_{i}\}_{i=1}^{n_1}\times\{\tilde{\sigma}_{j}\}_{j=1}^{n_1}$ in Lemma \ref{fLasy2}.
					For the continuously differentiable $f_{\boldsymbol{L}}$,  $f'_{\boldsymbol{L}}(x)$ is continuous and $f_{\boldsymbol{L}}$ is local Lipschitz continious in $[0,\frac{\sigma_{1,\boldsymbol{L}}}{\sqrt{v}}+1+\sqrt{\beta}]$. Then, $\left( f_{\boldsymbol{L}}(\!\sigma\!)-\!\sigma\!\right) ^2$ and $g_{\boldsymbol{L}}(\sigma)$ are  bounded and continuous in bounded domain.
					We conclude that $H_{\boldsymbol{L}}(x,y)$ is bounded and continuous in compact domain $[0,\frac{\sigma_{1,\boldsymbol{L}}}{\sqrt{v}}+1+\sqrt{\beta}]\times[0,\frac{\sigma_{1,\boldsymbol{L}}}{\sqrt{v}}+1+\sqrt{\beta}]$ in Lemma \ref{lemmaH}. At last, with Definition \ref{defi_weakcon} of  weak convergence and  Lemma \ref{fLasy2}, we conclude the proof. 
					\item 	We next show that $\frac{v}{n_1}\left\| \D_{\boldsymbol{L}}\left( \frac{\boldsymbol{L}+\sqrt{v}\boldsymbol{G}}{\sqrt{n_2 v}}\right) -\frac{\boldsymbol{L}}{\sqrt{n_2 v}}\right\|_{F}^2$ converges to a deterministic limit. {We show   the analytical form of the asymptotic MSE of the denoisers under Gaussian noise.}
				    Denote
				    \BE
				    \begin{aligned}
				    	F_3(\operatorname{vec}(\boldsymbol{G}))\triangleq\frac{v}{n_1}\left\| \D_{\boldsymbol{L}}\left( \frac{\boldsymbol{L}+\sqrt{v}\boldsymbol{G}}{\sqrt{n_2 v}}\right) -\frac{\boldsymbol{L}}{\sqrt{n_2 v}}\right\|_{F}^2.
				    \end{aligned}
			        \EE
	\\	
					To show the Gaussian concentration of $F_3(\boldsymbol{G})$, we can bound $\operatorname{Var}[F_3(\operatorname{vec}(\boldsymbol{G}))]$ by Gaussian Poincaré inequality\cite{10.1093/acprof:oso/9780199535255.001.0001}. 
					\BE
					\begin{aligned}
						\operatorname{Var}[F_3(\operatorname{vec}(\boldsymbol{G}))] \leq c \mathbb{E}\left[\|\nabla F_3(\operatorname{vec}(\boldsymbol{G})) \|_2^2\right]
					\end{aligned} 
				\EE
\\					 
					Note that $F_1(\sqrt{v}\boldsymbol{G})=\sqrt{n_2 v}\D_{\boldsymbol{L}}\left(\frac{\boldsymbol{L}+\sqrt{v}\boldsymbol{G}}{\sqrt{n_2 v}}\right) $ is a uniformly pseudo-Lipschitz function of order $k=1$ as shown in Lemma \ref{Lip_lemma_f}. By \cite[Lemma C.5 ]{berthier2020state}, $\frac{1}{n}\left\|\sqrt{n_2 v}\D_{\boldsymbol{L}}\left(\frac{\boldsymbol{L}+\sqrt{v}\boldsymbol{G}}{\sqrt{n_2 v}}\right)-\boldsymbol{L} \right\|^2_F$ is a uniformly pseudo-Lipschitz function of order $k=2$. Then, by \cite[Lemma C.8]{berthier2020state}, $ \mathbb{E}\left[\|\nabla F_3(\operatorname{vec}(\boldsymbol{G})) \|_2^2\right]$ is bounded and  asymptotically converge to $0$. And we conclude 
					\BE
					\begin{aligned}
\frac{v}{n_1}\left\| \D_{\boldsymbol{L}}\left( \frac{\boldsymbol{L}+\sqrt{v}\boldsymbol{G}}{\sqrt{n_2 v}}\right) -\frac{\boldsymbol{L}}{\sqrt{n_2 v}}\right\|_{F}^2
\overset{p}{=}						\mathbb{E}\left[\frac{v}{n_1}\left\| \D_{\boldsymbol{L}}\left( \frac{\boldsymbol{L}+\sqrt{v}\boldsymbol{G}}{\sqrt{n_2 v}}\right) -\frac{\boldsymbol{L}}{\sqrt{n_2 v}}\right\|_{F}^2\right].
					\end{aligned}
					\EE
					With (\ref{expsure}) and the analytical form obtained in (\ref{nproof_1}) , we conclude that 
					\BE\label{lemma_eq_con_g}
					\begin{aligned}
						 \frac{v}{n_1}\left\| \D_{\boldsymbol{L}}\left( \frac{\boldsymbol{L}+\sqrt{v}\boldsymbol{G}}{\sqrt{n_2 v}}\right) -\frac{\boldsymbol{L}}{\sqrt{n_2 v}}\right\|_{F}^2
						\overset{p}{=}  v\mathbb{E} \left[ \left( f_{\boldsymbol{L}}(\sigma)-\sigma\right) ^2\right] -  v
						+2v(1-\beta)\mathbb{E}\left[  g_{\boldsymbol{L}}(\sigma)\right] 
						+ 2\beta  v\mathbb{E}\left[ H_{\boldsymbol{L}}(\sigma,\hat{\sigma})\right].
					\end{aligned}
					\EE
				\end{enumerate}
			\item{Proof of (Conclusion: (\ref{main_lemma_L}.b))}
			
			\textcolor{black}{	We have shown the asymptotic convergence of $F(\U_{\boldsymbol{L}},\U_{\boldsymbol{G}},\V_{\boldsymbol{L}},\V_{\boldsymbol{G}},\tilde{\boldsymbol{\sigma}}_{\boldsymbol{L}},\tilde{\boldsymbol{\sigma}}_{\boldsymbol{G}})$.
			Recall that, for an i.i.d. random Gaussian matrix $\boldsymbol{G}$, $\U_{\boldsymbol{G}}$, $\V_{\boldsymbol{G}}$ and $\tilde{\boldsymbol{\sigma}}_{\boldsymbol{G}}$ are independent with each other and  $\U_{\boldsymbol{G}}$ and $\V_{\boldsymbol{G}}$ are haar matrix. 
		    For any two independent $\U_{p}\in \mathcal{O}{n}$ and $\V_{p}\in \mathcal{O}{n}$, $\U_{p}\boldsymbol{G}\V_{p}^T$  is still an  i.i.d. random Gaussian matrix, where $\U_{p}\U_{\boldsymbol{G}}$ and $\V_{p}\V_{\boldsymbol{G}}$ are haar matrix and  independent with each other.
\\		    
		    For four independent haar matrix $\U_{1}$, $\V_{1}$,  $\U_{2}$ and $\V_{2}$, 
		    \BS\label{proof_le_rot}
		    \begin{align}
	&F(\U_{1}\U_{\boldsymbol{L}},\U_{1}\U_{2}\U_{\boldsymbol{G}},\V_{1}\V_{\boldsymbol{L}},\V_{1}\U_{2}\V_{\boldsymbol{G}},\tilde{\boldsymbol{\sigma}}_{\boldsymbol{L}},\tilde{\boldsymbol{\sigma}}_{\boldsymbol{G}})\\
		   =
		   &F(\U_{\boldsymbol{L}},\U_{2}\U_{\boldsymbol{G}},\V_{\boldsymbol{L}},\V_{2}\V_{\boldsymbol{G}},\tilde{\boldsymbol{\sigma}}_{\boldsymbol{L}},\tilde{\boldsymbol{\sigma}}_{\boldsymbol{G}})\\
		   \overset{p}{=} &\operatorname{AMSE}(v),
		    \end{align}
	        \ES
	        where (\ref{proof_le_rot}b) is easy to check with the definition in (\ref{def_pr_F}) and (\ref{proof_le_rot}c) follows from the Conclusion: (\ref{main_lemma_L}.a).
		Note that	$\U_{\boldsymbol{L}}$ and $\V_{\boldsymbol{L}}$ are independent with 	$\U_{\boldsymbol{G}}$ and $\V_{\boldsymbol{G}}$. Then, $\U_{1}\U_{\boldsymbol{L}},\U_{1}\U_{2}\U_{\boldsymbol{G}},\V_{1}\V_{\boldsymbol{L}},\V_{1}\V_{2}\V_{\boldsymbol{G}}$  are four independent haar matrices. We replace the four matrices by $\tilde{\U}_{\boldsymbol{L}},\tilde{\U}_{\N},\tilde{\V}_{\boldsymbol{L}},\tilde{\V}_{\N}$ and reformulate the Gaussian concentration as
		\BE
		\begin{aligned}
			F(\tilde{\U}_{\boldsymbol{L}},\tilde{\U}_{\N},\tilde{\V}_{\boldsymbol{L}},\tilde{\V}_{\N},\tilde{\boldsymbol{\sigma}}_{\boldsymbol{L}},\tilde{\boldsymbol{\sigma}}_{\boldsymbol{G}})
			\overset{p}{=} \operatorname{AMSE}(v). 
		\end{aligned}
	     \EE
	     Then, we obtain 
	     \BE
	     \begin{aligned}
	     		F(\U_{\boldsymbol{L}},\U_{\boldsymbol{G}},\V_{\boldsymbol{L}},\V_{\boldsymbol{G}},\tilde{\boldsymbol{\sigma}}_{\boldsymbol{L}},\tilde{\boldsymbol{\sigma}}_{\boldsymbol{G}})\overset{p}{=}	F(\tilde{\U}_{\boldsymbol{L}},\tilde{\U}_{\N},\tilde{\V}_{\boldsymbol{L}},\tilde{\V}_{\N},\tilde{\boldsymbol{\sigma}}_{\boldsymbol{L}},\tilde{\boldsymbol{\sigma}}_{\boldsymbol{G}})
	     \end{aligned}
         \EE}
         
         \item {Proof of (Conclusion: (\ref{main_lemma_L}.c))}
         
        	\textcolor{black}{ Consider the MSE under i.i.d. random matrix $\N$
         \BE
         \begin{aligned}
         		F(\U_{\boldsymbol{L}},\U_{\N},\V_{\boldsymbol{L}},\V_{\N},\tilde{\boldsymbol{\sigma}}_{\boldsymbol{L}},\tilde{\boldsymbol{\sigma}}_{\N})=\frac{v}{n_1}\left\| \D_{\boldsymbol{L}}\left(\U_{\boldsymbol{L}}\frac{\boldsymbol{\Sigma}_{\boldsymbol{L}}}{\sqrt{n_2v}}\V^T_{\boldsymbol{L}}+\U_{\N}\frac{\boldsymbol{\Sigma}_{\N}}{\sqrt{n_2}}\V^T_{\N}\right)-\U_{\boldsymbol{L}}\frac{\boldsymbol{\Sigma}_{\boldsymbol{L}}}{\sqrt{n_2v}}\V^T_{\boldsymbol{L}}\right\| ^2_F.
         \end{aligned}
        \EE
        By the rotationally invariant assumption for $\boldsymbol{L}$, $\U_{\boldsymbol{L}}$ and $\V_{\boldsymbol{L}}$ are two independent haar matrices.
        \\
        For two independent haar matrix $\U_{1}$ and $\V_{1}$,
         \BE
         \begin{aligned}
                  		F(\U_{\boldsymbol{L}},\U_{\N},\V_{\boldsymbol{L}},\V_{\N},\tilde{\boldsymbol{\sigma}}_{\boldsymbol{L}},\tilde{\boldsymbol{\sigma}}_{\N})=     	F(\U_{1}\U_{\boldsymbol{L}},\U_{1}\U_{\N},\V_{1}\V_{\boldsymbol{L}},\V_{1}\V_{\N},\tilde{\boldsymbol{\sigma}}_{\boldsymbol{L}},\tilde{\boldsymbol{\sigma}}_{\N})
         \end{aligned}
        \EE
        where $\U_{1}\U_{\boldsymbol{L}},\U_{1}\U_{\N},\V_{1}\V_{\boldsymbol{L}},\V_{1}\V_{\N}$ are four independent haar matrices. Then, by (Conclusion: (\ref{main_lemma_L}.b)), we obtain
        \BE
        \begin{aligned}
        	F(\U_{1}\U_{\boldsymbol{L}},\U_{1}\U_{\N},\V_{1}\V_{\boldsymbol{L}},\V_{1}\V_{\N},\tilde{\boldsymbol{\sigma}}_{\boldsymbol{L}},\tilde{\boldsymbol{\sigma}}_{\boldsymbol{G}})	\overset{p}{=} \operatorname{AMSE}(v).
        \end{aligned} 
        \EE
        Let
      \BE\label{asy_def_f111}
      \begin{aligned}
F(\U_{1}\U_{\boldsymbol{L}},\U_{1}\U_{\N},\V_{1}\V_{\boldsymbol{L}},\V_{1}\V_{\N},\tilde{\boldsymbol{\sigma}}_{\boldsymbol{L}},\tilde{\boldsymbol{\sigma}}_{\N})=\frac{1}{\sqrt{n_1}}F_2(\tilde{\boldsymbol{\sigma}}_{\boldsymbol{N}})
      \end{aligned}
  \EE
  where $F_2$ is defined in Lemma \ref{Lip_lemma_f}. Note that 
  the empirical distribution of $\tilde{\boldsymbol{\sigma}}_{\boldsymbol{N}}$ and $\tilde{\boldsymbol{\sigma}}_{\boldsymbol{G}}$ both converge weakly to a same deterministic distribution due to the Marchenko–Pastur law \cite{marchenko1967distribution}. By Lemma \ref{fLasy}, for any $\tilde{\boldsymbol{\sigma}}_{\boldsymbol{N}}$ and $\tilde{\boldsymbol{\sigma}}_{\boldsymbol{G}}$, there always exists a  permutation $\boldsymbol{P^*}$ matrix, such that, as $n_1\to 0$,
  \BE\label{asy_NG}
  \begin{aligned}
  \frac{1}{\sqrt{n_1}}F_2(\tilde{\boldsymbol{\sigma}}_{\boldsymbol{N}})-\frac{1}{\sqrt{n_1}}F_2(\tilde{\boldsymbol{\sigma}}_{\boldsymbol{G}}\boldsymbol{P^*}) \to 0.
  \end{aligned}
\EE
Note that 
\BE
\begin{aligned}
	\frac{1}{\sqrt{n_1}}F_2(\tilde{\boldsymbol{\sigma}}_{\boldsymbol{G}}\boldsymbol{P^*})=&        	F(\U_{1}\U_{\boldsymbol{L}},\U_{1}\U_{\N},\V_{1}\V_{\boldsymbol{L}},\V_{1}\V_{\N},\tilde{\boldsymbol{\sigma}}_{\boldsymbol{L}},\tilde{\boldsymbol{\sigma}}_{\boldsymbol{G}}\boldsymbol{P^*})\\
	=&         	F(\U_{1}\U_{\boldsymbol{L}},\U_{1}\U_{\N}\boldsymbol{P^*},\V_{1}\V_{\boldsymbol{L}},\V_{1}\V_{\N}	\boldsymbol{P}^*_{\V},\tilde{\boldsymbol{\sigma}}_{\boldsymbol{L}},\tilde{\boldsymbol{\sigma}}_{\boldsymbol{G}}).
\end{aligned}
\EE
where
				\begin{equation}
	\boldsymbol{P}^*_{\V}=\left(
	\begin{matrix}
		\boldsymbol{P}^* & \boldsymbol{0}_{n_1,n_2-n_1}  \\
		\boldsymbol{0}_{n_2-n_1,n_1} & \boldsymbol{I}_{n_2-n_1,n_2-n_1}  \\
	\end{matrix}
	\right).
\end{equation}
and $\U_{1}\U_{\boldsymbol{L}},\U_{1}\U_{\N}\boldsymbol{P^*},\V_{1}\V_{\boldsymbol{L}},\V_{1}\V_{\N}	\boldsymbol{P}^*_{\V}$ are four independent haar matrix. Then, with (Conclusion: (\ref{main_lemma_L}.b)),  we conclude
\BE\label{asyG}
\begin{aligned}
	\frac{1}{\sqrt{n_1}}F_2(\tilde{\boldsymbol{\sigma}}_{\boldsymbol{G}}\boldsymbol{P^*})\overset{p}{=} \operatorname{AMSE}(v).
\end{aligned}
\EE
Combining (\ref{asy_def_f111}), (\ref{asy_NG}) and (\ref{asyG}), we obtain
\BE
\begin{aligned}
F(\U_{1}\U_{\boldsymbol{L}},\U_{1}\U_{\N},\V_{1}\V_{\boldsymbol{L}},\V_{1}\V_{\N},\tilde{\boldsymbol{\sigma}}_{\boldsymbol{L}},\tilde{\boldsymbol{\sigma}}_{\N})\overset{p}{=} \operatorname{AMSE}(v),
\end{aligned}
\EE
which completes the proof.  }   
	\end{enumerate}	
			\end{proof}

				\subsection{Proof of Lemma \ref{LMSEL}}\label{proof_LMSEL}
				\textcolor{black}{
				Assume Assumption \ref{as2} and Assumption (\ref{as3}.iv-vi) hold. Lemma (\ref{LMSEL}.1) is straight forward a application of Lemma \ref{main_lemma_L}. Note that
				\BE
				\begin{aligned}
				\frac{1}{n_1n_2}\left\langle \boldsymbol{L},	\D_{\boldsymbol{L}}\left(\boldsymbol{L}+\sqrt{v}\N\right) \right\rangle=-\frac{1}{2}	\left( \frac{1}{n_1n_2}\left\| \D_{\boldsymbol{L}}\left(\boldsymbol{L}+\sqrt{v}\N\right)-\boldsymbol{L}\right\| ^2_F-\frac{1}{n_1n_2}\left\| \D_{\boldsymbol{L}}\left(\boldsymbol{L}+\sqrt{v}\N\right)\right\| ^2_F-\frac{1}{n_1n_2}\left\| \boldsymbol{L}\right\| ^2_F\right), 
				\end{aligned}
			   \EE
			   where $ \frac{1}{n_1n_2}\left\| \D_{\boldsymbol{L}}\left(\boldsymbol{L}+\sqrt{v}\N\right)-\boldsymbol{L}\right\| ^2_F\overset{p}{=} AMSE(v)$ and $\frac{1}{n_1n_2}\left\| \D_{\boldsymbol{L}}\left(\boldsymbol{L}+\sqrt{v}\N\right)\right\| ^2_F\overset{p}{=}\mathbb{E}f^2_{\boldsymbol{L}}(\sigma)$. Then, 
			   \BE
			   \begin{aligned}
			   	\frac{1}{n_1n_2}\left\langle \boldsymbol{L},	\D_{\boldsymbol{L}}\left(\boldsymbol{L}+\sqrt{v}\N\right)\right\rangle  \overset{p}{=}(1-\beta)\mathbb{E} g_{\boldsymbol{L}}(\sigma) +\beta\mathbb{E}\left[ H_{\boldsymbol{L}}(\sigma,\hat{\sigma}) \right].
			   \end{aligned}
		   \EE
		   Applying the above conclusion to $\hat{a}^{(t)}_{\boldsymbol{L}}$, we derive the asymptotic results presented in Lemma (\ref{LMSEL}.2).
		   With similar straightforward calculations, it is easy to demonstrate the convergence of the parameters $\hat{c}^{(t)}_{\boldsymbol{L}}$ and $v^{(t)}_{p(\boldsymbol{L})\to\boldsymbol{L}}$ in Lemma (\ref{LMSEL}.3) and the asymptotic relations in Lemma (\ref{LMSEL}.4). This completes the proof. }

			\begin{lemma}\label{smooth_lemma}
				Let $x,\sigma^*,\epsilon>0$ and $\epsilon<\sigma^*$. Then,
				$\operatorname{h}_{\epsilon}(x; \sigma^*)$ defined in (\ref{def_smooth_h}) is continuous differentiable, where $\operatorname{h}_{\epsilon}(0)=0$ and $\operatorname{h}'_{\epsilon}(0)=0$. It's also   a  Lipschitz function with  Lipschitz constant $L_{\operatorname{h}_{\epsilon}}\triangleq\max(1,2+\frac{1.2\sigma^*}{e\epsilon})$. $\lim_{\epsilon\to 0}\operatorname{h}_{\epsilon}(x; \sigma^*)=\operatorname{h}(x; \sigma^*)$ whenever $\operatorname{h}(x; \sigma^*)$ is continuous at x.
			\end{lemma}
			\begin{proof}
				Note that $\operatorname{h}_ {\epsilon}(x; \sigma^*)$ is defined  as
						\BE\label{proof_21-hard}
				\begin{aligned}
			\operatorname{h}(x; \sigma^*)= \begin{cases} 0 & \text { if } 0\leq x< \sigma^* \\ 
						x & \text {  if  }  \sigma^*\leq x. \end{cases}
				\end{aligned}
				\EE
				Consider the bump function $B(x):\mathbb{R}\to\mathbb{R}$
				\BE
				\begin{aligned}
					B(x)=\left\{
					\begin{array}{cl}
						e^{\frac{1}{x^2-1}}    & |x|<1\\
						0.     & otherwise\\
					\end{array}\right.
				\end{aligned}
				\EE
				$\operatorname{h}_{\epsilon}(x; \sigma^*)$ is given by
				\BE\label{smooth_H_EXP}
				\begin{aligned}
					\operatorname{h}_{\epsilon}(x; \sigma^*)=\int 	\operatorname{h}(t; \sigma^*)B_{\epsilon}(x-t)dt 
					=\frac{\int_{\sigma^*}^{+\infty} t	e^{\frac{\epsilon^2}{(x-t)^2-\epsilon^2}}\boldsymbol{1}_{|x-t|\leq \epsilon}	}{\int_{-\infty}^{+\infty} 	e^{\frac{\epsilon^2}{x^2-\epsilon^2}}\boldsymbol{1}_{|x|\leq \epsilon}	dx}	dt= \begin{cases} 0 & \text { if } 0\leq x\leq\sigma^*-\epsilon \\ 
						\frac{\int_{\sigma^*}^{x+\epsilon} t	e^{\frac{\epsilon^2}{(x-t)^2-\epsilon^2}}dt 	}{\int_{-\epsilon}^{+\epsilon} 	e^{\frac{\epsilon^2}{x^2-\epsilon^2}}	dx}		& \text {  if  } \sigma^*-\epsilon<x<\sigma^*+\epsilon\\
						x	 & \text { if } \sigma^*+\epsilon\leq x   \end{cases}
				\end{aligned}
				\EE			
				As shown in (\ref{smooth_H_EXP}), for a $\epsilon_1$ with $\sigma^*-\epsilon>\epsilon_1>0$, we only need to prove the continuous differentiable property in a bounded domain $x\in[\sigma^*-\epsilon-\epsilon_1,\sigma^*-\epsilon+\epsilon_1]$. In the above domain, the integral only related to $\operatorname{h}(t; \sigma^*)$, where  $t\in[\sigma^*-\epsilon_1,\sigma^*+2\epsilon+\epsilon_1]$. Define 
				\BE
				\begin{aligned}
					\operatorname{h}_{p}(x; \sigma^*)= \begin{cases} 0 & \text { if } 0\leq x< \sigma^*\\ 
						x & \text {  if  } \sigma^*\leq x\leq \sigma^*+2\epsilon+\epsilon_1\\
						0 & \text {  if  } \sigma^*+2\epsilon+\epsilon_1<x.\end{cases}
				\end{aligned}
				\EE
				Then, we only need to prove the the continuous differentiable property for $\operatorname{h}_p(x, \sigma^*)\star 	B_{\epsilon}(x)$ in $[\sigma^*-\epsilon-\epsilon_1,\sigma^*-\epsilon+\epsilon_1]$, where
				\BE
				\begin{aligned}
					\operatorname{h}_p(x, \sigma^*)\star 	B_{\epsilon}(x)=\operatorname{h}(x, \sigma^*)\star 	B_{\epsilon}(x).
				\end{aligned}
				\EE
				Note that $\operatorname{h}_p(x, \sigma^*)\in L^1(\mathbb{R})$ is of compact support. The bump function $B_{\epsilon}(x)$ is infinitely differentiable and its all derivations are bounded. Then, by the Dominated Convergence Theorem, $x\in(\sigma^*-\epsilon-\epsilon_1,\sigma^*-\epsilon+\epsilon_1)$, we obtain for any $k\geq 1$
				\BE\label{smooth_H_EXP22}
				\begin{aligned}
					\frac{\partial ^{k} \operatorname{h}(x, \sigma^*)\star 	B_{\epsilon}(x)}{\partial x^{k}}&=	\frac{\partial ^{k} \operatorname{h}_p(x, \sigma^*)\star 	B_{\epsilon}(x)}{\partial x^{k}}
					= \operatorname{h}_p(x, \sigma^*)\star\frac{\partial ^{k}	B_{\epsilon}(x)}{\partial x^{k}}.
				\end{aligned}
				\EE

				Note that $\operatorname{h}_{\epsilon}(x; \sigma^*)$ is infinitely differentiable in $0\geq x<\sigma^*$ and $x>\sigma^*+\epsilon$. Combing the two parts, we conclude that $\operatorname{h}_{\epsilon}(x; \sigma^*)$ is infinitely differentiable in $\mathbb{R}_+$ with bounded $k$-th derivations. Concretely, the first-order derivative is given in (\ref{smooth_H_EXP3}).
					\BE\label{smooth_H_EXP3}
				\begin{aligned}
					\operatorname{h}'_{\epsilon}(x; \sigma^*)&= \begin{cases} 0 & \text { if } 0\leq x\leq\sigma^*-\epsilon \\ 
						\frac{\int_{\sigma^*}^{x+\epsilon} \frac{-2t\epsilon^2(x-t)}{\left( (x-t)^2-\epsilon^2\right)^2 }	e^{\frac{\epsilon^2}{(x-t)^2-\epsilon^2}}dt 	}{\int_{-\epsilon}^{+\epsilon} 	e^{\frac{\epsilon^2}{x^2-\epsilon^2}}	dx}		& \text {  if  } \sigma^*-\epsilon<x<\sigma^*+\epsilon\\
						1	 & \text { if } \sigma^*+\epsilon\leq x   \end{cases}
				\end{aligned}
				\EE
				 Consider $\operatorname{h}'_{\epsilon}(x; \sigma^*)$ in $(\sigma^*-\epsilon,\sigma^*+\epsilon)$. Let $y=\frac{t-x}{\epsilon}$. Then 
				\BE
				\begin{aligned}
					\operatorname{h}'_{\epsilon}(x; \sigma^*)=\frac{\epsilon\int^{1}_{\frac{\sigma^*-x}{\epsilon}} \frac{y(y+\frac{x}{\epsilon})}{\left( y^2-1\right)^2 }	e^{\frac{1}{y^2-1}}dy 	}{\int_{-\epsilon}^{+\epsilon} 	e^{\frac{\epsilon^2}{x^2-\epsilon^2}}	dx}	.
				\end{aligned}
				\EE
				where $0.44\epsilon<\int_{-\epsilon}^{+\epsilon} 	e^{\frac{\epsilon^2}{x^2-\epsilon^2}}	dx<0.45\epsilon$.
				Note that for any $x\in(\sigma^*-\epsilon,\sigma^*+\epsilon)$
				\BE
				\begin{aligned}
					\int^{1}_{\frac{\sigma^*-x}{\epsilon}} \frac{y(y+\frac{x}{\epsilon})}{\left( y^2-1\right)^2 }	e^{\frac{1}{y^2-1}}dy. 
				\end{aligned}
				\EE
				Expand the exponential
				\BE
				\begin{aligned}
					\frac{\left( y^2-1\right)^2 }{y^2}	e^{\frac{1}{1-y^2}}=\sum_{0}^{\infty}\frac{1}{n!}\frac{(1-y^2)^2}{y^2}\frac{1}{(1-y^2)^n}.
				\end{aligned}
				\EE	
				It's easy to see that as $|y|\to 1$, $\frac{y(y+\frac{x}{\epsilon})}{\left( y^2-1\right)^2 }	e^{\frac{1}{y^2-1}}\to 0$. Then
				\BE
				\begin{aligned}
				\int^{1}_{\frac{\sigma^*-x}{\epsilon}} \frac{y(y+\frac{x}{\epsilon})}{\left( y^2-1\right)^2 }	e^{\frac{1}{y^2-1}}dy &=\int^{1}_{\frac{\sigma^*-x}{\epsilon}} \frac{y^2}{\left( y^2-1\right)^2 }	e^{\frac{1}{y^2-1}}dy+\frac{x}{\epsilon}\int^{1}_{\frac{\sigma^*-x}{\epsilon}} \frac{y}{\left( y^2-1\right)^2 }	e^{\frac{1}{y^2-1}}dy\\
					&<2\int^{1}_{0} \frac{1}{\left( y^2-1\right)^2 }	e^{\frac{1}{y^2-1}}dy+\frac{x}{2\epsilon}e^\frac{\epsilon^2}{(\sigma^*-x)^2-\epsilon^2}\\
					&<0.8+\frac{x}{2\epsilon}e^\frac{\epsilon^2}{(\sigma^*-x)^2-\epsilon^2}<0.875+\frac{\sigma^*}{2e\epsilon}
				\end{aligned}
				\EE
				and 
				\BE
				\begin{aligned}
					&\quad\int^{1}_{\frac{\sigma^*-x}{\epsilon}} \frac{y(y+\frac{x}{\epsilon})}{\left( y^2-1\right)^2 }	e^{\frac{1}{y^2-1}}dy >0+\frac{x}{\epsilon}\int^{1}_{\frac{\sigma^*-x}{\epsilon}} \frac{y}{\left( y^2-1\right)^2 }	e^{\frac{1}{y^2-1}}=\frac{x}{2\epsilon}e^\frac{\epsilon^2}{(\sigma^*-x)^2-\epsilon^2}>0.
				\end{aligned}
				\EE
				Then, $\operatorname{h}_{\epsilon}(x; \sigma^*)$ is monotonically  increasing  and  continuously differentiable function. It's also a  Lipschitz function with  Lipschitz constant $L_{\operatorname{h}_{\epsilon}}\triangleq\max(1,2+\frac{1.2\sigma^*}{e\epsilon})$.

\textcolor{black}{			
Recall that $\operatorname{h}(x; \sigma^*)$ is defined as (\ref{proof_21-hard}) and is continuous at $(0,\sigma^*)$ and $(\sigma^*,+\infty)$. The remaining task is to prove $\lim_{\epsilon\to 0}\operatorname{h}_{\epsilon}(x; \sigma^*)$ for any $x\in(0,\sigma^*) \cup(\sigma^*,+\infty)$. Recall that $\operatorname{h}_{\epsilon}(x; \sigma^*)$ is given in
(\ref{smooth_H_EXP}). Then, 
				\BE
				\begin{aligned}
					\lim_{\epsilon\to 0}\operatorname{h}_{\epsilon}(x; \sigma^*)=\left\{
					\begin{array}{cl}
						0   & x<\sigma^*\\
						x     & x>\sigma^*.\\
					\end{array}\right.
				\end{aligned}
				\EE
This imply $\lim_{\epsilon\to 0}\operatorname{h}_{\epsilon}(x; \sigma^*)$  continuous at $(0,\sigma^*)$ and $(\sigma^*,+\infty)$, which completes the proof.}			
			\end{proof}
			\section{Proof of Lemma \ref{smooth_bestrankr}}\label{proof_smooth_bestrankr}
			In Lemma \ref{smooth_lemma}, we prove $\operatorname{h}_{\epsilon}(x; \sigma^*)$  satisfies the requirements in Assumption (\ref{as3}.vi). By Lemma \ref{LMSEL}, we conclude that
			\BE
			\begin{aligned}
				\frac{1}{n}\left\| \D_{\boldsymbol{L},\epsilon}(\boldsymbol{L}_{\N})-\boldsymbol{L}\right\|^2_F\overset{p}{=} \operatorname{MSE}_{\boldsymbol{L},\epsilon}(v)
			\end{aligned}
			\EE
			where $\operatorname{MSE}_{\boldsymbol{L}}(v,\epsilon)$ is given in (\ref{MSEL}) with $f_{\boldsymbol{L}}(x)=\operatorname{h}_{\epsilon}(x, \tilde{\sigma}^*)$.
			
			Then, we prove the convergence of $		\lim_{i\to\infty}\operatorname{MSE}_{\boldsymbol{L},\epsilon_i}(v)$. Note that, for $k,l>0$,
			we obtain
			\BE
			\begin{aligned}
			\left| \frac{1}{n}\left\|\D_{\boldsymbol{L},\epsilon_k}(\boldsymbol{L}_{\N})-\boldsymbol{L} \right\| ^2_F-\frac{1}{n}\left\|\D_{\boldsymbol{L},\epsilon_l}(\boldsymbol{L}_{\N})-\boldsymbol{L} \right\|^2_F\right| 
				&\leq\frac{1}{n}\left\|\D_{\boldsymbol{L},\epsilon_k}(\boldsymbol{L}_{\N})-\D_{\boldsymbol{L},\epsilon_l}(\boldsymbol{L}_{\N})\right\|
				^2_F\\
				&	\overset{p}{=}\frac{1}{v}\mathbb{E}[	(f_{\boldsymbol{L},\epsilon_k}(\sigma)-f_{\boldsymbol{L},\epsilon_l}(\sigma))^2 ]\\
				&\leq \frac{2}{v} \max(\epsilon_l^2+\epsilon_l\tilde{\sigma}^*,\epsilon_k^2+\epsilon_k\tilde{\sigma}^*).\\
			\end{aligned}
			\EE
			For any $\xi_1>0$, by $\epsilon_i\to 0$,  there exists a $N$ such that for any $i>N$, such that $\epsilon_i<\frac{v\xi_1}{4\sigma^*}$. Then
			for any $k,l>N$, the following holds
			\BE
			\begin{aligned}
				|\operatorname{MSE}_{\boldsymbol{L},\epsilon_k}(v)-	\operatorname{MSE}_{\boldsymbol{L},\epsilon_l}(v)|\leq\xi_1.
			\end{aligned}
			\EE
			Then, by Cauchy principle of convergence, we conclude that $\lim_{i\to\infty}\operatorname{MSE}_{\boldsymbol{L},\epsilon_i}(v)$ converges weakly to a deterministic value.
			For the best-rank-$r$ denoiser, by the triangle inequality, we obtain
			\BE
			\begin{aligned}
			\frac{1}{n}\left\|\D_{\boldsymbol{L},\epsilon_i}(\boldsymbol{L}_{\N})-\boldsymbol{L} \right\| ^2_F-\frac{1}{n}\left\|\D_{\boldsymbol{L},b}(\boldsymbol{L}_{\N})-\D_{\boldsymbol{L},\epsilon_i}(\boldsymbol{L}_{\N}) \right\|^2_F
				\leq\frac{1}{n}\left\|\D_{\boldsymbol{L},b}(\boldsymbol{L}_{\N})-\boldsymbol{L} \right\|
				^2_F\\
			\end{aligned}
			\EE
			and
			\BE
			\begin{aligned}
				\frac{1}{n}\left\|\D_{\boldsymbol{L},b}(\boldsymbol{L}_{\N})-\boldsymbol{L} \right\|\leq\frac{1}{n}\left\|\D_{\boldsymbol{L},\epsilon_i}(\boldsymbol{L}_{\N})-\boldsymbol{L} \right\| ^2_F+\frac{1}{n}\left\|\D_{\boldsymbol{L},b}(\boldsymbol{L}_{\N})-\D_{\boldsymbol{L},\epsilon_i}(\boldsymbol{L}_{\N}) \right\|^2_F.
			\end{aligned}
			\EE
			If we can show  
			\BE
			\begin{aligned}
				\lim_{i\to\infty}	\lim_{n\to\infty}\frac{1}{n}\left\|\D_{\boldsymbol{L},b}(\boldsymbol{L}_{\N})-\D_{\boldsymbol{L},\epsilon_i}(\boldsymbol{L}_{\N}) \right\|^2_F 	\overset{p}{=}0,
			\end{aligned}	
			\EE then (\ref{lemma_smooth_eqlim}) holds. Note that 
			\BE
			\begin{aligned}
				\frac{1}{n}\left\|\D_{\boldsymbol{L},b}(\boldsymbol{L}_{\N})-\D_{\boldsymbol{L},\epsilon_i}(\boldsymbol{L}_{\N}) \right\|^2_F
				=&\frac{1}{n_1v}\left\|\D_{\boldsymbol{L},b}(\frac{\boldsymbol{L}_{\N}}{\sqrt{vn_2}})-\D_{\boldsymbol{L},\epsilon_i}(\frac{\boldsymbol{L}_{\N}}{\sqrt{vn_2}}) \right\|^2_F\\
				=&\frac{1}{v}\frac{1}{n_1}\sum^{r}_{i=1}(f_{\boldsymbol{L},\epsilon_i}(\tilde{\sigma}_i)-\tilde{\sigma}_i)^2+\frac{1}{v}\frac{1}{n_1}\sum^{n_1}_{i>r}f^2_{\boldsymbol{L},\epsilon_i}(\tilde{\sigma}_i)
				\overset{p}{=}\frac{1}{v}\mathbb{E}[F_{\boldsymbol{L},\epsilon_i,\gamma}(\sigma)]
			\end{aligned}
			\EE
			with
			\BE
			\begin{aligned}
				F_{\boldsymbol{L},\epsilon_i,\gamma}(\sigma)=\left\{
				\begin{array}{cl}
					\left( f_{\boldsymbol{L},\epsilon_i}(\sigma)-\sigma\right) ^2  &   \tilde{\sigma}^*\leq\sigma\\
					f^2_{\boldsymbol{L},\epsilon_i}(\sigma)   &  0\geq\sigma< \tilde{\sigma}^*,\\
				\end{array}\right.
			\end{aligned}
			\EE
			where the expectation is taken over the limiting empirical singular value distribution of $\frac{\boldsymbol{L}_{\N}}{\sqrt{vn_1}}$.
			
			Since $F_{\boldsymbol{L},\epsilon_i,\gamma}(\sigma)< f^2_{\boldsymbol{L},\epsilon_i}(\sigma)+\sigma^2$ and  $\mathbb{E}[f^2_{\boldsymbol{L},\epsilon_i}(\sigma)+\sigma^2]<\infty$, if we can show 
			\BE
			\begin{aligned}
				\lim_{\epsilon_i\to 0+}	F_{\boldsymbol{L},\epsilon_i,\gamma}(\sigma)=0,\quad \text{a.s.}
			\end{aligned}
			\EE
			then by the dominated convergence theorem, we can conclude (\ref{lemma_smooth_eqlim}). Combine with  (\ref{smooth_H_EXP}), it is straightforward to conclude
			\BE
			\begin{aligned}
				&\lim_{\epsilon_i\to 0+} f^2_{\boldsymbol{L},\epsilon_i}(\sigma)=0,&\quad&\text{for}\quad \sigma\in(0,\tilde{\sigma}^*)\\		
				&\lim_{\epsilon_i\to 0+} \left( f_{\boldsymbol{L},\epsilon_i}(\sigma)-\sigma\right) ^2 =0,&\quad&\text{for}\quad \sigma\in(\tilde{\sigma}^*,+\infty).
			\end{aligned}
			\EE
			
			Then, we conclude the proof for (\ref{lemma_smooth_eqlim}).
			Similar to the above arguments, we can conclude
			\BE
			\begin{aligned}
				\lim_{i\to\infty}	\lim_{n\to\infty} \frac{1}{n}\left\| \D_{\boldsymbol{L},\epsilon_i}(\boldsymbol{L}_{\N})\right\|^2_F\overset{p}{=} \frac{1}{n}\left\|\D_{\boldsymbol{L},b}(\boldsymbol{L}_{\N})\right\|
				^2_F;\quad
				\lim_{i\to\infty}	\lim_{n\to\infty} \frac{1}{n}\left\langle  \D_{\boldsymbol{L},\epsilon_i}(\boldsymbol{L}_{\N}),\boldsymbol{L}_{\N}\right\rangle \overset{p}{=} \frac{1}{n}\left\langle  \D_{\boldsymbol{L},b}(\boldsymbol{L}_{\N}),\boldsymbol{L}_{\N}\right\rangle.
			\end{aligned}
			\EE
			Then, combining  (\ref{lemma_smooth_eqlim}),  it is straightforward to obtain
			\BE
			\begin{aligned}
				\lim_{i\to\infty}	\lim_{n\to\infty} \frac{1}{n}\left\langle  \D_{\boldsymbol{L},\epsilon_i}(\boldsymbol{L}_{\N}),\boldsymbol{L}\right\rangle \overset{p}{=} \frac{1}{n}\left\langle  \D_{\boldsymbol{L},b}(\boldsymbol{L}_{\N}),\boldsymbol{L}\right\rangle;\quad
				\lim_{i\to\infty}	\lim_{n\to\infty} \frac{1}{n}\left\langle  \D_{\boldsymbol{L},\epsilon_i}(\boldsymbol{L}_{\N}),\N\right\rangle \overset{p}{=} \frac{1}{n}\left\langle  \D_{\boldsymbol{L},b}(\boldsymbol{L}_{\N}),\N\right\rangle.
			\end{aligned}
			\EE
			Then, with above results, it's easy to check our claims.
			\section{Proof of Lemma \ref{L8} }\label{proof_L8}

			Under Assumptions 1-5, we obtain 
			\BE
			\begin{aligned}
				\quad\left \langle \X^{(t)}_{\X\to \y} -\X, \X\right \rangle=\left\langle \boldsymbol{L}_{p(\boldsymbol{L})\to\boldsymbol{L}}^{(t)} -\boldsymbol{L}, \boldsymbol{L}+\S\right \rangle+\left\langle \S_{p(\S)\to\S}^{(t)} -\S, \S+\boldsymbol{L}\right \rangle\overset{p}{=} -\phi(v_{\delta\to\boldsymbol{L}}^{(t)})-\varphi(v_{\delta\to\S}^{(t)}),
			\end{aligned}
			\EE
			by (\ref{as3ext}b) of Lemma \ref{asymS} and (\ref{asLext}b)  Lemma \ref{LMSEL}, where $\boldsymbol{L}_{p(\boldsymbol{L})\to\boldsymbol{L}}^{(t)} -\boldsymbol{L}$ and $\S_{p(\S)\to\S}^{(t)} -\S$ are separately independent with $\S$ and $\boldsymbol{L}$.
			Then 
			\BE
			\begin{aligned}
				\hat{v}_{\X\to\y}^{(t)}\overset{p}{=}\phi(v_{\delta\to\boldsymbol{L}}^{(t)})+\varphi(v_{\delta\to\S}^{(t)}).
			\end{aligned}
			\EE
			Note that $\frac{x}{x\hat{v}_{\X\to\y}^{(t)}+\sigma_{\boldsymbol{n}}^2}$ is continuous and bounded with respect to $x$ on $[0,+\infty)$ for any given $\hat{v}_{\X\to\y}^{(t)}$ and $\sigma_{\boldsymbol{n}}^2$. Recall $v^{(t)}_{\X}$ in (\ref{var_LMMSE_lemma3}) and define empirical value of $v_{\X}^{(t)}$ by $\hat{v}^{(t)}_{\X}$.
			For the  LMMSE denoiser, 
			\BE
			\begin{aligned}
				\quad\hat{v}^{(t)}_{\X}
				\overset{p}{=}v^{(t)}_{\boldsymbol{X} \rightarrow \boldsymbol{y}} \!-\!\alpha(v^{(t)}_{\boldsymbol{X} \rightarrow \boldsymbol{y}})^{2} \int\!\!\frac{\!\theta_{\A\A^T}\!}{\!\theta_{\A\A^T}\!v^{(t)}_{\!\boldsymbol{X} \rightarrow \boldsymbol{y}\!}\!+\!\sigma_{\boldsymbol{n}}^2\!}d p(\theta_{\A\A^T})
			\end{aligned}
			\EE
			by using  Assumption \ref{as3} and definition of weak convergence.
			Similarly, we prove that $\hat{a}_{\X}^{(t)}$, $\hat{c}_{\X}^{(t)}$ and $v^{(t+1)}_{\y\to\X}$ converge weakly and obtain their asymptotic forms.

			\section{Proofs Related to the Convergence Analysis of the State Evolution}
			
			\subsection{Proof of Lemma \ref{neccon}}\label{pneccon}
			Consider the following SE:
			\BE
			\begin{aligned}
				\tau_{\boldsymbol{S}}^{(t)}=\varphi\left(\left( \frac{1}{\alpha}-1\right)\tau_{\boldsymbol{S}}^{(t-1)}+ \frac{1}{\alpha}\tau_{\boldsymbol{L}}^{(t-1)}\right), \quad
				\tau_{\boldsymbol{L}}^{(t)}=\phi\left( \left( \frac{1}{\alpha}-1\right)\tau_{\boldsymbol{L}}^{(t-1)}+ \frac{1}{\alpha}\tau_{\boldsymbol{S}}^{(t-1)}\right).
			\end{aligned}
			\EE
			Note that by Assumption \ref{asphi}, $\varphi\left(\left( \frac{1}{\alpha}-1\right)\tau_{\boldsymbol{S}}^{(t-1)}+ \frac{1}{\alpha}\tau_{\boldsymbol{L}}^{(t-1)}\right)$ and $\phi\left( \left( \frac{1}{\alpha}-1\right)\tau_{\boldsymbol{L}}^{(t-1)}+ \frac{1}{\alpha}\tau_{\boldsymbol{S}}^{(t-1)}\right)$  are monotonically increasing functions. It follows that
			\BE\label{necSE}
			\begin{aligned}
				&\varphi\left(\left( \frac{1}{\alpha}-1\right)\tau_{\boldsymbol{S}}+ \frac{1}{\alpha}\tau_{\boldsymbol{L}}\right)	\ge\varphi\left( \left(\frac{1}{\alpha}-1\right)\left( \tau_{\boldsymbol{L}}+\tau_{\boldsymbol{S}}\right) \right)\\
				& \phi\left(\left( \frac{1}{\alpha}-1\right)\tau_{\boldsymbol{L}}+ \frac{1}{\alpha}\tau_{\boldsymbol{S}}\right)\ge\phi\left( \left(\frac{1}{\alpha}-1\right)\left( \tau_{\boldsymbol{L}}+\tau_{\boldsymbol{S}}\right) \right).
			\end{aligned}
			\EE
			Consider the sequences $\left\{\tau_{\boldsymbol{L},nec}^{(t)}\right\}_{t\geq 1}$ and $\left\{\tau_{\S,nec}^{(t)}\right\}_{t\geq 1}$ generated by the following iteration (under the same initial values as the original SE):
			\BS
			\begin{align}
				&\tau_{\boldsymbol{S},nec}^{(t)}=\varphi\left(\left( \frac{1}{\alpha}-1\right)\left( \tau_{\boldsymbol{S},nec}^{(t-1)}+ \tau_{\boldsymbol{L},nec}^{(t-1)}\right) \right) \\
				&\tau_{\boldsymbol{L},nec}^{(t)}=\phi\left( \left( \frac{1}{\alpha}-1\right)\left( \tau_{\boldsymbol{L},nec}^{(t-1)}+ \tau_{\boldsymbol{S},nec}^{(t-1)}\right) \right).
			\end{align}
			\ES
			From \eqref{necSE} and a simple induction argument, we obtain
			\BE
			\begin{aligned}
				\tau_{\S}^{(t)}\ge\tau_{\S,nec}^{(t)}\quad \text{and}\quad \tau_{\boldsymbol{L}}^{(t)}\ge\tau_{\boldsymbol{L},nec}^{(t)}.
			\end{aligned}
			\EE
			Hence, a necessary condition for 
			\[
			\lim_{t\to\infty}\tau_{\boldsymbol{L}}^{(t)}=0\quad\text{and}\quad \lim_{t\to\infty} \tau_{\S}^{(t)}=0
			\]
			is
			\[
			\lim_{t\to\infty}\tau_{\boldsymbol{L},nec}^{(t)}=0\quad\text{and}\quad \lim_{t\to\infty} \tau_{\S,nec}^{(t)}=0.
			\]
			The rest of the proof is to analyze the condition under which $\{\tau_{\boldsymbol{L},nec}^{(t)}\}$ and $\{\tau_{\boldsymbol{L},nec}^{(t)}\}$ converge to zero. To this end, we sum up (\ref{necSE}a) and (\ref{necSE}b) and get
			\BE\label{Eqn:aux}
			\begin{aligned}
				\tau_{\boldsymbol{S},nec}^{(t)}+\tau_{\boldsymbol{L},nec}^{(t)}&=\varphi\left(\left( \frac{1}{\alpha}-1\right)\left( \tau_{\boldsymbol{S},nec}^{(t-1)}+ \tau_{\boldsymbol{L},nec}^{(t-1)}\right) \right)
			+\phi\left( \left( \frac{1}{\alpha}-1\right)\left( \tau_{\boldsymbol{L},nec}^{(t-1)}+\tau_{\boldsymbol{S},nec}^{(t-1)}\right) \right).
			\end{aligned} 
			\EE
			We define $\tau_{\boldsymbol{X},nec}^{(t)}:=\tau_{\boldsymbol{S},nec}^{(t)}+\tau_{\boldsymbol{L},nec}^{(t)}$ and rewrite \eqref{Eqn:aux} as
			\[
			\tau_{\boldsymbol{X},nec}^{(t)}=\varphi\left( \left( \frac{1}{\alpha}-1\right)\tau_{\boldsymbol{X},nec}^{(t-1)} \right)+\phi\left( \left( \frac{1}{\alpha}-1\right)\tau_{\boldsymbol{X},nec}^{(t-1)} \right).
			\]
			It is easy to see that $\left\{\tau_{\X,nec}^{(t)}\right\}_{t\geq 1}$ converges to $0$ from any initial value if and only if the following condition holds
			\BE
			\begin{aligned}
				x>\varphi\left(\left( \frac{1}{\alpha}-1\right)x\right)+\phi\left( \left( \frac{1}{\alpha}-1\right) x \right),\quad \forall x>0.
			\end{aligned}
			\EE
			Let $y=(\frac{1}{\alpha}-1)x$, Then, the above equation is  rewritten as 
			\BE
			\begin{aligned}
				\frac{\alpha}{1-\alpha}	y>\varphi\left( y\right)+\phi\left(y \right),\quad \forall y>0.
			\end{aligned}
			\EE
			Thus, we obtain the necessary condition that $\alpha$ satisfies
			\BE
			\begin{aligned}
				\alpha > \alpha_{nec}:=\sup_{x >0}\ 	\frac{\varphi(x)+\phi(x)}{\varphi(x)+\phi(x)+x} \quad \forall x>0.
			\end{aligned}
			\EE
			
			On the other hand, the sequences $\left\{\tau_{\boldsymbol{L},nec}^{(t)}\right\}_{t\geq 1}$ and $\left\{\tau_{\S,nec}^{(t)}\right\}_{t\geq 1}$ converge to $0$ if and only if
			the sequence $\left\{\tau_{\X,nec}^{(t)}\right\}_{t\geq 1}$ converges to $0$.

			\subsection{Proof of Lemma \ref{SECON1}}  \label{Sec:Proof_Lemma9}

			Our proof strategy is as follows. We will show that if $(\tau_{\S}^{(t_{0})},\tau_{\boldsymbol{L}}^{(t_{0})})\in\mathcal{R}_{1}$ but not at the origin, then both $\tau_{\S}$ and $\tau_{\boldsymbol{L}}$ improve in the next iteration, i.e.,
			\BE\label{Eqn:R1_monotonicity}
			\begin{split}
				\tau_{\S}^{(t_{0}+1)} \le \tau_{\S}^{(t_{0})},\quad
				\tau_{\boldsymbol{L}}^{(t_{0}+1)} \le \tau_{\boldsymbol{L}}^{(t_{0})}.
			\end{split}
			\EE
			Further, we show that if $(\tau_{\S}^{(t_{0})},\tau_{\boldsymbol{L}}^{(t_{0})})\in\mathcal{R}_{1}$, then $(\tau_{\S}^{(t_{0}+1)},\tau_{\boldsymbol{L}}^{(t_{0}+1)})\in\mathcal{R}_{1}$. We then use an induction argument to conclude that the sequences $\{\tau_{\S}^{(t)}\}_{t\ge t_0}$ and $\{\tau_{\boldsymbol{L}}^{(t)}\}_{t\ge t_0}$ monotonically converge, as both sequences are bounded. Finally, we show the limit points must be $(0,0)$ when $\alpha>\max\{\alpha_1,\alpha_2,\alpha_3\}$.
			
			\subsubsection{Step 1}
			Our first goal is to prove \eqref{Eqn:R1_monotonicity}. From \eqref{regdi}, $(\tau_{\S}^{(t_{0})},\tau_{\boldsymbol{L}}^{(t_{0})})\in\mathcal{R}_{1}$ implies
			\BS\label{Eqn:Region1}
			\begin{align}
				\Psi_{2}^{-1}(\tau_{\boldsymbol{L}}^{(t_{0})})\geq\tau_{\S}^{(t_{0})}\geq\Psi_{1}(\tau_{\boldsymbol{L}}^{(t_{0})}),\\
				\Psi_{1}^{-1}(\tau_{\S}^{(t_{0})})\geq	\tau_{\boldsymbol{L}}^{(t_{0})}\geq\Psi_{2}(\tau_{\S}^{(t_{0})}).
			\end{align}
			\ES
			Since $\Psi_{1}$  and $\Psi_{2}$ are monotonically increasing (from Lemma \ref{propertyvarphi}), the second inequalities of \eqref{Eqn:Region1} yield
			\BE\label{Eqn:Region1_2}
			\begin{aligned}
				\Psi_{1}^{-1}\left( \tau_{\S}^{(t_{0})}\right) \geq\tau_{\boldsymbol{L}}^{(t_{0})},\quad
				\Psi_{2}^{-1}\left( \tau_{\boldsymbol{L}}^{(t_{0})}\right) \geq\tau_{\S}^{(t_{0})}.
			\end{aligned}
			\EE
			Recall the definitions of $\Psi_1^{-1}$ and $\Psi_2^{-1}$:
			\BE\label{dePsi_second}
			\begin{aligned}
				&\Psi^{-1}_{1}(\tau_{\S},\alpha)\triangleq\alpha \varphi^{-1}(\tau_{\S})-(1-\alpha)\tau_{\S}\\			
				&\Psi^{-1}_{2}(\tau_{\boldsymbol{L}},\alpha) \triangleq\alpha \phi^{-1}(\tau_{\boldsymbol{L}})-(1-\alpha)\tau_{\boldsymbol{L}}.
			\end{aligned}
			\EE
			Using \eqref{dePsi_second} and noting that $\varphi$ and $\phi$ are increasing (see Assumption \ref{asphi} ), we can rewrite \eqref{Eqn:Region1_2} as
			
			\BS\label{Eqn:Region1_3}
			\begin{align}
				\tau_{\boldsymbol{S}}^{(t_{0})} \geq	\varphi\left(\left(\frac{1}{\alpha}-1\right)\tau_{\boldsymbol{S}}^{(t_{0})}+\frac{1}{\alpha}\tau_{\boldsymbol{L}}^{(t_{0})}\right)=\tau_{\boldsymbol{S}}^{(t_{0}+1)}, \\
				\tau_{\boldsymbol{L}}^{(t_{0})}\geq\phi\left(\left( \frac{1}{\alpha}-1\right)\tau_{\boldsymbol{L}}^{(t_{0})}+ \frac{1}{\alpha}\tau_{\boldsymbol{S}}^{(t_{0})}\right)=	\tau_{\boldsymbol{L}}^{(t_{0}+1)},
			\end{align}
			\ES
			where the two equalities follow from the SE equation \eqref{SEoverall}. This concludes the proof of \eqref{Eqn:R1_monotonicity}.
			
			\subsubsection{Step 2} 
			We next prove $(\tau_{\S}^{(t_{0}+1)},\tau_{\boldsymbol{L}}^{(t_{0}+1)})\in\mathcal{R}_{1}$. 	Note that 
			\BS
			\BE\label{Eqn:Region1_4}
			\begin{aligned}
				\Psi_1^{-1}\left(\tau_{\boldsymbol{S}}^{(t_{0}+1)}\right) -\tau_{\boldsymbol{L}}^{(t_{0})}
				&\overset{(a)}{=}\Psi^{-1}_{1}\left( \varphi\left(\left( \alpha^{-1}-1\right)\tau^{(t_{0})}_{\boldsymbol{S}}+\alpha^{-1}\tau^{(t_{0})}_{\boldsymbol{L}}\right)\right) -\tau_{\boldsymbol{L}}^{(t_{0})}\\
				&\overset{(b)}{=}(1-\alpha)\tau^{(t_{0})}_{\boldsymbol{S}}-(1-\alpha) \varphi\left(\left( \frac{1}{\alpha}-1\right)\tau^{(t_{0})}_{\boldsymbol{S}}+ \frac{1}{\alpha}\tau^{(t_{0})}_{\boldsymbol{L}}\right)\\
				&=(1-\alpha)\cdot\left[\tau^{(t_{0})}_{\boldsymbol{S}}-\varphi\left(\left( \frac{1}{\alpha}-1\right)\tau^{(t_{0})}_{\boldsymbol{S}}+ \frac{1}{\alpha}\tau^{(t_{0})}_{\boldsymbol{L}}\right)\right] \\
				&\overset{(c)}{=}(1-\alpha)\cdot\left[ \tau^{(t_{0})}_{\boldsymbol{S}}-\tau^{(t_{0}+1)}_{\boldsymbol{S}} \right]\\
				&\overset{(d)}{\ge}0,\\
			\end{aligned}
			\EE
			where step (a) is the definition of $\tau_{\boldsymbol{S}}^{(t_{0}+1)}$ (see, e.g.,  \eqref{Eqn:Region1_3}); step (b) is based on the definition of $\Psi_1$ (see \eqref{dePsi_second}); step (c) is again due to the definition of $\tau_{\boldsymbol{S}}^{(t_{0}+1)}$; and step (d) is from \eqref{Eqn:Region1_3}. As $\Psi_1$ is an increasing function, \eqref{Eqn:Region1_4} implies 
			\BE\label{Eqn:Region1_5}
			\tau_{\boldsymbol{S}}^{(t_{0}+1)}\ge \Psi_1\left(\tau_{\boldsymbol{L}}^{(t_{0})}\right).
			\EE
			\ES
			Similarly, it can be shown that
			\BS
			\begin{align}
				&\Psi_2^{-1}\left(\tau_{\boldsymbol{L}}^{(t_{0}+1)}\right) -\tau_{\boldsymbol{S}}^{(t_{0})} \ge0,\\
				&\tau_{\boldsymbol{L}}^{(t_{0}+1)} \ge \Psi_2\left(\tau_{\boldsymbol{S}}^{(t_{0})} \right).
			\end{align}
			\ES
			Combining \eqref{Eqn:Region1_3}, \eqref{Eqn:Region1_4} and \eqref{Eqn:Region1_5} yields
			\BE\label{EQSE}
			\begin{aligned}
				&\Psi_{2}^{-1}\left(\!\tau_{\boldsymbol{L}}^{(t_{0}+1)}\!\right) \geq\tau_{\boldsymbol{S}}^{(t_{0}+1)}\geq\Psi_{1}\left(\! \tau_{\boldsymbol{L}}^{(t_{0}+1)}\!\right),\\
				&\Psi_{1}^{-1}\left(\!\tau_{\boldsymbol{S}}^{(t_{0}+1)}\!\right)\geq\tau_{\boldsymbol{L}}^{(t_{0}+1)}\geq\Psi_{2}\left(\! \tau_{\S}^{(t_{0}+1)}\!\right),
			\end{aligned}
			\EE
			which means that $(\tau_{\S}^{(t_0+1)},\tau_{\boldsymbol{L}}^{(t_0+1)})\in\mathcal{R}_{1}$. 
			
			\subsubsection{Step 3}
			Step 1 and Step 2 above already imply the convergence of the sequences $\{\tau_{\S}^{(t)}\}_{t\ge t_0}$ and $\{\tau_{\boldsymbol{L}}^{(t)}\}_{t\ge t_0}$, under the condition $(\tau_{\S}^{(t_{0})},\tau_{\boldsymbol{L}}^{(t_{0})})\in\mathcal{R}_{1}$. Clearly, the limit point must be a fixed point of the SE equations ($\varphi$ and $\phi$ are continuous functions):
			\begin{align*}
				&\tau_{\boldsymbol{S}}^{(t)}=\varphi\left(\left( \frac{1}{\alpha}-1\right)\tau_{\boldsymbol{S}}^{(t-1)}+ \frac{1}{\alpha}\tau_{\boldsymbol{L}}^{(t-1)}\right), \\
				&\tau_{\boldsymbol{L}}^{(t)}=\phi\left( \left( \frac{1}{\alpha}-1\right)\tau_{\boldsymbol{L}}^{(t-1)}+ \frac{1}{\alpha}\tau_{\boldsymbol{S}}^{(t-1)}\right).
			\end{align*}
			From Lemma \ref{alphaglobal}, the above SE equations have a unique solution at $(0,0)$ when $\alpha >\max\{\alpha_1,\alpha_2,\alpha_3\}$. 
			Hence,
			\[
			\lim_{t\to\infty}\tau_{\boldsymbol{L}}^{(t)}=0\quad\text{and}\quad \lim_{t\to\infty} \tau_{\S}^{(t)}=0.
			\]

			\subsection{Proof of Theorem \ref{TH3}}\label{proof_TH3}
			We start with the following Lemmas. 
			
			\begin{lemma}\label{propertyvarphi} 
				Suppose Assumption \ref{asphi} holds, and $\alpha>\max\{\alpha_1,\alpha_2\}$ where
				\BE
				\begin{split}
					\alpha_1:= \sup_{x>0}\ \frac{\varphi'(x)}{\varphi'(x)+1},
					\alpha_2:= \sup_{x>0}\ \frac{\phi'(x)}{\phi'(x)+1}.
				\end{split}
				\EE
				Then, the following statements hold. 
				\begin{itemize}
					\item[(i)] For any fixed $\tau_{\boldsymbol{L}}>0$, the following equation
					\BE\label{lemma_25_1}
					\begin{aligned}
						\varphi\left(\left( \frac{1}{\alpha}-1\right)\tau_{\boldsymbol{S}}+ \frac{1}{\alpha}\tau_{\boldsymbol{L}}\right)=\tau_{\boldsymbol{S}}.
					\end{aligned} 
					\EE
					has a unique globally attractive fixed point $\Psi_{1}(\tau_{\boldsymbol{L}})$  in $(0,M_\varphi)$, meaning that 
					\BS
					\BE
					\begin{aligned}\label{EQ}
						\tau_{\S}\!<\!\varphi\left(\!\left(\! \frac{1}{\alpha}\!-\!1\!\right)\!\tau_{\boldsymbol{S}}\!+\! \frac{1}{\alpha}\!\tau_{\boldsymbol{L}}\!\right)\!<\!\Psi_{1}(\tau_{\boldsymbol{L}}), 	\tau_{\S}\in (0,\Psi_{1}(\tau_{\boldsymbol{L}}))
					\end{aligned}
					\EE
					and
					\BE
					\begin{aligned}
						\Psi_{1}(\tau_{\boldsymbol{L}})	\!<\!\varphi\left(\!\left(\! \frac{1}{\alpha}\!-\!1\!\right)\!\tau_{\boldsymbol{S}}\!+\! \frac{1}{\alpha}\!\tau_{\boldsymbol{L}}\!\right)\!<\!\tau_{\S}, 	\tau_{\S}\in (\Psi_{1}(\tau_{\boldsymbol{L}}),\infty),
					\end{aligned}
					\EE
					\ES
					where $M_\varphi=\sup \varphi$.

					\item[(ii)] For any fixed $\tau_{\boldsymbol{S}}>0$, the following equation 
					\BE
					\phi\left(\left( \frac{1}{\alpha}-1\right) \tau_{\boldsymbol{L}}+ \frac{1}{\alpha}\tau_{\boldsymbol{S}}\right)= \tau_{\boldsymbol{L}}.
					\EE
					has a unique globally attractive fixed point  $\Psi_{2}(\tau_{\S})$ in $(0,M_\phi)$, where   $M_\phi=\sup \phi$.
					
					\item[(iii)] The fixed point function $\Psi_{1}(\tau_{\boldsymbol{L}})$ obtained in (i) is strictly monotonically increasing $[0,\infty)$ with $\Psi_{1}(0)=0$.
					
					\item[(iv)] The fixed point function $\Psi_{2}(\tau_{\S})$ obtained in (ii) is strictly monotonically increasing $[0,\infty)$ with $\Psi_{2}(0)=0$.
				\end{itemize}
			\end{lemma} 
			\begin{proof}
				
				Part (i):
				Based on Assumption \ref{asphi},  	$\varphi(x)$ is continuous and monotonically increasing on $[0,\infty)$. Then, for any fixed $\tau_{\boldsymbol{L}}\in(0,\infty)$, $\varphi\left(\left( \frac{1}{\alpha}-1\right)\tau_{\boldsymbol{S}}+ \frac{1}{\alpha}\tau_{\boldsymbol{L}}\right)$ is strictly monotonically increasing with respect to $\tau_{\boldsymbol{S}}$ on $\tau_{\boldsymbol{S}}\in[0,+\infty)$. 
				Then define
				\BE
				\begin{aligned}
					\hat{\varphi}(\tau_{\S};\tau_{\boldsymbol{L}})=	\varphi\left(\left( \frac{1}{\alpha}-1\right)\tau_{\boldsymbol{S}}+ \frac{1}{\alpha}\tau_{\boldsymbol{L}}\right)-\tau_{\boldsymbol{S}}.
				\end{aligned}
				\EE
				For any fixed $\tau_{\boldsymbol{L}}\in(0,\infty)$, $\hat{\varphi}(0)>0$ and 
				\BE
				\begin{aligned}
					\hat{\varphi}'(\tau_{\S};\tau_{\boldsymbol{L}})=	\frac{1-\alpha}{\alpha}\varphi'\left( \left( \frac{1}{\alpha}-1\right)\tau_{\boldsymbol{S}}+ \frac{1}{\alpha}\tau_{\boldsymbol{L}}\right) -1.
				\end{aligned}
				\EE
				Note that there exists a positive $\epsilon$ such that $\alpha>\alpha_1+\epsilon>\alpha_1$.  Then, $\hat{\varphi}'(\tau_{\S};\tau_{\boldsymbol{L}})<-\epsilon$ and there exist one $\tau_{\S}^{*}=\Psi_{1}(\boldsymbol{L})\in(0,\frac{\hat{\varphi}'(0)}{\epsilon}]$ such that $	\hat{\varphi}(\tau_{\S}^{*})=0$.  Hence, (\ref{lemma_25_1}) has exactly one fixed point. The corresponding $\tau_{\S}^{*}$ for any given $\tau_{\boldsymbol{L}}$ is defined as $\tau_{\S}^{*}=\Psi_{1}(\tau_{\boldsymbol{L}})$.  Furthermore, for $	0<\tau_{\S}<\Psi_{1}(\tau_{\boldsymbol{L}})$, we obtain
				\BE
				\begin{aligned}
					\tau_{\S}&\overset{a}{<}\varphi\left(\left( \frac{1}{\alpha}-1\right)\!\tau_{\boldsymbol{S}}+ \frac{1}{\alpha}\!\tau_{\boldsymbol{L}}\right)\\
					&\overset{b}{<}\varphi\left(\left( \frac{1}{\alpha}-1\right)\!\tau_{\boldsymbol{S}}^{*}+ \frac{1}{\alpha}\!\tau_{\boldsymbol{L}}\right)\\
					&\overset{c}{=}\tau_{\S}^{*}=\Psi_{1}(\tau_{\boldsymbol{L}}),
				\end{aligned} 
				\EE
				where step (a) is from the property of $	\hat{\varphi}'(\tau_{\S};\tau_{\boldsymbol{L}})$; step (b) is from the strictly monotonicity increasing property of $\varphi$; step (c) is  the definition of $\tau_{\S}^{*}$.
				
				For $\tau_{\S}\in (\Psi_{1}(\tau_{\boldsymbol{L}}),\infty)$, we obtain
				\BE
				\begin{aligned}
					\tau_{\S}&>\varphi\left(\left( \frac{1}{\alpha}-1\right)\!\tau_{\boldsymbol{S}}+ \frac{1}{\alpha}\!\tau_{\boldsymbol{L}}\right)>\varphi\left(\left( \frac{1}{\alpha}-1\right)\!\tau_{\boldsymbol{S}}^{*}+ \frac{1}{\alpha}\!\tau_{\boldsymbol{L}}\right)=\tau_{\S}^{*}=\Psi_{1}(\tau_{\boldsymbol{L}}).
				\end{aligned} 
				\EE
				
				Part (ii): The proof is similar to the proof in Part (i);
				
				Part (iii): Recall that 
				\begin{align*}
					\Psi_{1}^{-1}(\tau_{\S},\alpha)=\alpha \varphi^{-1}(\tau_{\S})-(1-\alpha)\tau_{\S}.
				\end{align*} 
				Then, $\Psi_{1}^{-1}(0,\alpha)=0$.
				For $\alpha>\alpha_{1}$, $\Psi_{1}^{-1}(\tau_{\S})$ is strictly monotonicity increasing.
				\BE
				\begin{aligned}
					\frac{ \partial \Psi_{1}^{-1}(\tau_{\S}) }{ \partial \tau_{\S} }=\alpha \frac{1}{\varphi'(\varphi^{-1}(\tau_{\S}))}-1+\alpha>0.
				\end{aligned}
				\EE
				Then, $\Psi_{1}(\tau_{\boldsymbol{L}})$ is strictly monotonicity increasing.
				
				The proof of part (iv) is similar to part (iii).
			\end{proof}
			Lemma \ref{propertyvarphi} above is concerned with the fixed points of the SE equations in \eqref{SEoverall}. Then, we present the crucial points for $\alpha$ in the following Lemma .

			\begin{lemma}\label{alphaglobal}
				Suppose that  Assumption \ref{asphi} holds. If $\alpha>\max(\alpha_{1},\alpha_{2},\alpha_{3})$, where
				\BS
				\begin{align}
					\alpha_1&= \sup_{x>0}\ \frac{\varphi'(x)}{\varphi'(x)+1},\\[3pt]
					\alpha_2&= \sup_{x>0}\ \frac{\phi'(x)}{\phi'(x)+1},\\[3pt]
					\alpha_3&=\inf\{\alpha\in(0,1); x\geq\Psi_1(\Psi_2(x,\alpha),\alpha), \forall 0<x<\sup\varphi\},
				\end{align}
				\ES
				$(\tau_{\S},\tau_{\boldsymbol{L}})=(0,0)$ is the only fixed point of (\ref{pfSE}).
			\end{lemma}
			\begin{proof}
				Note that from Lemma \ref{propertyvarphi}, $\alpha> \max(\alpha_{1},\alpha_{2})$ guarantee $\Psi_{1}(\tau_{\boldsymbol{L}})$ and $\Psi_{2}(\tau_{\S})$ are respectively the only global attractive fixed point of 
				\BS
				\BE
				\begin{aligned}
					\varphi\left(\left( \frac{1}{\alpha}-1\right)\tau_{\boldsymbol{S}}+ \frac{1}{\alpha}\tau_{\boldsymbol{L}}\right)=\tau_{\boldsymbol{S}}
				\end{aligned} 
				\EE
				and
				\BE
				\begin{aligned}
					\phi\left(\left( \frac{1}{\alpha}-1\right) \tau_{\boldsymbol{L}}+ \frac{1}{\alpha}\tau_{\boldsymbol{S}}\right)= \tau_{\boldsymbol{L}}.
				\end{aligned}
				\EE
				\ES
				Furthermore,  $\Psi_{1}(\tau_{\boldsymbol{L}})$ and $\Psi_{2}(\tau_{\S})$ are both continuous and strictly monotonically increasing. Then, intersect points of the two curve are the fixed point of (\ref{pfSE}). To guarantee $(0,0)$ is the only fixed point. We need $\alpha$ satisfy
				\BS
				\BE
				\begin{aligned}
					\Psi_1^{-1}(x,\alpha)-	\Psi_2(x,\alpha)>0, \forall 0<x<\sup\varphi
				\end{aligned}
				\EE
				or 
				\BE\label{Lemma26_2}
				\begin{aligned}
					\Psi_1^{-1}(x,\alpha)-	\Psi_2(x,\alpha)<0, \forall 0<x<\sup\varphi
				\end{aligned}
				\EE
				\ES
				Consider that
				\BS
				\begin{align}
				\frac{\partial \Psi_2^{-1}(\Psi_1^{-1}(x,\alpha),\alpha) }{\partial \alpha}
					&=\phi^{-1}\left(\alpha \varphi^{-1}(x)-(1-\alpha)x \right) +\alpha\left( \alpha \varphi^{-1}(x)\!-\!(\!1\!-\!\alpha\!)\!x\right)\notag\\
					&\quad+\!\alpha\frac{\varphi^{-1}\!(\!x\!)\!+\!x}{\phi'\!\left(\!\varphi^{-1}\!\left(\!\alpha\! \varphi^{-1}\!(\!x)\!-\!(\!1\!-\!\alpha)x \right)  \right)}\!-\!(\!1\!-\!\alpha\!)\left( \!\varphi^{-1}(x)+x\!\right).
				\end{align}
				\ES
				$\alpha>\alpha_{2}$ implies $\alpha \varphi^{-1}(x)-(1-\alpha)x>0$ for any $x>0$. Then, the first and the second terms are both positive. Note that $\alpha > \sup(\frac{\phi'(x)}{\phi'(x)+1}:x\geq0)$. Then, the sum of the third term and the forth term is positive.
				Then, for any given $x$, $\Psi_2^{-1}(\Psi_1^{-1}(x,\alpha),\alpha)$ is strictly monotonically increasing with respect to $\alpha$  in $(\max(\alpha_{1},\alpha_{2}),1)$.
				And there exist no $\alpha>\max(\alpha_{1},\alpha_{2})$ such that (\ref{Lemma26_2}) hold. Note that $\alpha_3$ is defined as
				\BE
				\begin{aligned}
					\alpha_3=\inf\{\alpha\in(0,1); x\geq\Psi_1(\Psi_2(x,\alpha),\alpha), \forall 0<x<\sup\varphi\}
				\end{aligned}
				\EE
				Thus, for any $\alpha>\alpha_3$, we obtain 
				\BE
				\begin{aligned}
					\Psi_2^{-1}(\Psi_1^{-1}(x,\alpha),\alpha)>x, \forall 0<x<\sup\varphi.
				\end{aligned}
				\EE
				Then, we complete the proof.
			\end{proof}
			
			We know ready to prove Theorem \ref{TH3}. Note that $\alpha>\max\{\alpha_1,\alpha_2\}$ ensure that both  (\ref{pfSE}a) and   (\ref{pfSE}b) only have one globally attractive fixed point by Lemmas \ref{propertyvarphi}.  Moreover, as $\alpha>\max\{\alpha_1,\alpha_2,\alpha_3\}$, we prove that $(\tau_{\S},\tau_{\boldsymbol{L}})=(0,0)$ is the only fixed point of the SE in Lemma \ref{alphaglobal}. The rest is to prove the SE dynamics converges to this fixed point from any starting point. In Lemma \ref{SECON1}, we prove that for any initialization in $\mathbb{R}_+^2\backslash\{0\}$, the SE converge to $(\tau_{\S},\tau_{\boldsymbol{L}})=(0,0)$. This completes   the overall proof.	
			
			\section{Comparison of RPCA/CRPCA Algorithms}\label{comp_alg}
			\begin{table}[h]
				\centering
				\resizebox{\textwidth}{!}{%
					\begin{tabular}{@{}p{2.5cm} p{2.5cm} p{2cm} p{12cm} p{3cm} p{2cm}@{}}
						\toprule
						\textbf{Category} & \makecell{\textbf{Method}} & \makecell{\textbf{Compressive}} & \makecell{\textbf{Optimization Model}} & \makecell{\textbf{Algorithm}} & \makecell{\textbf{Convergence}\\\textbf{ Guarantee}} \\
						\midrule
						\multirow{5}{*}{\makecell{\textbf{Convex/non-convex}\\\textbf{Optimization}}}
						& \textbf{SpaRCS\cite{waters2011sparcs}} & Yes & 
						\parbox{12cm}{
							\begin{flushleft}
								\begin{equation}
									\begin{aligned}
										&\min_{\boldsymbol{L}, \S} \|\y - \mathcal{A}(\boldsymbol{L} + \S)\|_F \\
										&\text{s.t.}\ \text{rank}(\boldsymbol{L}) \leq r,\ \|\S\| \leq K
									\end{aligned}
								\end{equation}
							\end{flushleft}
						} 
						& greedy algorithm & Yes \\
						
						& \textbf{R-SpaRCS\cite{ramesh2015r}} & Yes & 
						\parbox{12cm}{
							\begin{flushleft}
								\begin{equation}
									\begin{aligned}
										&\min_{\boldsymbol{L}, \S} \|\y - \mathcal{A}(\boldsymbol{L} + \S)\|_F \\
										&\text{s.t.}\ \text{rank}(\boldsymbol{L}) \leq r,\ \|\S\| \leq K
									\end{aligned}
								\end{equation}
							\end{flushleft}
						}
						& Model-based greedy algorithm & No \\
						
						& \textbf{ACOS/SACOS\cite{li2015identifying}} & No & 
						\parbox{12cm}{
							\begin{flushleft}
								\begin{equation}
									\begin{aligned}
										&\min_{\boldsymbol{L}, \S} \|\boldsymbol{L}\|_* + \lambda \|\S\|_{1,2} \\
										&\text{s.t.}\ \mathbf{Y} = \boldsymbol{L} + \S
									\end{aligned}
								\end{equation}
							\end{flushleft}
						}
						& two-step optimization algorithm & Yes \\
						
						& \textbf{COCRPCA\cite{8309163}} & Yes &
						\parbox{12cm}{
							\begin{flushleft}
								\begin{equation}
									\begin{aligned}
										&\min_{\mathbf{s}_t, \mathbf{\ell}_t} \left\{ H(\mathbf{s}_t, \mathbf{\ell}_t \mid \mathbf{y}_t, \mathbf{Z}_{t-1}, \mathbf{B}_{t-1}) \right. \\
										&\left. = \frac{1}{2} \left\| \mathbf{\Phi}(\mathbf{s}_t + \mathbf{\ell}_t) - \mathbf{y}_t \right\|_2^2 \right. \\
										&\left. + \lambda \mu \sum_{j=0}^{J} \beta_j \left\| \mathbf{W}_j (\mathbf{s}_t - \mathbf{z}_j) \right\|_1 + \mu \left\| \mathbf{B}_{t-1}, \mathbf{\ell}_t \right\|_* \right\}
									\end{aligned}
								\end{equation}
							\end{flushleft}
						}
						& proximal gradient methods & Yes \\
						
						& \textbf{CODA \cite{8416578}} & Yes &
						\parbox{12cm}{
							\begin{flushleft}
								\begin{equation}
									\begin{aligned}
										&\min_{\mathbf{s}_t, \mathbf{\ell}_t} \left\{ H(\mathbf{s}_t, \mathbf{\ell}_t) = \frac{1}{2} \left\| \mathbf{\Phi}(\mathbf{s}_t + \mathbf{\ell}_t) - \mathbf{y}_t \right\|_2^2 \right. \\
										&\left. + \lambda \mu \sum_{j=0}^{J} \beta_j \sum_{c=1}^{C} \gamma_{j|c} \left\| \mathbf{W}_{j|\Omega_{j,c}} (\mathbf{s}_t|_{\Omega_{j,c}} - \mathbf{z}_{j|\Omega_{j,c}}) \right\|_1 + \mu \left\| \mathbf{B}_{t-1}, \mathbf{\ell}_t \right\|_* \right\}
									\end{aligned}
								\end{equation}
							\end{flushleft}
						}
						& proximal gradient methods & Yes \\
						
						\midrule
						\multirow{5}{*}{\makecell{\textbf{Bayesian}\\\textbf{Inference}}} 
						& \textbf{BiG-AMP \cite{parker2014bilinear}} & No &
						\parbox{12cm}{
							\begin{flushleft}
								\begin{equation}
									\begin{aligned}
										p_{\mathbf{X}, \mathbf{A} | \mathbf{Y}}(\mathbf{X}, \mathbf{A} &\mid \mathbf{Y}) \propto \left[ \prod_m \prod_l p_{y_{ml} \mid z_{ml}} \left( y_{ml} \mid \sum_k a_{mk} x_{kl} \right) \right] \\
										&\times \left[ \prod_n \prod_l p_{x_{nl}}(x_{nl}) \right] \left[ \prod_m \prod_n p_{a_{mn}}(a_{mn}) \right]
									\end{aligned}
								\end{equation}
							\end{flushleft}
						}
						& Sum-product algorithm & No \\
						
						& \textbf{VB \cite{han2017bayesian}} & No &
						\parbox{12cm}{
							\begin{flushleft}
								\begin{equation}
									\begin{aligned}
										p(\mathbf{Y} \mid \mathbf{U}, \mathbf{V}, &\mathbf{W}, \mathbf{Z}, \beta) \\
										&= \mathcal{N}(\mathbf{Y} \mid \mathbf{U} \mathbf{V}^\top + \mathbf{W} \circ \mathbf{Z}, \beta^{-1} \mathbf{I}_{mn})
									\end{aligned}
								\end{equation}
							\end{flushleft}
						}
						& Variational Bayesian approximation & No \\
						
						& \textbf{VBSE \cite{chen2013variational}} & No &
						\parbox{12cm}{
							\begin{flushleft}
								\begin{equation}
									\begin{aligned}
										p(\mathbf{Y}, \mathbf{A}, \mathbf{B}, \mathbf{E}, \gamma, \alpha, \beta) &= p(\mathbf{Y} \mid \mathbf{A}, \mathbf{B}, \mathbf{E}, \beta) p(\mathbf{A} \mid \gamma) p(\mathbf{B} \mid \gamma) \\
										&\times p(\gamma) p(\mathbf{E} \mid \alpha) p(\alpha) p(\beta)
									\end{aligned}
								\end{equation}
							\end{flushleft}
						}
						& Variational Bayesian approximation & No \\
						
						& \textbf{UAMP-MF\cite{yuan2022unitary}} & No &
						\parbox{12cm}{
							\begin{flushleft}
								\begin{equation}
									\begin{aligned}
										p(\mathbf{X}, \mathbf{H}, \lambda &\mid \mathbf{Y}) f_{\mathbf{Y}}(\mathbf{Y}, \mathbf{X}, \mathbf{H}, \lambda) f_{\mathbf{X}}(\mathbf{X}) f_{\mathbf{H}}(\mathbf{H}) f_{\lambda}(\lambda)
									\end{aligned}
								\end{equation}
							\end{flushleft}
						}
						& Variational Inference \& UAMP \cite{guo2015approximate} & No \\
						
						& \textbf{ITMP}(proposed) & Yes &
						\parbox{12cm}{
							\begin{flushleft}
								\begin{equation}
									\begin{aligned}
										p(\y, \mathbf{X}, \S, \boldsymbol{L}) &= \mathcal{N}(\y; \mathcal{A}(\mathbf{X}), \sigma_n^2 \mathbf{I}) p(\S) p(\boldsymbol{L}) \prod_{i,j} \delta(X_{i,j} - S_{i,j} - L_{i,j})
									\end{aligned}
								\end{equation}
							\end{flushleft}
						}
						& Sum-product rule \& extrinsic message approximation & Yes \\
						
						\bottomrule
					\end{tabular}%
				}	
				\caption{Comparison of RPCA/CRPCA Algorithms}
				\label{tab:optimization}
			\end{table}
		\end{appendices} 
		\bibliographystyle{IEEEtran}  
		\bibliography{ITMP}

	\end{document}